\documentclass{mscs}
\usepackage{bussproofs}

\usepackage{amsmath}
\usepackage{amssymb}
\usepackage{amsfonts}
\usepackage{graphicx}
\usepackage{enumerate}
\usepackage{color}
\usepackage[utf8]{inputenc}
\usepackage{hyperref}
\usepackage{MnSymbol}

\usepackage{microtype}

\usepackage{tikz}
\usetikzlibrary{matrix}

\newcommand\partitle[1]{\vspace{0.75ex}\mbox{}\noindent{\bf #1.}}

\newtheorem{theorem}{Theorem}[section]
\newtheorem{lemma}[theorem]{Lemma}
\newtheorem{corollary}[theorem]{Corollary}
\newtheorem{proposition}[theorem]{Proposition}
\newtheorem{conjecture}[theorem]{Conjecture}

\newtheorem{example}[theorem]{Example}
\newtheorem{definition}[theorem]{Definition}

\newtheorem{remark}[theorem]{Remark}

\newcommand{\bindcaptchain}{bin\-ding--cap\-tu\-ring chain}
\newcommand{\bindcaptchains}{bin\-ding--cap\-tu\-ring chains}

\newcommand\fig[1]{\includegraphics[scale=0.8]{figs/{{#1}}}}

\newcommand\hcentered[1]{\hfill{#1}\hfill}
\newcommand\vcentered[1]{\raisebox{-0.5\height}{#1}}
\newcommand\centered[1]{\hcentered{\vcentered{#1}}}

\newcommand\lambdacal\blambda
\newcommand\blambda{\ensuremath{\bs{\lambda}}} 

\newcommand{\inflambdacal}{\ensuremath{\bs{\lambda}^{\hspace*{-1pt}\bs{\infty}}}}
\newcommand{\lambdaletreccal}{\ensuremath{\bs{\lambda}_{\sletrec}}}
\newcommand{\lambdaletrecprefixcal}{\ensuremath{\bs{({\lambda}_{\sletrec})}}}

\newcommand{\lambdaprefixcal}{\ensuremath{\bs{(\lambda)}}}
\newcommand{\lambdaprefixposcal}{\ensuremath{\bs{(\lambda)_{\spos}}}}
\newcommand{\inflambdaprefixcal}{\ensuremath{\bs{(\inflambdacal)}}}

\newcommand{\lambdacalculus}{$\lambda$\nb-cal\-cu\-lus}
\newcommand{\lambdaletreccalculus}{\lambdacalculus\ with $\sletrec$}
\newcommand{\lambdaabstraction}{$\lambda$\nb-ab\-strac\-tion}
\newcommand{\lambdaabstractions}{\lambdaabstraction{s}}

\let\oldlambda\lambda
\renewcommand\lambda{\ensuremath\oldlambda}

\newcommand\inflambda{\ensuremath{\lambda^{\hspace*{-1pt}\bs\infty}}}
\newcommand\lambdaletrec{\ensuremath{\lambda_\sletrec}}

\newcommand{\lambdaterm}{$\lambda$\nb-term}
\newcommand{\lambdaterms}{\lambdaterm{s}}
\newcommand{\inflambdaterm}{$\lambda^{\infty}$\nb-term}
\newcommand{\inflambdaterms}{\inflambdaterm{s}}
\newcommand{\lambdaletrecterm}{$\lambda_{\sletrec}$\nb-term}
\newcommand{\lambdaletrecterms}{\lambdaletrecterm{s}}
\newcommand{\lambdaletrecexpressible}{$\lambda_{\sletrec}$\nb-ex\-press\-ible}

\newcommand{\slabs}{\lambda}
\newcommand{\labs}[2]{\slabs{#1}.\hspace*{0.5pt}{#2}}
\newcommand{\lapp}[2]{{#1}\,{#2}}
\newcommand\lappbreak[2]{\begin{array}{l}(#1)\\~~(#2)\end{array}}

\newcommand{\lvars}{{\cal X}}
\newcommand{\llrecvars}{{\cal R}}

\newcommand{\snlvar}{\ensuremath{\mathsf{0}}}
\newcommand{\nlvar}{\snlvar}
\newcommand{\snlvarsucc}{\ensuremath{\mathsf{S}}}

\newcommand{\snllabsfo}{\slabs}
\newcommand{\nllabsfo}{\funap{\slabs}}
\newcommand{\snllappfo}{@}
\newcommand{\nllappfo}{\bfunap{\snllappfo}}
\newcommand{\nlvarfo}{0}
\newcommand{\snlvarsuccfo}{\snlvarsucc}
\newcommand{\nlvarsuccfo}{\funap{\snlvarsuccfo}}

\newcommand{\bnllter}{t}

\newcommand{\bnllteri}{\indap{\bnllter}}

\newcommand{\sletrec}{\ensuremath{\textsf{letrec}}}
\newcommand{\letrec}[2]{\sletrec\hspace*{2pt}{#1}\hspace*{2pt}\textsf{in}\hspace*{2pt}{#2}}
\newcommand{\letrecin}[2]{\sletrec\hspace*{2pt}{#1}\hspace*{2pt}\textsf{in}\hspace*{2pt}{#2}}

\newcommand{\abindgroup}{B}
\newcommand{\bbindgroup}{C}
\newcommand{\cbindgroup}{D}
\newcommand{\abindgroupi}{\indap{\abindgroup}}

\newcommand{\sletCRS}[1]{\ensuremath{\textsf{let}_{#1}}}
\newcommand{\srecinCRS}[1]{\ensuremath{\textsf{rec-in}_{#1}}}
\newcommand{\letrecCRS}[3]{\funap{\sletCRS{#1}}{\absCRS{#2}{\funap{\srecinCRS{#1}}{#3}}}}
\newcommand{\letrecemptyCRS}[1]{\funap{\sletCRS{0}}{{\funap{\srecinCRS{0}}{#1}}}}
\newcommand{\letrecCRSbare}[3]{\funap{\sletCRS{#1}}{\absCRS{#2}{{\srecinCRS{#1}}{#3}}}}

\newcommand{\stxtletrec}{{\sf letrec}}
\newcommand{\stxtlambdaletreccal}{{\lambda_{\text{\normalfont letrec}}}}

\newcommand{\allter}{L}
\newcommand{\bllter}{P}
\newcommand{\cllter}{Q}
\newcommand{\allteri}{\indap{\allter}}
\newcommand{\bllteri}{\indap{\bllter}}

\newcommand{\alltertilde}{\tilde{L}}

\newcommand{\alltertildei}{\indap{\alltertilde}}

\newcommand{\arecvar}{f}
\newcommand{\brecvar}{g}
\newcommand{\crecvar}{h}
\newcommand{\arecvari}{\indap{\arecvar}}

\newcommand{\sareachrel}{\rightarrowtail}
\newcommand{\areachrel}{\mathrel{\sareachrel}}

\makeatletter
\def\a#1{\reflectbox{$\m@th#1{\lambda}$}}

\makeatother

\newcommand{\sflabs}[1]{(\slabs{#1})}
\newcommand{\flabs}[2]{\sflabs{#1}\hspace*{0.5pt}{#2}}
\newcommand{\femptylabs}[1]{()\hspace*{0.5pt}{#1}}

\newcommand{\flabspos}[4]{\bpap{(\slabs{#1})}{#2}{#3}\hspace*{0.5pt}{#4}}
\newcommand{\femptylabspos}[3]{\bpap{()}{#1}{#2}\hspace*{0.5pt}{#3}}

\newcommand{\sflabsposCRS}[2]{{\mathsf{pre}}^{#2}_{\enumsequence{#1}}}
\newcommand{\flabsposCRS}[4]{\funap{\sflabsposCRS{#2}{#3}}{\absCRS{#1}{#4}}}

\newcommand{\avar}{x}
\newcommand{\bvar}{y}
\newcommand{\cvar}{z}
\newcommand{\dvar}{u}

\newcommand{\fvar}{w}

\newcommand{\avari}[1]{\avar_{#1}}
\newcommand{\bvari}[1]{\bvar_{#1}}

\newcommand{\avarname}{\mathsf{x}}

\newcommand{\avarnamei}{\indap{\avarname}}

\newcommand{\ater}{M}
\newcommand{\bter}{N}
\newcommand{\cter}{O}
\newcommand{\ateri}[1]{\ater_{#1}}
\newcommand{\bteri}[1]{\bter_{#1}}

\newcommand{\aiter}{M}
\newcommand{\biter}{N}
\newcommand{\citer}{P}

\newcommand{\eiter}{T}
\newcommand{\aiteri}[1]{\aiter_{#1}}
\newcommand{\biteri}[1]{\biter_{#1}}

\newcommand\sacxt{C}

\newcommand\acxt[1]{\sacxt[#1]}

\newcommand{\aflpter}{M}

\newcommand{\sann}{{:}}
\newcommand{\annexp}[2]{{{#1}\hspace*{2pt}\sann\hspace*{2pt}{#2}}}

\newcommand{\flabsann}[3]{(\slabs{#1})\hspace*{3pt}\annexp{#2}{#3}}
\newcommand{\femptylabsann}[2]{()\hspace*{2pt}\annexp{#1}{#2}}

\newcommand{\equat}[2]{{{#1}={#2}}}
\newcommand{\femptylabsequat}[2]{\equat{\femptylabs{#1}}{\femptylabs{#2}}}

\newcommand{\labflabs}[3]{\annexp{#1}{\flabs{#2}{#3}}}
\newcommand{\labfemptylabs}[2]{\annexp{#1}{\femptylabs{#2}}}

\newcommand{\alabel}{l}
  
\newcommand{\aannvar}{u}

\newcommand{\sTer}{{\mit Ter}}
\newcommand{\Ter}{\funap{\sTer}}
\newcommand{\siTer}{{\sTer^{\infty}}}
\newcommand{\iTer}{\funap{\siTer}}
\newcommand{\Terbot}{\funap{\sTer_\bot}}

\newcommand{\sTerlab}{\widehat{\textit{Ter}}}
\newcommand{\Terlab}{\funap{\sTerlab}}

\newcommand{\asig}{\Sigma}

\newcommand{\sametavar}{X}

\newcommand{\scmetavar}{Z}
\newcommand{\sametavari}{\indap{X}}
\newcommand{\sbmetavari}{\indap{Y}}
\newcommand{\scmetavari}{\indap{Z}}

\newcommand{\cmetavar}{\funap{\scmetavar}}
\newcommand{\ametavari}[1]{\funap{\sametavari{#1}}}
\newcommand{\bmetavari}[1]{\funap{\sbmetavari{#1}}}
\newcommand{\cmetavari}[1]{\funap{\scmetavari{#1}}}

\newcommand{\aconst}{c}
\newcommand{\aconstname}{\mathsf{c}}

\newcommand{\aconstnamei}{\subap{\aconstname}}

\newcommand{\sbt}{\,\begin{picture}(-1,1)(-1,-2)\circle*{1.5}\end{picture}\ }
\newcommand{\scirc}{\,\begin{picture}(-1,1)(-1,-2)\circle{3.5}\end{picture}\ }

\newcommand\TL{{}}
\newcommand\TR{{\sbt}}
\newcommand\BL{{\scirc}}
\newcommand\BR{{\odot}}

\newcommand{\absCRS}[2]{[{#1}]\hspace*{1pt}{#2}}
\newcommand{\fabsCRS}[2]{\{{#1}\}\hspace*{1pt}{#2}}

\newcommand{\CRS}{\ensuremath{\text{CRS}}}
\newcommand{\CRSs}{\ensuremath{\text{CRSs}}}
\newcommand{\iCRS}{\ensuremath{\text{iCRS}}}

\newcommand{\arulename}{\rho}

\newcommand{\brulename}{\sigma}

\newcommand{\rules}{R}

\newcommand{\aCRS}{{\cal C}}

\newcommand{\acxthole}{\Box}

\newcommand{\subst}[3]{{#1}[{#2}\defdby{#3}]}

\newcommand{\slabsCRS}{\mathsf{abs}}
\newcommand{\labsCRS}[2]{\funap{\slabsCRS}{\absCRS{#1}{#2}}}
\newcommand{\slappCRS}{\mathsf{app}}
\newcommand{\lappCRS}[2]{\funap{\slappCRS}{#1,#2}}

\newcommand{\sflabsCRS}[1]{\mathsf{pre}_{#1}}
\newcommand{\flabsCRS}[3]{\funap{\sflabsCRS{#1}}{\absCRS{#2}{#3}}}
\newcommand{\femptylabsCRS}[1]{\funap{\sflabsCRS{0}}{#1}}

\newcommand{\llunfCRS}{\textit{\textbf R}_{\hspace*{-0.5pt}\bs{\sunf}}}

\newcommand{\llunfARS}{\textit{R}_{\hspace*{-0.5pt}\sunf}}

\newcommand{\llunfredCRS}{R_{\subrule\sunf{\text{red}}}}

\newcommand\sRegCRS{\textit{\textbf{Reg}}}

\newcommand{\RegCRS}{\sRegCRS}
\newcommand{\stRegCRS}{\sRegCRS^{\bs{+}}}
\newcommand{\RegzeroCRS}{\sRegCRS^{\bs{-}}}

\newcommand{\RegposCRS}{\sRegCRS_{\textit{\textbf{pos}}}}
\newcommand{\stRegposCRS}{\sRegCRS_{\textit{\textbf{pos}}}^{\bs{+}}}
\newcommand{\RegposzeroCRS}{\sRegCRS_{\textit{\textbf{pos}}}^{\bs{-}}}

\newcommand{\RegARS}{\textit{Reg}}
\newcommand{\stRegARS}{\textit{Reg}^{+}}
\newcommand{\RegzeroARS}{\textit{Reg}^{-}}

\newcommand{\stRegARSlab}{\widehat{\stRegARS}} 

\newcommand{\RegposARS}{\textit{Reg}_{\textit{pos}}}
\newcommand{\stRegposARS}{\textit{Reg}_{\textit{pos}}^+}
\newcommand{\RegposzeroARS}{\textit{Reg}_{\textit{pos}}^-}

\newcommand{\stParseCRS}{\ensuremath{\textit{\textbf{Parse}}^{\bs{+}}}}
\newcommand{\stParseUnfCRS}{\ensuremath{\textit{\textbf{Parse}}_{\sunf}^{\bs{+}}}}
\newcommand{\ParseCRS}{\ensuremath{\textit{\textbf{Parse}}}}

\newcommand{\RegletrecCRS}{\sRegCRS_{\sletrec}}
\newcommand{\stRegletrecCRS}{\sRegCRS^{+}_{\sletrec}}

\newcommand{\RegletrecARS}{\textit{Reg\/}_{\sletrec}}
\newcommand{\stRegletrecARS}{\textit{Reg\/}^{+}_{\sletrec}}

\newcommand{\siglcCRS}{\asig_{\lambda}}
\newcommand{\sigllcCRS}{\asig_{\lambda_{\sletrec}}}
\newcommand{\siglpcCRS}{\asig_{(\lambda)}}
\newcommand{\siglpposcCRS}{\asig_{(\lambda)_{\spos}}}
\newcommand{\sigllpcCRS}{\asig_{(\lambda_\sletrec)}}

\newcommand{\siglpcparseCRS}{\ensuremath{\asig_{(\lambda),\sparse}}}

\newcommand{\ARS}{ARS}
\newcommand{\ARSs}{ARSs}

\newcommand{\InducedSubARS}[1]{\ensuremath{(#1\,\smred)}}
\newcommand{\InducedSubARSmred}[2]{\ensuremath{(#1\hspace*{0.5pt}#2)}}

\newcommand{\objects}{A}
\newcommand{\steps}{\Phi}
\newcommand{\ssrc}{\textsf{src}}
\newcommand{\stgt}{\textsf{tgt}}

\newcommand{\aARS}{{\cal A}}
\newcommand{\aobj}{a}
\newcommand{\bobj}{a'}
\newcommand{\astep}{\phi}

\newcommand\wDepth[2]{{#1}_{#2}}
\newcommand{\objectsi}{\indap{A}}
\newcommand{\stepsi}{\indap{\steps}}

\newcommand{\stepslabon}[1]{\indap{\stepslab}{\text{on-}{#1}}}
\newcommand{\stepslabnoton}[1]{\indap{\stepslab}{\text{not-on-}{#1}}}

\newcommand{\stepslab}{\widehat{\steps}}

\newcommand{\ssrclab}{\widehat{\ssrc}}
\newcommand{\stgtlab}{\widehat{\stgt}}

\newcommand{\ssrci}[1]{\textsf{src}_{#1}}
\newcommand{\stgti}[1]{\textsf{tgt}_{#1}}
\newcommand{\srci}[1]{\funap{\ssrci{#1}}}
\newcommand{\tgti}[1]{\funap{\stgti{#1}}}
\newcommand{\aobji}{\indap{a}}
\newcommand{\bobji}[1]{a'_{#1}}
\newcommand{\astepi}{\indap{\astep}}

\newcommand{\soutgoing}{\text{out}}
\newcommand{\sincoming}{\text{in}}

\newcommand{\soutgoingst}[1]{{#1}_{\soutgoing}}
\newcommand{\sincomingst}[1]{{#1}_{\sincoming}}
\newcommand{\outgoingst}[1]{\funap{\soutgoingst{#1}}}
\newcommand{\incomingst}[1]{\funap{\sincomingst{#1}}}

\newcommand{\aARSi}{\indap{{\cal A}}}

\newcommand{\alabelling}{\mathbb{L}}
\newcommand{\aARSbisim}{\mathbb{B}}

\newcommand{\arewlabelling}{\mathfrak{L}}
\newcommand{\ainitlabelling}{\mathfrak{l}}

\newcommand{\aocc}{o}

\newcommand{\apath}{\pi}

\newcommand{\sgST}{{\mit ST}}
\newcommand{\sgSTstrat}[1]{\indap{\sgST}{#1}}
\newcommand{\gSTstrat}[1]{\funap{\sgSTstrat{#1}}}

\newcommand{\sgSTreg}{{\mit ST}}
\newcommand{\sgSTregstrat}[1]{\indap{\sgSTreg}{#1}}
\newcommand{\gSTregstrat}[1]{\funap{\sgSTregstrat{#1}}}

\newcommand{\sgSTstregstrat}[1]{\bpap{\sgST}{#1}{+}}
\newcommand{\gSTstregstrat}[1]{\funap{\sgSTstregstrat{#1}}}

\newcommand{\sgSTregpos}{{\mit ST}^{\textit{pos}}}
\newcommand{\sgSTregposstrat}[1]{\indap{\sgSTregpos}{#1}}
\newcommand{\gSTregposstrat}[1]{\funap{\sgSTregposstrat{#1}}}

\newcommand{\sgSTstregpos}{\sgST^{+,\textit{pos}}}
\newcommand{\sgSTstregposstrat}[1]{\indap{\sgSTstregpos}{#1}}
\newcommand{\gSTstregposstrat}[1]{\funap{\sgSTstregposstrat{#1}}}

\newcommand{\sdroppos}{\textit{drop}^{\text{pos}}}
\newcommand{\droppos}{\funap{\sdroppos}}

\newcommand{\arewseq}{\tau}
\newcommand{\brewseq}{\xi}
\newcommand{\crewseq}{\pi}

\newcommand{\Checkreg}{\Check}
\newcommand{\Checkstreg}{\check}

\newcommand{\sred}{{\to}}
\newcommand{\red}{\mathrel{\to}}
\newcommand{\seqred}{\to^{=}}

\newcommand{\smred}{{\twoheadrightarrow}}
\newcommand{\mred}{\mathrel{\smred}}
\newcommand{\smorestepred}{\supap{\sred}{+}}

\newcommand\sconvred{{\leftarrow}}

\newcommand{\sconvmred}{{\twoheadleftarrow}}

\newcommand{\snfred}{\to^{\scriptstyle !}} 

\newcommand{\sredp}[1]{{\subap{\to}{#1}}}

\newcommand{\smredp}[1]{{\twoheadrightarrow^{#1}}}

\newcommand{\sredb}[1]{{\subap{\sred}{#1}}}
\newcommand{\redb}[1]{\mathrel{\sredb{#1}}}
\newcommand{\smredb}[1]{{\smred_{#1}}}

\newcommand{\sconvmredb}[1]{{\sconvmred_{#1}}}

\newcommand{\sparredp}[1]{{\subap{\hspace{1ex}||\hspace{-2.8ex}\longrightarrow}{#1}}}

\newcommand{\sreg}{\text{reg}}
\newcommand{\sstreg}{\text{reg}^+}
\newcommand{\sregzero}{\text{reg}^-}

\newcommand{\srred}{\subap{\sred}{r}}
\newcommand{\sregred}{\subap{\sred}{\reg}}
\newcommand{\sstregred}{\subap{\sred}{\streg}}
\newcommand{\sregzerored}{\subap{\sred}{\sregzero}}

\newcommand{\sregmred}{\subap{\smred}{\reg}}
\newcommand{\sstregmred}{\subap{\smred}{\streg}}
\newcommand{\sregzeromred}{\subap{\smred}{\sregzero}}

\newcommand{\sregzeroeqred}{\bpap{\sred}{\sregzero}{=}}
\newcommand{\rred}{\mathrel{\srred}}  
\newcommand{\regred}{\mathrel{\sregred}}  
\newcommand{\stregred}{\mathrel{\sstregred}}

\newcommand{\regmred}{\mathrel{\sregmred}}  
\newcommand{\stregmred}{\mathrel{\sstregmred}} 
\newcommand{\regzeromred}{\mathrel{\sregzeromred}}

\newcommand{\reg}{\ensuremath{\text{\normalfont reg}}}
\newcommand{\streg}{\ensuremath{\text{\normalfont reg}^+}}

\newcommand\unfold\bigtriangledown
\newcommand{\sunfoldred}{\indap{\sred}{\unfold}}
\newcommand{\unfoldred}{\mathrel{\sunfoldred}}
\newcommand{\sunfoldmred}{\indap{\smred}{\unfold}}

\newcommand{\sunfolddepthred}[1]{{\sred}^{#1}_{\unfold}}

\newcommand{\sunfoldconvred}{\indap{\sconvred}{\unfold}}

\newcommand{\sunfolddepthconvred}[1]{\sconvred^{#1}_{\unfold}}

\newcommand{\sunfoldinfred}{\indap{\sinfred}{\unfold}}
\newcommand{\unfoldinfred}{\mathrel{\sunfoldinfred}}
\newcommand{\sunfoldconvinfred}{\indap{\sconvinfred}{\unfold}}
\newcommand{\unfoldconvinfred}{\mathrel{\sunfoldconvinfred}}

\newcommand{\sunfoldomegared}{\indap{\somegared}{\unfold}}
\newcommand{\unfoldomegared}{\mathrel{\sunfoldomegared}}
\newcommand{\sunfoldconvomegared}{\indap{\sconvomegared}{\unfold}}
\newcommand{\unfoldconvomegared}{\mathrel{\sunfoldconvomegared}}

\newcommand{\sunfoldomeganfred}{\indap{\someganfred}{\unfold}}

\newcommand{\sunfoldconvomeganfred}{\indap{\sconvomeganfred}{\unfold}}

\newcommand{\sunfoldbotred}{\indap{\sred}{\bot}}

\newcommand{\sunfoldbotomegared}{\indap{\somegared}{\bot}}
\newcommand{\unfoldbotomegared}{\mathrel{\sunfoldbotomegared}}

\newcommand{\sunf}{\bigtriangledown\hspace{-0.3ex}}
\newcommand{\sunflabs}{\sunf.\slabs}
\newcommand{\sunflapp}{\sunf.@}

\newcommand{\sunflabsred}{\subap{\sred}{\sunflabs}}
\newcommand{\unflabsred}{\mathrel{\sunflabsred}}
\newcommand{\sunflappred}{\subap{\sred}\sunflapp}

\newcommand{\sunflabsstratred}[1]{\subap{\sred}{\sunflabs.{#1}}}

\newcommand{\sunflappstratred}[1]{\subap{\sred}{\sunflapp.{#1}}}

\newcommand\snil{\text{nil}}
\newcommand{\sunfnil}{\sunf.\snil}
\newcommand{\sunfnilred}{\subap{\sred}{\sunfnil}}

\newcommand{\sunfnilstratred}[1]{\subap{\sred}{\sunfnil.{#1}}}

\newcommand{\unfemptyred}{\mathrel{\unfemptyred}}

\newcommand\sfree{\text{free}}

\newcommand{\unffreered}{\mathrel{\unffreered}}

\newcommand\rec{\text{rec}}
\newcommand\srecsubst{\text{rec}}
\newcommand{\sunfrec}{\sunf.\srecsubst}
\newcommand{\sunfrecred}{\subap{\sred}{\sunfrec}}
\newcommand{\unfrecred}{\mathrel{\sunfrecred}}
\newcommand{\sunfrecstratred}[1]{\subap{\sred}{\sunfrec.{#1}}}

\newcommand\smergeletrec{\text{letrec}}
\newcommand{\sunfletrec}{\sunf.\text{\smergeletrec}}
\newcommand{\sunfletrecred}{\subap{\sred}{\sunfletrec}}
\newcommand{\unfletrecred}{\mathrel{\sunfletrecred}}
\newcommand{\sunfletrecstratred}[1]{\subap{\sred}{\sunfletrec.{#1}}}
\newcommand{\unfletrecstratred}[1]{\mathrel{\unfletrecstratred{#1}}}

\newcommand{\sreduce}{\text{red}}
\newcommand{\sunfreduce}{\sunf.\sreduce}
\newcommand{\sunfreducered}{\subap{\sred}{\sunfreduce}}
\newcommand{\unfreducered}{\mathrel{\sunfreducered}}
\newcommand{\sunfreducestratred}[1]{\subap{\sred}{\sunfreduce.{#1}}}
\newcommand{\unfreducestratred}[1]{\mathrel{\unfreducestratred{#1}}}

\newcommand{\sUnf}{{\cal U}}  
\newcommand{\Unf}{\funap{\sUnf}}

\newcommand{\spUnf}{{\cal U'}}  

\newcommand{\scompressregnf}{{\downarrow_{\scompressreg}}}
\newcommand{\compressregnf}[1]{{#1}\scompressregnf}
\newcommand{\scompressregnfred}{{\snfred_{\scompressreg}}}
\newcommand{\compressregnfred}{\mathrel{\scompressregnfred}}

\newcommand{\scompressstregnf}{{\downarrow_{\scompressstreg}}}
\newcommand{\compressstregnf}[1]{{#1}\scompressstregnf}
\newcommand{\scompressstregnfred}{{\snfred_{\scompressstreg}}}
\newcommand{\compressstregnfred}{\mathrel{\scompressstregnfred}}

\newcommand\thsp{-1.785ex}
\newcommand\threeheadrightarrow{\twoheadrightarrow\hspace*\thsp\twoheadrightarrow}
\newcommand\threeheadleftarrow{\twoheadleftarrow\hspace*\thsp\twoheadleftarrow}

\newcommand{\sinfred}{\threeheadrightarrow}

\newcommand{\infredb}[1]{\mathrel{\sinfred_{#1}}}
\newcommand{\sconvinfred}{\threeheadleftarrow}

\newcommand{\somegared}{\sred^{\omega}}

\newcommand{\sconvomegared}{\sconvred^{\omega}}

\newcommand{\someganfred}{\sred^{!\,\omega}}

\newcommand{\sconvomeganfred}{\sconvred^{!\,\omega}}

\newcommand{\astrat}{\mathbb{S}}
\newcommand{\astratplus}{\mathbb{S}^+}
\newcommand{\astrati}{\indap{\astrat}}

\newcommand{\sstratred}{\indap{\sred}}
\newcommand{\stratred}[1]{\mathrel{\sstratred{#1}}}
\newcommand{\sstratmred}{\indap{\smred}}
\newcommand{\stratmred}[1]{\mathrel{\sstratmred{#1}}}

\newcommand{\slabsdecomp}{\slabs}
\newcommand{\slabsdecompred}{{\subap{\sred}{\slabsdecomp}}}
\newcommand{\labsdecompred}{\mathrel{\slabsdecompred}}
\newcommand{\slabsdecompstratred}[1]{{\subap{\sred}{\slabsdecomp,{#1}}}}
\newcommand{\labsdecompstratred}[1]{\mathrel{\slabsdecompstratred{#1}}}

\newcommand{\slappdecomp}{@}

\newcommand{\slappdecompi}{\subap{@}}
\newcommand{\slappdecompired}[1]{{\subap{\sred}{\slappdecompi{#1}}}}
\newcommand{\lappdecompired}[1]{\mathrel{\slappdecompired{#1}}}

\newcommand{\slappdecompistratred}[2]{{\subap{\sred}{\slappdecompi{#1},{#2}}}}
\newcommand{\lappdecompistratred}[2]{\mathrel{\slappdecompistratred{#1}{#2}}}

\newcommand{\slappdecompiregred}[1]{\pbap{\sred}{\slappdecompi{#1}}{\reg}}

\newcommand{\scompress}{\text{\normalfont del}}
\newcommand{\scompressreg}{\scompress}
\newcommand{\scompressregred}{{\subap{\sred}{\scompress}}}
\newcommand{\compressregred}{\mathrel{\scompressregred}}
\newcommand{\scompressregstratred}[1]{{\subap{\sred}{\scompress,{#1}}}}

\newcommand{\scompressregconvred}{{\subap{\sconvred}{\scompress}}}

\newcommand{\scompressregconveqred}{{\bpap{\sconvred}{\scompress}{=}}}

\newcommand{\scompressregmred}{{\subap{\smred}{\scompress}}}
\newcommand{\compressregmred}{\mathrel{\scompressregmred}}

\newcommand{\scompressregconvmred}{{\subap{\sconvmred}{\scompress}}}

\newcommand{\scompressstreg}{\snlvarsucc} 
\newcommand{\scompressstregred}{{\subap{\sred}{\snlvarsucc}}}
\newcommand{\compressstregred}{\mathrel{\scompressstregred}}
\newcommand{\scompressstregstratred}[1]{{\subap{\sred}{\snlvarsucc,{#1}}}}

\newcommand{\scompressstregconvred}{{\subap{\sconvred}{\snlvarsucc}}}

\newcommand{\scompressstregeqred}{{\bpap{\sred}{\snlvarsucc}{=}}}

\newcommand{\scompressstregmred}{{\subap{\smred}{\snlvarsucc}}}
\newcommand{\compressstregmred}{\mathrel{\scompressstregmred}}

\newcommand{\constred}{\mathrel{\constred}}

\newcommand{\srule}{\varrho}

\newcommand{\srulep}{\supap{\srule}}
\newcommand{\srulebp}{\bpap{\srule}}

\newcommand\srulepos{\varrho_{pos}}
\newcommand\rulepos[1]{\srulepos^{#1}}

\newcommand\subrule[2]{{#1}.{#2}}

\newcommand{\lTG}{$\lambda$\nb-transition graph}

\newcommand{\agraph}{G}
\newcommand{\vertices}{V}
\newcommand{\edges}{E}
\newcommand{\avert}{v}
\newcommand{\bvert}{w}

\newcommand{\strans}{\rightarrowtail}

\newcommand{\scope}{scope}
\newcommand{\extscope}{\scope$^+$}
\newcommand{\extscopes}{\scope{s}$^+$}
\newcommand{\scopedelimiting}{\scope\nb-de\-li\-mi\-ting}
\newcommand{\Scopedelimiting}{Scope\nb-de\-li\-mi\-ting}
\newcommand{\extscopedelimiting}{\extscope-delimiting}
\newcommand{\eager}{\ensuremath{\text{\normalfont eag}}}
\newcommand{\lazy}{\ensuremath{\text{\normalfont lazy}}}

\newcommand{\ascdelstratreg}{\astrat}
\newcommand{\ascdelstratstreg}{\astrat^+}

\newcommand{\eagscdelstrat}{\subap{\astrat}{\eager}}
\newcommand{\lazyscdelstrat}{\subap{\astrat}{\lazy}}

\newcommand{\eagscdelstratreg}{\eagscdelstrat}              
\newcommand{\eagscdelstratstreg}{\bpap{\astrat}{\eager}{+}} 
\newcommand{\lazyscdelstratreg}{\lazyscdelstrat}            
\newcommand{\lazyscdelstratstreg}{\bpap{\astrat}{\lazy}{+}} 

\newcommand{\ainst}{\iota}
\newcommand{\binst}{\kappa}

\newcommand{\addrule}[2]{{#1}{+}{#2}}

\newcommand\Terinflambda{\ensuremath{\normalfont\textbf{T}^\infty_{\lambdaprefixcal}}}
\newcommand\Terlambdaletrec{\ensuremath{\normalfont\textbf{T}_{\lambdaletrecprefixcal}}}

\newcommand{\Reg}{\ensuremath{\normalfont\textbf{Reg}}}
\newcommand{\stReg}{\ensuremath{\normalfont\textbf{Reg}^{\bs{+}}}}
\newcommand{\stRegzero}{\ensuremath{\textbf{Reg}_{\bs{0}}^{\bs{+}}}}

\newcommand{\annstRegzero}{\ensuremath{\normalfont\textbf{ann-Reg}_{\bs{0}}^{\bs{+}}}}

\newcommand{\stRegeq}{\ensuremath{\normalfont\textbf{Reg}^{\bs{+}}_{\bs{=}}}}

\newcommand{\stReglab}{\ensuremath{\normalfont\widehat{\textbf{Reg}^{\bs{+}}}}}

\newcommand{\Regletrec}{\ensuremath{\normalfont\textbf{Reg}_{\sletrec}}}
\newcommand{\stRegletrec}{\ensuremath{\normalfont\textbf{Reg}^{\bs{+}}_{\sletrec}}}
\newcommand{\ugRegletrec}{\ensuremath{\normalfont\textbf{ug-Reg}_{\sletrec}}}
\newcommand{\ugstRegletrec}{\ensuremath{\normalfont\textbf{ug-Reg}^{\bs{+}}_{\sletrec}}}

\newcommand{\annstRegletrec}{\ensuremath{\normalfont\textbf{ann-Reg}^{\bs{+}}_{\sletrec}}}

\newcommand{\Vacreg}{\ensuremath{\text{\normalfont del}}} 
\newcommand{\Vacstreg}{\ensuremath{\snlvarsucc}} 

\newcommand{\annVacstreg}{\ensuremath{\snlvarsucc}} 

\newcommand\AlphaInfPreterm[1]{\ensuremath{\funap{{\normalfont\it A}^{\bs{\infty}}}{#1}}}

\newcommand{\Alpha}{\ensuremath{{\normalfont\bf A}^{\bs{\hspace{-1.5px}\infty}}}}
\newcommand{\Alphasig}{\funap{\Alpha}}
\newcommand{\EqTer}{\ensuremath{{\normalfont\bf EQ}^{\bs{\infty}}}}
\newcommand{\AlphaPreTer}{\ensuremath{{EQ}_{\alpha}^{\infty}}}

\newcommand{\salphaequiv}{{\equiv_{\alpha}}}
\newcommand{\alphaequiv}{\mathrel{\salphaequiv}}

\newcommand{\sFIX}{\text{FIX}}
\newcommand{\FIX}{(\ensuremath\sFIX)}
\newcommand{\sFIXletrec}{\ensuremath{\sFIX_{\sletrec}}}
\newcommand{\FIXletrec}{(\sFIXletrec)}
\newcommand{\sFIXletrecmin}{\ensuremath{\sFIX^{-}_{\sletrec}}}
\newcommand{\FIXletrecmin}{(\sFIXletrecmin)}

\newcommand{\labscomp}{\slabs}
\newcommand{\lappcomp}{@}
\newcommand{\bvarax}{\nlvar}

\newcommand{\sderivable}{\vdash}
\newcommand{\derivablein}[2]{\sderivable_{#1}\hspace*{1.5pt}{#2}}
\newcommand{\sinfderivable}{\mbox{}^\infty\hspace*{-3pt}\vdash}
\newcommand{\sinfderivablein}[1]{\sinfderivable_{#1}}
\newcommand{\infderivablein}[2]{\sinfderivablein{#1}\hspace*{1.5pt}{#2}}

\newcommand{\sDeriv}{{\cal D}}
\newcommand{\Deriv}{\sDeriv}
\newcommand{\Derivi}[1]{\sDeriv_{#1}}
\newcommand{\Derivlab}{\supap{\Deriv}{(\text{\normalfont lb})}}
\newcommand{\Derivann}{\hat{\Deriv}}

\newcommand{\infDeriv}{\Deriv^{\infty}}

\newcommand{\Derivanni}{\subap{\hat{\Deriv}}}

\newcommand{\depth}[1]{\left|{#1}\right|}

\newcommand{\amarker}{u}
\newcommand{\bmarker}{v}
\newcommand{\cmarker}{w}
\newcommand{\amarkeri}{\indap{\amarker}}

\newcommand{\funin}{\mathrel{:}}
\newcommand{\funap}[2]{{#1}({#2})}
\newcommand{\bfunap}[3]{{#1}({#2},\hspace*{0.5pt}{#3})}
\newcommand{\indap}[2]{#1_{#2}}

\newcommand{\sdefdby}{{:=}}
\newcommand{\defdby}{\mathrel{\sdefdby}}
\newcommand{\length}[1]{\left|{#1}\right|}
\newcommand{\nb}{\nobreakdash}
\newcommand{\nbd}{\nb-}
\newcommand{\bs}{\boldsymbol}
\newcommand\dom[1]{\funap{\textit{dom}}{#1}}
\newcommand{\ran}[1]{\funap{\textit{ran}}{#1}}

\newcommand\alignbreak{\displaybreak[0]\\}

\newcommand{\srestrictto}[2]{{#1}\!\mid_{#2}}
\newcommand{\restrictto}[2]{\funap{\srestrictto}}

\newcommand{\converse}[1]{{#1}^{\smile}}

\newcommand{\niks}{}

\newcommand{\sidon}[1]{\text{\normalfont id}_{#1}}

\newcommand{\slogand}{{\wedge}}
\newcommand{\logand}{\mathrel{\slogand}}
\newcommand{\slogor}{{\vee}}
\newcommand{\logor}{\mathrel{\slogor}}

\newcommand{\ssbinrelcomp}{\cdot}
\newcommand{\sbinrelcomp}[2]{{#1}\mathrel{\ssbinrelcomp}{#2}}
\newcommand{\binrelcomp}[2]{\mathrel{\sbinrelcomp{#1}{#2}}}

\newcommand{\subap}[2]{#1_{#2}}
\newcommand{\supap}[2]{#1^{#2}}
\newcommand{\bpap}[3]{{#1}_{#2}^{#3}}
\newcommand{\pbap}[3]{{#1}^{#2}_{#3}}

\newcommand{\slb}{\mathit{lb}}
\newcommand{\lb}{\funap{\slb}}
\newcommand{\spl}{\mathit{p}}
\newcommand{\pl}{\funap{\spl}}
\newcommand{\sst}{\mathit{st}}
\newcommand{\st}{\funap{\sst}}

\newcommand{\descsetexpmid}{\mathrel{\vert}}
\newcommand{\descsetexpbigmid}{\mathrel{\big\vert}}
\newcommand{\descsetexpBigmid}{\mathrel{\Big\vert}}

\newcommand{\descsetexp}[2]{\left\{{#1}\descsetexpmid{#2}\right\}}
\newcommand{\descsetexpnormalsize}[2]{\{{#1}\descsetexpmid{#2}\}}
\newcommand{\descsetexpbig}[2]{\bigl\{{#1}\descsetexpbigmid{#2}\bigr\}}
\newcommand{\descsetexpBig}[2]{\Bigl\{{#1}\descsetexpBigmid{#2}\Bigr\}}

\newcommand{\setexp}[1]{\left\{{#1}\right\}}

\newcommand{\spowersetof}{\powerset}
\newcommand{\powersetof}{\funap{\spowersetof}}

\newcommand{\aset}{A}
\newcommand{\bset}{B}
\newcommand{\cset}{C}

\newcommand{\arel}{R}
\newcommand{\brel}{S}

\newcommand{\sequence}[2]{\{{#1}\}_{#2}}

\newcommand{\enumsequence}[1]{\langle{#1}\rangle}

\newcommand{\redsuccs}[2]{{({#2}\,{#1}})}

\newcommand{\punc}[1]{\ensuremath{\hspace*{1pt}{#1}}}

\newcommand{\todo}[1]{}

\newcommand{\pair}[2]{\langle{#1},\hspace*{0.5pt}{#2}\rangle}
\newcommand{\triple}[3]{\langle{#1},\hspace*{0.5pt}{#2},\hspace*{0.5pt}{#3}\rangle}
\newcommand{\quadruple}[4]{\langle{#1},\hspace*{0.5pt}{#2},\hspace*{0.5pt}{#3},\hspace*{0.5pt}{#4}\rangle}
\newcommand{\tuple}[1]{\langle{#1}\rangle}

\newcommand{\nats}{\mathbb{N}}
\newcommand{\posnats}{\mathbb{N}^+}

\newcommand{\vecsub}[2]{\vec{#1}_{#2}}

\newcommand\nmvec[3]{{#1}_{#2} \dots {#1}_{#3}}
\newcommand\nvec[2]{\nmvec{#1}{1}{#2}}
\newcommand\vecOneToN[1]{\nvec{#1}{n}}

\newcommand\sep[1]{&#1&} 

\newcommand{\noopsort}[1]{}

\newcommand{\defd}[1]{{#1}{\downarrow}}
\newcommand{\undefd}[1]{{#1}{\uparrow}}

\newcommand\alphaeq{\ensuremath{=_\alpha}}

\newcommand{\sbisim}[1][]{%
    \setbox0=\hbox{\kern-.1ex{$\leftrightarrow$}\kern-.1ex}
    \setbox1=\vbox{\hbox{\raise .1ex \box0}\hrule}%
    \ensuremath{\mathrel{\hbox{\kern.1ex\box1\kern.1ex}_{#1}}}
  }
\newcommand{\bisim}{\mathrel{\sbisim}}

\newcommand\sbisimstep\rightsquigarrow
\newcommand\sconvbisimstep\leftsquigarrow

\newcommand{\abisim}{R}

\newcommand{\LTG}{LTG}

\newcommand{\slts}{{\cal L}}
\newcommand{\alts}{\lts{\aARS}}
\newcommand{\lts}{\indap{\slts}}
\newcommand{\SilentLTS}[2]{\indap{\slts}{{#1},{#2}}}

\newcommand{\sltg}{{\cal G}}
\newcommand\altg\sltg

\newcommand\ltg[2]{\sltg_{#1}(#2)}
\newcommand{\SilentLTG}[3]{\indap{\sltg}{{#1},{#2}}(#3)}

\newcommand{\lambdatg}{$\lambda$\nb-transition graph}

\newcommand{\LTS}{LTS}
\newcommand{\aLTS}{{\cal L}}

\newcommand{\aLTSi}{\indap{\aLTS}}

\newcommand{\aLTG}{G}
\newcommand{\aLTGi}{\indap{\aLTG}}

\newcommand{\states}{S}
\newcommand{\labels}{A}
\newcommand{\transitions}{\strans}
\newcommand{\statesi}{\indap{S}}

\newcommand{\transitionsi}{\subap{\sred}}

\newcommand{\ainitialstate}{x}
\newcommand{\ainitialstatei}{\indap{\ainitialstate}}

\newcommand{\slabtrans}[1]{\overset{#1}{\to}}
\newcommand{\labtrans}[3]{{#1}\mathrel{\slabtrans{#2}}{#3}}
\newcommand{\slabtransi}[2]{\overset{#2}{\to}_{#1}}
\newcommand{\labtransi}[4]{{#2}\mathrel{\slabtransi{#1}{#3}}{#4}}

\newcommand{\astate}{s}
\newcommand{\bstate}{t}
\newcommand{\cstate}{u}
\newcommand{\astatei}{\indap{\astate}}

\newcommand{\alab}{a}

\definecolor{azure}{rgb}{0.94,1.00,1.00}
\definecolor{blue}{rgb}{0,0,0.5}
\definecolor{brown}{rgb}{.75,.25,.25}
\definecolor{cyan}{rgb}{0.25,0.88,0.82}
\definecolor{chocolate}{rgb}{0.82,0.41,0.12}
\definecolor{darkcyan}{rgb}{0.5,0,1}
\definecolor{darkgreen}{rgb}{0,0.39,0}
\definecolor{darkmagenta}{rgb}{0.5,0,0.5}
\definecolor{firebrick}{RGB}{175,25,25}
\definecolor{forestgreen}{rgb}{0.13,0.55,0.13}
\definecolor{lightcyan}{rgb}{0.88,1.00,1.00}
\definecolor{lightpink}{rgb}{1.00,0.71,0.76}
\definecolor{lightyellow}{rgb}{1.00,1.00,0.88}
\definecolor{lightgoldenrod}{rgb}{0.83,0.97,0.51}
\definecolor{lightgoldenrodyellow}{rgb}{0.98,0.98,0.82}
\definecolor{lightskyblue}{rgb}{0.53,0.81,0.98}
\definecolor{moccasin}{rgb}{1.00,0.89,0.71}
\definecolor{magenta}{rgb}{1,0,1}
\definecolor{navyblue}{rgb}{0,0,0.5}
\definecolor{orange}{rgb}{1.0,0.65,0.0}
\definecolor{orangered}{rgb}{1.0,0.27,0.0}
\definecolor{palegreen}{rgb}{0.60,0.98,0.60}
\definecolor{powderblue}{rgb}{0.69,0.88,0.90}
\definecolor{purple}{rgb}{1,0.5,1}
\definecolor{royalblue}{RGB}{65,105,225}
\definecolor{mediumblue}{RGB}{0,0,205}
\definecolor{cornflowerblue}{RGB}{100,149,237}
\definecolor{springgreen}{rgb}{0.0,1.0,0.5}
\definecolor{turquoise}{rgb}{0.25,0.88,0.82}
\definecolor{snow}{rgb}{1.00,0.98,0.98}
\definecolor{tan}{rgb}{0.82,0.71,0.55}
\definecolor{red}{rgb}{1,0,0}

\newcommand{\sbinds}{{\leftspoon}}
\newcommand{\binds}{\mathrel{\sbinds}}
\newcommand{\sbindseq}{{\leftspoon^{=}}}
\newcommand{\bindseq}{\mathrel{\sbindseq}}
\newcommand{\scaptures}{{\dashleftarrow}}
\newcommand{\captures}{\mathrel{\scaptures}}

\newcommand{\siscapturedby}{{\dashrightarrow}}
\newcommand{\iscapturedby}{\mathrel{\siscapturedby}}

\newcommand{\spos}{\text{\normalfont pos}}
\newcommand{\sPositions}{{\mit Pos}}
\newcommand{\Positions}{\funap{\sPositions}}

\newcommand{\positions}{\nats^*}
\newcommand{\vecpositions}{\vec{\nats}^*}

\newcommand{\sscopeof}{\indap{\mathit{scope}}}
\newcommand{\scopeof}[1]{\funap{\sscopeof{#1}}}
\newcommand{\sextscopeof}{\indap{\mathit{scope}^+}}
\newcommand{\extscopeof}[1]{\funap{\sextscopeof{#1}}}

\newcommand{\sscopered}{\indap{\sred}{\mathit{sc}}}
\newcommand{\scopered}{\mathrel{\sscopered}}

\newcommand{\sscopemorestepred}{\indap{\smorestepred}{\mathit{sc}}}
\newcommand{\scopemorestepred}{\mathrel{\sscopemorestepred}}

\newcommand{\rootpos}{\epsilon}

\newcommand{\apos}{p}
\newcommand{\bpos}{q}
\newcommand{\cpos}{r}
\newcommand{\dpos}{s}

\newcommand{\aposi}{\indap{\apos}}
\newcommand{\bposi}{\indap{\bpos}}
\newcommand{\cposi}{\indap{\cpos}}

\newcommand{\apreter}{s}
\newcommand{\bpreter}{t}
\newcommand{\aipreter}{s}
\newcommand{\bipreter}{t}

\newcommand{\apreteri}{\indap{\apreter}}
\newcommand{\bpreteri}{\indap{\bpreter}}

\newcommand{\safun}{f}

\newcommand{\afun}{\funap{\safun}}

\newcommand{\sreadback}{\textsf{readback}}
\newcommand{\readback}{\funap{\sreadback}}

\newcommand{\sreadwriten}{\indap{\textsf{rw}}}
\newcommand{\readwriten}[1]{\bfunap{\sreadwriten{#1}}}
\newcommand{\sreadwritezero}{\indap{\textsf{rw}}{0}}
\newcommand{\readwritezero}{\funap{\sreadwritezero}}

\newcommand{\sparse}{\ensuremath{\mathsf{parse}^+}}
\newcommand{\sparsen}{\ensuremath{\indap{\sparse}}}
\newcommand{\parsen}[1]{\bfunap{\sparsen{#1}}}
\newcommand{\sparseempty}{\indap{\sparse}{0}}
\newcommand{\parseempty}{\funap{\sparseempty}}

\newcommand{\subparse}{\text{parse}^+}
\newcommand{\sparsered}{\indap{\sred}{\subparse}}

\newcommand{\sparsemred}{\indap{\smred}{\subparse}}
\newcommand{\parsemred}{\mathrel{\sparsemred}}
\newcommand{\sparseinfred}{\indap{\sinfred}{\subparse}}
\newcommand{\parseinfred}{\mathrel{\sparseinfred}}

\newcommand{\sunfoldparsered}{\subap{\sred}{\subrule\sunf\subparse}}

\newcommand{\sunfoldparseomegared}{\indap{\somegared}{\subrule\sunf\subparse}}
\newcommand{\unfoldparseomegared}{\mathrel{\sunfoldparseomegared}}

\newcommand{\sparselabsred}{\indap{\sred}{\subrule\subparse\slabs}}
\newcommand{\sparselappred}{\indap{\sred}{\subrule\subparse@}}
\newcommand{\sparsecompressstregred}{\indap{\sred}{\subrule\subparse\scompressstreg}}
\newcommand{\sparsenlvarred}{\indap{\sred}{\subrule\subparse\nlvar}}

\title[Expressibility 
                      in the Lambda Calculus with $\protect\sletrec$]%
      {Expressibility 
                      in the Lambda Calculus\\ with $\protect\sletrec$
       \thanks{This work was supported by \emph{NWO~(Nederlandse Organisatie voor Wetenschappelijk Onderzoek)}
         in the framework of the project \emph{Realising Optimal Sharing (ROS)} under the direction of Doaitse Swierstra and Vincent von Oostrom.}}

\author[Clemens Grabmayer and Jan Rochel]{
\begin{tabular}{ll}
  Clemens Grabmayer & Jan Rochel\\
  Department of Philosophy&
  Department of Information and Computing Sciences\\
  Utrecht University, The Netherlands&
  Utrecht University, The Netherlands
\end{tabular}
}
  
\begin{document}

\maketitle

\begin{abstract}
  We investigate the relationship between finite terms in $\lambdaletreccal$,
the lambda calculus with $\sletrec$, and the infinite lambda terms they express.
As there are easy examples of infinite \lambdaterms\ that, intuitively, are not
unfoldings of terms in $\lambdaletreccal$, we consider the question: 
How can those infinite lambda terms be characterised that are $\lambdaletreccal$\nb-expressible 
in the sense that they can be obtained as infinite unfoldings of terms in $\lambdaletreccal$?

\vspace{0.4ex}
For `observing' infinite \lambdaterms\ through repeated `experiments' carried out at the head of the term 
we introduce two rewrite systems (with rewrite relations) $\sregred$ and $\sstregred$ that decompose the term structure,
and produce `generated subterms' in two notions.
Thereby the sort of the step can be observed as well as its target, a generated subterm. 
In both systems there are four sorts of decomposition steps:
$\slabsdecompred$\nb-steps (decomposing a $\lambda$\nb-abstraction),
$\slappdecompired{0}$- and $\slappdecompired{1}$\nb-steps (decomposing an application into its function and argument),
and respectively, $\scompressregred$-steps (delimiting the scope of an abstraction, for $\sregred$),
and $\scompressstregred$ (delimiting of scopes, for $\sstregred$). 
These steps take place on infinite \lambdaterms\ furnished with a leading prefix of abstractions
for gathering previously encountered $\lambda$\nb-abstractions and keeping the generated subterms closed.
We call an infinite \lambdaterm\ `regular'\discretionary{/}{}{/}`strongly regular'
if its set of $\sregred$-reachable/$\sstregred$-reachable generated subterms
is finite.  
Furthermore, we analyse the binding structure of infinite \lambdaterms\
with the concept of `\bindcaptchain'.


\vspace{0.4ex}
Utilising these concepts, we answer the question above by providing two characterisations
of $\lambdaletreccal$\nb-expressibility. 
For all infinite \lambdaterms~$\aiter$, the following statements are equivalent:
(i):~$\aiter$~is $\lambdaletreccal$\nb-expressible;
(ii):~$\aiter$~is strongly regular;
(iii):~$\aiter$~is regular, and it only has finite \bindcaptchains.

\end{abstract}

\section{Introduction}
  \label{sec:intro}

A prevalent enrichment of the \lambda-calculus is the extension by a $\mu$
or $\sletrec$ bin\-ding-con\-struct, the latter being a generalisation of the
first. Such constructs facilitate finite representations for infinite
\lambdaterms, and for typed \lambda-calculi (and thus for most functional
programming languages) are crucial for enabling definitions by unbounded recursion 
and thereby for ensuring Turing completeness. 
A term in the \lambda{}-calculus with \sletrec{}, which we denote
by \lambdaletreccal, is usually understood as a finite representation
of its denotation in the `unfolding semantics', 
the infinite \lambdaterm\ that is obtained by completely unfolding
all occurring recursive definitions (all \sletrec-bindings).  
The unfolding semantics is practically relevant since
compiler builders frequently argue (at least in an initial phase) for the correctness of program transformations 
defined on $\stxtlambdaletreccal$\nb-ex\-press\-ions only intuitively, 
via their denotation in the unfolding semantics.%
  \footnote{%
    Note that while optimising program transformations used for compiling a functional programs
    typically change the unfolding semantics of a program, 
    they are intended to preserve its `behavioural semantics', and that of the denoted infinite \lambdaterm.}
    
It turns out, however, that not every infinite \lambdaterm\ that exhibits a regular structure
(more precisely, that has a regular syntax tree) 
can be expressed as the infinite unfolding of a term in \lambdaletreccal.
For an example with an informal explanation, see Example~\ref{ex:entangled} below.
We say that an infinite \lambdaterm\ is
`\lambdaletrecexpressible' if it has a representation as \emph{finite} term in \lambdaletreccal.
Note that infinite \lambdaterms\ typically have infinitely many different representations as
$\stxtlambdaletreccal$\nb-terms, see also Example~\ref{ex:ll-expressible} below.
Reformulating the observation above, one can say that 
the infinite \lambdaterms\ that can be expressed by terms in $\lambdaletreccal$,
and hence are the denotation of $\stxtlambdaletreccal$\nb-terms in the unfolding semantics,
form only a proper subclass of those infinite \lambdaterms\ with a regular structure. 
In this report we prove this fact, and we provide two different characterisations of the class of
\lambdaletrecexpressible\ infinite \lambdaterms, by means of the concept of `strong regularity',
and `\bindcaptchains'.


\begin{example}[not $\lambdaletrec$-expressible]\label{ex:entangled}\label{ex:not:ll-expressible}
  Consider the infinite \lambdaterm\ of the form
  $\aiter = \labs{a}{\labs{b}{\lapp{(\labs{c}{\lapp{(\labs{d}{\lapp{\dots}{c}})}{b}})}{a}}}$
  with syntax tree as shown in Figure~\ref{fig:ex:entangled} on the left. 
  While this syntax tree has a regular structure, 
  the scopes of the abstractions in it are infinitely entangled: 
  the scope of $\slabs a$ reaches into the scope of $\slabs b$, the scope of $\slabs b$ into the one of $\slabs c$,
  and so on.  
  This feature of $\aiter$ can suggest the idea that it is impossible that $\aiter$ is the result 
  of `unrolling' a \lambdaletrecterm\ step by step in a manner that respects scopes.
  Such a process would namely `tile' its result in a regular manner 
  with finite term-context tiles having a bounded scope-nesting depth 
  such that furthermore the scopes of abstractions contained in different context tiles do not overlap.
  This excludes, intuitively, the formation of the infinite entanglement of successively overlapping scopes 
  that can be observed in $\aiter$.
  The term $\aiter$ will indeed be recognised as not \lambdaletrec-expressible.
\end{example}

\begin{figure}[t]
\centering
\fig{pstricks/entangled.lgr}
\hspace{0.5cm}
\fig{cons.lgr}
\hspace{0.2cm}
\fig{conscons.lgr}
\hspace{0.2cm}
\fig{pstricks/consinf.lgr}
\caption{\label{fig:ex:entangled}\label{fig:not:ll-expressible:ll-expressible}
         Syntax tree for the infinite \protect\lambdaterm~$\aiter$ in Example~\ref{ex:entangled},
         and two term graphs, and the syntax tree for the infinite \protect\lambdaterm~$\biter$ in Example~\ref{ex:equiv-letrec-terms}.}
\end{figure}

\begin{example}[$\lambdaletrec$-expressible]\label{ex:equiv-letrec-terms}\label{ex:ll-expressible}
The infinite \lambda-term~$\biter = \labs{x}{x:x:\dots}$ can be expressed by
$\labs{x}{\letrec{r=x:r}{r}}$ as well as by
$\labs{x}{\letrec{r=x:(x:r)}{r}}$
(for the term graphs of these \lambdaletrecterms\ as well as the syntax tree of $\biter$ of their infinite unfolding,
 see Figure~\ref{fig:ex:entangled} on the right).
The colon ($:$) used here is inspired from
the functional programming language Haskell. It is a binary infix operator
constructing a list from its arguments by putting the first argument (here:
$x$) in front of its second argument. This example is an implementation of what
is known as the $repeat$ function. The \lambda-calculus we deal with here does
not explicitly include operators. Therefore in the context of this work the
colon can be viewed as a free variable or a closed \lambdaterm.

Note that, if colon is viewed as a closed \lambdaterm\ that contains \lambda-abstractions,
then, contrary to the term $\aiter$ from Example~\ref{ex:entangled}, abstraction scopes in the syntax tree of $\biter$
are organised in blocks of of bounded scope-nesting depth
(for example each occurrence of colon then is such a block), and no infinite entanglement of scopes takes place. 
\end{example}

Over a first order signature, the simplest kind of infinite terms 
are those that are regular in the sense that they possess only a finite number of different subterms.
They correspond to trees over ranked alphabets that are regular \cite{cour:1983}.
Like regular trees also regular terms can be expressed finitely by systems of recursion equations \cite{cour:1983}
or by `rational expressions' \cite[Def.\hspace*{1pt}4.5.3]{cour:1983},
which correspond to $\mu$\nb-terms (see e.g.\ \cite{endr:grab:klop:oost:2011}).
Hereby finite expressions denote infinite terms either
via a mathematical definition (a fixed-point construction, or induction on paths) 
or as the limit of a rewrite sequence consisting of unfolding steps. 

For infinite higher-order terms such as infinite \lambdaterms\
the concept of regularity is less clear-cut from the outset,
due to the presence of variable binding. 
Frequently, regularity has been used with as meaning
the existence of a first-order syntax tree with named variables that is regular
(e.g.\ in \cite{ario:klop:1997,ario:blom:1997}).
For example, the infinite \lambdaterm\ $\biter$ that is represented by the syntax tree
in Figure~\ref{fig:not:ll-expressible:ll-expressible} on the right is regular in this sense.
However, such a definition of regularity has the drawback that 
it depends on a property of a first-order representation (as syntax trees, or
correspondingly as pseudo- or pre-terms) that is not invariant under
$\alpha$\nb-con\-ver\-sion, the re\-na\-ming of bound variables.
Note that the syntax tree left in Figure~\ref{fig:not:ll-expressible:ll-expressible} of the infinite \lambdaterm~$\aiter$ 
contains infinitely many variables, and for this reason is not regular as a first-order tree.
Still, the infinite \lambdaterm~$\aiter$ viewed as higher-order term (and hence up to $\alpha$\nb-con\-ver\-sion)
that is described by this syntax tree would be considered to be regular according to the 
understanding mentioned above. This is because there is an $\alpha$\nb-variant of the syntax tree left in Figure~\ref{fig:not:ll-expressible:ll-expressible}
that uses only two variable names, which appear alternatingly.  
It is therefore desirable to obtain a definition that generalises the condition for the first-order case
by adapting the notion of subterm for \lambdaterms, 
and that pertains directly to a formulation of infinite \lambdaterms\ as higher-order terms. 

Viable notions of subterm for \lambdaterms\ in a higher-order formalisation require a stipulation 
on how to treat variable binding 
when descending into the body of a \lambdaabstraction.
For this purpose we enrich the syntax of \lambdaterms\ with a bracketed prefix of
abstractions (similar to a proof system for weak $\mu$\nb-equality in
\cite[Fig.\hspace*{1.5pt}12]{endr:grab:klop:oost:2011}).
An expression $\flabs{\avari{1}\ldots\avari{n}}{\aiter}$ 
represents a partially decomposed \lambdaterm: 
the body $\aiter$ typically contains free occurrences of variables 
which in the original \lambdaterm\ were bound by \lambdaabstractions\
that have been split off by decomposition steps.
The role of such abstractions has then been taken over by abstractions in the prefix $\sflabs{\avari{1}\ldots\avari{n}}$.
In this way expressions with abstraction prefixes can be kept closed under decomposition steps. 

On these prefixed \lambdaterms, we define two closely related rewrite systems $\RegCRS$ and $\stRegCRS$.
Rewrite sequences in $\RegCRS$ and in $\stRegCRS$ deconstruct infinite \lambdaterms\ by steps that typically decompose
applications and \lambdaabstractions, and take place just below the bracketed abstractions.
$\RegCRS$ and $\stRegCRS$ differ with respect to the steps for removing vacuous prefix bindings they facilitate:
while such bindings can always be removed by pertinent steps in $\RegCRS$,
the system $\stRegCRS$ only enables steps that remove vacuous bindings at the end of the abstraction prefix. 
For each of these systems we consider a family of strategies 
that make deterministic choices concerning the application of the steps for removing vacuous prefix bindings:
we call these strategies `\scopedelimiting\ strategies' for $\RegCRS$, and `\extscopedelimiting\ strategies' for $\stRegCRS$. 
\Scopedelimiting\ strategies $\astrat$ for $\RegCRS$ and \extscopedelimiting\ strategies $\astratplus$ for $\stRegCRS$ 
induce rewrite relations $\sstratred{\astrat}$ and $\sstratred{\astratplus}$, respectively. 
These families of rewrite strategies 
define respective notions of `generated subterm', 
and they give rise to differently strong concepts of regularity:
a \lambdaterm\ $\citer$ is called regular (strongly regular)
if there is a rewrite strategy $\astrat$ for $\RegCRS$ 
            (a rewrite strategy $\astratplus$ for $\stRegCRS$)
such that the set of from $\citer$
$\sstratred{\astrat}$\nb-reach\-able ($\sstratred{\astratplus}$\nb-reach\-able) generated subterms is finite.

\begin{figure}[tbp]
\setlength\tabcolsep{-0.5mm}
\begin{center}
\begin{tabular}{c
                 c}
\vcentered{\fig{simpleletrec-unf}}
& \hspace*{17.5ex}
  \vcentered{\fig{simpleletrec-ltg}}
\\[1ex]
{syntax tree} 
              & 
                \hspace*{18.5ex}
                \parbox[t]{30ex}{\centering
                                   $\sstratred{\eagscdelstratstreg}$-generated subterms
                                   }
\end{tabular}
\end{center}
\caption{\label{fig:ll:expressible}%
         Syntax tree 
         and $\sstratred{\eagscdelstratstreg}$\nb-ge\-ne\-ra\-ted subterms for
         the infinite \protect\lambdaterm~$\eiter$
         that is strongly regular, and expressible by the $\stxtlambdaletreccal$-term
         $\letrecin{\arecvar = \labs{\avar\bvar}{\lapp{\lapp{\arecvar}{\bvar}}{\avar}}}{\arecvar}$.}
\end{figure}

\begin{example}[$\stRegCRS$-decomposition]
To illustrate decomposition with respect to the `eager \scopedelimiting' strategy $\eagscdelstratstreg$ for $\stRegCRS$,
let us consider the infinite \lambdaterm~$\eiter$ with syntax tree in Figure~\ref{fig:ll:expressible}.
This term can be represented by the equation
$\eiter = \labs{\avar\bvar}{\lapp{\lapp{\eiter}{\bvar}}{\avar}}$
or by
$\letrecin{\arecvar = \labs{\avar\bvar}{\lapp{\lapp{\arecvar}{\bvar}}{\avar}}}{\arecvar}$
(in fact $\eiter$ is the infinite unfolding of this $\stxtlambdaletreccal$\nb-term).
Using the former description, decomposition by $\sstratred{\eagscdelstratstreg}$\nb-steps proceeds as follows,
repetitively:
\vspace{-0.5ex}
  \begin{equation*}\label{eq:decomposition:ll:expressible}  
    \femptylabs{\eiter} ~~~
    \flabs{\avar}{\labs{\bvar}{\lapp{\lapp{\eiter}{\bvar}}{\avar}}} ~~~
    \flabs{\avar\bvar}{\lapp{\lapp{\eiter}{\bvar}}{\avar}} ~~~
    \begin{array}{lll}
    \flabs{\avar\bvar}{{{\lapp{\eiter}{\bvar}}}} &
    \begin{array}{llll}
    \flabs{\avar\bvar}{\eiter} ~~~ &
    \flabs{\avar}{\eiter}  ~~~ &
    \femptylabs{\eiter} ~~~ &
    \ldots
    \\
    \flabs{\avar\bvar}{\bvar}
    \end{array}
    \\[2ex]
    \flabs{\avar\bvar}{\avar} & \;\;
    \flabs{\avar}{\avar}
    \end{array}
  \end{equation*}
For a rendering of this reduction graph that also records the sort of the applied decomposition steps, 
see Figure~\ref{fig:ll:expressible}. 
Note that removal steps for vacuous bindings take place here only at the end of the prefix,
and are applied eagerly, that is, always as soon as possible.
Since $\femptylabs{\eiter}$ has only 9 different $\sstratred{\eagscdelstratstreg}$\nb-reducts,
the term $\eiter$ is strongly regular. 
\end{example}

The generalisations of the concept of regularity to infinite \lambdaterms\ suggest the question:
do the expressibility results in \cite{cour:1983} for regular first-order trees
with respect to systems of recursion equations, rational expressions, or $\mu$\nb-terms
also generalise in an appropriate way?
We tackle only the case of strong regularity here, 
and obtain an expressibility result with respect to the $\lambdaletreccal$, 
the \lambdacalculus\ with $\stxtletrec$.
We say that a term $\allter$ in $\lambdaletreccal$ expresses an infinite \lambdaterm~$\aiter$
if $\aiter$ is the infinite unfolding of $\allter$.
We show that an infinite unfolding is unique if it exists, and it can be obtained 
as the limit of an infinite rewrite sequence of unfolding steps. 
We prove that an infinite \lambdaterm\ is \lambdaletrecexpressible, that is, expressible by a term in $\lambdaletreccal$,
if and only if it is strongly regular.
This  result settles a conjecture%
  \footnote{Confer the last sentence of Sec.~1.2.4 in \cite{blom:2001}:
    `We conjecture that the set of regular lambda-trees is precisely the set
     of lambda-trees that can be obtained as the unwinding of terms with $\stxtletrec$'.
    Mind that `regular lambda-trees' there correspond to strongly regular \lambdaterms\ in our sense,
    and that the notion of `sub-tree' of a `lambda-tree' informally defined there corresponds 
    to our notion of $\sstratred{\astratplus}$\nb-ge\-ne\-ra\-ted subterm with respect to a \extscopedelimiting\ strategy $\astratplus$ for $\stRegCRS$.}  
by Blom in \cite
                {blom:2001}. 

\vspace*{2ex}
\noindent
{\bf Overview of the paper.} 
Section~\ref{sec:prelims} is concerned with terminology and notation used for
known formalisms. In Section~\ref{sec:letrec} a rewrite system
for unfolding terms in the \lambdacalculus\ with $\sletrec$ is formulated.  
In Section~\ref{sec:regular} we introduce rewriting systems for
decomposing infinite \lambdaterms\ into `generated subterms', and
we show some properties of these systems in connection to so-called scope-delimiting strategies. 
Also in this section, we define regularity and strong regularity for infinite \lambdaterms\
employing the concepts of generated subterms and \scopedelimiting\ strategies.
In Section~\ref{sec:letrec2} we adapt the rewrite systems for decomposing
infinite \lambdaterms\ and the notions of scope-delimiting strategies to the \lambdacalculus\ with \sletrec.
In Section~\ref{sec:proofs}, we develop proof systems that are sound and complete
for the notions of regularity and strong regularity,
for equality of strongly regular infinite \lambdaterms,
and for the property of a \lambdaletrec\nb-term to unfold to an infinite \lambdaterm.
In Section~\ref{sec:chains} we examine the binding
structure of infinite \lambda-terms (\bindcaptchains) and connect to
the concepts introduced so far. 
In Section~\ref{sec:express} we establish 
the correspondence between strong regularity and \lambdaletrec-expressibility
for infinite \lambdaterms.
In Section~\ref{sec:ltgs} we introduce `\lTG{s}' of infinite \lambdaterms\ 
and of \lambdaletrecterms\ as
labelled transition graphs in which the edges carry one of the four different labels
$\slappdecompi{0}$, $\slappdecompi{1}$, $\slabsdecomp$, and $\scompressstreg$.
In Section~\ref{sec:conclusion} we summarise and provide an outlook
on possible results that are related, and on potential applications
of the concepts we introduce. 

\vspace*{2ex}
\noindent
{\bf Note.} 
  The results in this paper concern the relationship
  between infinite \lambdaterms\ and $\stxtlambdaletreccal$\nb-terms
  that is induced by the finite unfolding rewrite operation on $\stxtlambdaletreccal$\nb-terms. 
  These results and the analysis we use to prove them are `static' with respect to $\beta$\nb-reduction,
  that is, $\beta$\nb-reduction rewrite steps play no part in them at all.
  For this reason we will not even define $\beta$\nb-reduction for infinite \lambdaterms\ and
  for $\stxtlambdaletreccal$\nb-terms.
  Note that in particular
  the confluence result for the unfolding rewrite relation on $\stxtlambdaletreccal$\nb-terms
  (see Proposition~\ref{prop:unf-confluence} below)
  does not conflict
  with the fact that some rewrite calculi for cyclic \lambda-calculi such as $\lambdaletreccal$ 
  are non-confluent with respect to $\beta$\nb-reduction
  (see e.g.\ \cite{ario:klop:1997}).

\section{Preliminaries}
  \label{sec:prelims}

In this section we gather most of the basic known concepts that play a vital part in the rest of our paper.
Some central notions concerning rewriting are recapitulated from \cite{terese:2003}, while for
others we provide references. Some variations of known concepts that are tailor-made
for our purposes are formulated in definition environments. 

We let $\nats = \setexp{0,1,2,\ldots}$ and $\posnats = \nats\setminus{0}$.
For a partial function $\safun \funin A \rightharpoonup B$, and $a\in A$ 
we denote by $\defd{\afun{a}}$ the fact that $\safun$ is defined for $a$,
and by $\undefd{\afun{a}}$ that $\safun$ is not defined for $a$.
The \emph{domain} of $\safun$ is the set $\dom{\safun} \defdby \descsetexp{a\in A}{\defd{\afun{a}}} \subseteq A$,
and the \emph{range} of $\safun$ the set $\ran{\safun} \defdby \descsetexp{\afun{a}}{a\in A,\, \defd{\afun{a}}} \subseteq B$.

For relations $\arel\subseteq\aset\times\bset$ and $\brel\subseteq\bset\times\cset$
we denote by $\sbinrelcomp{\arel}{\brel}$ the \emph{composition of $\arel$ with $\brel$}
defined by 
$\sbinrelcomp{\arel}{\brel} \defdby \descsetexp{\pair{x}{z}}{(\exists y\in\bset) \pair{x}{y}\in\arel \logand \pair{y}{z}\in\brel }
 \subseteq \aset\times\cset$,
and by $\arel^*$ the reflexive and transitive closure of $\arel$ under composition,
which is defined by 
$\arel^* \defdby \bigcup_{i\in\nats} \arel^i$
where $\arel^0 \defdby \sidon{\aset} \defdby \descsetexp{\pair{x}{x}}{x\in\aset}$
and, for all $i\in\nats$, $\arel^{i+1} \defdby \sbinrelcomp{\arel}{\arel^i}$.
 
\partitle{Abstract Rewriting Systems} 
An \emph{abstract rewriting system (\ARS)} is a quadruple $\quadruple{\objects}{\steps}{\ssrc}{\stgt}$
consisting of a set $\objects$ of \emph{objects}, a set $\steps$ of \emph{steps}, 
and $\ssrc, \stgt \funin \steps\to\objects$, the \emph{source} and \emph{target} functions. 
We will always assume that $\objects\cap\steps = \emptyset$. 
For objects $\aobj\in\objects$ we denote by $\outgoingst{\steps}{\aobj}$ and by $\incomingst{\steps}{\aobj}$
the set of steps in $\steps$ that depart (are outgoing steps) from $\aobj$, and that arrive (are incoming steps) at $\aobj$, respectively.
We say that an \ARS\ is \emph{finite} if its set of steps is finite.
For \ARS{s} $\aARSi{i} = \tuple{\objectsi{i},\stepsi{i},\ssrci{i},\stgti{i}}$ for $i\in\setexp{1,2}$
we say that $\aARSi{1}$ is a \emph{sub\nb-\ARS} of $\aARSi{2}$
if $\objectsi{1} \subseteq \objectsi{2}$, $\stepsi{1} \subseteq \stepsi{2}$, and $\ssrci{1}$, $\stgti{1}$ are the restrictions
of $\ssrci{2}$ and $\stgti{2}$, respectively, to $\stepsi{1}$, which are required to be total functions (this implies that,
for all $\astep\in\stepsi{1}$, it holds that $\srci{2}{\astep} = \srci{1}{\astep} \in \objectsi{1}$,
and $\tgti{2}{\astep} = \tgti{1}{\astep} \in \objectsi{1}$.

\partitle{Induced sub-\ARS}
For an object $\aobj\in\objects$ of an \ARS\ $\aARS =
\tuple{\objects,\steps,\ssrc,\stgt}$ we denote by $\InducedSubARS{\aobj} \defdby
\tuple{\objects',\steps',\ssrc',\stgt'}$ the \emph{sub-\ARS{} of $\aARS$
induced by $\aobj$}, where $\objects'$ comprises only the objects from
$\objects$ that are reachable from $\aobj$ by an arbitrary number of steps (or
no steps) and with $\steps'$, $\ssrc'$, $\stgt'$ being the restrictions of
$\steps$, $\ssrc$, $\stgt$ to the objects in $\objects'$ and to steps between objects in $\objects'$.

\partitle{Bisimulations between \ARSs} 
Let $\aARSi{i} = \quadruple{\objectsi{i}}{\stepsi{i}}{\ssrci{i}}{\stgti{i}}$ for $i\in\setexp{1,2}$ be \ARSs.
A relation $\aARSbisim\subseteq (\objectsi{1}\times\objectsi{2}) \cup (\stepsi{1}\times\stepsi{2})$,
which relates objects with objects and steps with steps, is called an \emph{\ARS-bisimulation}
if:
\begin{itemize}
  \item if $\aobji{1} \mathrel{\aARSbisim} \aobji{2}$, then $\aARSbisim$ relates each step
    from $\aobji{1}$ to some step from $\aobji{2}$ (\emph{forth} condition),
    and each step from $\aobji{2}$ to some step from $\aobji{1}$ (\emph{back} condition);
  \item if $\astepi{1} \mathrel{\aARSbisim} \astepi{2}$ with
    $\astepi{1} \funin \aobji{1} \to \bobji{1}$ and $\astepi{2} \funin \aobji{2} \to \bobji{2}$,
    then $\aobji{1} \mathrel{\aARSbisim} \aobji{2}$ and $\bobji{1} \mathrel{\aARSbisim} \bobji{2}$.
\end{itemize}                                    
                                 
\partitle{Labellings of \ARSs} 
  Let $\aARS = \tuple{\objects,\steps,\ssrc,\stgt}$ and $\aARS' = \tuple{\objects',\steps',\ssrc',\stgt'}$ be \ARSs.                              
  \begin{enumerate}[(i)]
    \item 
      An \ARS-bisimulation~$\alabelling$ between $\aARS$ and $\aARS'$ is called a \emph{labelling of $\aARS$ to $\aARS'$},
      and $\aARS'$ \emph{the $\alabelling$\nb-labelled version of}~$\aARS$, 
      if the converse $\converse{\alabelling}$ of $\alabelling$ is a function $\converse{\alabelling} \funin \objects'\cup\steps' \to \objects\cup\steps$,
      and if additionally, for all $\aobj'\in\objects'$ and $\aobj\in\objects$ with $\aobj \mathrel{\alabelling} \aobj'$,
      the restriction 
      $\srestrictto{\converse{\alabelling}}{\funap{\soutgoingst{\steps}'}{\aobj'}} \funin 
       \funap{\soutgoingst{\steps}'}{\aobj'}
          \to
        \funap{\soutgoingst{\steps}}{\aobj}$ 
      of $\converse{\alabelling}$ to the steps departing from $a'$   
      is bijective. 
    \item     
      A \emph{rewrite labelling}~$\arewlabelling$ of $\aARS$ to $\aARS'$ 
      is a pair~$\pair{\alabelling}{\ainitlabelling}$ consisting of a labelling~$\alabelling$ of $\aARS$ to $\aARS'$
      together with an \emph{initial} labelling function~$\ainitlabelling$ mapping objects of $\aARS$ to bisimilar objects of $\aARS'$.
  \end{enumerate}

\partitle{Strategies}
A \emph{history-free strategy} for an abstract rewriting system $\aARS$ is a sub-\ARS\ of $\aARS$
that has the same objects, and the same normal forms as $\aARS$.
A \emph{history-aware strategy} for an abstract rewriting system~$\aARS$ 
is a history-free strategy for the $\alabelling$-labelled version of $\aARS$
with respect to, and together with, a rewrite labelling $\pair{\alabelling}{\ainitlabelling}$ of
$\aARS$. By a \emph{strategy} for $\aARS$ we will mean a history-free strategy
or a history-aware strategy for $\aARS$. 

\partitle{Remark}
Let $\astrat$ be a history-aware strategy of $\aARS$, and let $\aARS'$ be that
$\alabelling$-labelled version of $\aARS$ which $\astrat$ is a history-free
strategy of. Then $\astrat$ projects to a history-free strategy
$\check\astrat$ of $\aARS$. The projection is defined by $\alabelling$, which
induces a local bijective correspondence of outgoing steps of related sources
of $\aARS$ and $\aARS'$. Mind that for deterministic $\astrat$, $\check\astrat$
may become non-deterministic.
Furthermore, every rewrite sequence according to $\astrat$ in $\aARS'$
projects to a unique rewrite sequence in $\aARS$ (which is a rewrite sequence according to $\check{\astrat}$).

\vspace{0.3ex}
The last mentioned fact makes it possible to speak, for a given rewrite labelling,
of rewrite sequences of a history-aware strategy on the objects of the original \ARS. 
%
%
Let $\astrat$ be a history-aware strategy for an \ARS~$\aARS$, and $\aobj$ an object of $\aARS$.
Suppose that $\astrat$ is a sub\nb-\ARS\ of the $\alabelling$\nb-labelled version $\aARS'$ of $\aARS$
for some rewrite labelling $\pair{\alabelling}{\ainitlabelling}$ of $\aARS$. 
Then by a \emph{rewrite sequence of $\astrat$ on\/ $\aobj$} (in $\aARS$) we will
mean the projection to $\aARS$ of a rewrite sequence of $\astrat$ (in $\aARS'$) on the result $\funap{\ainitlabelling}{\aobj}$
of the initial labelling applied to $\aobj$. 


\vspace{1.25ex}
\partitle{Rewrite relations: notation and properties}
  For the single-step rewrite relation induced by an \ARS\ we use the arrow symbol $\sred$
  possibly subscripted with appropriate names.
  Let $\sred \subseteq \objects\times\objects$ be a rewrite relation. 
  We denote by $\smred$ the \emph{many-step} rewrite relation induced by $\sred$,
    by which we mean the reflexive and transitive closure of $\sred$.
  By $\smorestepred$ we denote the \emph{one-or-more-step} rewrite relation of $\sred$,
  the transitive closure of $\sred$.
  By $\seqred$ we mean the \emph{zero-or-one-step} rewrite relation of $\sred$,
  the reflexive closure of $\sred$.
  By a normal form of $\sred$ we mean an $\aobj\in\objects$ such that there is no $\bobj\in\objects$ with $\aobj \red \bobj$.
  By $\snfred$ we mean the \emph{reduction to normal form} rewrite relation induced by $\sred$,
  the restriction of $\smred$ to a relation with the normal forms of $\sred$ as codomain:
  ${\snfred} = \descsetexp{\pair{\aobj}{\bobj}}{\aobj \mred \bobj,\, \text{$\bobj$ is normal form of $\sred$}}$.

\vspace*{0.75ex}
The rewriting properties below are reformulations, and some are slight variations,
of known properties of rewrite relations.
   
\begin{definition}  
  Let $\sredb{1}$, $\sredb{2}$, $\sredb{3}$ be rewrite relations.
  The rewrite relation $\sredb{1}$ is called \emph{cofinal for} $\sredb{2}$ 
    if $\smredb{2} \subseteq {\binrelcomp{\smredb{1}}{\sconvmredb{2}}}$.
  We say that 
  $\sredb{1}$ is \emph{cofinal for $\sredb{2}$ with trailing $\sredb{3}$\nb-steps}
    if $\smredb{2} \subseteq {\binrelcomp{\smredb{1}}{\sconvmredb{3}}}$.
  Furthermore we say that 
  $\sredb{1}$ \emph{factors into} $\sredb{2}$ and $\sredb{3}$ 
    if $\sredb{1} \subseteq {\binrelcomp{\sredb{2}}{\sredb{3}}}$.
\end{definition}    

\vspace{1ex}
\noindent We will several times use the following specific version of K\H{o}nig's Lemma.

\partitle{\bf K\H{o}nig's Lemma}
Let $\agraph = \pair{\vertices}{\edges}$ be an undirected graph with set $\vertices$ of vertices and set $\edges$ of edges. 
Suppose that $\agraph$ has infinitely many vertices ($\vertices$ is infinite),
that it is connected (for all vertices $\avert,\bvert\in\vertices$ there exists a path in $\agraph$ from $\avert$ to $\bvert$)
and that every vertex has finite degree (it is adjacent to only finitely many other vertices in $\agraph$).
Then for every vertex $\avert\in\vertices$,
$\agraph$ contains an infinitely long simple path from $\avert$, that is, a path starting at $\avert$ without repetition of vertices.%
    \footnote{\label{footnote:Koenigs:Lemma}%
              This formulation corresponds to the following original formulation
              by D\'{e}nes K\H{o}nig on page~80 in \cite{koen:2001}:
              ``\emph{Satz 3: Jeder unendliche zusammenh\"{a}ngende Graph~$G$ endlichen Grades
                  besitzt einen einseitig unendlichen Weg, wobei der Anfangspunkt $P_0$ dieses
                  Weges beliebig vorgeschrieben werden kann.}''
              This in connection with the definition on page~10 in \cite{koen:2001}:
              ``Eine unendliche Menge von Kanten $P_i P_{i+1}$ ($i=0,1,2,\ldots$ \emph{in inf.}),
                bzw.\ der durch sie gebildete Graph, hei\ss{}t ein \emph{einseitig unendlicher Weg},
                falls f\"{u}r $i\neq j$ stets $P_i\neq P_j$ ist.''} 

\partitle{Combinatory Reduction Systems}
Many of the formalisations we introduce are based on the framework of
Combinatory Reduction Systems (\CRS{s}) \cite{klop:1980},
\cite{klop:oost:raam:1993} \cite[Sec.~11.3]{terese:2003}, and, in particular,
on infinitary Combinatory Reduction Systems (\iCRS{s}) \cite{kete:simo:2011}.

\partitle{Infinite \CRS{s}}
When speaking of `infinite terms' for \CRS{s} over some signature we draw on
\cite[12.4]{terese:2003} and \cite{kete:simo:2011} where meta-terms of \iCRS{s}
are defined by means of metric completion. The metric is defined on 
$\alpha$-equivalence classes of finite preterms dependent on the minimal depth 
at which two finite preterms belonging to the equivalence classes have a `conflict'.
The objects formed by the metric completion process can be represented as
equivalence classes of infinite preterms, we call them \emph{\iCRS\nb-preterms},
with respect to a notion of $\alpha$\nb-equivalence that again is based
on the notion of `conflict' (see also \cite[Def.~12.4.1]{terese:2003}).
Hereby \iCRS\nb-preterms are infinite ordered dyadic trees in which each node 
is either labelled by a variable name, and then the node does not have a successor,
or by named abstractions $\,\slabs{\avar}\,$ (with some variable name $\avar$), and then the node has a single successor node,
or by an abstraction symbol, and then the node has a right and a left successor node.

For denoting infinite preterms (and later terms) we use, as much as possible,
usual notation for dealing with finite terms. A slight exception is our use,
in certain situations, of a finite-\CRS\nb-based notation 
for infinite \lambdaterms\ that are not \lambdaletrec\nb-expressible
(e.g.\ see Example~\ref{ex:entangled-infinite-path}).

By \emph{\iCRS\nb-terms} we will mean $\alpha$\nb-equivalence classes of \iCRS\nb-preterms.
The notion of $\alpha$\nb-equivalence on \iCRS\nb-preterms based on the
absence of conflicts can be described by provability in the proof system $\AlphaInfPreterm\asig$
in Figure~\ref{fig:alpha-infpreterm:Schroer} which is a variant of a proof system due to
Schroer (see \cite{hend:oost:2003}).
 
\begin{definition}[$\alpha$\nb-equivalence for \iCRS\ preterms, Schroer-style proof system]\label{def:AlphaInfPreterm}
  The proof system~$\AlphaInfPreterm\asig$ for $\alpha$\nb-equivalence 
  on \iCRS\nb-preterms over signature~$\asig$
  consists of the axioms and rules displayed in Figure~\ref{fig:alpha-infpreterm:Schroer},
  each of which contains a rule $\safun$ for every $\safun\in\asig$. 
  Provability in $\AlphaInfPreterm{\asig}$ of an equation between preterms 
  is defined as the existence
  of a possibly infinite, \emph{completed} derivation: for example, 
  by $\infderivablein{\Alpha}{\apreter = \bpreter}$ we mean the existence
  of a possibly infinite prooftree $\infDeriv$ with conclusion $\apreter = \bpreter$
  such that maximal threads from the conclusion upwards either have length $\omega$,
  or have finite length and end at a leaf that carries an axiom $\bvarax$.  
  (We will generally use the decorated turnstyle symbol $\sinfderivable$ to indicate
   provability by completed, possibly infinite derivations.)
\end{definition}
\begin{figure}[t!]
\begin{center}
  \framebox{
\begin{minipage}{350pt}
\begin{center}
  \mbox{}
  \\[1ex]
  \mbox{
    \AxiomC{}
    \RightLabel{\text{const}}
    \UnaryInfC{$\aconst = \aconst$}
    \DisplayProof
        }
  \hspace{3.5ex}
  \mbox{
    \AxiomC{$\subst\apreter\avar\aconst =  \subst\bpreter\bvar\aconst$}
    \RightLabel{$\absCRS{\hspace*{1pt}}{\niks}$}
    \UnaryInfC{$\absCRS{\avar}{\apreter} = \absCRS{\bvar}{\bpreter}$}
    \DisplayProof
        }
  \hspace*{3.5ex}
  \mbox{
    \AxiomC{$ \apreteri{1}  = \bpreteri{1} $}
    \AxiomC{$ \ldots\ldots $}
    \AxiomC{$ \apreteri{n}  =  \bpreteri{n} $}
    \RightLabel{$\safun$}
    \TrinaryInfC{$ \afun{\apreteri{1},\ldots,\apreteri{n}}
                   =
                  \afun{\bpreteri{1},\ldots,\bpreteri{n}} $}
    \DisplayProof
    }
  \\[1ex]
  \mbox{}
\end{center}
\end{minipage}
   }
\end{center} 
  \vspace*{-1.25ex}  
  \caption{\label{fig:alpha-infpreterm:Schroer}%
  Schroer-style proof system $\AlphaInfPreterm\asig$ for $\alpha$-equivalence of \iCRS-preterms
           over signature $\asig$: for every $\safun\in\asig$ with arity $n$,
           $\AlphaInfPreterm{\asig}$ contains a rule $\safun$. 
           In instances of the rule $\absCRS{\hspace*{1pt}}{\niks}$,
           the constant $\aconst$ is chosen fresh for $\apreter$ and $\bpreter$.
           Substitution which occurs in the assumption of
           $\absCRS{\hspace*{1pt}}{\niks}$ denotes substitution by variable
           replacement on \iCRS-preterms. It needs not to be capture-avoiding
           because of the freshness of the substituant.}
\end{figure}   

However, closer to coinductive proof systems for infinite \lambdaterms\ that we develop
is the following different, but equivalent characterisation of $\alpha$\nb-equivalence for infinite \iCRS\nb-preterms,
a variant for \iCRS\nb-terms of a proof system for $\alpha$\nb-equivalence between finite \lambdaterms\
due to Kahrs (see \cite{hend:oost:2003}). 

\begin{samepage}
\begin{definition}[$\alpha$\nb-equivalence for \iCRS\ preterms, Kahrs-style proof system]\label{def:Alpha}%
  The proof system $\Alphasig{\asig}$ for $\alpha$\nb-equivalence 
  on \iCRS\nb-preterms over signature~$\asig$
  consists of the axioms and the rules in Figure~\ref{fig:alpha-infpreterms:Kahrs}
  with, for every $\safun\in\asig$, a rule $\safun$. 
  Provability of an equation between preterms in $\Alphasig{\asig}$ is defined,
  analogously as in \ref{fig:alpha-infpreterm:Schroer}, as the existence
  of a possibly infinite, completed derivation.   
\end{definition}
\end{samepage}

\begin{figure}[t!]
\begin{center}  
  \framebox{
\begin{minipage}{350pt}
\begin{center}
  \mbox{}
  \\[1.5ex]
  \mbox{ 
    \AxiomC{}
    \RightLabel{\bvarax}
    \UnaryInfC{$ \fabsCRS{\vec{\avar}\bvar}{\bvar}  = \fabsCRS{\vec{\cvar}\dvar}{\dvar}  $}
    \DisplayProof
        } 
  \hspace*{5ex}     
  \mbox{
    \AxiomC{$ \fabsCRS{\vec{\avar}}{\apreter}  =  \fabsCRS{\vec{\cvar}}{\bpreter} $}
    \RightLabel{$\Vacstreg\;\;$ \parbox[c]{105pt}{(if $\bvar$ does not occur in $\apreter$,\\[-0.5ex]
                                                  and $\fvar$ does not occur in $\bpreter$)}}
    \UnaryInfC{$ \fabsCRS{\vec{\avar}\bvar}{\apreter} = \fabsCRS{\vec{\cvar}\fvar}{\bpreter} $}
    \DisplayProof
        }     
  \\[2.5ex]    
  \mbox{
    \AxiomC{$ \fabsCRS{\vec{\avar}\bvar}{\apreter} 
                   = 
              \fabsCRS{\vec{\cvar}\dvar}{\bpreter} $}
    \RightLabel{$\absCRS{\hspace*{1pt}}{\niks}$}
    \UnaryInfC{$ \fabsCRS{\vec{\avar}}{\absCRS{\bvar}{\apreter}} 
                   = 
                 \fabsCRS{\vec{\cvar}}{\absCRS{\dvar}{\bpreter}} $}
    \DisplayProof
        }
  \hspace*{3.5ex} 
  \mbox{
    \AxiomC{$ \fabsCRS{\vec{\avar}}{\apreteri{1}}  =  \fabsCRS{\vec{\bvar}}{\bpreteri{1}} $}
    \AxiomC{$ \ldots\ldots $}
    \AxiomC{$ \fabsCRS{\vec{\avar}}{\apreteri{n}}  =  \fabsCRS{\vec{\bvar}}{\bpreteri{n}} $}
    \RightLabel{$\safun$}
    \TrinaryInfC{$ \fabsCRS{\vec{\avar}}{\afun{\apreteri{1},\ldots,\apreteri{n}}} 
                   = 
                  \fabsCRS{\vec{\bvar}}{\afun{\bpreteri{1},\ldots,\bpreteri{n}}} $}
    \DisplayProof 
    }     
  \\[1ex]
  \mbox{}  
\end{center}
\end{minipage}
   }
\end{center} 
  \vspace*{-1.25ex}  
  \caption{\label{fig:alpha-infpreterms:Kahrs}%
           Kahrs-style proof system~$\Alphasig{\asig}$ for $\alpha$\nb-equivalence on \iCRS-preterms
           over signature $\asig$: for every $\safun\in\asig$ with arity $n$,
           $\Alphasig{\asig}$ contains a rule $\safun$.
           }
\end{figure}   

This formulation of $\alpha$\nb-equivalence for infinite \lambdaterms\ will be
the key to our formulation of a `coinduction principle' for infinite \lambdaterms\
in Theorem~\ref{thm:coinduction-principle}.

\partitle{Infinite rewrite relation} 
  For an \iCRS\ with rewrite relation $\sred$, we denote by $\sinfred$
  the infinitary rewrite relation induced by strongly convergent and continuous rewrite sequences of arbitrary
  (countable) ordinal length. Hereby strong convergence means that, at every limit ordinal,
  the rewrite activity in the terms of the rewrite sequence tends to infinity.
  Continuity means that the terms of the rewrite sequence converge, in the metric space of infinite terms,
  at every limit ordinal.
  By $\somegared$ we will denote the rewrite relation induced by strongly continuous $\sred$\nb-rewrite sequences
  of length $\omega$.

\partitle{Labelled transition systems, labelled transition graphs} 
A \emph{labelled transition system (\LTS)} 
is a triple $\aLTS = \triple{\states}{\labels}{\transitions}$
consisting of a set $\states$ of \emph{states},
a set $\labels$ of \emph{labels},
and a set ${\transitions} \subseteq \states \times\labels\times\states$
of $\labels$\nb-labelled transitions.
Labelled transitions $\triple{\astatei{1}}{\alab}{\astatei{2}} \in \states\times\labels\times\states$ 
will be indicated as $\labtrans{\astatei{1}}{\alab}{\astatei{2}}$.
 
A \emph{labelled transition graph (LTG)} $\aLTG$ is a pointed \LTS, that is,
$\aLTG = \quadruple{\states}{\labels}{\ainitialstate}{\transitions}$ where
$\triple{\states}{\labels}{\transitions}$ is an \LTS,
and  $x\in\states$, which is called the \emph{initial state}.

\partitle{Bisimulation between \LTS{s}, \LTG{s}}
Let $\aLTSi{1} = \triple{\statesi{1}}{\labels}{\transitionsi{1}}$ 
and $\aLTSi{2} = \triple{\statesi{2}}{\labels}{\transitionsi{2}}$ be a \LTS{s} over a common set of labels.
A \emph{bisimulation on $\aLTS$} is a binary relation $\abisim \subseteq \statesi{1}\times\statesi{2}$
that satisfies, for all $\astate\in\statesi{1}$ and $\bstate\in\statesi{2}$:
\begin{enumerate}[(i)]
  \item if $\astate \mathrel{\abisim} \bstate$ and $\labtransi{1}{\astate}{\alab}{\astate'}$, then there exists
          $\bstate'\in\statesi{2}$ such that $\labtransi{2}{\bstate}{\alab}{\bstate'}$ and $\astate'\mathrel{\abisim}\bstate'$;
  \item if $\astate \mathrel{\abisim} \bstate$ and $\labtransi{2}{\bstate}{\alab}{\bstate'}$, then there exists
          $\astate'\in\statesi{1}$ such that $\labtransi{1}{\astate}{\alab}{\astate'}$ and $\astate'\mathrel{\abisim}\bstate'$.
\end{enumerate}
Two states $\astate\in\statesi{1}$ and $\bstate\in\statesi{2}$ are \emph{bisimilar}, denoted by $\astate \bisim \bstate$,
if there exists a bisimulation $\abisim$ such that $\astate \mathrel{\abisim} \bstate$.

Two \LTG{s} $\aLTGi{1} = \triple{\statesi{1},\labels}{\ainitialstatei{1}}{\transitionsi{1}}$ 
and $\aLTGi{1} = \triple{\statesi{2},\labels}{\ainitialstatei{2}}{\transitionsi{2}}$
are \emph{bisimilar} if there is a bisimulation on the underlying \LTS{s} that relates
the initial state $\ainitialstatei{1}$ of $\aLTGi{1}$ with the initial state $\ainitialstatei{2}$ of $\aLTGi{2}$. 

%



\section{The \lambda-calculus and the \lambdaletrec-calculus}
  \label{sec:letrec}
This section provides the definitions for terms in the \lambda-calculus and the
\lambdaletrec-calculus that we will be using throughout this work. We define CRS
signatures and a rewriting system $\llunfCRS$ for unfolding \lambdaletrec-terms
to obtain infinite \lambda-terms.

\begin{definition}[first-order representation of \lambdacal{} and \lambdaletreccal]
Let $\lvars = \{\avarnamei{0},\avarnamei{1},\avarnamei{2},\ldots\}$ be a set of variable names for $\slabs$\nb-abstrac\-tions,
and $\llrecvars$ be a set of names for recursion variables. 
We will use $\avar$, $\bvar$, $\cvar$ as syntactical variables for variable names bound by
$\slabs$-abstraction, and $\arecvar$, $\brecvar$, $\crecvar$
for recursion variable names; and similarly, we will use $\allter$, $\bllter$,
$\cllter$ for terms.
The set of \lambdaletrec{}-terms is inductively defined by the following grammar:
\begin{equation*}
\begin{array}{lllll}
(\textit{term})        & \allter    \sep{::=} \labs{\avar}{\allter}           & (\textit{abstraction}) \\
                       &            \sep{ | } \lapp{\allter}{\allter}         & (\textit{application}) \\ 
                       &            \sep{ | } \avar                           & (\textit{variable})    \\
                       &            \sep{ | } \letrec{\abindgroup}{\allter}   & (\textit{letrec})      \\
(\textit{binding group}) & \mathit{\abindgroup} \sep{::=} \arecvari{1}=\allter~\dots~\arecvari{n}=\allter   & (\textit{equations})   \\
                       &      \sep{   } (\arecvari{1},\dots,\arecvari{n} \in \llrecvars~\text{all distinct})
\end{array}
\end{equation*}
The set of \lambda-terms is defined by a reduced form of the grammar with the \textit{letrec} alternative and the \textit{binding group} rule left out.
\end{definition}

On this grammar the rules in Figure~\ref{fig:llunfCRS:informal} 
describe unfolding of \lambdaletrec{}-terms using an informal notation.
The names of the first four rules are chosen to reflect the kind of term that
resides inside of the $\textsf{in}$\nb-part of the $\sletrec$-term, which helps to see
that the rules are complete in the sense that every term of the form
$\letrec\abindgroup\allter$ is a redex.
\begin{figure}[tp]
  \[\begin{array}{rrcl}
    (\srulebp{\sunf}{@}):&
    \letrec{\abindgroup}{\lapp{\allteri{0}}{\allteri{1}}}
      &\sred&
    \lapp{(\letrec{\abindgroup}{\allteri{0}})}{(\letrec{\abindgroup}{\allteri{1}})}
    \\[1ex]
    (\srulebp{\sunf}{\slabs}):&
    \letrec{\abindgroup}{\labs{\avar}{\allteri{0}}}
      &\sred&
    \labs{\avar}{\letrec{\abindgroup}{\allteri{0}}}
    \\[1ex]
    (\srulebp{\sunf}{\smergeletrec}):&
    \letrec{\abindgroupi{0}}{\letrec{\abindgroupi{1}}{\allter}}
      &\sred&
    \letrec{\abindgroupi{0},\abindgroupi{1}}{\allter} 
    \\[1ex]
    (\srulebp{\sunf}{\rec}):&
    \letrec{\abindgroup}{\arecvari{i}}
      &\sred&
    \letrec{\abindgroup}{\allteri{i}} 
      \hspace*{2.5ex}
      \text{(if $\abindgroup$ is $\arecvari{1} = \allteri{1} \ldots \arecvari{n} = \allteri{n}$)}  
    \\[1ex]
    (\srulebp{\sunf}{\snil}):&
    \letrec{}{\allter}
      &\sred&
    \allter
    \\[1ex]
    (\srulebp{\sunf}{\sreduce}):&
    \letrec{\arecvari{1} = \allteri{1} \ldots \arecvari{n} = \allteri{n}}{\allter}
      &\sred&
    \letrec{\arecvari{j_1} = \allteri{j_1} \ldots \arecvari{j_{n'}} = \allteri{j_{n'}}}{\allter}  
    \\[0.25ex]
      & & & \hspace*{-3cm}\text{(if $\arecvari{j_1},\ldots,\arecvari{j_{n'}}$ are the recursion variables reachable from $\allter$)}
  \end{array}\]
  \caption{\label{fig:llunfCRS:informal}The rules of the \CRS~$\llunfCRS$ for unfolding \lambdaletrecterms\ in informal notation.}

  \vspace*{3ex}
  
  \begin{align*}
    (\srulebp{\sunf}{@}):\;\, &
        \letrecCRS{n}{\vec{\arecvar}}{\ametavari{1}{\vec{\arecvar}}, \ldots, \ametavari{n}{\vec{\arecvar}}, \lappCRS{\cmetavari{0}{\vec{\arecvar}}}{\cmetavari{1}{\vec{\arecvar}}}}
        \\& \hspace*{2ex}
        {} \red 
        \lappCRS{(\letrecCRS{n}{\vec{\arecvar}}{
                                                                               \ldots,\ametavari{n}{\vec{\arecvar}},\cmetavari{0}{\vec{\arecvar}}})}
                {\;(\letrecCRS{n}{\vec{\arecvar}}{
                                                                                 \ldots, \ametavari{n}{\vec{\arecvar}},\cmetavari{1}{\vec{\arecvar}}})}
    \displaybreak[0]\\[1.5ex]
    (\srulebp{\sunf}{\slabs}):\;\, &
        \letrecCRS{n}{\vec{\arecvar}}{\ametavari{1}{\vec{\arecvar}}, \ldots, \ametavari{n}{\vec{\arecvar}}, \labsCRS{\avar}{\cmetavar{\vec{\arecvar},\avar}}}
        \\& \hspace*{2ex} 
        {} \red
        \labsCRS{\avar}{\letrecCRS{n}{\vec{\arecvar}}{\ametavari{1}{\vec{\arecvar}}, \ldots, \ametavari{n}{\vec{\arecvar}}, \cmetavar{\vec{\arecvar},\avar}}}
    \displaybreak[0]\\[1.5ex]
    (\srulebp{\sunf}{\smergeletrec}):\;\, &
        \letrecCRS{n}{\vec{\arecvar}}{\ametavari{1}{\vec{\arecvar}}, \ldots, \ametavari{n}{\vec{\arecvar}}, 
                                      \letrecCRS{m}{\vec{\brecvar}}{\bmetavari{1}{\vec{\arecvar},\vec{\brecvar}}, \ldots, \bmetavari{m}{\vec{\arecvar},\vec{\brecvar}}
                                                                    \cmetavar{\vec{\arecvar},\vec{\brecvar}}}}
        \\
        & \hspace*{2ex} 
        {} \red
        \letrecCRS{n+m}{\vec{\arecvar}\vec{\brecvar}}{\ametavari{1}{\vec{\arecvar}}, \ldots, \ametavari{n}{\vec{\arecvar}},
                                                      \bmetavari{1}{\vec{\arecvar},\vec{\brecvar}}, \ldots, \bmetavari{m}{\vec{\arecvar},\vec{\brecvar}},
                                                      \cmetavar{\vec{\arecvar},\vec{\brecvar}}}
    \displaybreak[0]\\[1.5ex]
    (\srulebp{\sunf}{\rec}):\;\, &
      \letrecCRS{n}{\vec{\arecvar}}{\ametavari{1}{\vec{\arecvar}}, \ldots, \ametavari{n}{\vec{\arecvar}}, \arecvari{i}}
      {} \red
      \letrecCRS{n}{\vec{\arecvar}}{\ametavari{1}{\vec{\arecvar}}, \ldots, \ametavari{n}{\vec{\arecvar}}, \ametavari{i}{\vec{\arecvar}}}   
    \displaybreak[0]\\[1.5ex]
    (\srulebp{\sunf}{\sunfnil}):\;\, &
      \letrecemptyCRS{\scmetavar}
        {} \red
      \scmetavar
    \displaybreak[0]\\[1.5ex]
    (\srulebp{\sunf}{\sreduce}):\;\, &
        \letrecCRSbare{n}{\arecvari{1}\ldots\arecvari{n}}
                            {\bigl(  \ametavari{1}{\arecvari{i_{1,1}},\ldots,\arecvari{i_{m_1,1}}},
                                     \ldots, \ametavari{n}{\arecvari{i_{1,n}},\ldots,\arecvari{i_{m_n,n}}},}
                                     \\
        & \phantom{\letrecCRSbare{n}{\arecvari{1}\ldots\arecvari{n}}
                            {\bigr(  \ametavari{1}{\arecvari{i_{1,1}},\ldots,\arecvari{i_{1,m_1}}},\ldots,}} \,
                                     \cmetavar{\arecvari{i_{1,n+1}},\ldots,\arecvari{i_{m_{n+1},n+1}}}\bigl)
        \\[1ex]
        & \hspace*{2ex} {} \red
        \letrecCRSbare{n'}{\arecvari{j_1}\ldots\arecvari{j_{n'}}}
                            {\bigl(  \ametavari{j_1}{\allteri{1,j_1},\ldots,\allteri{m_{j_1},j_1}},
                                     \ldots}
                                     \\
        & \phantom{\letrecCRSbare{n}{\arecvari{j_1}\ldots\arecvari{j_{n'}}}{}} 
                                     \ldots                             
                                     \ametavari{j_{n'}}{\allteri{1,j_{n'}},\ldots,\allteri{m_{j_{n'}},j_{n'}}},
                                     \cmetavar{\arecvari{1,i_{n+1}},\ldots,\arecvari{i_{m_{n+1},n+1}}} \bigl)
  \end{align*}
  \begin{center}
  \vspace*{-1ex}
  
  \hspace*{10ex} 
  \begin{minipage}[c]{350pt}
    The terms that are denoted by the symbol~$\allter$ with different indices and appear
    on the right-hand side of a rule from the scheme ($\srulebp{\sunf}{\sreduce}$) are defined as:
  \begin{align*}
    \allteri{y,j_x} & \defdby \begin{cases}
                         f_{i_{y,x}}     & \text{if $i_{y,x}\in\setexp{j_1,\ldots,j_{n'}}$}  \\
                         \labs{\fvar}{\fvar}     & \text{else}
                       \end{cases}
                     &  &     
                     \text{(for all $1\le x\le n'$ and $1\le y\le m_{j_x}$),}
  \end{align*}                    
  relative to the indices $ j_1, \ldots, j_{n'}$ 
  with $1\le j_1 < \ldots < j_{n'}\le n$ that are  precisely the indices that are reachable from the indices in 
  $\setexp{i_{1,n+1},\ldots,i_{m_{n+1},n+1}}$
  (the indices $l$ of $\arecvari{l}$ that are applied to the variable $\sametavar$)
  via the
  binary reachability relation $\sareachrel$ on $\setexp{1,\ldots,n}$ that is 
  defined, for all $x,y\in\setexp{1,\ldots,n}$, by:
  $x \areachrel y$ if and only if there exist a $z\in\setexp{1,\ldots,m_x}$
  with $i_{z,x} = y$,
  formally that is:
  $ \setexp{j_1, \ldots, j_{n'}} = \descsetexpnormalsize{l}{i \areachrel^* l \text{ for some $i\in\setexp{i_{1,n+1},\ldots,i_{m_{n+1},n+1}}$}}$.
  \end{minipage}
  \end{center}
  \vspace*{2ex}
  
\caption{\label{fig:llunfCRS}The rules of the \CRS~$\llunfCRS$ for unfolding \lambdaletrecterms.}
\end{figure}

We will use higher-order notation and rules to reason about the
\lambdacalculus\ and the \lambdaletreccalculus, which immediately validates our
results for $\alpha$-equivalence classes instead of just preterms.

\begin{remark}[infinitary rewriting]
We use \CRS{s} as a rewriting framework since until now infinitary rewriting
theory has only been developed for \CRS{s} yet \cite{kete:simo:2009, kete:simo:2010, kete:simo:2011}.
\end{remark}

For formulating the above rules as a CRS, we provide \CRS\nb-signatures for \lambdacal\ and \lambdaletreccal.  

\begin{definition}[CRS signatures for $\lambdacal$ and $\lambdaletreccal$]
    \normalfont\label{def:sigs:lambdacal:lambdaletrec:CRS}  
  The \CRS\nb-signature for $\lambdacal$ consists of the set $\siglcCRS = \setexp{\slappCRS,\,\slabsCRS}$ 
  where $\slappCRS$ is a binary and $\slabsCRS$ a unary function symbol.
  The \CRS\nb-signature~$\sigllcCRS$ consists of the countably infinite set 
  $\sigllcCRS = \siglcCRS \cup \descsetexp{ \sletCRS{n}, \, \srecinCRS{n} }{ n\in\nats }$ 
  of function symbols, where, for $n\in\nats$, the symbols $\sletCRS{n}$ and $\srecinCRS{n}$ have arity~$n$. 

By $\Ter{\lambdacal}$ and $\Ter\lambdaletreccal$ we denote the set of
\emph{closed} CRS\nb-terms over $\siglcCRS$ and $\sigllcCRS$ respectively, with the
restriction that
\begin{itemize}
\item the symbols $\sletCRS{n}$ and $\srecinCRS{n}$ only occur as patterns of
the form
\\$\letrecCRS{n}{\arecvari1\dots\arecvari{n}}{\ateri1,\dots,\ateri{n},\ater}$
for some terms $\ateri1,\dots,\ateri{n},\ater \in \Ter\lambdacal$
\item and that otherwise a CRS abstraction can only occur directly beneath
an $\slabsCRS$-symbol.
\end{itemize}
$\Ter\lambdacal$ and $\Ter\lambdaletreccal$ will later be specified more
formally in Definition~\ref{def:Ter-inflambda} and Definition~\ref{def:Ter-lambdaletrec}.

  We will use $\ater$, $\bter$, $\cter$ used as syntactical variables for terms in $\Ter\lambdacal$. 
  And by $\Ter{\lambdaletreccal}$ we denote the set of CRS\nb-terms over $\siglcCRS$,
  for which we will use the symbols $\allter$, $\bllter$, $\cllter$ as syntactical variables.
\end{definition}

\begin{example}\label{ex:crs-notation:entangled}
The term in Example~\ref{ex:entangled} in CRS notation:
\\$\labsCRS{a}{\labsCRS{b}{\lappCRS{\labsCRS{c}{\lappCRS{\labsCRS{d}{\lappCRS{\dots}{c}}}{b}}}{a}}}$
\end{example}

\begin{example}\label{ex:crs-notation:equiv-letrec-terms}
The terms in Example~\ref{ex:equiv-letrec-terms} in CRS notation:
\\$\labsCRS{x}{\letrecCRS{1}{r}{\lappCRS{\lappCRS:x}{r}},r}$
\\$\labsCRS{x}{\letrecCRS{1}{r}{\lappCRS{\lappCRS:x}{\lappCRS{\lappCRS:x}{r}},r}}$
\end{example}

\begin{definition}[terms in \inflambdacal]
We denote by $\Ter{\inflambdaprefixcal}$ the set of finite and infinite
\CRS\nb-terms over the signature $\siglcCRS$. Note that the set of infinite
\lambda-terms subsume finite \lambda-terms, thus whenever we speak of an
infinite \lambda-term in fact we refer to a \emph{potentially} infinite
\lambda-term.
\end{definition}

\begin{definition}[the \CRS~$\llunfCRS$ for unfolding in \lambdaletreccal]\label{def:llunfCRS}
  The \emph{\CRS}~$\llunfCRS$ \emph{for unfolding $\lambdaletrec$\nb-terms}\/
  is the \CRS\ for terms over the signature $\sigllcCRS$ (see Definition~\ref{def:sigs:lambdacal:lambdaletrec:CRS}) 
  with the rule schemes in Figure~\ref{fig:llunfCRS} in which $n$ varies among numbers in $\posnats\,$.  

  The \ARS\ induced by the \CRS\ $\llunfCRS$ will be denoted by $\llunfARS$.
  We write $\sunfoldred$ for the rewrite relation induced by $\llunfCRS$.
  And by $\sunflappred\,$, $\sunflabsred\,$, $\sunfnilred\,$, $\sunfrecred\,$,
  $\sunfletrecred\,$, and $\sunfreducered$ we denote the rewrite relations of
  both $\llunfCRS$ and $\llunfARS$ that are induced by the rules
  $\srulebp{\sunf}{@}\,$, $\srulebp{\sunf}{\slabs}\,$,
  $\srulebp{\sunf}{\sunfnil}\,$, $\srulebp{\sunf}{\rec}\,$,
  $\srulebp{\sunf}{\smergeletrec}\,$, and
  $\srulebp{\sunf}{\sreduce}\,$, respectively.

\end{definition}

\begin{remark}[motivation of the rules $\srulebp{\sunf}{\sreduce}$ and $\srulebp{\sunf}{\snil}$ in $\llunfCRS$]\label{rem:reduce:nil:motivation}
  The purpose of taking up the rule $\srulebp{\sunf}{\sreduce}$, together with the rule $\srulebp{\sunf}{\snil}$ into the \CRS~$\llunfCRS$
  consists in preventing unbounded growth of binding groups during unfolding. 
  Consider for instance the outermost rewrite sequence
  on the term ${\letrec{f=\letrec{g={\lapp fg}}{g}}{f}}$ shown in Figure~\ref{fig:unbounded:growth:bindgroup}.
\begin{figure}
\begin{flushleft}
\hspace*{-3ex}
$  
\begin{array}{ll}
              & {\letrec{f=\letrec{g={\lapp fg}}{g}}{f}}
\\[2mm]
\sunfrecred & {\letrec{f=\letrec{g={\lapp fg}}{g}}{\letrec{g'={\lapp f{g'}}}{g'}}}
\\[2mm]
\sunfletrecred & {\letrec{
\begin{array}{l}
f=\letrec{g={\lapp fg}}{g}
\\
g'={\lapp f{g'}}
\end{array}
}{g'}}
\\[3mm]
\sunfrecred & {\letrec{
\begin{array}{l}
f=\letrec{g={\lapp fg}}{g}
\\
g'={\lapp f{g'}}
\end{array}
}{{\lapp f{g'}}}}
\\[3mm]
\sunflappred & \lapp{\Bigl(\,\letrec{
\begin{array}{l}
f=\letrec{g={\lapp fg}}{g}
\\
g'={\lapp f{g'}}
\end{array}
}{f}\,\Bigr)}
{\Bigl(\,\letrec{
\begin{array}{l}
f=\letrec{g={\lapp fg}}{g}
\\
g'={\lapp f{g'}}
\end{array}
}{g'}\,\Bigr)}
\\[3mm]
\sunfrecred &
\lapp{\Bigr(\,\letrec{\begin{array}{l}
  f=\letrec{g={\lapp fg}}{g}
  \\
  g'={\lapp f{g'}}
\end{array}}
{(\letrec{g''={\lapp f{g''}}}{g''})}\Bigr)}
{\Bigl(\,\letrec{
\begin{array}{l}
f=\letrec{g={\lapp fg}}{g}
\\
g'={\lapp f{g'}}
\end{array}
}{g'}\,\Bigr)}
\\[4mm]
\sunfletrecred&
\lapp
{\Biggl(\,\letrec{
\begin{array}{l}
f=\letrec{g={\lapp fg}}{g}
\\
g'={\lapp f{g'}}
\\
g''={\lapp f{g''}}
\end{array}
}{g''}\,\Biggr)}
{\Bigl(\,\letrec{
\begin{array}{l}
f=\letrec{g={\lapp fg}}{g}
\\
g'={\lapp f{g'}}
\end{array}
}{g'}\,\Bigr)}
\end{array}
$
\end{flushleft}
\caption{\label{fig:unbounded:growth:bindgroup}
         Unbounded growth of binding groups indicated by the initial segment of 
         an infinite $\sunfoldred$\nb-rewrite sequence 
                                                       that does not contain $\sunfreducered$\nb-steps.}
\end{figure}
    Applications of the rule $\srulebp{\sunf}{\sreduce}$ are able to 
    remove unreachable equations in binding groups. 
  
    While restricting the size of binding groups during unfolding is a sensible
    constraint on the unfolding process, it is not strictly necessary to define the unfolding of a \lambdaletrec\nb-term.
    We  could also use a rule
    $\srulebp\sunf\sfree$ in place of $\srulebp{\sunf}{\sreduce}$ and
    $\srulebp{\sunf}{\snil}$ where $\srulebp\sunf\sfree$ is (informally) defined as: 
\[
    (\srulebp{\sunf}{\sfree}):
    \letrec{\arecvari{1} = \allteri{1} \ldots \arecvari{n} = \allteri{n}}{\allter}
    \red
    \allter
    \hspace{5mm}
    \text{(if $\arecvari1,\ldots,\arecvari{n}$ do not occur in $\allter$)}
\]
  which allows steps like
  $ \labs{\avar}{\letrec{\abindgroup}{\avar}} \red \labs{\avar}{\avar}$,
  and thus allows to move bound variables out of the $\textsf{in}$\nb-part of $\sletrec$\nb-expressions. 
  (Note that in $\llunfCRS$ such a step can be simulated by $\sunfreducered$\nb-step followed by a $\sunfnilred$\nb-step.)
  We will, however, embed the unfolding rules into other rewriting
  systems of which we wish to perform unfolding in a lazy way such that the 
  number of derivable terms is bounded. the approach with the
  $\srulebp\sunf\sfree$-rule runs counter to that idea, as it easily leads to an
  unbounded growth of binding groups.
\end{remark}

\begin{remark}[shape of the rule $\srulebp{\sunf}{\sreduce}$ in $\llunfCRS$]\label{rem:reduce:shape}
  The rewrite rules $\srulebp{\sunf}{\sreduce}$ of $\llunfCRS$ are not `fully extended' as 
  the metavariables $\sametavari{i}$ occurring in the left-hand side of the rule 
  do not have to be instantiated with all recursion variables 
  $\arecvari{1}$, \ldots, $\arecvari{n}$ bound in the abstraction prefix. 
  This is due to the design of this rule scheme in which reachability of a recursion variable $\arecvari{i}$
  from the $\textsf{in}$\nb-part of the formalised $\sletrec$\nb-term
  is defined by extracting from the format of the specific instance which recursion variables $\arecvari{1}$, \ldots, $\arecvari{n}$
  occur in which of the metavariables~$\sametavari{1}$, \ldots, $\sametavari{n}$. 
  
  In applications of a rule $\srulebp{\sunf}{\sreduce}$, unreachable recursion variables are removed from the abstraction prefix
  on the right hand side. Since the format of \CRS{s} requires that corresponding metavariables
  on the left- and on the right-hand side of a rule must have the same arity,
  occurrences of unreachable recursion variables as arguments for metavariables 
  describing the binding group of a reachable recursion variable
  cannot simply disappear on the right-hand side.
  In the definition above, a specific, but arbitrarily chosen \lambdaterm, namely $\labs{\fvar}{\fvar}$,
  is substituted for such occurrences of unreachable recursion variables.
  This does not interfere with the unfolding operation defined later.
\end{remark}

Furthermore we profit from the property of normal forms w.r.t.\
$\srulebp{\sunf}{\sreduce}$ that the set of free variables of the
\lambdaletrec-term corresponds to the set of free variables of its unfolding,
which we will utilise in the mention rewriting systems later on.
  
\begin{definition}[reduced \lambdaletrecterms]\label{def:reduced:llter}
  A \lambdaletrec\nb-term~$\allter$ is called \emph{reduced} 
  if it is a normal form with respect to $\sunfnilred$ and $\sunfreducered$.
  If a \lambdaletrecterms{}~$\bllter$ reduces to a reduced term~$\allter$
  by exclusively $\sunfnilred$- and $\sunfreducered$-steps,
  then $\bllter$ is called a \emph{reduced form} of~$\allter$.
\end{definition}

\begin{proposition}[confluence of \sletrec-unfolding]\label{prop:unf-confluence}
$\llunfCRS$ is confluent.
\end{proposition}
  
\begin{proof}
We give a proof based on decreasing diagrams \cite[Sec.~14.2]{terese:2003} by
showing that \emph{parallel} steps are confluent. The proof involves a
comprehensive critical-pair analysis. It can be found in
Appendix~\ref{app:conf_proof} on page~\pageref{app:conf_proof}.
\end{proof}

Note that this confluence result concerns a rewriting system for unfolding
$\stxtlambdaletreccal$-terms, and therefore does not conflict
with non-confluence observations concerning versions of cyclic $\lambda$\nb-cal\-culi
with unfolding rules as well was with $\beta$\nb-reduction \cite{ario:klop:1997}.

\begin{example}\label{ex:llunfCRS}
$\llunfCRS$ when applied to $\letrec{\arecvar =
\labs{\avar\bvar}{\lapp{\lapp{\arecvar}{\bvar}}{\avar}}}{\arecvar}$
admits the following rewrite sequence:
\begin{align*}
& \letrec{\arecvar = \labs{\avar\bvar}{\lapp{\lapp{\arecvar}{\bvar}}{\avar}}}{\arecvar}
\alignbreak\sunfrecred ~&
\letrec{\arecvar = \labs{\avar\bvar}{\lapp{\lapp{\arecvar}{\bvar}}{\avar}}}{\labs{\avar\bvar}{\lapp{\lapp{\arecvar}{\bvar}}{\avar}}}
\alignbreak\sunflabsred ~&
\labs\avar{\letrec{\arecvar = \labs{\avar\bvar}{\lapp{\lapp{\arecvar}{\bvar}}{\avar}}}{\labs{\bvar}{\lapp{\lapp{\arecvar}{\bvar}}{\avar}}}}
\alignbreak\sunflabsred ~&
\labs{\avar\bvar}{\letrec{\arecvar = \labs{\avar\bvar}{\lapp{\lapp{\arecvar}{\bvar}}{\avar}}}{\lapp{\lapp{\arecvar}{\bvar}}{\avar}}}
\alignbreak\sunflappred ~&
\labs{\avar\bvar}{\lapp
  {(\letrec{\arecvar = \labs{\avar\bvar}{\lapp{\lapp{\arecvar}{\bvar}}{\avar}}}{\lapp{\arecvar}{\bvar}})}
  {(\letrec{\arecvar = \labs{\avar\bvar}{\lapp{\lapp{\arecvar}{\bvar}}{\avar}}}{\avar})}
}
\alignbreak\sunfreducered ~&
\labs{\avar\bvar}{\lapp
  {(\letrec{\arecvar = \labs{\avar\bvar}{\lapp{\lapp{\arecvar}{\bvar}}{\avar}}}{\lapp{\arecvar}{\bvar}})}
  {(\letrec{}{\avar})}
}
\alignbreak\sunfnilred ~&
\labs{\avar\bvar}{\lapp
  {(\letrec{\arecvar = \labs{\avar\bvar}{\lapp{\lapp{\arecvar}{\bvar}}{\avar}}}{\lapp{\arecvar}{\bvar}})}
  {\avar}
}
\alignbreak\sunflappred ~&
\labs{\avar\bvar}{\lapp
  {\lapp
      {(\letrec{\arecvar = \labs{\avar\bvar}{\lapp{\lapp{\arecvar}{\bvar}}{\avar}}}{\arecvar})}
      {(\letrec{\arecvar = \labs{\avar\bvar}{\lapp{\lapp{\arecvar}{\bvar}}{\avar}}}{\bvar})}
    }
  {\avar}
}
\alignbreak\sunfreducered ~&
\labs{\avar\bvar}{\lapp
  {\lapp
      {(\letrec{\arecvar = \labs{\avar\bvar}{\lapp{\lapp{\arecvar}{\bvar}}{\avar}}}{\arecvar})}
      {(\letrec{}{\bvar})}
    }
  {\avar}
}
\alignbreak\sunfnilred ~&
\labs{\avar\bvar}{\lapp{\lapp
      {(\letrec{\arecvar = \labs{\avar\bvar}{\lapp{\lapp{\arecvar}{\bvar}}{\avar}}}{\arecvar})}
      {\bvar}
    }
  {\avar}
}
\alignbreak\sunfrecred ~&
\labs{\avar\bvar}{\lapp{\lapp
      ~\dots~
      {\bvar}
    }
  {\avar}
}
\end{align*}
\end{example}

However, not every \lambdaletrecterm unfolds to an infinite \lambdaterm\ 
in the sense that it has an infinite \lambdaterm\ as its infinite $\sunfoldred$\nb-normal form.
For example, 
the \lambdaletrecterm~$\allter = \letrec{\arecvar = \arecvar}{\arecvar}$ admits only rewrite sequences of the form
$ \allter \unfrecred \allter \unfrecred \ldots $, and hence does not unfold to an infinite \lambdaterm.
Terms like this are unproductive in the sense that during all outermost-fair rewrite sequences 
the production of an infinite \lambdaterm\ stagnates due to an unproductive cycle. 
%
\begin{definition}[$\llunfCRS$-productivity]
Let $\allter$ be a \lambdaletrec\nb-term $\allter$.
We say that $\allter$ is \emph{$\llunfCRS$\nb-productive} 
if the following statement holds:
\begin{itemize}
\item $\allter$ does not have a $\sunfoldmred$\nb-reduct that
is the source of an infinite $\sunfoldred$\nb-rewrite sequence consisting exclusively of
outermost steps with respect to $\sunfrecred$, $\sunfnilred$, $\sunfletrecred$, or $\sunfreducered$.
\end{itemize}
\end{definition}

\begin{lemma}\label{lem:outermost-fair-sequences}
Let $\allter$ be a \lambdaletrec\nb-term of the form
$\letrec\abindgroup\bllter$.
Then exactly one of the following statements hold:
\begin{itemize}
\item All maximal outermost-fair $\llunfCRS$-rewriting sequences on $\allter$
      solely contain terms of the form $\letrec\bbindgroup\cllter$.
\item All maximal outermost-fair $\llunfCRS$-rewriting sequences on $\allter$
      only have finitely many terms of the form $\letrec\bbindgroup\cllter$.
\end{itemize}

\todo{ist das nicht einfacher?}
Alternatively:\\
Then the following statements holds either for all maximal outermost-fair
$\llunfCRS$-rewriting sequences or for none: the sequence contains only terms of the form $\letrec\bbindgroup\cllter$.
\end{lemma}

\begin{lemma}\label{lem:unfolding}
  For all \lambdaletrec\nb-terms~$\allter$ the following statements are equivalent:
  \begin{enumerate}[(i)]
    \item{}\label{lem:unfolding:item:i} 
      $\allter \unfoldomegared \aiter$ for some infinite \lambdaterm~$\aiter$.
    \item{}\label{lem:unfolding:item:ii} 
      $\allter$ is $\llunfCRS$\nb-productive.   
    \item{}\label{lem:unfolding:item:iii}
      Every maximal outermost-fair $\sunfoldred$\nb-rewrite sequence on $\allter$ is strongly convergent.
  \end{enumerate}
\end{lemma}

\begin{proof}
$\text{(ii)} \Rightarrow \text{(iii)}$, because if $\allter$ is
$\llunfCRS$\nb-productive then every outermost occurrence of a $\sletrec$ in
every $\sunfoldred$-reduct will be eventually pushed down to a higher position by
either a $\sunflabsred$- or a $\sunflappred$-step of any maximal outermost-fair
$\sunfoldred$-sequence. Since only $\sletrec$-terms are redexes in $\llunfCRS$ any
maximal outermost-fair rewrite sequence starting from $\allter$ converges to an
infinite normal form. (i) follows directly from (iii). $\text{(i)} \Rightarrow
\text{(ii)}$ follows from Lemma~\ref{lem:outermost-fair-sequences} by
contradiction. If $\allter$ is not $\llunfCRS$-productive then it has by
definition a $\sunfoldred$-reduct with at least one occurrence of a
\sletrec which cannot be pushed further down by any outermost application of
any $\llunfCRS$-rule. By Lemma~\ref{lem:outermost-fair-sequences} the same
holds for every other maximal outermost-fair rewrite sequence. Therefore
$\allter$ cannot unfold to an infinite \lambda-term $\ater$ because $\ater$ may
not contain any \sletrec{s}.
\end{proof}

\begin{lemma}[uniqueness of unfolding]\label{lem:unique:unfolding}  
  Unfolding normal forms of \lambdaletrecterms\ reachable in at most $\omega$ steps are unique.
  That is: if $ \aiteri{1} \unfoldconvomegared \allter \unfoldomegared \aiteri{2}$ 
    for a \lambdaletrecterm~$\allter$,
  and infinite \lambdaterms~$\aiteri{1}$ and $\aiteri{2}$, 
  then $\aiteri{1} = \aiteri{2}$. 
\end{lemma}

\begin{proof}[Proof sketch]
Let us assume that $\allter$ unfolds to $\aiteri1$ and $\aiteri2$ by the
$\sunfoldred$-reduction sequences $\arewseq_1$ and $\arewseq_2$.
Consider for $n \in \nats$ the first reduct $\bllteri1$ ($\bllteri2$) in the
sequence $\arewseq_1$ ($\arewseq_2$) that is stable above depth $n$. 
It follows from confluence of the rewrite relation
(Proposition~\ref{prop:unf-confluence}) that $\bllteri1$ and $\bllteri2$ have a
common $\sunfoldred$-reduct $\bllteri3$. 
Since the rules of $\llunfCRS$ do not create a redex at lower depth than the
occurrence of the left-hand side (`redexes are not pushed upwards') in the
rewriting sequence from $\bllteri1$ ($\bllteri2$) to $\bllteri3$, no
contractions take place above depth $n$.
In that sense $\bllteri1$ and $\bllteri2$ are `equal up to depth $n$'. Such a
notion is, however, still in need of precise formulation for CRS-terms.
The argument can be repeated for arbitrary $n \in \nats$, therefore $\aiteri1$
and $\aiteri2$ agree on arbitrarily large outermost contexts.
\end{proof}

As a consequence of the lemma above, the rewriting system $\llunfCRS$ 
defines a partial function $\sUnf$ for unfolding \lambdaletrec-terms. 

\begin{definition}[unfolding as a mapping]
We define the partial unfolding function:
\begin{align*}
\sUnf \funin \Ter{\lambdaletreccal} & {} \rightharpoonup \Ter{\inflambdacal}
\\
\allter & {} \mapsto \ater ~~~ \text{if}~\allter\unfoldomegared\ater 
\end{align*}
We say that $\allter$ \emph{expresses} $\ater$ if $\Unf\allter = \ater$.
Uniqueness of $\sUnf$ follows from Lemma~\ref{lem:unfolding}. 
\end{definition}

\begin{example}\label{ex:expresses}
The terms from Example~\ref{ex:equiv-letrec-terms} both express the same
\inflambda-term: \\
$\labs{x}{\letrec{r}{x:r}} \unfoldinfred \labs{x}{x:x:x:\dots} \unfoldconvinfred \labs{x}{\letrec{r}{x:(x:r)}}$
\end{example}

%

\begin{example}[$\llunfCRS$-unproductive \lambdaletrec-term]\label{ex:non-unfoldable}
As an example for a non-unfoldable \lambdaletrecterm, consider 
                                                               $\letrec{f=\letrec{g=f}{g}}{f}$
which is not in $\dom{\sUnf}$ and the cyclic rewriting sequence:
\[\begin{array}{ll}
               &  $\letrec{f = \letrec{g = f}{g}}{f}$ \\
\sunfrecred    &  $\letrec{f = \letrec{g = f}{g}}{\letrec{g = f}{g}}$ \\
\sunfletrecred &  $\letrec{f = \letrec{g = f}{g}; g = f}{g}$ \\
\sunfrecred    &  $\letrec{f = \letrec{g = f}{g}; g = f}{f}$ \\
\sunfreducered &  $\letrec{f = \letrec{g = f}{g}}{f}$ \\
\end{array}\]
We will revisit this example in Example~\ref{ex:stRegletrec} to illustrate the
cyclicity proof.
\end{example}

We also define an unfolding function which is complete on $\Ter{\lambdaletreccal}$
by mapping non-unfoldable subterms to $\bot$, which yields Böhm trees 
                                                                 in \inflambdacal.

\begin{definition}[partial unfolding]
$\Terbot{\inflambdacal}$ denotes infinite terms over $\siglcCRS\cup\{\bot\}$,
where $\bot$ is a constant symbol. 
\begin{align*}
\spUnf \funin \Ter{\lambdaletreccal} & {} \to \Terbot{\inflambdacal}
\\
\allter & {} \mapsto \ater ~~~ \text{if}~\allter 
                                                 \unfoldbotomegared \aiter
\end{align*}
Thereby $\sunfoldbotomegared$ is the infinitary rewrite relation induced by the
rewrite relation $\sunfoldbotred$ that extends $\sunfoldred$ by
mapping $\llunfCRS$\nb-root-active subterms to $\bot$.
\end{definition}

\section{Regular and strongly regular terms in $\inflambdacal$}
  \label{sec:regular}
\todo{Definieren unsere Systeme eine Pfadsemantik für \inflambdacal? Verweis
auf ``Paths in the \lambda-calculus}

For infinite first-order trees the concept of regularity is well-known and
well-studied \cite{cour:1983}. Regularity of a labelled tree%
  \footnote{By a `labelled tree' we here mean a finite or infinite tree
            whose nodes are labelled by function symbols from a first-order signature
            such that the arity of the function symbol in a node determines the number
            of successors of the node.} 
is defined
as the existence of only finitely many subtrees and implies the existence of a
finite graph that unfolds to that tree.
In this section we generalise the notion of regularity to trees with a binding
mechanism, the \inflambda-calculus specifically. We give a definition
for regularity which corresponds to regularity of a term when regarded as a
first-order tree, and for strong regularity, which will be shown in the following
sections to coincide with \lambdaletrec-expressibility.

We define regularity and strong regularity in terms of rewriting systems that
will be called $\RegCRS$ and $\stRegCRS$. Rewrite sequences in these systems
\emph{inspect} a given term coinductively in the sense that a rewrite sequence
corresponds to a decomposition of the term along one of its paths from the
root. Both $\RegCRS$ and $\stRegCRS$ extend a kernel system $\RegzeroCRS$
comprising three rewrite rules which denote whether the position just passed in
the tree is an abstraction or an application and in the second case whether the
application is being followed to the left or to the right.

The rewriting systems are defined
on \inflambda-terms enriched by what we call an \emph{abstraction
prefix}, by which the terms can be kept closed during the deconstruction. This
is crucial for the definition of the rewriting system as a CRS. While
intuitively it is clear that the \lambda-term $\lapp\ater\bter$ is composed of
the subterms $\ater$ and $\bter$, abstractions are more problematic. In a
first-order setting one could say that $\labs{x}\ater$ contains $\ater$ as a
subterm, but if $x$ occurs freely in $\ater$ then $\ater$ would be an open
term. That means that the scrutinisation of an abstraction would be able to go
from a closed term to an open term, which would run counter to the
interpretation of a higher-order term as an $\alpha$-equivalence class. In the
definition of the rewriting systems below this issue is resolved as follows.
When inspecting an abstraction, the binder is not left out but moved from the
scrutinised subterm into the prefix. That guarantees that the term as a whole
remains closed.

In $\labs{x}{\labs{y}{\lapp{\lapp xx}{y}}}$
for instance the path from the root to the second occurrence of $x$ then
corresponds to the rewrite sequence:
\[
\femptylabs{\labs{x}{\labs{y}{\lapp{\lapp xx}{y}}}}
\labsdecompred
\flabs{x}{\labs{y}{\lapp{\lapp xx}{y}}}
\labsdecompred
\flabs{xy}{\lapp{\lapp xx}{y}}
\lappdecompired0
\flabs{xy}{\lapp xx}
\lappdecompired1
\flabs{xy}{x}
\]

The $\stRegCRS$ and the $\RegCRS$ system extend $\RegzeroCRS$ by a
\emph{scope-delimiting} rule, which signifies that the scope of an abstraction
has ended, whereby both systems are based on different notions of scope.

$\RegCRS$ relies on what we simply call \emph{scope} of an abstraction:
the range from the abstraction up to the positions under which the bound variable does not occur anymore.

We base $\stRegCRS$ on a different notion of scope, called \extscope{}, which is
strictly nested. The \extscope{s} of an abstraction extends its scope by
encompassing all \extscope{s} that are opened within its range. As a consequence,
\extscope{s} do no overlap (see Figure~\ref{fig:extscope}). 

Precise definitions of scope and \extscope\ are given later in Definition~\ref{def:bind:capt:chain}.

When every \extscope{} is closed by the \extscope{}-delimiting rule then the sequence
of rewrite steps alone (i.e.~without the terms themselves) unambiguously
determines which abstraction a variable occurrence belongs to. The rewrite
sequence from above would then have one additional \extscope-delimiting step
asserting that the variable at the end of the path is indeed $x$ and not $y$:
\[
\femptylabs{}{\labs{x}{\labs{y}{\lapp{\lapp xx}{y}}}}
\labsdecompred
\flabs{x}{\labs{y}{\lapp{\lapp xx}{y}}}
\labsdecompred
\flabs{xy}{\lapp{\lapp xx}{y}}
\lappdecompired0
\flabs{xy}{\lapp xx}
\lappdecompired1
\flabs{xy}{x}
\compressstregred
\flabs{x}{x}
\]
\begin{figure}
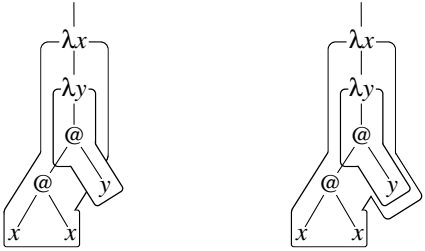

\fig{scope} \hspace{2cm} \fig{extscope}
\caption{\label{fig:extscope}The difference between scope and \extscope}
\end{figure}
The abstraction prefix not only keeps the term closed but also denotes
which \extscope\ is still open, which provides the information to decide
applicability of the \extscope-delimiting rule. The last step closes the
\extscope{} of $y$, therefore that variable is removed from the prefix. The
rewrite sequence for the path to the occurrence of $y$ does not include an
\extscope-delimiting step:
\[
\femptylabs{}{\labs{x}{\labs{y}{\lapp{\lapp xx}{y}}}}
\labsdecompred
\flabs{x}{\labs{y}{\lapp{\lapp xx}{y}}}
\labsdecompred
\flabs{xy}{\lapp{\lapp xx}{y}}
\lappdecompired1
\flabs{xy}{y}
\]

Ultimately, the $\stRegCRS$ rewriting system defines nameless representations
for \inflambda-terms related to the de-Bruijn notation. Considering the
de-Bruijn representation of the above term $\lambda\,\lambda\,\lapp{\lapp{(\lapp
S0)}{(\lapp S0)}}{0}$ we find that the position of the
$\scompressstregred$-steps indeed coincides with the position of the $S$
markers. However, the rewrite system $\stRegCRS$ permits more flexibility for the
placement of $\scompressstregred$-steps, an idea also used in \cite{oost:looi:zwit:2004}.
For example the path from above to the
second occurrence of $x$ can also be witnessed by another rewrite sequence
\[
\femptylabs{}{\labs{x}{\labs{y}{\lapp{\lapp xx}{y}}}}
\labsdecompred
\flabs{x}{\labs{y}{\lapp{\lapp xx}{y}}}
\labsdecompred
\flabs{xy}{\lapp{\lapp xx}{y}}
\lappdecompired0
\flabs{xy}{\lapp xx}
\compressstregred
\flabs{x}{\lapp xx}
\lappdecompired1
\flabs{x}{x}
\]
where the \extscope{} of $y$ is closed earlier. This would
correspond to $\lambda\,\lambda\,\lapp{(\lapp{S}{(\lapp 00)})}{0}$ in de-Bruijn
notation, more precisely, in a variant of the de-Bruijn notation which permits
the \scope/\extscope\nb-delimiter $S$ to be used also within the term, before abstractions and applications.

\begin{definition}[\CRS-terms with abstraction prefixes]
  \label{def:sig:lambdaprefixcal:CRS}
  The \CRS\nb-sig\-na\-ture for $\lambdaprefixcal$, the $\lambda$\nb-calculus
  with abstraction prefixes, extends the \CRS\nb-signature~$\siglcCRS$ for
  $\lambdacal$ 
  (see Definition~\ref{def:sigs:lambdacal:lambdaletrec:CRS})
  and consists of the set 
  $\siglpcCRS = \siglcCRS 
                          \cup \descsetexp{ \sflabsCRS{n} }{ n\in\nats }$ 
  of function symbols, 
  where 
  for $n\in\nats$ the function symbols $\sflabsCRS{n}$ 
  for \emph{prefix $\lambda$\nb-abstractions} of length~$n$ are unary (have arity one). 
  Using the syntactical variables for terms in $\lambdacal$,
  \CRS\nb-terms with leading prefixes
  $\flabsCRS{n}{\avari{1}}{\ldots\absCRS{\avari{n}}{\ater}}$
  will informally be denoted by
  $\flabs{\avari{1}\ldots\avari{n}}{\ater}$,
  abbreviated as
  $\flabs{\vec{\avar}}{\ater}$.
\end{definition}   

\begin{definition}[The \CRSs~$\RegzeroCRS$, $\RegCRS$, $\stRegCRS$ for decomposing $\lambda$-terms]%
  \label{def:RegCRS:stRegCRS}
  Consider the following \CRS\nbd-rules over signature $\siglpcCRS$: 
  \begin{align*}
    (\srulep{\slappdecompi{i}}): 
      & & \hspace*{-25ex}
    \flabsCRS{n}{\avari{1}\ldots\avari{n}}{\lappCRS{\cmetavari{0}{\vec{\avar}}}{\cmetavari{1}{\vec{\avar}}}}
      & {} \red 
    \flabsCRS{n}{\avari{1}\ldots\avari{n}}{\cmetavari{i}{\vec{\avar}}}  
      \hspace*{4ex} (i\in\{0,1\})
    \\
    (\srulep{\slabsdecomp}):
      & & 
    \flabsCRS{n}{\avari{1}\ldots\avari{n}}{\labsCRS{\avari{n+1}}{\cmetavar{\vec{\avar}}}}
      & {} \red 
    \flabsCRS{n}{\avari{1}\ldots\avari{n+1}}{\cmetavar{\vec{\avar}}} 
    \displaybreak[0]\\
    (\srulep{\snlvarsucc}):
      & & 
    \flabsCRS{n+1}{\avari{1}\ldots\avari{n+1}}{\cmetavar{\avari{1},\ldots,\avari{n}}}
      & {} \red 
    \flabsCRS{n}{\avari{1}\ldots\avari{n}}{\cmetavar{\avari{1},\ldots,\avari{n}}}
    \displaybreak[0]\\  
    (\srulep{\scompress}):
      & & & \hspace*{-35ex}
      \begin{aligned}
        & 
        \flabsCRS{n+1}{\avari{1}\ldots\avari{n+1}}{\cmetavar{\avari{1},\ldots,\avari{i-1},\avari{i+1},\ldots,\avari{n+1}}}
        \red
        \\
        & \hspace*{16ex} 
        \flabsCRS{n}{\avari{1}\ldots\avari{i-1}\avari{i+1}\ldots\avari{n}}{\cmetavar{\avari{1},\ldots,\avari{i-1},\avari{i+1},\ldots,\avari{n+1}}}
      \end{aligned}
    %
  \end{align*}
  By $\RegzeroCRS$ we denote the \CRS\ with rules $\srulep{\slappdecompi{i}}$ and $\srulep{\slabsdecomp}$.
  By $\RegCRS$ (and respectively, by $\stRegCRS$) we denote the \CRS\ consisting of all of the above rules 
  \emph{except} the rule $\srulep{\snlvarsucc}$ (\emph{except} the rule $\srulep{\scompress}$).
  The rewrite relations of $\RegzeroCRS$, $\RegCRS$, and $\stRegCRS$ are denoted by $\sregzerored$, $\sregred$ and $\sstregred$, respectively.
  And by $\slappdecompired{0}$, $\slappdecompired{1}$,
  $\slabsdecompred$, $\scompressstregred$, $\scompressregred$,
  we respectively denote the rewrite relations
  induced by each of the single rules 
  $\srulep{\slappdecompi{0}}$, $\srulep{\slappdecompi{1}}$, $\srulep{\slabsdecomp}$,
  $\srulep{\snlvarsucc}$, and $\srulep{\scompress}$. 
\end{definition}

Assuming that the translation between the formal and the informal notation is
facile, for better readability we will from now on rely on the latter. Here are
the rules from above in informal notation:
  \begin{align*}
    (\srulep{\slappdecompi{i}}): 
      & &
    \flabs{\avari{1}\ldots\avari{n}}{\lapp{\ateri{0}}{\aiteri{1}}}
      & {} \red 
    \flabs{\avari{1}\ldots\avari{n}}{\aiteri{i}}  
      & & \hspace*{-0.5ex} (i\in\{0,1\})
    \\
    (\srulep{\slabsdecomp}):
      & &
    \flabs{\avari{1}\ldots\avari{n}}{\labs{\avari{n+1}}{\aiteri{0}}}
      & {} \red 
    \flabs{\avari{1}\ldots\avari{n+1}}{\aiteri{0}} 
    \displaybreak[0]\\
    (\srulep{\snlvarsucc}):
      & &
    \flabs{\avari{1}\ldots\avari{n+1}}{\aiteri{0}}
      & {} \red 
    \flabs{\avari{1}\ldots\avari{n}}{\aiteri{0}}  
      & & 
      \hspace*{-15ex}
                           (\text{if the binding $\lambda\avari{n+1}$ is vacuous}
                                                 )
    \displaybreak[0]\\  
    (\srulep{\scompress}):
      & &
    \flabs{\avari{1}\ldots\avari{n+1}}{\aiteri{0}}
      & {} \red 
    \flabs{\avari{1}\ldots\avari{i-1}\avari{i+1}\ldots\avari{n+1}}{\aiteri{0}}   
      & & 
    \\   
      & & & & &
      \hspace*{-15ex} 
                           (\text{if the binding $\lambda\avari{i}$ is vacuous})
    %
  \end{align*}

\begin{figure}
\begin{tabular}{ccc}
\fig{lxy_xxy-regmin.ltg-l} & \fig{lxy_xxy-cprgraph.ltg-l} & \fig{lxy_xxy-redgraph.ltg-l} \\
$\RegzeroCRS$ & $\RegCRS$ & $\stRegCRS$
\end{tabular}
\caption{\label{fig:three_redgraphs}
The sub-ARSs induced by $\femptylabs{\labs{x}{\labs{y}{\lapp{\lapp
xx}{y}}}}$. Note, that in the pictures the nodes do not display the entire
\inflambdaprefixcal-terms but only their prefixes.
}
\end{figure}

\begin{remark}[only $\stRegCRS$ defines nameless representations]
Considering the graphs from Figure~\ref{fig:three_redgraphs} without labels
on their nodes we see that only from the $\stRegCRS$ graph the original
term could be reconstructed unambiguously. For example, the path to the rightmost occurrence of
$x$ has the rewrite sequence in $\RegCRS\,$ \[\labsdecompred .
\labsdecompred . \lappdecompired0 . \compressregred . \lappdecompired1\]
which would also be an admissible to witness an occurrence of $y$ at the same
position. This ambiguity plays a role for the definition of \lambdatg{s}
in Section~\ref{sec:ltgs}, and is discussed in that context in Remark~\ref{rem:nameless-repr}.
\end{remark}

Note that the following relationships between rewrite relations:
$\sregzerored$ is contained in $\sregred$ and $\sstregred$,
and since the rule $\srulep{\scompress}$ generalises the rule $\srulep{\snlvarsucc}$,
$\sstregred$ is contained in $\sregred$. 
\begin{proposition}\label{prop:rewseqs:stRegCRS:2:RegCRS}
  Every rewrite sequence in $\stRegCRS$ corresponds directly to a rewrite sequence in $\RegCRS$
  by exchanging $\scompressstregred$\nb-steps with $\scompressregred$\nb-steps.
\end{proposition}
    
Our interest will focus on the subset of terms with an outermost abstraction
prefix symbol and no other occurrences of such symbols. Note that the
rules in $\RegCRS$ and $\stRegCRS$ guarantee that every reduct of a term of the
form $\flabs{\vec{\avar}}{\aiter}$ is again a term of this form. Therefore we define:
  
\begin{definition}[prefixed \inflambda-terms]\label{def:CRSterms}
By $\Ter{\inflambdaprefixcal}$ we denote the set of 
\emph{closed} \iCRS\nb-terms over $\siglpcCRS$ with the restriction that
\begin{itemize}
\item every term $\ater\in\Ter{\inflambdaprefixcal}$ has a prefix at its root
and nowhere else, or in other words: $\ater$ is of the form
$\flabsCRS{n}{\avari{1}}{\ldots\absCRS{\avari{n}}\ater}$
where $\ater$ does not contain any
occurrences of function symbols $\sflabsCRS{i}$ for $i\in\nats$
\item and that otherwise a CRS abstraction can only occur directly beneath
an $\slabsCRS$-symbol.
\end{itemize}
$\Ter{\inflambdaprefixcal}$ is more formally specified in Definition~\ref{def:Ter-lambdaletrec}.
\end{definition}

\begin{proposition}\label{prop:oneRegARS}
$\Ter{\inflambdaprefixcal}$ is closed under 
$\sregzerored$, $\sregred$, and $\sstregred$.
\end{proposition}
 
\begin{definition}[the \ARSs~${\normalfont \RegzeroARS}$, ${\normalfont \RegARS}$, ${\normalfont \stRegARS}$]%
                  \label{RegARS:stRegARS:oneRegARS:onestRegARS}%
We denote by $\RegzeroARS$, $\RegARS$ and $\stRegARS$ the infinite
abstract rewriting systems (\ARS{s}) induced by the iCRSs derived from
$\RegzeroCRS$, $\RegCRS$, $\stRegCRS$, restricted to terms in
$\Ter\inflambdaprefixcal$.

The rewrite relations of $\RegzeroARS$, $\RegARS$ and $\stRegARS$ will
be denoted by the same symbols used for $\RegzeroCRS$, $\RegCRS$ and
$\stRegCRS$. Since all of our considerations will refer to the restricted set
of terms, this should create no confusion.
\end{definition}

%

\begin{proposition}\label{prop:rewprops:RegCRS:stRegCRS}
  The restrictions of the rewrite relations as defined in Definition~\ref{def:RegCRS:stRegCRS}
  to $\Ter{\inflambdaprefixcal}$, the set of objects of $\RegzeroARS$, $\RegARS$, and $\stRegARS$,
  have the following properties:%
    \footnote{Mind the restriction here to terms in $\Ter{\inflambdacal}$.}
  \begin{enumerate}[(i)]
    \item\label{prop:rewprops:RegCRS:stRegCRS:item:i}
       $\scompressregred$ is confluent, and 
                                        terminating.                                                                                                    
       %
    \item\label{prop:rewprops:RegCRS:stRegCRS:item:ii}
      $\scompressregred$ one-step commutes with $\slabsdecompred$, $\slappdecompired{0}$, $\slappdecompired{1}$, 
      and one-step sub-commutes with $\scompressstregred\,$:
      \begin{align*}
        \binrelcomp{\scompressregconvred}{\slabsdecompred} 
          & \;\subseteq\; 
        \binrelcomp{\slabsdecompred}{\scompressregconvred}
        & 
        \binrelcomp{\scompressregconvred}{\slappdecompired{i}} 
          & \;\subseteq\; 
        \binrelcomp{\slappdecompired{i}}{\scompressregconvred}
           \;\;\;\;\; (i\in\setexp{0,1})
        \\   
        \binrelcomp{\scompressregconvred}{\scompressstregred} 
          & \;\subseteq\; 
        \binrelcomp{\scompressstregeqred}{\scompressregconveqred}
      \end{align*}
    \item\label{prop:rewprops:RegCRS:stRegCRS:item:iii-0} 
      $\scompressstregred \subseteq \scompressregred$,
      and consequently, $\sstregred \subseteq \sregred$. 
      Furthermore, $\sregzerored \subsetneqq \sstregred \subsetneqq \sregred\:$.
      %
    \item\label{prop:rewprops:RegCRS:stRegCRS:item:iii} 
      $\scompressstregred$ is deterministic, hence confluent,
      and terminating.
      %
    \item\label{prop:rewprops:RegCRS:stRegCRS:item:iv}
      $\scompressstregred$ one-step commutes with 
      $\slabsdecompred$, $\slappdecompired{0}$, and $\slappdecompired{1}$:
      \begin{align*}
        \binrelcomp{\scompressstregconvred}{\slabsdecompred} 
          & \;\subseteq\; 
        \binrelcomp{\slabsdecompred}{\scompressstregconvred}
        & 
        \binrelcomp{\scompressstregconvred}{\slappdecompired{i}} 
          & \;\subseteq\; 
        \binrelcomp{\slappdecompired{i}}{\scompressstregconvred}
           \;\;\;\;\;\text{($i\in\setexp{0,1}$)}
      \end{align*}
    \item{}\label{prop:rewprops:RegCRS:stRegCRS:item:iv-1}
      $\flabs{\avar}{\avar}$ is the sole term in $\sregred$\nb-normal form.
      Every $\sstregred$\nb-normal form is of the form
      $\flabs{\avari{1}\ldots\avari{n}}{\avari{n}}$. 
    \item\label{prop:rewprops:RegCRS:stRegCRS:item:v}
      $\sregred$ and $\sstregred$ are finitely branching,
      and, on finite terms, terminating. 
  \end{enumerate}
\end{proposition}

\begin{proof}
  Most properties, including those concerning commutation of steps, are easy to verify
  by analysing the behaviour of the rewrite rules in $\RegARS$ on terms of $\Ter{\lambdaprefixcal}$.
  
  For $\sstregred\subsetneqq\sregred$ in (\ref{prop:rewprops:RegCRS:stRegCRS:item:iii-0})
  note that, for example, $\flabs{\avar\bvar}{\bvar} \regred \flabs{\bvar}{\bvar}$ by a $\scompressregred$\nb-step, 
  but that $\flabs{\avar\bvar}{\bvar}$ is a $\scompressstregred$-normal form, and hence also a $\sstregred$\nb-normal form.
  
  In item~(\ref{prop:rewprops:RegCRS:stRegCRS:item:v}) we first argue for
  finite branchingness of $\sregred$ and $\sstregred\,$ on $\Ter{\inflambdacal}$:
  this property is entailed by the fact that, 
  on a term $\flabs{\vec{\avar}}{\aiter}$ with just one abstraction in its prefix,
  of the constituent rewrite relations
  $\slappdecompired{0}$, $\slappdecompired{1}$,
  $\slabsdecompred$, $\scompressstregred$, $\scompressregred$ of $\sregred$ and $\sstregred$
  only $\scompressregred$ can have branching degree greater than one, which in this case then
  also is bounded by the length $\length{\vec{\avar}}$ of the abstraction prefix. 
  For termination of $\sregred$ and $\sstregred\,$ on finite terms with just a leading abstraction prefix
  we can restrict to $\sregred$, due to (\ref{prop:rewprops:RegCRS:stRegCRS:item:iii-0}),
  and argue as follows:
  On finite terms in $\Ter{\inflambdaprefixcal}$,
  in every $\sregred$\nb-rewrite step
  either the size of the body of the term decreases strictly,
  or the size of the body stays the same, but the length of the prefix decreases by one.
  Hence in every rewrite step 
  the measure $\pair{\text{body size}}{\text{prefix length}}$ on terms
  decreases strictly in the (well-founded) lexicographic ordering on $\nats\times\nats$.
\end{proof}

%
%

As a consequence of Proposition~\ref{prop:rewprops:RegCRS:stRegCRS}, 
(\ref{prop:rewprops:RegCRS:stRegCRS:item:i}) and (\ref{prop:rewprops:RegCRS:stRegCRS:item:iii}),
the rewrite relations $\scompressregred$ and $\scompressstregred$ are normalizing on $\Ter{\inflambdacal}$.
For every term $\aiter\in\Ter{\inflambdaprefixcal}$ we will denote by
$\compressregnf{\aiter}$ and $\compressstregnf{\aiter}$ the normal forms of $\aiter$
with respect to $\scompressregred$ and $\scompressstregred$, respectively.
And by $\scompressregnfred$ (by $\scompressstregnfred$) 
we denote the many-step rewrite relation for $\scompressregred$ (for $\scompressstregred$)
that leads to a $\scompressregred$\nb-normal form (to a $\scompressstregred$\nb-normal form),
that is: $\aiter \compressregnfred \biter$ if $\compressregnf{\aiter} = \biter$
        (and respectively,
         $\aiter \compressstregnfred \biter$ if $\compressstregnf{\aiter} = \biter$), for all terms $\aiter$ and $\biter$.

\begin{proposition}\label{prop:compress:prefix:RegARS}
  The following statements hold:
  \begin{enumerate}[(i)]
    \item\label{prop:compress:prefix:RegARS:item:i}
      Let $\flabs{\vec{\avar}}{\aiter}$ a term in $\RegARS$
      with $\length{\vec{\avar}} = n\in\nats$.
      Then the number of terms $\flabs{\vec{\bvar}}{\biter}$ in $\RegARS$ 
      with $\flabs{\vec{\bvar}}{\biter} 
                                        \compressregmred \flabs{\vec{\avar}}{\aiter}$
      and $\length{\vec{\bvar}} = n+k\in\nats$ 
      is $\binom{n+k}{n}$.
    \protect\item\label{prop:compress:prefix:RegARS:item:ii}
      Let $T$ be a finite set of terms in $\RegARS$, and $k\in\nats$.
      Then also the set of terms in $\RegARS$ that are the form $\flabs{\vec{\bvar}}{\biter}$ 
      with $\length{\vec{\bvar}}\le k$ and that have a $\scompressregmred$\nb-reduct in $T$
      is finite.                                            
  \end{enumerate}         
\end{proposition}

\begin{proof}
  If 
  $\flabs{\vec{\bvar}}{\funap{\aiter}{\bvar}}
     = 
   \flabs{\bvari{1}\ldots\bvari{n+k}}{\funap{\biter}{\bvari{1},\ldots,\bvari{n+k}}} 
     \compressregmred 
   \flabs{\avari{1}\ldots\avari{n}}{\funap{\aiter}{\avari{1},\ldots,\avari{n}}}
     =
   \flabs{\vec{\avar}}{\funap{\aiter}{\vec{\avar}}}$,
  then it follows that there are
  $i_1,\ldots,i_n\in\setexp{1,\ldots,n+k}$ with $i_1 < i_2 < \ldots < i_n$
  such that the term $\flabs{\vec{\bvar}}{\funap{\aiter}{\bvar}}$ is actually of the form
  $\flabs{\bvari{1}\ldots\bvari{n+k}}{\funap{\biter}{\bvari{i_1},\ldots,\bvari{i_n}}}$
  and furthermore 
  $\flabs{\bvari{i_1}\ldots\bvari{i_n}}{\funap{\biter}{\bvari{i_1},\ldots,\bvari{i_n}}} 
     =
   \flabs{\avari{1}\ldots\avari{n}}{\funap{\aiter}{\avari{1},\ldots,\avari{n}}}$.
  Hence the number of terms
  $\flabs{\vec{\bvar}}{\funap{\aiter}{\bvar}}$
  with 
       $\scompressregmred$\nb-reduct
  $\flabs{\vec{\avar}}{\funap{\aiter}{\vec{\avar}}}$
  is equal to the number of choices $i_1,\ldots,i_n\in\setexp{1,\ldots,n+k}$ such that $i_1 < i_2 < \ldots < i_n$. 
  This establishes statement~(\ref{prop:compress:prefix:RegARS:item:i}).
  Statement~(\ref{prop:compress:prefix:RegARS:item:ii}) is an easy consequence.
\end{proof}

\begin{lemma}\label{lem:rewprops:projection:RegCRS:stRegCRS}
  On $\Ter{\inflambdaprefixcal}$, the rewrite relations $\sregred$ and $\sstregred$ have the following
  further properties with respect to $\scompressregred$, $\scompressstregred$, $\scompressregnfred$, and $\scompressstregnfred$:
  \vspace{-2.5ex}
  \begin{center}
    \begin{minipage}{189pt}
      \begin{align}
        {\binrelcomp{\scompressregconvmred}{\sregred}}
          \; & \subseteq\;
        {\binrelcomp{(\binrelcomp{\compressregnfred}{\sregzeroeqred})}{\scompressregconvmred}}
        \\
        {\binrelcomp{\scompressregconvmred}{\sregmred}}
          \; & \subseteq\;
        {\binrelcomp{(\binrelcomp{\compressregnfred}{\sregzeroeqred})^*}{\scompressregconvmred}}
      \end{align}
    \end{minipage}
    \begin{minipage}{189pt}
       \begin{align}
        {\binrelcomp{\scompressregconvmred}{\sstregred}}
          \; & \subseteq\;
        {\binrelcomp{(\binrelcomp{\compressstregnfred}{\sregzeroeqred})}{\scompressregconvmred}}
        \\
        {\binrelcomp{\scompressregconvmred}{\sstregmred}}
          \; & \subseteq\;
        {\binrelcomp{(\binrelcomp{\compressstregnfred}{\sregzeroeqred})^*}{\scompressregconvmred}}
      \end{align}  
    \end{minipage}
  \end{center}
\end{lemma}

\begin{proof}
  These commutation properties, which can be viewed as projection properties, can be shown by arguments with diagrams
  using the commutation properties in Proposition~\ref{prop:rewprops:RegCRS:stRegCRS}. 
\end{proof}

\begin{remark}
  The commutation properties in Lemma~\ref{lem:rewprops:projection:RegCRS:stRegCRS}
  can be refined to state that $\slabsdecompred$\nb-steps project to $\slabsdecompred$\nb-steps,
  and $\slappdecompired{0}$- and $\slappdecompired{1}$\nb-steps project to
  $\slappdecompired{0}$- and $\slappdecompired{1}$\nb-steps, accordingly.
\end{remark}

As an immediate consequence of Proposition~\ref{lem:rewprops:projection:RegCRS:stRegCRS} we obtain the following lemma,
which formulates a connection via projection between rewrite sequences in $\RegARS$ (in $\stRegARS$) 
and \emph{$\scompressregred$\nb-eager} (\emph{$\scompressstregred$\nb-eager}) rewrite sequences in $\RegARS$ (in $\stRegARS$)
that do not contain $\slabsdecompred$-, $\slappdecompired{0}$, or $\slappdecompired{1}$\nb-steps 
on terms that allow $\scompressregred$\nb-steps ($\scompressstregred$\nb-steps).

\begin{lemma}\label{lem:projection:RegCRS:stRegCRS}
  The following statements hold:
  \begin{enumerate}[(i)] 
    \item\label{lem:projection:RegCRS:stRegCRS:item:reg}
      Every (finite or infinite) rewrite sequence in $\RegARS$ of the form:
      \begin{equation*}
        \arewseq \funin
          \flabs{\vec{\avar}_0}{\aiteri{0}}
            \regred
         \flabs{\vec{\avar}_1}{\aiteri{1}}
            \regred
         \ldots
           \regred  
        \flabs{\vec{\avar}_k}{\aiteri{k}}
           \regred
        \ldots
      \end{equation*} 
      projects over a rewrite sequence
      $\crewseq \funin \flabs{\vec{\avar}_0}{\aiteri{0}} \compressregmred \flabs{\vec{\avar}'_0}{\aiteri{0}}$
      to a $\scompressregred$\nb-eager rewrite sequence in $\RegARS\,$ of the form:
      \begin{multline*}
       \Checkstreg{\arewseq} \funin
         \flabs{\vec{\avar}'_0}{\aiteri{0}}
           \binrelcomp{\compressregnfred}{\sregzeroeqred}
         \flabs{\vec{\avar}'_1}{\aiteri{1}}
           \binrelcomp{\compressregnfred}{\sregzeroeqred}
         \;\:\ldots\;\:
         \\
         \;\:\ldots\;\:
           \binrelcomp{\compressregnfred}{\sregzeroeqred}
         \flabs{\vec{\avar}'_k}{\aiteri{k}}
           \binrelcomp{\compressregnfred}{\sregzeroeqred}
         \ldots 
      \end{multline*}
      in the sense that
      $\flabs{\vec{\avar}_i}{\aiteri{i}} \compressregmred \flabs{\vec{\avar}'_i}{\aiteri{i}}\,$
      for all $i\in\nats$ less or equal to the length of $\arewseq$.
    \item\label{lem:projection:RegCRS:stRegCRS:item:streg}
      Every (finite or infinite) rewrite sequence in $\stRegARS$ of the form:
      \begin{equation*}
        \arewseq \funin
          \flabs{\vec{\bvar}_0}{\biteri{0}}
            \stregred
         \flabs{\vec{\bvar}_1}{\biteri{1}}
            \stregred
         \ldots
           \stregred  
         \flabs{\vec{\bvar}_k}{\biteri{k}}
           \stregred
         \ldots  
      \end{equation*}  
      projects over a rewrite sequence
      $\crewseq \funin \flabs{\vec{\bvar}_0}{\biteri{0}} \compressregmred \flabs{\vec{\bvar}'_0}{\biteri{0}}$
      to a $\scompressstregred$\nb-eager rewrite sequence in $\stRegARS\,$:
     \begin{multline*}
       \Checkstreg{\arewseq} \funin
         \flabs{\vec{\bvar}'_0}{\biteri{0}}
           \binrelcomp{\compressstregnfred}{\sregzeroeqred}
         \flabs{\vec{\bvar}'_1}{\biteri{1}}
           \binrelcomp{\compressstregnfred}{\sregzeroeqred}
         \;\:\ldots\;\:
         \\
         \;\:\ldots\;\:
         \binrelcomp{\compressstregnfred}{\sregzeroeqred}
         \flabs{\vec{\bvar}'_k}{\biteri{k}}
           \binrelcomp{\compressregnfred}{\sregzeroeqred}
         \ldots  
      \end{multline*} 
      in the sense that
      $\flabs{\vec{\bvar}_i}{\biteri{i}} \compressstregmred \flabs{\vec{\bvar}'_i}{\biteri{i}}\,$
      for all $i\in\nats$ less or equal to the length of $\arewseq$.
  \end{enumerate}      
\end{lemma}

\begin{remark}[non-determinism of ${\normalfont \sregred}$ and ${\normalfont \sstregred}$]
On terms in $\Ter{\inflambdaprefixcal}$, which have just one prefix at the top of the term, 
there are two different causes for non-determinism of the rewrite relations in $\Reg$ and $\stReg$:
First, since the left-hand sides of the rules $\srulep{\slappdecompi{0}}$ and $\srulep{\slappdecompi{1}}$
coincide, these rules enable different steps on the same term, producing the left- and respectively the right subterm
of the application immediately below the prefix. 
Second, the rules $\srulep{\scompress}$ and $\srulep{\snlvarsucc}$ can be applicable
in situations where also one of the rules $\srulep{\slappdecompi{0}}$, $\srulep{\slappdecompi{1}}$,
or $\srulep{\slabsdecomp}$ is applicable. 
Whereas the first kind of non-determinism is due to the `observer' having to observe the two different subterms
of an application in a $\lambda$\nb-term, the second is due to a freedom of the observer as to
when to attest the end of a scope (of some kind) in the analysed $\lambda$\nb-term.
\end{remark}

In the definition below we define strategies on $\RegARS$ and $\stRegARS$ that
resolves the second source of non-determinism while leaving the first kind
intact. 
As a result the sub-ARS induced by some term $\ater$ with respect to
$\stRegARS$ correspond structurally to the term graph of $\ater$.
%

\begin{definition}[scope/\extscopedelimiting\ strategy]%
  \label{def:scope:delimiting:strat:Reg:stReg}
  We call a strategy $\astrat$ for $\RegARS$ (for $\stRegARS$) 
  a \emph{\scopedelimiting\ strategy} (\emph{a \extscopedelimiting\ strategy})
  if the source of a step is non-deterministic (that is, it is the source of more than one step)
  if and only if it is the source of precisely a $\slappdecompired{0}$\nb-step and a $\slappdecompired{1}$\nb-step.
  
  For every such strategy~$\astrat$, we denote by
  $\slappdecompistratred{0}{\astrat}\,$, $\slappdecompistratred{1}{\astrat}\,$, $\slabsdecompstratred{\astrat}\,$,
  and  $\scompressregstratred{\astrat}$ ($\scompressstregstratred{\astrat}\,$), 
  the rewrite relations that are induced by those steps according to $\astrat$
  that are induced by applications of the rules
  $\srulep{\slappdecompi{0}}$, $\srulep{\slappdecompi{1}}$, $\srulep{\slabsdecomp}$,
  and $\srulep{\scompress}$ ($\srulep{\snlvarsucc}$), 
  respectively.
\end{definition}

\begin{remark}
  Note the following more verbose formulation of the condition for a strategy $\astrat$ for $\RegARS$ (for $\stRegARS$) 
  to be called a \scopedelimiting\ (\extscopedelimiting) strategy:
  \begin{itemize}
    \item every source of a step is one of three kinds: 
      the source of a
      $\slabsdecompred$-step, the source of a $\scompressregred$-step (a
      $\scompressstregred$\nb-step), or the source of both a $\slappdecompired{0}$\nb-step and a
    $\slappdecompired{1}$\nb-step with the restriction that (in all three
    cases) it is not the source of any other step.
  \end{itemize} 
  Mindful of the fact that sources of $\slabsdecompred$\nb-steps are not sources of $\slappdecompired{i}$\nb-steps, or vice versa,
  in $\RegARS$ (in $\stRegARS$),
  this condition can be relaxed to the equivalent formulation:  
  \begin{itemize}
    \item no source of a $\labsdecompred$\nb-step or a $\lappdecompired{i}$\nb-step for $i\in\setexp{0,1}$
      is also the source of a $\scompressregred$\nb-step (a $\scompressstregred$\nb-step),
      and every source of a $\lappdecompired{i}$\nb-step for $i\in\setexp{0,1}$
      is the source of both a $\lappdecompired{0}$- and a $\lappdecompired{1}$\nb-step, but not the source of any other step.
  \end{itemize}
\end{remark}

\begin{definition}[eager and lazy scope/\extscopedelimiting\ strategies]
  \label{def:eager:lazy:scopedelimiting:strategy:Reg:stReg}
  The \emph{eager \scopedelimiting\ strategy\/~$\eagscdelstratreg$}
      (\emph{lazy \scopedelimiting\ strategy~$\lazyscdelstratreg$})
  \emph{for $\RegARS$} 
  is defined as the restriction of rewrite steps in $\RegARS$ 
  to eager (and respectively, to lazy) application of the rule $\srulep{\scompress}$:
  on a term $\aflpter\in\Ter{\lambdaprefixcal}$, 
  applications of 
  other rules
  are only allowed if $\srulep{\scompress}$ is not applicable
  (applications of $\srulep{\scompress}$ are only allowed when other rules are not applicable). 
  %
  Analogously, 
  the \emph{eager \extscopedelimiting\ strategy\/~$\eagscdelstratstreg$}
              (the \emph{lazy  \extscopedelimiting\ strategy $\lazyscdelstratstreg$})
  \emph{for $\stRegARS$}    
  is defined as the restriction of rewrite steps in $\stRegARS$ 
  to eager (respectively, to lazy) application of the rule~$\srulep{\snlvarsucc}$. 
\end{definition}

\begin{figure}
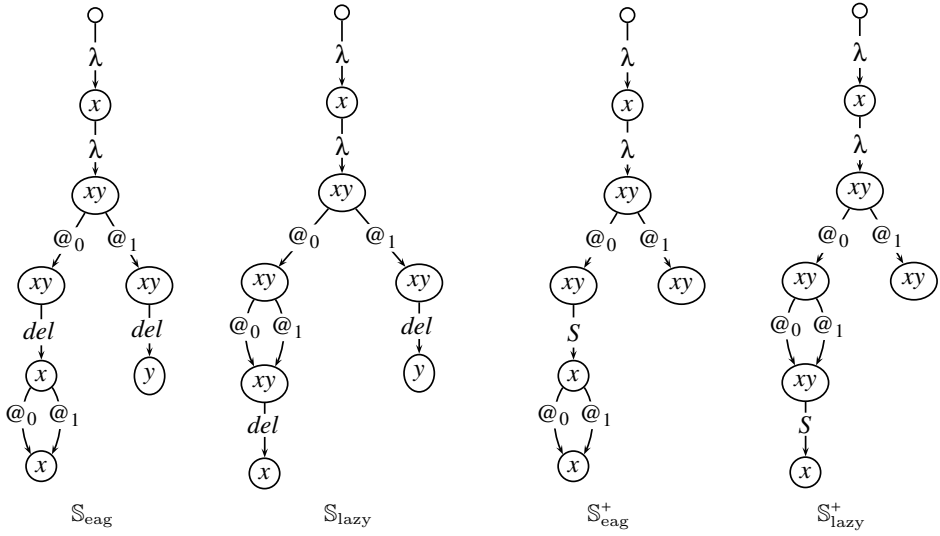

\begin{tabular}{cccc}
\fig{lxy_xxy-cpreager.ltg-l}&
\fig{lxy_xxy-cprlazy.ltg-l}&
\fig{lxy_xxy-eager.ltg-l}&
\fig{lxy_xxy-lazy.ltg-l}
\\
$\eagscdelstratreg$&
$\lazyscdelstratreg$&
$\eagscdelstratstreg$&
$\lazyscdelstratstreg$
\end{tabular}
\caption{\label{fig:four-strategies}
sub-ARSs induced by $\labs{x}{\labs{y}{\lapp{\lapp xx}{y}}}$ with respect to
different $\RegARS$ and $\stRegARS$ strategies; compare:
Figure~\ref{fig:three_redgraphs}. Again note, that the labels do only show the
prefixes associated with each term.}
\end{figure}

\begin{samepage}
\begin{remark}[history-aware versus history-free \scopedelimiting\ strategies]
The history-free strategy obtained by projection from a history-aware
\scopedelimiting\ strategy is not in general a \scopedelimiting\ strategy.
This is due to the non-determinism which may be introduced by the projection.
Consider for example the term $\lapp MM$ with $M=\labs{x}{\labs{y}{\lapp{\lapp
xx}{y}}}$ and the history-aware strategy $\astrat$ constructed by using
$\eagscdelstratreg$ on the left component and $\lazyscdelstratreg$ on the right
component of $\lapp MM$. The sub-ARS induced by $MM$ then corresponds to the
graphs of the sub-ARSs for $\eagscdelstratreg$ and $\lazyscdelstratreg$ as
depicted in Figure~\ref{fig:four-strategies} placed side by side with an
additional connecting node at the top (also some leafs are merged). The
induced sub-ARSs of the history-free strategy obtained by projection, however,
resembles the graph from Figure~\ref{fig:three_redgraphs} (with one
additional application node at the top). That graph however bears the
non-determinism which is not permitted for a scope-delimiting strategy, in the
form of the existence of a source of both a $\scompressstregred$ and a
$\slappdecompiregred{i}$ step.
\end{remark}
\end{samepage}

The following proposition formulates a property of the eager \scopedelimiting\ (\extscopedelimiting) strategy
in $\RegARS$ (in $\stRegARS$) that assigns it a special status:
the target of every rewrite sequence with respect to $\eagscdelstratreg$ (with respect to $\eagscdelstratstreg$)
can be reached, modulo some final $\scompressregred$\nb-steps ($\scompressstregred$\nbd-steps),
also by a rewrite sequence with respect to an arbitrary \scopedelimiting\ (\extscopedelimiting) strategy.
Furthermore, rewrite sequences with respect to $\eagscdelstratreg$ (with respect to $\eagscdelstratstreg$)
are able to mimic rewrite sequences with respect to an arbitrary \scopedelimiting\ (\extscopedelimiting) strategy,
up to trailing $\scompressregred$\nb-steps ($\scompressstregred$\nbd-steps) applied to the latter.

\begin{proposition}\label{prop:eagscdelstrat}
  For all \scopedelimiting\ strategies $\astrat$ on $\RegARS$,
  and \extscopedelimiting\ strategies $\astratplus$ on $\stRegARS$ the following statements hold:
  \begin{enumerate}[(i)]
    \item\label{prop:eagscdelstrat:item:i}
      $\sstratmred{\eagscdelstratreg}$ factors into $\sstratmred{\astrat}$ and $\scompressregmred\,$, and
      $\sstratmred{\eagscdelstratstreg}$ factors into $\sstratmred{\astratplus}$ and $\scompressstregmred$. 
    \item\label{prop:eagscdelstrat:item:ii} 
      $\sstratred{\eagscdelstratreg}$ is cofinal for $\sstratred{\astrat}$ with trailing $\scompressregred$\nb-steps, and
      $\sstratred{\eagscdelstratstreg}$ is cofinal for $\sstratred{\astratplus}$ with trailing $\scompressstregred$\nb-steps.
  \end{enumerate}
\end{proposition}


\begin{figure}
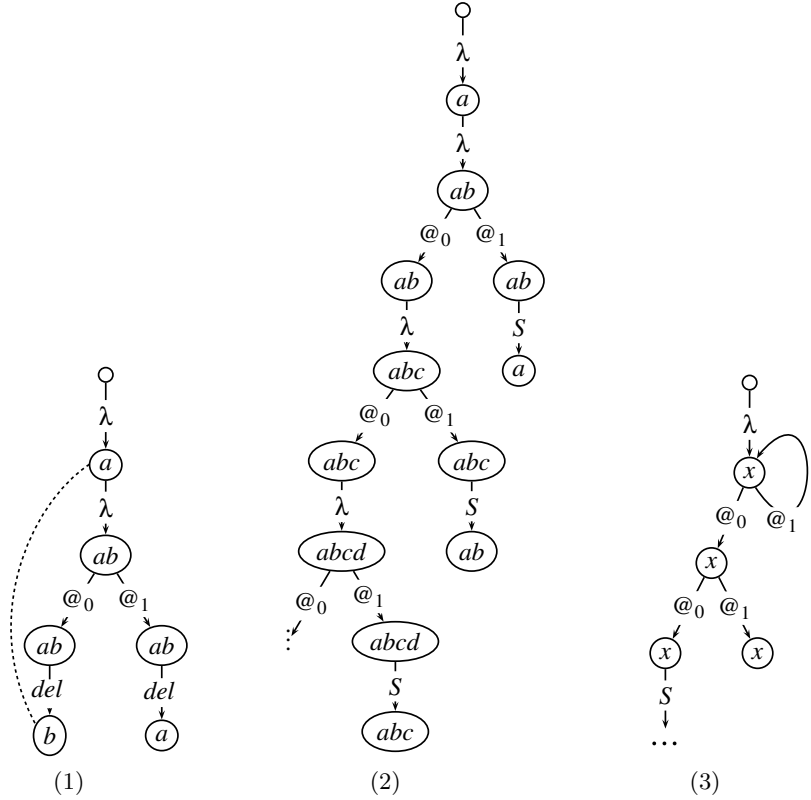

\begin{tabular}{ccc}
\hspace{0.7cm}
\fig{entangled-cpreager.ltg-l} &
\hspace{0.7cm}
\fig{pstricks/entangled.ltg-l} &
\hspace{0.7cm}
\fig{cons.ltg-l} \\
(1) & (2) & (3)
\end{tabular}
\caption{\label{fig:entangled-cons-ltg}%
Induced sub-ARS as graphs (only with prefixes as node labels) for:\protect\\
(1)~Example~\ref{ex:entangled} with the eager $\eagscdelstrat$ scope-delimiting strategy for $\RegARS$.
    The dotted line denotes node equality. The two connected nodes represent
    identical ($\alpha$-equivalent) terms. They are represented by two separate
    nodes instead of a single one to avoid confusion as to which variable (it is
    $a$) is deleted from the prefix by the adjacent $\scompress$-edge.\protect\\
(2)~Example~\ref{ex:entangled} with an arbitrary \extscope\nb-delimiting strategy for
    $\stRegARS$. The vertical dots denote an infinite growth of the graph.\protect\\
(3)~Example~\ref{ex:expresses} with an arbitrary \extscopedelimiting{}
    strategy for $\stRegARS$. The horizontal dots denote an unspecified continuation
    of the graph, which cannot be made explicit as the nature of the `$:$'-operator
    has been left open.}
\end{figure}

\begin{definition}[generated subterms of \inflambdaterms]\label{def:STred:Reg:stReg}
  Let $\astrat$ be a \scopedelimiting\ strategy for $\RegARS$ with rewrite relation~$\stratred{\astrat}$
  (let $\astratplus$ be a \scopedelimiting\ strategy for $\stRegARS$ with rewrite relation~$\stratred{\astratplus}$).  
  For every $\aiter\in\Ter{\inflambdacal}$, 
  the set $\gSTregstrat{\astrat}{\aiter}$ of \emph{generated subterms} of $\aiter$ with respect to~$\RegARS$ and~$\astrat$ 
  (the set $\gSTstregstrat{\astrat}{\aiter}$ with respect to~$\stRegARS$ and~$\astrat$) 
  is defined 
  as the set of $\stratmred{\astrat}$\nb-reducts
    (the set of $\stratmred{\astratplus}$\nb-reducts) of $\femptylabs{\aiter}$ via the mappings:
  \begin{align*}
    \sgSTregstrat{\astrat} \funin \Ter{\inflambdacal} & {} \longrightarrow \powersetof{\Ter{\inflambdaprefixcal}}
     &
    \hspace*{-2ex} 
    \sgSTstregstrat{\astratplus} \funin \Ter{\inflambdacal} & {} \longrightarrow \powersetof{\Ter{\inflambdaprefixcal}} 
    \\
    \aiter & {} \longmapsto \gSTregstrat{\astrat}{{\aiter}} 
                               \defdby
                            \redsuccs{\stratmred{\astrat}}{\femptylabs{\aiter}}   
    &
    \aiter & {} \longmapsto \gSTstregstrat{\astratplus}{{\aiter}} 
                               \defdby
                            \redsuccs{\stratmred{\astratplus}}{\femptylabs{\aiter}}   
  \end{align*}
\end{definition}

\begin{definition}\label{def:regularity:ARS}
  Let $\aARS$ be an abstract rewriting system, and $\astrat$ a strategy for $\aARS$.
  We say that an object $\aobj$ in $\aARS$ is \emph{$\astrat$\nb-regular in $\aARS$} 
  if the \ARS\ $\redsuccs{\stratmred{\astrat}}{\aobj}$ induced by $\aobj$ 
  (and consequently the set of $\sstratmred{\astrat}$\nb-reducts of $\aobj$ in $\aARS$) is finite. 
\end{definition}

\begin{definition}[regular and strongly regular infinite \lambdaterms]\label{def:reg:streg}
  An infinite \lambdaterm~$\aiter$ is called \emph{regular} (\emph{strongly regular})
  if there exists a \scopedelimiting\ strategy $\ascdelstratreg$ for $\RegARS$ 
                 (a \extscopedelimiting\ strategy $\ascdelstratstreg$ for $\stRegARS$)
  such that $\aiter$ is $\astrat$\nb-regular (is $\ascdelstratstreg$\nb-regular).
\end{definition}

  Note that an infinite \lambdaterm~$\aiter$ is regular (strongly regular) 
  if and only if
  the set $\gSTregstrat{\ascdelstratreg}{\aiter}$ (the set $\gSTstregstrat{\ascdelstratstreg}{\aiter}$) 
  of generated subterms of $\aiter$ with respect to some \scopedelimiting\ strategy~$\ascdelstratreg$ 
  (with respect to some \extscopedelimiting\ strategy~$\ascdelstratstreg$) 
  is finite. 

\begin{samepage}
\begin{proposition}\label{prop:def:reg:streg}
  The following statements hold:
  \begin{enumerate}[(i)] 
    \item\label{prop:def:reg:streg:item:i}
      Every strongly regular infinite \lambdaterm\ is also regular. 
    \item\protect\label{prop:def:reg:streg:item:ii}
      Finite $\lambda$\nb-terms are both regular and strongly regular.   
  \end{enumerate}    
\end{proposition}
\end{samepage}

\begin{proof}  
  For statement (\ref{prop:def:reg:streg:item:i}), 
  let $\aiter$ be an infinite \lambdaterm, 
  and let $\astrat$ be a \scopedelimiting\ strategy for $\stRegARS$ such that $\gSTstregstrat{\astrat}{\aiter}$ is finite.
  Due to Proposition~\ref{prop:rewseqs:stRegCRS:2:RegCRS}, 
  $\astrat$ can be modified with the result of a \scopedelimiting\ strategy $\astrat'$ for $\RegARS$
  by exchanging $\scompressstregred$\nb-steps with $\scompressregred$\nb-steps.  
  Then there is a stepwise correspondence between $\sstratred{\astrat}$-rewrite sequences
  and $\sstratred{\astrat'}$\nb-rewrite sequences that pass through the same terms. 
  Consequently, the sets of $\sstratmred{\astrat}$- and $\sstratmred{\astrat'}$\nb-reducts
  of $\femptylabs{\aiter}$ coincide: 
  $\gSTstrat{\astrat'}{\aiter} = \gSTstrat{\astrat}{\aiter}$. It follows that $\gSTstrat{\astrat'}{\aiter}$ is finite. 
   
  For statement (\ref{prop:def:reg:streg:item:ii}), note that
  by Proposition~\ref{prop:rewprops:RegCRS:stRegCRS}, (\ref{prop:rewprops:RegCRS:stRegCRS:item:v}),
  and K\H{o}nig's Lemma,
  every finite term in $\Ter{\lambdaprefixcal}$ 
  has only finitely many reducts with respect to $\sregmred$, or $\sstregmred$.
  It follows that for every finite \lambdaterm~$\aiter$ and every \scopedelimiting\ strategy $\astrat$
  on $\RegARS$, or on $\stRegARS$,
  the number of $\sstratmred{\astrat}$\nb-reducts of $\femptylabs{\aiter}$ 
  is finite, too. 
\end{proof}

The \scopedelimiting\ strategies in the above definition could be fixed to the respective eager versions
without changing the notions of regular and strongly regular infinite $\lambda$\nb-term.

\begin{proposition}\label{prop:eager:strat:in:def:reg:streg}
  For all infinite \lambdaterm{s} $\aiter$ the following statements hold:
  \begin{enumerate}[(i)]
    \item\label{prop:eager:strat:in:def:reg:streg:item:i}
      $\aiter$ is regular if and only if 
                                         $\aiter$ is $\eagscdelstratreg$\nb-regular. 
    \item\label{prop:eager:strat:in:def:reg:streg:item:ii} 
      $\aiter$ is strongly regular if and only if 
                                                  $\aiter$ is $\eagscdelstratstreg$\nb-regular.
  \end{enumerate}
\end{proposition}

\begin{proof}
  We only prove (\ref{prop:eager:strat:in:def:reg:streg:item:i}), because
  (\ref{prop:eager:strat:in:def:reg:streg:item:ii}) can be established analogously.
  The implication ``$\Leftarrow$'' in (\ref{prop:eager:strat:in:def:reg:streg:item:i})
  follows directly from the definition of regularity.  
  For showing ``$\Rightarrow$'', let $\aiter$ be an infinite \lambdaterm\ that is regular. 
  Then there exists a \scopedelimiting\ strategy~$\astrat$ so that $\gSTstrat{\astrat}{\aiter}$ is finite.
  Since by Proposition~\ref{prop:eagscdelstrat}, (\ref{prop:eagscdelstrat:item:i}),
  every $\sstratmred{\eagscdelstrat}$\nb-rewrite sequence factors into an 
  $(\binrelcomp{\sstratmred{\astrat}}{\scompressregmred})$-rewrite sequence, it follows that every term
  in $\gSTstrat{\eagscdelstratreg}{\aiter}$ is the $\scompressregmred$\nb-reduct of a term in $\gSTstrat{\astrat}{\aiter}$.
  As every term in $\Ter{\inflambdaprefixcal}$ has only finitely many $\scompressregmred$\nb-reducts,
  it follows that also $\gSTstrat{\eagscdelstratreg}{\aiter}$ is finite.
\end{proof}

\begin{example}[regular and strongly regular terms]\label{ex:reg-and-streg-terms}
The following examples demonstrate the connection between (strong) regularity
and, as illustrated in Figure \ref{fig:entangled-cons-ltg}, the finiteness of
the \ARS{s} induced by $\RegARS$ ($\stRegARS$) strategies.
\begin{itemize}
\item Example~\ref{ex:entangled} is regular but not strongly regular.
\item Example~\ref{ex:expresses} is strongly regular.
\end{itemize}
\end{example}

For further illustration of the statements made in
Example~\ref{ex:reg-and-streg-terms} let us consider various $\RegARS$ and
$\stRegARS$ rewrite sequences corresponding to infinite paths through the
terms.  

\begin{example}\label{ex:entangled-infinite-path}
For the term $\ater$ from Example~\ref{ex:entangled}, however we first
introduce a finite CRS-based notation, as a `higher-order recursive program scheme'. 
We can represent $\ater$ by
$\labs{a}{\funap{rec_M}{a}}$ together with the \CRS-rule
$\funap{rec_M}{\sametavar} \redb{rec_M}
\labs{\avar}{\lapp{\funap{rec_M}{\avar}}{\sametavar}}$. It holds that
$\labs{a}{\funap{rec_M}{a}} \infredb{rec_M} \ater$.
Using this notation we can finitely describe the infinite path down the spine
of Example~\ref{ex:entangled} by the cyclic $\eagscdelstrat$ rewrite sequence:
\[\begin{array}{ll}
                    & \femptylabs{\labs{a}{\funap{rec_M}{a}}}
\\\slabsdecompred   & \flabs{a}{\funap{rec_M}{a}}
\\\sredp{rec_M}     & \flabs{a}{\labs{b}{\lapp{\funap{rec_M}{b}}a}}
\\\slabsdecompred   & \flabs{ab}{\lapp{\funap{rec_M}{b}}a}
\\\slappdecompired0 & \flabs{ab}{\funap{rec_M}{b}}
\\\scompressregred  & \flabs{b}{\funap{rec_M}{b}}
\\=                 & \flabs{a}{\funap{rec_M}{a}}
\end{array}\]
In $\stRegARS$ the rewriting sequence for the same path is invariant over all
\extscope-delimiting strategies and necessarily infinite:
\[\begin{array}{llllll}
\femptylabs{\labs{a}{\funap{rec_M}{a}}}
& \slabsdecompred   & \flabs{a}{\funap{rec_M}{a}}
& \sredp{rec_M}    \\ \flabs{a}{\labs{b}{\lapp{\funap{rec_M}{b}}a}}
& \slabsdecompred   & \flabs{ab}{\lapp{\funap{rec_M}{b}}a}
& \slappdecompired0 & \flabs{ab}{\funap{rec_M}{b}}
& \sredp{rec_M}    \\ \flabs{ab}{\labs{c}{\lapp{\funap{rec_M}{c}}b}}
& \slabsdecompred   & \flabs{abc}{\lapp{\funap{rec_M}{c}}b}
& \slappdecompired0 & \flabs{abc}{\funap{rec_M}{c}}
& \sredp{rec_M}    \\ \flabs{abc}{\labs{d}{\lapp{\funap{rec_M}{d}}c}}
& \slabsdecompred   & \flabs{abcd}{\lapp{\funap{rec_M}{d}}c}
& \slappdecompired0 & \flabs{abcd}{\funap{rec_M}{d}}
& \sredp{rec_M}    \\ \flabs{abcd}{\labs{e}{\lapp{\funap{rec_M}{e}}d}}
& \slabsdecompred   & \flabs{abcde}{\lapp{\funap{rec_M}{e}}d}
& \dots
\end{array}\]
\end{example}

\begin{example}\label{ex:simpleletrec-infinite-path}
As an illustration of a regular term we study the $\ater$ defined as the
unfolding of $\letrec{\arecvar =
\labs{\avar\bvar}{\lapp{\lapp{\arecvar}{\bvar}}{\avar}}}{\arecvar}$ from
Example~\ref{ex:llunfCRS}.
It is strongly regular since the infinite path through the term can be witnessed by this
cyclic $\stRegARS$ rewriting sequence:

\[\begin{array}{ll}
                      \femptylabs\ater
&=                  \\ \femptylabs{\labs{\avar\bvar}{\lapp{\lapp\ater\bvar}{\avar}}}
&\slabsdecompred    \\ \flabs{\avar}{\labs{\bvar}{\lapp{\lapp\ater\bvar}{\avar}}}
&\slabsdecompred    \\ \flabs{\avar\bvar}{\lapp{\lapp\ater{\bvar}}{\avar}}
&\slappdecompired0  \\ \flabs{\avar\bvar}{\lapp\ater{\bvar}}
&\slappdecompired0  \\ \flabs{\avar\bvar}\ater
&\scompressstregred \\ \flabs{\avar}\ater
&\scompressstregred \\ \femptylabs\ater
&\dots
\end{array}\]
See also Figure~\ref{fig:distance} for a graphical illustration of the induced sub-ARS.
\end{example}

\begin{figure}
\begin{tabular}{cc}
\hspace{1cm}\fig{simpleletrec-eager.ltg-l}\hspace{1cm} & \hspace{1cm}\fig{pstricks/simpleletrec-lazy.ltg-l}\hspace{1cm} \\
$\eagscdelstratstreg$ & $\lazyscdelstratstreg$
\end{tabular}
\caption{
\label{fig:distance}
Two $\stRegCRS$ sub-ARSs induced by the infinite unfolding of
$\letrec{\arecvar = \labs{\avar\bvar}{\lapp{\lapp{\arecvar}{\bvar}}{\avar}}}{\arecvar}$.}
\end{figure}

\begin{remark}
The restriction of Proposition~\ref{prop:eager:strat:in:def:reg:streg} to the
eager \scopedelimiting{} strategy cannot be relaxed to arbitrary
\scopedelimiting{} strategies. The term in
Example~\ref{ex:simpleletrec-infinite-path} for instance is
$\eagscdelstratstreg$\nb-regular but not $\lazyscdelstratstreg$\nb-regular
(Figure~\ref{fig:distance}).
\end{remark}

\begin{definition}[grounded cycles in $\RegARS$, $\stRegARS$]\label{def:grounded:cycle}
Let 
$ \arewseq \funin
  \flabs{\vecsub{\avar}{0}}{\aiteri{0}}
    \red
  \flabs{\vecsub{\avar}{1}}{\aiteri{1}}
    \red
  \ldots $
be a finite or infinite rewrite sequence with respect to $\sregred$ or $\sstregred$. 
By a \emph{grounded cycle} in $\arewseq$ we mean
a cycle 
$ \flabs{\vecsub{\avar}{i}}{\aiteri{i}}
    \red
  \flabs{\vecsub{\avar}{i+1}}{\aiteri{i+1}}
    \red
  \ldots
     \red
  \flabs{\vecsub{\avar}{i+k}}{\aiteri{i+k}} = \flabs{\vecsub{\avar}{i}}{\aiteri{i}} $,
in $\arewseq$, where $i\in\nats$ and $k\ge 1$, 
with the additional property that $\length{\vecsub{\avar}{i+j}} \ge \length{\vecsub{\avar}{i}}$ for all $j\in\setexp{0,1,\ldots,k}$
(i.e.\ the lengths of the abstraction prefixes in the terms of the cycle is greater or equal to
 the length of the abstraction prefix at the first and final term of the cycle).
\end{definition}

\begin{proposition}\label{prop:grounded:cycle}
  Let $\aiter$ be an infinite \lambdaterm\ that is $\astrat$\nb-regular ($\astratplus$\nb-regular) for some 
  \scopedelimiting\ strategy~$\astrat$ (\extscopedelimiting\ strategy $\astratplus$). 
  Then every infinite rewrite sequence with respect to $\astrat$ (with respect to $\astratplus$) contains a grounded cycle.
\end{proposition}

\begin{proof}
  Since the argument is analogous in both cases, we only treat strongly regular terms. 
  Let $\aiter$ be an infinite \lambdaterm\ that is $\astratplus$\nb-regular for some \extscopedelimiting\ strategy~$\astratplus$, 
  and let 
  $ \arewseq \funin
  \aiter 
    = 
  \flabs{\vecsub{\avar}{0}}{\aiteri{0}}
    \stratred{\astratplus}
  \flabs{\vecsub{\avar}{1}}{\aiteri{1}}
    \stratred{\astratplus}
  \ldots $
  be an infinite rewrite sequence.
  As $\aiter$ is $\astratplus$\nb-regular, 
  the sequence $\sequence{\flabs{\vecsub{\avar}{i}}{\aiteri{i}}}{i\in\nats}$ of terms on $\arewseq$ contains only finitely many different terms.  
  Let $ l \defdby \liminf \sequence{\length{\vecsub{\avar}{i}}}{i\in\nats}$, that is, the minimum of abstraction prefix lengths 
  that appears infinitely often on $\arewseq$. Let $\sequence{\flabs{\vecsub{\avar}{i_j}}{\aiteri{i_j}}}{j\in\nats}$
  be the subsequence of $\sequence{\flabs{\vecsub{\avar}{i}}{\aiteri{i}}}{i\in\nats}$
  consisting of terms with prefix length $l$, and such that, for all $k\ge i_0$,
  $\length{\vecsub{\avar}{k}} \ge l$.
  Since also this subsequence 
  contains only finitely many terms, there exist $j_1,j_2\in\nats$, $j_1 < j_2$ such that 
  $ \flabs{\vecsub{\avar}{i_{j_1}}}{\aiteri{i_{j_1}}} = \flabs{\vecsub{\avar}{i_{j_2}}}{\aiteri{i_{j_2}}} $.
  By the choice of the subsequence it follows that
  $ \flabs{\vecsub{\avar}{i_{j_1}}}{\aiteri{i_{j_1}}} 
      \stratred{\astratplus}
    \ldots
      \stratred{\astratplus}
    \flabs{\vecsub{\avar}{i_{j_2}}}{\aiteri{i_{j_2}}} $
  is a grounded cycle in $\arewseq$. 
\end{proof}

Rounding off this section, we provide a motivation for the system~$\stRegARS$
in terms of an operation `parse' that at the same time 
(i)~decomposes an infinite \lambdaterm\ into its generated subterms, and 
(ii)~recombines the generated subterms encountered in the decomposition analysis
     with, in the limit, the original term as the result. 
With this purpose in mind, we define a \CRS~$\stParseCRS$.

\begin{definition}[$\stParseCRS$]\label{def:stParseCRS}
  Let $\siglpcparseCRS = \siglpcCRS \cup \descsetexp{ \sparsen{n} }{ n\in\nats }$ 
  be the extension of the \CRS\nb-sig\-na\-ture for $\lambdaprefixcal$, where 
  for $n\in\nats$, the symbols $\sparsen{n}$ have arity~$n$.
  By $\stParseCRS$ we denote the \CRS\ with the following rules:
\begin{align*}
    (\srulebp{\sparse}{@}):
      & \hspace*{1ex}
      \begin{aligned}[t]
        \parsen{n}{\vecsub{\sametavar}{n}}{\flabsCRS{n}{\vecsub{\avar}{n}}{\lappCRS{\cmetavari{0}{\vecsub{\avar}{n}}}{\cmetavari{1}{\vecsub{\avar}{n}}}}}  
          & {} \;\red\;
        \\[-0.5ex]
        & \hspace*{-35ex}  
        \lappCRS{\parsen{n}{\vecsub{\sametavar}{n}}{\flabsCRS{n}{\vecsub{\avar}{n}}{\cmetavari{0}{\vecsub{\avar}{n}}}}}{\parsen{n}{\vecsub{\sametavar}{n}}{\flabsCRS{n}{\vecsub{\avar}{n}}{\cmetavari{1}{\vecsub{\avar}{n}}}}}  
      \end{aligned}
    \\[0.5ex]
    (\srulebp{\sparse}{\slabs}):
      & \hspace*{1ex}
      \begin{aligned}[t]
        \parsen{n}{\vecsub{\sametavar}{n}}{\flabsCRS{n}{\vecsub{\avar}{n}}{\labsCRS{\bvar}{\cmetavar{\vecsub{\avar}{n},\bvar}}}}
          & {} \;\red\;
        \\[-0.5ex]
        & \hspace*{-13ex}  
        \labsCRS{\bvar}{\parsen{n+1}{\vecsub{\sametavar}{n},\bvar}{\flabsCRS{n+1}{\vecsub{\avar}{n}}{\absCRS{\bvar}{\cmetavar{\vecsub{\avar}{n},\bvar}}}}}
      \end{aligned}
    \\[0.5ex]
    (\srulebp{\sparse}{\scompressstreg}):
      & \hspace*{1ex}
      \begin{aligned}[t]
        \parsen{n+1}{\sametavari{1},\ldots,\sametavari{n+1}}{\funap{\sflabsCRS{n+1}}{\absCRS{\vecsub{\avar}{n+1}}{\cmetavar{\vecsub{\avar}{n}}}}}
        & {} \;\red\;
        \\[-0.5ex]
        & \hspace*{-5.9ex}
       \parsen{n}{\sametavari{1},\ldots,\sametavari{n}}{\funap{\sflabsCRS{n}}{\absCRS{\vecsub{\avar}{n}}{\cmetavar{\vecsub{\avar}{n}}}}}
      \end{aligned}
    \\[0.5ex]  
    (\srulebp{\sparse}{\nlvar}):
      & \hspace*{1ex}
      \parsen{n}{\sametavari{1},\ldots,\sametavari{n}}{\flabsCRS{n}{\vecsub{\avar}{n}}{\avari{n}}}
        \;\red\;
      \sametavari{n}
\end{align*}
  We denote by $\sparsered$ the rewrite relation induced by this \CRS.
\end{definition}

\begin{remark}
Observe that the rules ($\srulep{\slappdecompi{i}}$) for $i\in\setexp{0,1}$, 
                           ($\srulep{\slabsdecomp}$), and ($\srulep{\scompressstreg}$)
of the \CRS~$\stRegCRS$ are contained within the rules 
                           ($\srulebp{\sparse}{@}$), ($\srulebp{\sparse}{\slabs}$), ($\srulebp{\sparse}{\Vacstreg}$), respectively,
of the \CRS~$\stParseCRS$,
in the sense that applications of the latter rules include applications of the former.

This has as a consequence that repeated $\sparsered$\nb-steps on a term $\femptylabs{\aiter}$  
lead to terms that contain generated subterms of $\aiter$ as closed subexpressions. Furthermore 
$\sparsered$\nb-rewrite sequences on $\femptylabs{\aiter}$ are possible that move redexes simultaneously
deeper and deeper, analysing ever larger parts of $\aiter$, and at the same time recreating a larger and larger  
\lambdaterm\ parts (stable prefix contexts) of $\aiter$, the original term.   
\end{remark}

\begin{proposition}\label{prop:stParseCRS}
   For every term $\aiter\in\Ter{\lambdacal}$ it holds:
  \begin{enumerate}[(i)]
    \item{}\label{prop:stParseCRS:item:i} 
      $\stParseCRS$ analyses $\aiter$ into its generated subterms:
      If $\parseempty{\femptylabs{\aiter}} \parsemred \aiter'$, then all subexpressions starting with $\sflabsCRS{n}$ (for some $n\in\nats$) in $\aiter'$
      are generated subterms of $\aiter$.
      Moreover, for every generated subterm $\flabs{\vec{\bvar}}{\biter}$ of $\aiter$, there exists a 
      $\parsemred$\nb-reduction $\aiter''$ of $\aiter$ such that $\flabs{\vec{\bvar}}{\biter}$ is contained in $\aiter''$.
    \item{}\label{prop:stParseCRS:item:ii} 
      $\stParseCRS$ reconstructs $\aiter$:
      $\parseempty{\femptylabs{\aiter}} \parseinfred \aiter$. 
  
  \end{enumerate}
\end{proposition}
 
\begin{example}\label{ex:stParseCRS}
Let $\ater$ be the infinite unfolding of $\letrec{\arecvar =
\labs{\avar\bvar}{\lapp{\lapp{\arecvar}{\bvar}}{\avar}}}{\arecvar}$ for which
we use as a finite representation the equation
$\ater = \labs{\avar\bvar}{\lapp{\lapp{\ater}{\bvar}}{\avar}}$.
In $\stParseCRS$, $\ater$ is decomposed, and composed again, by the infinite rewrite sequence (see also Example~\ref{ex:llunfCRS}):
\begin{align*} &
\parseempty{\femptylabs{\ater}}
\alignbreak = &
\parseempty{\femptylabs{\labs\avar{\labs\bvar{\lapp{\lapp{\ater}{\bvar}}{\avar}}}}}
\alignbreak\sparselabsred &
\labs{\avar'}{\parsen1{\avar'}{\flabs{\avar'}{\labs\bvar{\lapp{\lapp{\ater}{\bvar}}{\avar'}}}}}
\alignbreak\sparselabsred &
\labs{\avar'}{\labs{\bvar'}{\parsen2{\avar',\bvar'}{\flabs{\avar'\bvar'}{\lapp{\lapp{\ater}{\bvar'}}{\avar'}}}}}
\alignbreak\sparselappred &
\labs{\avar'}{\labs{\bvar'}{\lapp
  {\parsen2{\avar',\bvar'}{\flabs{\avar'\bvar'}{\lapp{\ater}{\bvar'}}}}
  {\parsen2{\avar',\bvar'}{\flabs{\avar'\bvar'}{\avar'}}}}}
\alignbreak\sparsecompressstregred &
\labs{\avar'}{\labs{\bvar'}{\lapp
  {\parsen2{\avar',\bvar'}{\flabs{\avar'\bvar'}{\lapp{\ater}{\bvar'}}}}
  {\parsen1{\avar'}{\flabs{\avar'}{\avar'}}}}}
\alignbreak\sparsenlvarred &
\labs{\avar'}{\labs{\bvar'}{\lapp
  {\parsen2{\avar',\bvar'}{\flabs{\avar'\bvar'}{\lapp{\ater}{\bvar'}}}}
  {\avar'}}}
\alignbreak\sparselappred &
\labs{\avar'}{\labs{\bvar'}{\lapp
  {\lapp
    {\parsen2{\avar',\bvar'}{\flabs{\avar'\bvar'}{\ater}}}
    {\parsen2{\avar',\bvar'}{\flabs{\avar'\bvar'}{\bvar'}}}}
  {\avar'}}}
\alignbreak\sparsenlvarred &
\labs{\avar'}{\labs{\bvar'}{\lapp
  {\lapp
    {\parsen2{\avar',\bvar'}{\flabs{\avar'\bvar'}{\ater}}}
    {\bvar'}}
  {\avar'}}}
\alignbreak\sparsecompressstregred &
\labs{\avar'}{\labs{\bvar'}{\lapp
  {\lapp
    {\parsen1{\avar'}{\flabs{\avar'}{\ater}}}
    {\bvar'}}
  {\avar'}}}
\alignbreak\sparsecompressstregred &
\labs{\avar'}{\labs{\bvar'}{\lapp
  {\lapp
    {\parseempty{\femptylabs{\ater}}}
    {\bvar'}}
  {\avar'}}}
\alignbreak= &
\labs{\avar'}{\labs{\bvar'}{\lapp
  {\lapp
    {\parseempty{\femptylabs{\labs\avar{\labs\bvar{\lapp{\lapp{\ater}{\bvar}}{\avar}}}}}}
    {\bvar'}}
  {\avar'}}}
\alignbreak\sparselabsred &
\ldots  
\end{align*}
Note that the generated subterms of $\aiter$ appear as the last arguments of the 
functions $\sparsen{i}$ in this rewrite sequence.
\end{example}

\section{Observing \lambdaletrec-terms by their generated subterms}
  \label{sec:letrec2}
In this section we adapt the concepts developed so far for the infinitary \lambda-calculus
to \lambdaletrec{}. By combining the rules of
$\llunfCRS$ with those of $\RegCRS$ and $\stRegCRS$, respectively, we
obtain the \CRS{s} $\RegletrecCRS$ and $\stRegletrecCRS$ for the deconstruction of
\lambdaletrec-terms furnished with an abstraction prefix.
We define scope-delimiting and \extscope-delimiting strategies for
$\RegletrecCRS$ and $\stRegletrecCRS$ as before by 
                                                   excluding all non-determinism
except for sources of $\slappdecompired{0}$- and $\slappdecompired{1}$\nb-steps.

\begin{definition}[The \CRS{s} $\RegCRS$ and $\stRegCRS$ for decomposing $\lambdaletrec$-terms]\label{def:RegletrecCRS:stRegletrecCRS}
We extend $\sigllcCRS$ (see
Definition~\ref{def:sigs:lambdacal:lambdaletrec:CRS}) by function symbols
$\sflabsCRS{n}$ with arity one to obtain the signature $\sigllpcCRS =
\sigllcCRS \cup \descsetexp{\sflabsCRS{n}}{n\in\nats}$ on which we define the
\lambdaletrec-calculus with abstraction prefix. We denote the induced
set of (finite) terms by $\Ter{\lambdaletrecprefixcal}$ and we adopt the same informal
notation for \lambdaletrecprefixcal{} as for \inflambdaprefixcal{} (see
Definition~\ref{def:sig:lambdaprefixcal:CRS}).

On the signature $\sigllcCRS$ we define the \CRS{} $\RegletrecCRS$ (the
\CRS~$\stRegletrecCRS$) with the rules as the union of the rules of
$\llunfCRS$ and $\RegCRS$ (the union of the rules of $\llunfCRS$ and
$\stRegCRS$).
\end{definition}

\begin{definition}[The \ARSs~${\normalfont \RegletrecARS}$, ${\normalfont \stRegletrecARS}$]%
          \label{RegletrecARS:stRegletrecARS:oneRegletrecARS:onestRegletrecARS}%
We denote by $\RegletrecARS$ and
$\stRegletrecARS$ the \emph{abstract rewriting systems (\ARS{s})} induced by
the \CRS{s} $\RegletrecCRS$ and $\stRegletrecCRS$,
respectively. By $\RegletrecARS$, and
$\stRegletrecARS$ we denote the \ARS{s} that result by restricting the
\ARS{s}~$\RegletrecARS$, and $\stRegletrecARS$,
respectively, to the subset $\Ter{\lambdaletreccal}$ of terms.
  
Like in Definition~\ref{RegARS:stRegARS:oneRegARS:onestRegARS} the
symbols used for the rewrite relations of $\RegletrecCRS$ and $\stRegletrecCRS$
are overloaded; we use the same symbols to denote the restrictions of said
rewrite relations to $\Ter{\lambdaletrecprefixcal}$.
\end{definition}

The lemma below formulates a number of simple rewrite properties for
$\RegletrecARS$ and $\stRegletrecARS$, mainly concerning the interplay
between unfolding and decomposition steps.

\begin{lemma}\label{prop:rewprops:RegletrecCRS:stRegletrecCRS}
  On $\Ter{\lambdaletrecprefixcal}$, the rewrite relations in $\RegletrecCRS$ and $\stRegletrecCRS$
  have the following properties: 
  \begin{enumerate}[(i)]
    \item{}\label{prop:rewprops:RegletrecCRS:stRegletrecCRS:item:i}
      $\sunfoldred$ one-step commutes with each of
      $\slabsdecompred$, $\slappdecompired{0}$, $\slappdecompired{1}$, $\scompressstregred$, and $\scompressregred$:
      \begin{align*}
        \binrelcomp{\sunfoldconvred}{\slabsdecompred} 
          & \;\subseteq\; 
        \binrelcomp{\slabsdecompred}{\sunfoldconvred}
        & 
        \binrelcomp{\sunfoldconvred}{\slappdecompired{i}} 
          & \;\subseteq\; 
        \binrelcomp{\slappdecompired{i}}{\sunfoldconvred}
           \;\;\;\;\; (i\in\setexp{0,1})
        \\   
        \binrelcomp{\sunfoldconvred}{\scompressstregred} 
          & \;\subseteq\; 
        \binrelcomp{\scompressstregred}{\sunfoldconvred}
        &   
        \binrelcomp{\sunfoldconvred}{\scompressregred} 
          & \;\subseteq\; 
        \binrelcomp{\scompressregred}{\sunfoldconvred}
      \end{align*}
    \item\label{prop:rewprops:RegletrecCRS:stRegletrecCRS:item:ii}
      $\sregred$ in $\RegletrecARS$ and $\stregred$ in $\stRegletrecARS$
      have the same normal forms as $\sregred$ in $\RegARS$ and $\stregred$ in $\stRegARS$, respectively:
      $\flabs{\avar}{\avar}$ is the single term of $\lambdaletrecprefixcal$ in $\sregred$\nb-normal form.
      Every $\sstregred$\nb-normal form a term in $\lambdaletrecprefixcal$ is of the form
      $\flabs{\avari{1}\ldots\avari{n}}{\avari{n}}$.
  \end{enumerate}
\end{lemma}

\begin{proof}
  The commutation properties in (\ref{prop:rewprops:RegletrecCRS:stRegletrecCRS:item:i})
  are easy to verify by analysing the behaviour of the rewrite rules in $\RegletrecARS$ and in $\stRegletrecARS$ 
  on the terms of $\Ter{\lambdaletrecprefixcal}$.
  
  The statement in (\ref{prop:rewprops:RegletrecCRS:stRegletrecCRS:item:ii}) follows from
  Proposition~\ref{prop:rewprops:RegCRS:stRegCRS}, (\ref{prop:rewprops:RegCRS:stRegCRS:item:i}),
  and (\ref{prop:rewprops:RegCRS:stRegCRS:item:v}):
  Normal forms with respect to $\sregred$ in $\RegletrecARS$ and $\stregred$ in $\stRegletrecARS$
  can only be \lambdaterms\ without occurrences of $\sletrec$,
  since every occurrence of $\sletrec$ in a \lambdaletrecterm\ gives rise to a $\sunfoldred$\nb-redex.
\end{proof}

\begin{lemma}\label{lem:commute:unfoldomegared:stregred}
  The rewrite relation $\sunfoldomeganfred$ (to infinite $\sunfoldred$\nb-normal form in at most $\omega$ steps) 
  one-step commutes with $\slabsdecompred$, $\slappdecompired{0}$, $\slappdecompired{1}$, $\scompressstregred$, and $\scompressregred$:
      \begin{align*}
        \binrelcomp{\sunfoldconvomeganfred}{\slabsdecompred} 
          & \;\subseteq\; 
        \binrelcomp{\slabsdecompred}{\sunfoldconvomeganfred}
        & 
        \binrelcomp{\sunfoldconvomeganfred}{\slappdecompired{i}} 
          & \;\subseteq\; 
        \binrelcomp{\slappdecompired{i}}{\sunfoldconvomeganfred}
           \;\;\;\;\; (i\in\setexp{0,1})
        \\   
        \binrelcomp{\sunfoldconvomeganfred}{\scompressstregred} 
          & \;\subseteq\; 
        \binrelcomp{\scompressregred}{\sunfoldconvomeganfred}
        &   
        \binrelcomp{\sunfoldconvomeganfred}{\scompressregred} 
          & \;\subseteq\; 
        \binrelcomp{\scompressregred}{\sunfoldconvomeganfred}
      \end{align*}
      This implies, for prefixed terms that have unfoldings that:
      \begin{gather*}
        \begin{aligned}
          \Unf{\flabs{\vec{\avar}}{\labs{\bvar}{\allteri{0}}}}
             & {} \labsdecompred
          \Unf{\flabs{\vec{\avar}\bvar}{\allteri{0}}}
            & \hspace*{5ex}   
          \Unf{\flabs{\vec{\avar}}{\lapp{\allteri{0}}{\allteri{1}}}}
             & {} \lappdecompired{i}
          {\Unf{\flabs{\vec{\avar}}{\allteri{i}}}}
          \;\;\;\;\; (i\in\setexp{0,1})
        \end{aligned}
        \displaybreak[0]\\
        \flabs{\avari{1}\ldots\avari{n+1}}{\aiter} 
          \compressstregred
        \flabs{\avari{1}\ldots\avari{n}}{\aiter}
          \;\;\Rightarrow\;\;
        \Unf{\flabs{\avari{1}\ldots\avari{n+1}}{\aiter}} 
          \compressstregred
        \Unf{\flabs{\avari{1}\ldots\avari{n}}{\aiter}}
        \displaybreak[0]\\
        \flabs{\vec{\avar}}{\aiter} 
          \compressregred
        \flabs{\vec{\avar}'}{\aiter}
          \;\;\Rightarrow\;\;
        \Unf{\flabs{\vec{\avar}}{\aiter}} 
          \compressregred
        \Unf{\flabs{\vec{\avar}'}{\aiter}}
      \end{gather*} 
  Furthermore it holds:      
      $
        \binrelcomp{\sunfoldconvomeganfred}{\sunfoldred} 
          \;\subseteq\; 
        \sunfoldconvomeganfred
      $
      .
\end{lemma}

\begin{proof}
  The commutation properties with the rewrite relation $\sunfoldomeganfred$ 
  can be shown by using refined versions of the commutation properties
  in Lemma~\ref{prop:rewprops:RegletrecCRS:stRegletrecCRS}, (\ref{prop:rewprops:RegletrecCRS:stRegletrecCRS:item:i}),
  in which the minimal depth of unfolding steps is taken account of.
  When denoting by $\sunfolddepthred{\ge n}$ the rewrite relation that is generated by $\sunfoldred$\nb-steps
  of depth $\ge n$, then the following properties hold:
  \begin{align*}  
        \binrelcomp{\sunfolddepthconvred{\ge n}}{\slabsdecompred} 
          & \;\subseteq\; 
        \binrelcomp{\slabsdecompred}{\sunfolddepthconvred{\ge n}}
        & 
        \binrelcomp{\sunfolddepthconvred{\ge n+1}}{\slappdecompired{i}} 
          & \;\subseteq\; 
        \binrelcomp{\slappdecompired{i}}{\sunfolddepthconvred{\ge n}}
        \\   
        \binrelcomp{\sunfolddepthconvred{\ge n+1}}{\scompressstregred} 
          & \;\subseteq\; 
        \binrelcomp{\scompressstregred}{\sunfolddepthconvred{\ge n}}
        &   
        \binrelcomp{\sunfolddepthconvred{\ge n+1}}{\scompressregred} 
          & \;\subseteq\; 
        \binrelcomp{\scompressregred}{\sunfolddepthconvred{\ge n}}
  \end{align*}     
  Using these properties, strongly convergent $\sunfoldred$\nb-rewrite sequences 
  can be shown to project, over $\slabsdecompred$-, $\slappdecompired{i}$-, $\scompressregred$-, and $\scompressstregred$-steps,
  to again strongly convergent $\sunfoldred$\nb-rewrite sequences.
       
  The property     
      $
      \binrelcomp{\sunfoldconvomeganfred}{\sunfoldred} 
        \;\subseteq\; 
      \sunfoldconvomeganfred
      $
  can be shown by using refined versions of the elementary diagrams from the confluence proof that
  take the the minimal depths of steps into account. 
\end{proof}

As for $\RegARS$ ($\stRegARS$) we require of scope-delimiting
strategies to have deterministic $\scompressregred$-steps
($\scompressstregred$\nb-steps). 
As was the case for \scopedelimiting\ and \extscopedelimiting\ strategies for $\RegARS$ and $\stRegARS$,
these strategies will also here fix all
non-determinism except for the choice between $\slappdecompi0$ and $\slappdecompi1$.

\begin{definition}[scope/\extscopedelimiting\ strategy for 
                                                           \lambdaletrecterms]%
  \label{def:scope:delimiting:strat:Regletrec:stRegletrec}
  A strategy $\astrat$ for $\RegletrecARS\,$(for $\stRegletrecARS$) will be called
  a \emph{\scopedelimiting} (\emph{\extscopedelimiting}) \emph{strategy} if:
\begin{itemize}
  \item $\astrat$ is deterministic for sources of $\slabsdecompred$-steps,
  $\scompressregred$-steps ($\scompressstregred$\nb-steps), and all
  \sletrec-unfolding steps (i.e.\ all $\srred$-steps for every rule
  $\srulebp{\sunf}{r}$ of $\llunfCRS$).
 \item $\astrat$ enforces eager application of $\srulebp{\sunf}{\sreduce}$:
   every source of a step in $\astrat$ according to a rule different from $\srulebp{\sunf}{\sreduce}$
   is not the source
   of a $\sredp\sreduce$-step.
\end{itemize}

%

  We say that such a strategy~$\astrat$ is a \emph{lazy-unfolding} \scopedelimiting\ strategy 
  (a \emph{lazy-unfolding} \extscopedelimiting\ strategy)
  if furthermore:
  \begin{itemize}                 
    \item $\astrat$ applies the rules of $\llunfredCRS$ only at the root of the
    term, i.e.\ directly beneath the abstraction prefix except for
    $\srulebp{\sunf}{\sreduce}$.
    \item $\astrat$ uses the rules of $\llunfCRS$ other than 
    reduction rule $\srulebp{\sunf}{\sreduce}$ in a lazy way: for every source
    $s$ of a step in $\astrat$ with respect to a rule of $\llunfCRS$ other than
    $\srulebp{\sunf}{\sreduce}$ it holds that $s$ is not also the source of a
    step, in the \ARS\ underlying the strategy,
    with respect to one of the rules of $\RegARS$ (of $\stRegARS)$ 
  \end{itemize}    
  
  For every \scopedelimiting~strategy~$\astrat$ on $\RegletrecARS$ (on $\stRegletrecARS$), we denote by
  $\slappdecompistratred{0}{\astrat}\,$, $\slappdecompistratred{1}{\astrat}\,$, $\slabsdecompstratred{\astrat}\,$,
  $\scompressstregstratred{\astrat}\,$, $\scompressregstratred{\astrat}\,$,
  and by $\sunflappstratred{\astrat}\,$, $\sunflabsstratred{\astrat}\,$, 
  $\sunfnilstratred{\astrat}\,$, $\sunfrecstratred{\astrat}\,$, $\sunfletrecstratred{\astrat}\,$, and $\sunfreducestratred{\astrat}$
  the rewrite relations that are induced by those steps according to $\astrat$
  that result by applications of the rules 
  $\srulep{\slappdecompi{0}}$, $\srulep{\slappdecompi{1}}$, $\srulep{\slabsdecomp}$,
  $\srulep{\snlvarsucc}$, $\srulep{\scompress}$ of $\RegCRS$ and $\stRegCRS$,
  and the rules  
  $\srulebp{\sunf}{@}\,$, $\srulebp{\sunf}{\slabs}\,$, $\srulebp{\sunf}{\sunfnil}\,$,
  $\srulebp{\sunf}{\text{rec}}\,$, $\srulebp{\sunf}{\smergeletrec}\,$, and $\srulebp{\sunf}{\sreduce}\,$
  of $\llunfCRS$, respectively.
\end{definition}

\begin{remark}
We need to permit the application of $\srulebp{\sunf}{\sreduce}$ anywhere
inside the term to handle terms that contain inaccessible bindings in binding groups. 
Otherwise the possibility to apply $\srulep\scompress$ ($\srulep\snlvarsucc$) 
may be blocked in the case of a prefixed \lambdaletrecterm~$\flabs{\vec{\avar}}{\allter}$
in which a variable from the abstraction prefix is
vacuous with respect to the infinite unfolding of $\flabs{\vec{\avar}}{\allter}$
but is bound by a term in an inaccessible binding of some binding group in $\allter$.   

We can however restrict applicability of $\srulebp{\sunf}{\sreduce}$ to 
outermost redexes if we start unfolding on reduced terms, that is, $\sunfreducered$\nb-normal forms, 
because then $\srulebp{\sunf}{\sreduce}$-redexes can only
arise at outermost positions (by application of the $\srulebp{\sunf}{@}$-rule).
\end{remark}

\begin{remark}[non-deterministic unfolding]\label{rem:nondet:unfolding}
Note that in Definition~\ref{def:scope:delimiting:strat:Regletrec:stRegletrec} we do not
only require of a strategy to eliminate the non-determinism with respect to
$\srulep{\scompress}$-transitions ($\srulep\snlvarsucc$-transition) but all
non-determinism except for $\srulep{\slappdecompi0}$/$\srulep{\slappdecompi1}$
non-determinism. 
This restriction will play a role later for the definition of \lambdatg{s}
in Section~\ref{sec:ltgs}, and here below for the definition of the
projection of \scopedelimiting\ (\extscopedelimiting) strategies on \lambdaletrecterms\
to \scopedelimiting\ (\extscopedelimiting) strategies on infinite \lambdaterms.
\end{remark}

\begin{remark}[eager application of $\srulebp{\sunf}{\sreduce}$]
By requiring scope/\extscope{} delimiting strategies to apply
$\srulebp{\sunf}{\sreduce}$ eagerly we can exploit a useful property with
respect to free variables of a term: if $\flabs{\vec\avar}{\allter\in\Ter\lambdaletreccal}$ is in
$\sunfreducered$\nb-normal form, 
then the free variables
occurring in $M$ correspond to the free variables of $\Unf \allter$.
\end{remark}

Note that here we have applied $\sUnf$ to a prefixed term which requires an
extension of $\sUnf$ to terms in $\Ter\lambdaletrecprefixcal$.

\begin{definition}[unfolding of prefixed terms] We redefine $\sUnf$ as a
partial function over the domain $\Ter\lambdaletrecprefixcal$:
\begin{align*}
\sUnf \funin \Ter{\lambdaletrecprefixcal} & {} \rightharpoonup \Ter{\inflambdacal}
\\
\allter & {} \mapsto \ater ~~~ \text{if}~\allter\unfoldomegared\ater 
\end{align*}
Uniqueness of $\sUnf$ follows from Lemma~\ref{lem:unfolding}.
\end{definition}

\begin{definition}[productive terms w.r.t.\ scope-/scope$^{\bs{+}}$-del.\ strat.]%
  \label{def:productive:scdelstrats}
  Let $\allter$ be a \lambdaletrecterm, and
  $\astrat$ a \scopedelimiting~strategy for $\RegletrecARS$,
  or a \extscopedelimiting~strategy for $\stRegletrecARS$.
  We say that $\allter$ is \emph{$\astrat$\nb-productive}
  if every infinite rewrite sequence on $\allter$ with respect to $\astrat$ 
  contains infinitely many steps according to
  $\slappdecompistratred{0}{\astrat}\,$, $\slappdecompistratred{1}{\astrat}\,$,  or $\slabsdecompstratred{\astrat}\,$.
\end{definition}

\begin{definition}[generated subterms of \lambdaletrecterms]\label{def:STred:Regletrec:stRegletrec}
  Let $\astrat$ be a \scopedelimiting\ strategy for $\RegletrecARS$ (for $\stRegletrecARS$). 
  For every $\allter\in\Ter{\lambdaletreccal}$, 
  the set $\gSTregstrat{\astrat}{\allter}$ (the set $\gSTstregstrat{\astrat}{\allter}$)
  of \emph{generated subterms} of $\allter$ with respect to $\RegletrecARS$ and $\astrat$ 
                                           (with respect to $\stRegletrecARS$ and $\astrat$) 
  is defined 
  as the set of $\stratred{\astrat}$\nb-many-step reducts 
     (the set of $\stratred{\astrat}$\nb-many-step reducts)
  of $\femptylabs{\allter}$
  via the mappings:
  \begin{align*}
    \sgSTregstrat{\astrat} \funin \Ter{\lambdaletreccal} & {} \rightarrow \powersetof{\Ter{\lambdaletrecprefixcal}}
     \hspace*{2ex} &
    \sgSTstregstrat{\astratplus} \funin \Ter{\lambdaletreccal} & {} \rightarrow \powersetof{\Ter{\lambdaletrecprefixcal}} 
    \\
    \allter & {} \mapsto \gSTregstrat{\astrat}{\allter} 
                               \defdby
                            \redsuccs{\stratmred{\astrat}}{\femptylabs{\allter}}   
    &
    \allter & {} \mapsto \gSTstregstrat{\astratplus}{\allter} 
                               \defdby
                            \redsuccs{\stratmred{\astratplus}}{\femptylabs{\allter}}   
  \end{align*}

\end{definition}

The following lemma states that every \extscopedelimiting\ strategy for $\stRegletrecARS$ (on \lambdaletrecterms), 
when restricted to the reducts of a \lambdaletrecterm~$\allter$ that expresses an infinite \lambdaterm~$\aiter$,
projects to the restriction of a \extscopedelimiting\ strategy for $\stRegARS$ (on infinite \lambdaterms)  
to reducts of $\aiter$. And it asserts a similar statement for \scopedelimiting\ strategies.
The projection hereby makes use of the commutation properties described in Lemma~\ref{lem:commute:unfoldomegared:stregred}
between the infinite unfolding $\sunfoldomeganfred$ and the decomposition rewrite relations
$\slabsdecompred$, $\slappdecompired{0}$, $\slappdecompired{1}$, $\scompressstregred$, and $\scompressregred$.

\begin{lemma}\label{lem:proj:scdelstrat:letrec:lambda}
  Let $\astrat$ be a \extscopedelimiting~strategy $\astrat$ for $\RegletrecARS$ (for $\stRegletrecARS$),
  and let $\allter$ be a \lambdaletrecterm\ that is $\astrat$\nb-productive.  
  Then there exists a (history-aware)
  \scopedelimiting\ strategy (\extscopedelimiting~strategy) $\Check{\astrat}$ for $\RegARS$ (for $\stRegARS$)
  such that the induced sub-\ARS\ $\InducedSubARSmred{\Unf{\allter}}{\sstratmred{\Check{\astrat}}}$ of $\Unf{\allter}$
  is the projection (under the unfolding mapping $\sUnf$) of the induced sub\nb-\ARS\ $\InducedSubARSmred{\allter}{\sstratmred{\astrat}}$ 
  of $\allter$, in the sense that
  for all $\allter'$ in $\InducedSubARSmred{\allter}{\sstratmred{\astrat}}$ it holds:
  \begin{gather*}
    \allter' \binrelcomp{\sstratmred{\sunf,\astrat}}{\sstratred{\slabsdecomp/\slappdecompi{i}/\scompressstreg/\scompressreg,\astrat}} \allter''
      \;\;\Longrightarrow\;\;
    \Unf{\allter'} \stratred{\Check{\astrat}} \Unf{\allter''} 
    \\
    \Unf{\allter'} \stratred{\Check{\astrat}} \aiter''
      \;\;\Longrightarrow\;\;
    (\exists \allter'')\, [\,
    \allter' \binrelcomp{\sstratmred{\sunf,\astrat}}{\sstratred{\slabsdecomp/\slappdecompi{i}/\scompressstreg/\scompressreg,\astrat}} \allter''
      \;\logand\;
    \aiter'' = \Unf{\allter''}  \,] 
  \end{gather*}
  As a consequence, $\Unf{\allter}$ is $\Check{\astrat}$\nb-regular if $\allter$ is $\astrat$\nb-regular.
\end{lemma}

\begin{proof}[Proof sketch]
We can utilise Lemma~\ref{lem:commute:unfoldomegared:stregred} to make
commuting diagrams out of the two formulas above (for any given $\allter'$),
which allows us to determine $\Check\astrat$ with respect to all terms in
$\InducedSubARSmred{\allter}{\sstratmred{\astrat}}$.
This freedom in the definition of $\Check{\astrat}$ also guarantees the property in the second implication in the lemma.
\end{proof}

\begin{definition}[$\stParseUnfCRS$]\label{def:stParseUnfCRS}
By $\stParseUnfCRS$ we denote the CRS comprising the rules of $\stParseCRS$ as well
as the rules from $\llunfCRS$. 
\end{definition}

\begin{example}\label{ex:stParseUnfCRS}
When applied to
$\letrec{\arecvar = \labs{\avar\bvar}{\lapp{\lapp{\arecvar}{\bvar}}{\avar}}}{\arecvar}$,
the rewrite relation in
$\stParseUnfCRS$ unfolds and decomposes, but at the same time recreates the
corresponding infinite \lambda-term (see also Example~\ref{ex:llunfCRS} and Example~\ref{ex:stParseCRS}):
\begin{align*}
& \parseempty{\femptylabs{\letrec{\arecvar = \labs{\avar\bvar}{\lapp{\lapp{\arecvar}{\bvar}}{\avar}}}{\arecvar}}}
\alignbreak\sunfrecred ~&
\parseempty{\femptylabs{\letrec{\arecvar = \labs{\avar\bvar}{\lapp{\lapp{\arecvar}{\bvar}}{\avar}}}{\labs{\avar\bvar}{\lapp{\lapp{\arecvar}{\bvar}}{\avar}}}}}
\alignbreak\sunflabsred ~&
\parseempty{\femptylabs{\labs\avar{\letrec{\arecvar = \labs{\avar\bvar}{\lapp{\lapp{\arecvar}{\bvar}}{\avar}}}{\labs{\bvar}{\lapp{\lapp{\arecvar}{\bvar}}{\avar}}}}}}
\alignbreak\sparselabsred ~&
\labs\avar{\parsen1{\avar}{\flabs{\avar}{\letrec{\arecvar = \labs{\avar\bvar}{\lapp{\lapp{\arecvar}{\bvar}}{\avar}}}{\labs{\bvar}{\lapp{\lapp{\arecvar}{\bvar}}{\avar}}}}}}
\alignbreak\sunflabsred ~&
\labs\avar{\parsen1{\avar}{\flabs{\avar}{\labs{\bvar}{\letrec{\arecvar = \labs{\avar\bvar}{\lapp{\lapp{\arecvar}{\bvar}}{\avar}}}{\lapp{\lapp{\arecvar}{\bvar}}{\avar}}}}}}
\alignbreak\sparselabsred ~&
\labs{\avar\bvar}{\parsen2{\avar,\bvar}{\flabs{\avar\bvar}{\letrec{\arecvar = \labs{\avar\bvar}{\lapp{\lapp{\arecvar}{\bvar}}{\avar}}}{\lapp{\lapp{\arecvar}{\bvar}}{\avar}}}}}
\alignbreak\sunflappred ~&
\labs{\avar\bvar}{\parsen2{\avar,\bvar}{\flabs{\avar\bvar}{\lapp
  {(\letrec{\arecvar = \labs{\avar\bvar}{\lapp{\lapp{\arecvar}{\bvar}}{\avar}}}{\lapp{\arecvar}{\bvar}})}
  {(\letrec{\arecvar = \labs{\avar\bvar}{\lapp{\lapp{\arecvar}{\bvar}}{\avar}}}{\avar})}
}}}
\alignbreak\sunfreducered ~&
\labs{\avar\bvar}{\parsen2{\avar,\bvar}{\flabs{\avar\bvar}{\lapp
  {(\letrec{\arecvar = \labs{\avar\bvar}{\lapp{\lapp{\arecvar}{\bvar}}{\avar}}}{\lapp{\arecvar}{\bvar}})}
  {(\letrec{}{\avar})}
}}}
\alignbreak\sunfnilred ~&
\labs{\avar\bvar}{\parsen2{\avar,\bvar}{\flabs{\avar\bvar}{\lapp
  {(\letrec{\arecvar = \labs{\avar\bvar}{\lapp{\lapp{\arecvar}{\bvar}}{\avar}}}{\lapp{\arecvar}{\bvar}})}
  {\avar}
}}}
\alignbreak\sparselappred ~&
\labs{\avar\bvar}{\lapp
  {(\parsen2{\avar,\bvar}{\flabs{\avar\bvar}{\letrec{\arecvar = \labs{\avar\bvar}{\lapp{\lapp{\arecvar}{\bvar}}{\avar}}}{\lapp{\arecvar}{\bvar}}}})}
  {(\parsen2{\avar,\bvar}{\flabs{\avar\bvar}{\avar}})}
}
\alignbreak\sparsecompressstregred ~&
\labs{\avar\bvar}{\lapp
  {(\parsen2{\avar,\bvar}{\flabs{\avar\bvar}{\letrec{\arecvar = \labs{\avar\bvar}{\lapp{\lapp{\arecvar}{\bvar}}{\avar}}}{\lapp{\arecvar}{\bvar}}}})}
  {(\parsen1{\avar}{\flabs{\avar}{\avar}})}
}
\alignbreak\sparsenlvarred ~&
\labs{\avar\bvar}{\lapp
  {(\parsen2{\avar,\bvar}{\flabs{\avar\bvar}{\letrec{\arecvar = \labs{\avar\bvar}{\lapp{\lapp{\arecvar}{\bvar}}{\avar}}}{\lapp{\arecvar}{\bvar}}}})}
  {\avar}
}
\alignbreak\sunflappred ~&
\labs{\avar\bvar}{\lapp
  {(\parsen2{\avar,\bvar}{\flabs{\avar\bvar}{\lapp
      {(\letrec{\arecvar = \labs{\avar\bvar}{\lapp{\lapp{\arecvar}{\bvar}}{\avar}}}{\arecvar})}
      {(\letrec{\arecvar = \labs{\avar\bvar}{\lapp{\lapp{\arecvar}{\bvar}}{\avar}}}{\bvar})}
    }})}
  {\avar}
}
\alignbreak\sunfreducered ~&
\labs{\avar\bvar}{\lapp
  {(\parsen2{\avar,\bvar}{\flabs{\avar\bvar}{\lapp
      {(\letrec{\arecvar = \labs{\avar\bvar}{\lapp{\lapp{\arecvar}{\bvar}}{\avar}}}{\arecvar})}
      {(\letrec{}{\bvar})}
    }})}
  {\avar}
}
\alignbreak\sunfnilred ~&
\labs{\avar\bvar}{\lapp
  {(\parsen2{\avar,\bvar}{\flabs{\avar\bvar}{\lapp
      {(\letrec{\arecvar = \labs{\avar\bvar}{\lapp{\lapp{\arecvar}{\bvar}}{\avar}}}{\arecvar})}
      {\bvar}
    }})}
  {\avar}
}
\alignbreak\sparselappred ~&
\labs{\avar\bvar}{\lapp{\lapp
      {(\parsen2{\avar,\bvar}{\flabs{\avar\bvar}{\letrec{\arecvar = \labs{\avar\bvar}{\lapp{\lapp{\arecvar}{\bvar}}{\avar}}}{\arecvar}}})}
      {(\parsen2{\avar,\bvar}{\flabs{\avar\bvar}{\bvar}})}
    }
  {\avar}
}
\alignbreak\sparsenlvarred ~&
\labs{\avar\bvar}{\lapp{\lapp
      {(\parsen2{\avar,\bvar}{\flabs{\avar\bvar}{\letrec{\arecvar = \labs{\avar\bvar}{\lapp{\lapp{\arecvar}{\bvar}}{\avar}}}{\arecvar}}})}
      {\bvar}
    }
  {\avar}
}
\alignbreak\sparsecompressstregred ~&
\labs{\avar\bvar}{\lapp{\lapp
      {(\parsen1{\avar}{\flabs{\avar}{\letrec{\arecvar = \labs{\avar\bvar}{\lapp{\lapp{\arecvar}{\bvar}}{\avar}}}{\arecvar}}})}
      {\bvar}
    }
  {\avar}
}
\alignbreak\sparsecompressstregred ~&
\labs{\avar\bvar}{\lapp{\lapp
      {(\parseempty{\femptylabs{\letrec{\arecvar = \labs{\avar\bvar}{\lapp{\lapp{\arecvar}{\bvar}}{\avar}}}{\arecvar}}})}
      {\bvar}
    }
  {\avar}
}
\alignbreak\sunfrecred ~&
\labs{\avar\bvar}{\lapp{\lapp
      ~\dots~
      {\bvar}
    }
  {\avar}
}
\end{align*}
\end{example}

\begin{lemma}\label{lem:unfolding:versus:scdelstrats}
  For all closed $\allter\in\Ter{\lambdaletreccal}$ 
  the following statements are equivalent:
  \begin{enumerate}[(i)]
    \item{}\label{lem:unfolding:versus:scdelstrats:item:i}
      $\allter$ expresses an infinite \lambdaterm\ $\aiter$, that is, $\allter \unfoldomegared \aiter$.
    \item{}\label{lem:unfolding:versus:scdelstrats:item:ii}
      $\parseempty{\femptylabs{\allter}} \unfoldparseomegared \aiter$, for some infinite \lambdaterm\ $\aiter$. 
    \item{}\label{lem:unfolding:versus:scdelstrats:item:iii}
      $\allter$ is $\astratplus$\nb-productive for some \extscopedelimiting\ strategy~$\astratplus$. 
    \item{}\label{lem:unfolding:versus:scdelstrats:item:iv}
      $\allter$ is $\astratplus$\nb-productive for every \extscopedelimiting\ strategy~$\astratplus$. 
  \end{enumerate}
\end{lemma}

\begin{proof}
  Let $\allter\in\Ter{\lambdaletreccal}$. We show the lemma by establishing the implications 
  in the following order:
       ``$(\text{\ref{lem:unfolding:versus:scdelstrats:item:iv}})
            \Rightarrow 
          (\text{\ref{lem:unfolding:versus:scdelstrats:item:iii}})
            \Rightarrow 
          (\text{\ref{lem:unfolding:versus:scdelstrats:item:ii}})
            \Rightarrow 
          (\text{\ref{lem:unfolding:versus:scdelstrats:item:i}})
            \Rightarrow
          (\text{\ref{lem:unfolding:versus:scdelstrats:item:iv}}) 
          $''.
         
  The implication         
      ``$(\text{\ref{lem:unfolding:versus:scdelstrats:item:iv}})
            \Rightarrow 
         (\text{\ref{lem:unfolding:versus:scdelstrats:item:iii}})$''
  is clear: (\text{\ref{lem:unfolding:versus:scdelstrats:item:iv}})
  implies that $\allter$ is productive for e.g.\ the lazy-unfolding, eager \extscopedelimiting\ strategy on $\stRegletrecARS$.
          
  For showing  the implication                  
      ``$(\text{\ref{lem:unfolding:versus:scdelstrats:item:iii}})
            \Rightarrow 
         (\text{\ref{lem:unfolding:versus:scdelstrats:item:ii}})$'',
  let $\astratplus$ be a \extscopedelimiting~strategy for $\stRegARS$ such that $\allter$ is $\astratplus$\nb-productive.
  Then the strategy $\astratplus$ defines a $\sunfoldparsered$\nb-rewrite sequence $\arewseq$ on $\parseempty{\femptylabs{\allter}}$ 
  by using $\astratplus$ to define next steps on subexpressions that are of the form 
  $\parsen{n}{\ldots}{\flabsCRS{n}{\avari{1},\ldots,\avari{n}}{\bllter}}$ in already obtained reducts:
  if on a term $\flabs{\avari{1},\ldots,\avari{n}}{\bllter}$ the strategy $\astratplus$ 
  prescribes a $\sunfoldred$\nb-step, then this step is taken over in $\arewseq\,$;
  if $\astratplus$ prescribes a $\labsdecompstratred{\astratplus}$-step, 
  then $\arewseq$ can continue with a $\indap{\sred}{\subparse,\slabs}$\nb-step;
  if $\astratplus$ prescribes a $\lappdecompistratred{0}{\astratplus}$- and a $\lappdecompistratred{1}{\astratplus}$-step, 
  then $\arewseq$ can continue with a $\indap{\sred}{\subparse,@}$\nb-step.
  For the construction of $\arewseq$, possible steps in subexpressions 
  $\parsen{n}{\ldots}{\flabsCRS{n}{\avari{1},\ldots,\avari{n}}{\bllter}}$ at parallel positions
  have to be interleaved to ensure that the reduction work is done in an outer\-most-fair way.
  Productivity of $\astratplus$ on $\allter$ ensures that always after finitely many steps
  inside a subexpression $\parsen{n}{\ldots}{\flabsCRS{n}{\avari{1},\ldots,\avari{n}}{\bllter}}$
  the function symbol $\sparsen{n}$ disappears at this position (either entirely, or it is moved deeper over
  a $\lambda$\nb-abstraction or an application).
  In the terms of the rewrite sequence $\arewseq$ larger and larger \lambdaterm\ contexts appear at the head.
  Hence $\arewseq$ is strongly convergent, and it obtains, in the limit, an infinite \lambdaterm;
  thus it witnesses $\arewseq \funin \parseempty{\femptylabs{\allter}} \unfoldparseomegared \aiter$. 
         
  For the implication                  
      ``$(\text{\ref{lem:unfolding:versus:scdelstrats:item:ii}})
            \Rightarrow 
         (\text{\ref{lem:unfolding:versus:scdelstrats:item:i}})$'',
  suppose that $\arewseq$ is a rewrite sequence that witnesses  
  $\parseempty{\femptylabs{\allter}} \unfoldparseomegared \aiter$ for some infinite \lambdaterm\ $\aiter$.    
  Since the $\sparsered$\nb-steps require already unfolded parts of the term,
  they have to `shadow' unfolding steps.
  All $\sunfoldred$\nb-steps in $\arewseq$ take place beneath symbols $\sparsen{n}$.
  So the possibility of $\sparsered$\nb-steps during $\arewseq$ depends
  on the unfolding steps during $\arewseq$, but not vice versa.
  Hence a rewrite sequence $\arewseq'$ on $\parseempty{\femptylabs{\allter}}$
  can be constructed that only keeps the $\sunfoldred$\nb-steps from $\arewseq$. 
  Since $\arewseq$ is strongly convergent and converges to $\aiter$,
  $\arewseq'$ witnesses $\parseempty{\femptylabs{\allter}} \unfoldomegared \parseempty{\femptylabs{\aiter}}$.  
  By dropping the `non-participant' prefix context $\parseempty{\femptylabs{\acxthole}}$ 
  from all terms in $\arewseq'$, and adapting the steps accordingly, 
  a rewrite sequence $\arewseq''$ is obtained that witnesses $\arewseq'' \funin \allter \unfoldomegared \aiter$.  
  
  We show the implication 
      ``$(\text{\ref{lem:unfolding:versus:scdelstrats:item:i}})
            \Rightarrow 
         (\text{\ref{lem:unfolding:versus:scdelstrats:item:iv}})$''
  indirectly. So we assume that there is a \extscopedelimiting~strategy $\astratplus$
  such that $\allter$ is not $\astratplus$\nb-productive.
  As in the proof above of                
      ``$(\text{\ref{lem:unfolding:versus:scdelstrats:item:iii}})
            \Rightarrow 
         (\text{\ref{lem:unfolding:versus:scdelstrats:item:ii}})$'',
  $\astratplus$ defines an outermost-fair, maximal $\sunfoldparsered$\nb-rewrite sequence $\arewseq$ on $\parseempty{\femptylabs{\allter}}$.
  But since $\astratplus$ here is a strategy that is not productive for $\allter$,
  it follows that, due to its construction, $\arewseq$ does not succeed in `pushing' all function symbols $\sletrec$ to deeper and deeper depth,
  and thereby building up an infinite \lambdaterm. 
  Instead, this outermost-fair $\sunfoldparsered$\nb-rewrite sequence contains infinitely many steps 
  at the position of an outermost occurrence of $\sletrec$. 
  Since, other than the $\sunfoldred$\nb-steps, the $\sparsered$\nb-steps (which always take place above
  outermost occurrences of $\sletrec$\nb-symbols) cannot be the reason for this,
  the same stagnation of an outermost-fair unfolding process takes place if the $\sparsered$\nb-steps
  are postponed, that is dropped from $\arewseq$. In this way, 
  by again dropping the `non-participant' prefix context $\parseempty{\femptylabs{\acxthole}}$ from the terms of $\arewseq$, and adapting the steps accordingly, 
  we obtain an outermost-fair $\sunfoldred$\nb-rewrite sequence starting on $\femptylabs{\allter}$ 
  that does not converge to an infinite \lambdaterm. 
  But then Lemma~\ref{lem:unfolding} 
  implies that $\allter$ does not unfold to an infinite \lambdaterm.     
\end{proof}

\section{Proving regularity and strong regularity}
  \label{sec:proofs}

In this section we introduce proof systems that are 
                                                    sound and complete
for the notions of regular, and strongly regular, infinite \lambdaterm{s}.
In order to prove soundness and completeness,  we establish, as auxiliary results,
a correspondence between \scopedelimiting/\extscopedelimiting\ strategies for $\RegARS$/$\stRegARS$
and closed derivations in the corresponding proof systems.
Then we introduce a proof system that is sound and complete for 
equality between strongly regular infinite \lambdaterms.
Furthermore, we give two proof systems that are sound and complete for the property of \lambdaletrec\nb-terms to
unfold to infinite \lambdaterms. And finally, we show the following part of our characterisation result:
infinite \lambdaterms\ that are unfoldings of \lambdaletrec\nb-terms are strongly regular.    

We start with a more formal definition of \lambda-terms and \lambdaletrec-terms
than in Definition~\ref{def:sigs:lambdacal:lambdaletrec:CRS},
by means of derivability in proof system that formalises term decomposition.

\begin{figure}[t!]
\begin{center}\framebox{\begin{minipage}{330pt}\begin{center}
  \mbox{}
  \\[1.5ex]
  \mbox{ 
    \AxiomC{}
    \RightLabel{\bvarax}
    \UnaryInfC{$\flabs{\vec{\avar}\bvar}{\bvar}$}
    \DisplayProof
  } 
  \hspace*{3.5ex}    
  \mbox{
    \AxiomC{$ \flabs{\vec{x}\bvar}{\aiteri{0}} $}
    \RightLabel{$\labscomp$}
    \UnaryInfC{$ \flabs{\vec{\avar}}{\labs{\bvar}{\aiteri{0}}} $}
    \DisplayProof
        }
  \hspace*{3.5ex}
  \mbox{
    \AxiomC{$ \flabs{\vec{\avar}}{\aiteri{0}}$}
    \AxiomC{$ \flabs{\vec{\avar}}{\aiteri{1}}$}
    \RightLabel{$\lappcomp$}
    \BinaryInfC{$ \flabs{\vec{\avar}}{\lapp{\aiteri{0}}{\aiteri{1}}} $}
    \DisplayProof 
    }     
   \\[3.5ex]
  \mbox{
    \AxiomC{$ \flabs{\avari{1}\ldots\avari{n-1}}{\aiter} $}
    \RightLabel{\Vacstreg\ $\;$
                \parbox[c]{80pt}{\small (if the binding\\$\slabs\avari{n}$ is vacuous)
                                 }%
                }
    \UnaryInfC{$ \flabs{\avari{1}\ldots\avari{n}}{\aiter} $}
    \DisplayProof
        }  
  \\[1ex]
  \mbox{}  
\end{center}\end{minipage}}\end{center} 
  \vspace*{-1.25ex}  
  \caption{\label{fig:Terinflambda}%
           Proof system $\Terinflambda$ for defining the set of infinite \lambda-terms.}
\end{figure}   

\begin{definition}[infinite \lambda-terms]
\label{def:Ter-inflambda}
We define the set of prefixed infinite \lambda-terms as those terms in
$\Ter\siglpcCRS$ for which there exists a possibly infinite, completed (see Definition~\ref{def:Alpha}) derivation in the
proof system $\Terinflambda$ with axioms and rules as shown in Figure~\ref{fig:Terinflambda}:
\[\Ter\inflambdaprefixcal \defdby \descsetexpbig{\ater\in\Ter\siglpcCRS}{\infderivablein{\Terinflambda}\ater}\]
The set of plain infinite \lambda-terms are those terms that comply with the
previous definition when equipped with an empty prefix:
\[\Ter\inflambdacal \defdby \descsetexpbig{\ater\in\Ter\siglcCRS}{\infderivablein{\Terinflambda}{\femptylabs\ater\in\Ter\inflambdaprefixcal}}\]
\end{definition}

\begin{figure}[t!]
\begin{center}\framebox{\begin{minipage}{310pt}\begin{center}
  \mbox{}
  \\[1.5ex]
  \mbox{
    \AxiomC{$\flabs{\vec\avar\arecvari{1}\dots\arecvari{n}}{\ateri{1}}$}
    \AxiomC{$\dots$}
    \AxiomC{$\flabs{\vec\avar\arecvari{1}\dots\arecvari{n}}{\ateri{n}}$}
    \AxiomC{$\flabs{\vec\avar\arecvari{1}\dots\arecvari{n}}{\ater}$}
    \RightLabel{$\sletrec$}
    \QuaternaryInfC{$\flabs{\vec{\avar}}{\letrec{\arecvari1=\ateri1,\dots,\arecvari{n}=\ateri{n}}{\ater}}$}
    \DisplayProof
  }
  \\[1ex]
  \mbox{}  
\end{center}\end{minipage}}\end{center} 
  \vspace*{-1.25ex}  
  \caption{\label{fig:Terlambdaletrec}%
  Proof system $\Terlambdaletrec$ defined as an extension of $\Terinflambda$ by
  an additional rule to define the set of \lambdaletrec-terms}
\end{figure}   

\begin{definition}[\lambdaletrec-terms]
\label{def:Ter-lambdaletrec}
The set of prefixed \lambdaletrec-terms comprises those terms out of $\Ter{\sigllpcCRS}$
for which there exists a finite derivation in the proof system
$\Terlambdaletrec$ (Figure~\ref{fig:Terlambdaletrec}):
\[\Ter\lambdaletrecprefixcal \defdby \descsetexpbig{\ater\in\Ter\sigllpcCRS}{\derivablein{\Terlambdaletrec}{\ater}}\]
The set of plain \lambdaletrec-terms are those terms that comply with the
previous definition when equipped with an empty prefix:
\[\Ter\lambdaletreccal \defdby \descsetexpbig{\ater\in\Ter\sigllcCRS}{\derivablein{\Terlambdaletrec}{\femptylabs\ater\in\Ter{\lambdaletrecprefixcal}}}\]
\end{definition}

Building on rules already used in the proof systems for term formation in
\inflambdacal\ and \inflambdaprefixcal\ from the definition above, 
we now introduce proof systems for regularity and strong regularity of infinite \lambdaterms\ in $\inflambdacal$.

\begin{definition}[proof systems $\Reg$, and $\stReg$, $\stRegzero$]%
  \label{def:Reg:stReg:stRegzero}
  The natural-deduction style proof system $\stReg$ for recognising strongly regular, infinite $\lambda$\nb-terms
  contains the axioms and rules as shown in Figure~\ref{fig:stReg:stRegzero}.
  In particular, the rule \FIX\ is a natural-deduction style derivation rule
  in which marked assumptions from the top of the proof tree can be discharged.
  Instances of this rule carry the side-condition that the depth $\depth{\Derivi{0}}$ 
  of the immediate subderivation $\Derivi{0}$ of its premise is greater or equal to 1
  (hence this subderivation contains at least one rule instance,
   and, importantly, for a topmost occurrence of \FIX, $\Derivi{0}$ must have a 
   bottommost instance of one of the rules ($\labscomp$), ($\lappcomp$), or ($\Vacstreg$)). 
   
  The variant $\stRegzero$ of $\Reg$ contains the same axioms and rules as $\stReg$,
  but in it instances of \FIX\ are subject to the additional side-condition:
  for all $\flabs{\vec{y}}{\bter}$ on threads in $\Derivi{0}$ from
  open marked assumptions $(\flabs{\vec{\avar}}{\ater})^u$ downwards
  it holds that $\length{\vec{y}} \ge \length{\vec{x}}$.
    
  The natural-deduction style proof system $\Reg$ for recognising regular infinite $\lambda$\nb-terms
  differs from $\stReg$ by the absence of the rule ($\Vacstreg$), and the presence instead
  of the rule ($\Vacreg$) in Figure~\ref{fig:Reg}, and by the restriction of the axiom scheme ($\bvarax$)
  to the more restricted version displayed in Figure~\ref{fig:Reg}. 
  
  Provability of a term in $\inflambdaprefixcal$ in one of these proof systems is defined as the existence
  of a \emph{closed} derivation: for $\bs{R}\in\setexp{\Reg,\, \stReg,\, \stRegzero}$ we denote
  by $\derivablein{\bs{R}}{\flabs{\vec{\avar}}{\aiter}}$ the existence
  of a prooftree $\Deriv$ with conclusion $\aiter$ and with rule instances of $\bs{R}$
  such that all marked assumptions at the top of the $\Deriv$ are discharged 
  at some instance of the rule \FIX.
\end{definition}

\begin{figure}[t!]
\begin{center}\framebox{\begin{minipage}{330pt}\begin{center}
  \mbox{}
  \\[1ex]
  \mbox{ 
    \AxiomC{}
    \RightLabel{\bvarax}
    \UnaryInfC{$\flabs{\vec{\avar}\bvar}{\bvar}$}
    \DisplayProof
  } 
  \hspace*{3.5ex}    
  \mbox{
    \AxiomC{$ \flabs{\vec{x}\bvar}{\aiteri{0}} $}
    \RightLabel{$\labscomp$}
    \UnaryInfC{$ \flabs{\vec{\avar}}{\labs{\bvar}{\aiteri{0}}} $}
    \DisplayProof
        }
  \hspace*{3.5ex}
  \mbox{
    \AxiomC{$ \flabs{\vec{\avar}}{\aiteri{0}}$}
    \AxiomC{$ \flabs{\vec{\avar}}{\aiteri{1}}$}
    \RightLabel{$\lappcomp$}
    \BinaryInfC{$ \flabs{\vec{\avar}}{\lapp{\aiteri{0}}{\aiteri{1}}} $}
    \DisplayProof 
    }     
   \\[3.5ex]
  \mbox{
    \AxiomC{$ \flabs{\avari{1}\ldots\avari{n-1}}{\aiter} $}
    \RightLabel{\Vacstreg\ $\;$
                \parbox[c]{80pt}{\small (if the binding\\$\slabs\avari{n}$ is vacuous)
                                 }%
                }
    \UnaryInfC{$ \flabs{\avari{1}\ldots\avari{n}}{\aiter} $}
    \DisplayProof
        }  
  \hspace*{0.5ex}
  \mbox{
    \AxiomC{$ [\flabs{\vec{\avar}}{\aiter}]^u$}
    \noLine
    \UnaryInfC{$\Derivi{0}$}
    \noLine
    \UnaryInfC{$ \flabs{\vec{\avar}}{\aiter} $}
    \RightLabel{$\sFIX,u$ $\,$
                \parbox{49pt}{\small (if $\depth{\Derivi{0}} \ge 1$)}
                }                   
    \UnaryInfC{$ \flabs{\vec{\avar}}{\aiter} $}
    \DisplayProof
    }
  \\[1ex]
  \mbox{}  
\end{center}\end{minipage}}\end{center} 
  \vspace*{-1.25ex}  
  \caption{\label{fig:stReg:stRegzero}%
           The natural-deduction style proof system $\stReg$ for strongly regular infinite $\lambda$\nb-terms,
           which is an extension of $\Terinflambda$ by one additional rule \FIX.
           In the variant system $\stRegzero$, instances of \FIX\ are subject to the following side-condition:
           for all $\flabs{\vec{y}}{\bter}$ on threads in $\Derivi{0}$ from
           open marked assumptions $(\flabs{\vec{\avar}}{\ater})^u$ downwards
           it holds that $\length{\vec{y}} \ge \length{\vec{x}}$.
           }
\end{figure}

\begin{figure}[t!]
\begin{center}  
  \framebox{
\begin{minipage}{330pt}
\begin{center}
  \mbox{}
  \\[0.5ex]
  \mbox{ 
    \AxiomC{}
    \RightLabel{\bvarax}
    \UnaryInfC{$\flabs{\bvar}{\bvar}$}
    \DisplayProof
        } 
  \hspace*{3.5ex}      
  \mbox{
    \AxiomC{$ \flabs{\avari{1}\ldots\avari{i-1}\avari{i+1}\ldots\avari{n}}{\aiter} $} 
    \RightLabel{\Vacreg\ $\;$
                \parbox{65pt}{\small (if the binding\\[-0.5ex]\hspace*{\fill} $\lambda\avari{i}$ is vacuous)}
                }
    \UnaryInfC{$ \flabs{\avari{1}\ldots\avari{n}}{\aiter} $} 
    \DisplayProof
        }  
  \\[0.5ex]
  \mbox{}
\end{center}
\end{minipage}
            }
\end{center} 
  \vspace*{-1.25ex}  
  \caption{\label{fig:Reg}%
           The natural-deduction style 
           proof system $\Reg$ for regular infinite $\lambda$\nb-terms 
           arises from the proof system $\stReg$ 
           by replacing the rule ($\Vacstreg$) with the rule ($\Vacreg$)
           for the introduction of vacuous bindings in the $\lambda$\protect\nb-abstraction prefixes,
           and by replacing the axiom scheme ($\bvarax$) of $\stReg$ by the more restricted version here.
           }
\end{figure}

\begin{remark}
  The proof system $\stReg$ is related to
  a proof system for nameless, finite terms in the \lambdacalculus\
  that is used in \cite[sec.~2]{oost:looi:zwit:2004} 
  as part of a translation of \lambdaterms\ into `Lambdascope' interaction nets,
  which are used for optimal evaluation (in the sense of L\'{e}vy) of \lambdaterms.
\end{remark}

The proposition below explains that the side-condition on instances of \FIX\ from the proof systems above
to have immediate subderivations $\Derivi{0}$ with $\depth{\Derivi{0}}\ge 1$
entails a `guardedness' property for threads from such instances upwards to discharged instances. 

\begin{proposition}\label{prop:Reg:stReg}
  Let $\Deriv$ be a derivation in $\Reg$, in $\stReg$ or in $\stRegzero$ possibly with open marked assumptions.
  Then for all instances $\ainst$ of the rule \FIX\ in $\Deriv$ it holds:
  every thread from $\ainst$ upwards to a marked assumption that is discharged at $\ainst$
  passes at least one instance of a rule ($\labscomp$) or ($\lappcomp$).
\end{proposition}

\begin{proof}
  Since for $\stReg$ and $\stRegzero$ the argument is analogous, we only consider derivations in $\Reg$. 
  So, let $\Deriv$ be a derivation in $\Reg$. Furthermore,  
  let $\ainst$ be an instance of the rule \FIX\ in $\Deriv$ with conclusion $\flabs{\vec{\avar}}{\aiter}$,
  and let $\apath$ be a thread from the conclusion of $\ainst$ upwards to a marked assumption $(\flabs{\vec{\avar}}{\aiter})^{\amarker}$. 
  Let $\binst$ be the topmost instance of \FIX\ in $\Deriv$ that is passed on $\apath$.
  By its side-condition, the immediate subderivation of $\binst$ has depth greater or equal to 1,
  and hence there is at least one instance of a rule ($\labscomp$), ($\lappcomp$), or ($\Vacreg$) passed on $\apath$
  above $\binst$. If there is an instance of ($\labscomp$) or ($\lappcomp$) on this part of $\apath$, we are done.
  Otherwise only rules ($\Vacreg$) are passed on $\apath$ above $\binst$. But since the rule ($\Vacreg$)
  decreases the length of the abstraction prefix in the term occurrences in a pass from the conclusion to the premise,
  and since the length of the abstraction prefix at the start of $\apath$ is the same as at the end of $\apath$, namely $\length{\vec{\avar}}$,
  it follows that at least one occurrence of a rule that increases the length of the abstraction prefix 
  must have been passed on $\apath$, too, on the part from $\ainst$ to $\binst$. Since the only
  rule of $\Reg$ that increases the length of an abstraction prefix in a pass from conclusion to a premise
  is the rule ($\labscomp$), we have also in this case found a desired rule instance on $\apath$. 
\end{proof}

\begin{example}\label{example:stReg}
  Let $\ater$ be the infinite unfolding of $\letrec{\arecvar =
  \labs{\avar\bvar}{\lapp{\lapp{\arecvar}{\bvar}}{\avar}}}{\arecvar}$ for which
  we use as a finite representation the equation $\ater =
  \labs{\avar\bvar}{\lapp{\lapp{\ater}{\bvar}}{\avar}}$. This term admits the
  following derivations in $\stReg$, where as opposed to the left one the right one
  has some redundancy:
  \begin{gather*}
     \begin{aligned}[c]
     \scalebox{0.92}{
      \AxiomC{$ (\femptylabs{\ater})^{\amarker} $}
      \RightLabel{$\Vacstreg$}
      \UnaryInfC{$ \flabs{\avar}{\ater} $}
      \RightLabel{$\Vacstreg$}
      \UnaryInfC{$ \flabs{\avar\bvar}{\ater} $}
      \AxiomC{\mbox{}}
      \RightLabel{$\bvarax$}
      \UnaryInfC{$ \flabs{\avar\bvar}{\bvar} $}
      \RightLabel{$\lappcomp$}
      \BinaryInfC{$ \flabs{\avar\bvar}{\lapp{\ater}{\bvar}} $}
      \AxiomC{\mbox{}}
      \RightLabel{$\bvarax$}
      \UnaryInfC{$ \flabs{\avar}{\avar} $}
      \RightLabel{$\Vacstreg$}
      \UnaryInfC{$ \flabs{\avar\bvar}{\avar} $}
      \RightLabel{$\lappcomp$}
      \BinaryInfC{$ \flabs{\avar\bvar}{\lapp{\lapp{\ater}{\bvar}}{\avar}} $}
      \RightLabel{$\labscomp$}
      \UnaryInfC{$ \flabs{\avar}{\labs{\bvar}{\lapp{\lapp{\ater}{\bvar}}{\avar}}} $}
      \RightLabel{$\labscomp$}
      \UnaryInfC{$ \femptylabs{}{\labs{\avar\bvar}{\lapp{\lapp{\ater}{\bvar}}{\avar}}} $}
      \RightLabel{$\sFIX,\amarker$}
      \UnaryInfC{$ \femptylabs{}{\ater} $}
      \DisplayProof
       }
    \end{aligned}
    \hspace*{-2ex}
    \begin{aligned}[c]
     \scalebox{0.92}{
      \AxiomC{$ (\flabs{\avar}{\labs{\bvar}{\lapp{\lapp{\ater}{\bvar}}{\avar}}})^{\amarker} $}
      \RightLabel{$\labscomp$}
      \UnaryInfC{$ \femptylabs{\ater}) $}
      \RightLabel{$\Vacstreg$}
      \UnaryInfC{$ \flabs{\avar}{\ater} $}
      \RightLabel{$\Vacstreg$}
      \UnaryInfC{$ \flabs{\avar\bvar}{\ater} $}
      \AxiomC{\mbox{}}
      \RightLabel{$\bvarax$}
      \UnaryInfC{$ \flabs{\avar\bvar}{\bvar} $}
      \RightLabel{$\lappcomp$}
      \BinaryInfC{$ \flabs{\avar\bvar}{\lapp{\ater}{\bvar}} $}
      \AxiomC{\mbox{}}
      \RightLabel{$\bvarax$}
      \UnaryInfC{$ \flabs{\avar}{\avar} $}
      \RightLabel{$\Vacstreg$}
      \UnaryInfC{$ \flabs{\avar\bvar}{\avar} $}
      \RightLabel{$\lappcomp$}
      \BinaryInfC{$ \flabs{\avar\bvar}{\lapp{\lapp{\ater}{\bvar}}{\avar}} $}
      \RightLabel{$\labscomp$}
      \UnaryInfC{$ \flabs{\avar}{\labs{\bvar}{\lapp{\lapp{\ater}{\bvar}}{\avar}}} $}
      \RightLabel{$\sFIX,\amarker$}
      \UnaryInfC{$ \flabs{\avar}{\labs{\bvar}{\lapp{\lapp{\ater}{\bvar}}{\avar}}} $}
      \RightLabel{$\labscomp$}
      \UnaryInfC{$ \femptylabs{}{\ater} $}
      \DisplayProof
        }
    \end{aligned}
  \end{gather*}
  Note that the derivation on the left is also a derivation in $\stRegzero$,
  but that this is not the case for the derivation on the right. There,
  for the occurrence of \FIX\ the side-condition in the system $\stRegzero$ is violated:
  on the path from the marked assumption $(\flabs{\avar}{\labs{\bvar}{\lapp{\lapp{\ater}{\bvar}}{\avar}}})^{\amarker}$ 
  down to the instance of $\FIX$
  there is the occurrence $\femptylabs{}{\ater}$ of a term with shorter prefix than the term in the assumption and in the conclusion. 
  See Example~\ref{example:annstRegzero} for a comparison of annotated versions of the two derivations above.
 
  See Example~\ref{ex:simpleletrec-infinite-path} for a rewriting
  sequence in $\stRegARS$ corresponding to the leftmost path in both derivations, and
  also Figure~\ref{fig:distance} for the corresponding transition graph. \todo{LTG}
\end{example}

\begin{remark}[$\stReg$ versus $\stRegzero$]
  While it will be established in Proposition~\ref{prop:stReg:stRegzero} below that
  provability in $\stReg$ and $\stRegzero$ coincides, 
  the difference between these systems will come to the fore in annotated versions that are purpose-built
  for the extraction of \lambdaletrecterms\ that express infinite \lambdaterms.  
  This will be explained and illustrated later in Example~\ref{example:annstRegzero},
  using annotated versions of the two derivations in Example~\ref{example:stReg} above. 
\end{remark}

\begin{example}
The infinite \lambdaterm\ from Example~\ref{fig:ex:entangled}
with the \lTG\ shown in Figure~\ref{fig:entangled-cons-ltg}
is derivable in $\Reg$ by the following closed derivation using the
notation from Example~\ref{ex:entangled-infinite-path}:
  \begin{equation*}
    \mbox{
      \AxiomC{$  (\overbrace{\flabs{b}{\funap{rec_{\aiter}}{b}}}^{\flabs{a}{\funap{rec_{\aiter}}{a}}})^{\amarker}  $}
      \RightLabel{$\Vacreg$}
      \UnaryInfC{$ \flabs{ab}{\funap{rec_{\aiter}}{b}} $}
      \AxiomC{\mbox{}}
      \RightLabel{$\bvarax$}
      \UnaryInfC{$ \flabs{a}{a} $}
      \RightLabel{$\Vacreg$}
      \UnaryInfC{$ \flabs{ab}{a}$}
      \RightLabel{$\lappcomp$}
      \BinaryInfC{$ \flabs{ab}{\lapp{\funap{rec_{\aiter}}{b}}{a}} $}
      \RightLabel{$\labscomp$}
      \UnaryInfC{$\flabs{a}{\labs{b}{\lapp{\funap{rec_{\aiter}}{b}}{a}}}$}
      \RightLabel{\sFIX, $\amarker$}
      \UnaryInfC{$ \flabs{a}{\funap{rec_{\aiter}}{a}} $}
      \RightLabel{$\labscomp$}
      \UnaryInfC{$ \femptylabs{\labs{a}{\funap{rec_{\aiter}}{a}}} $}
      \DisplayProof
          }
  \end{equation*}
  When trying to construct a derivation for this term in $\stReg$ from the bottom upwards,
  the rules of $\stReg$ apart from (in first instance) \FIX\ offer only deterministic choices,
  with as outcome an infinite prooftree of the form: 
  \begin{equation*}
    \mbox{
      \AxiomC{$\vdots $}
      \noLine
      \UnaryInfC{$ \flabs{abcde}{\funap{rec_{\aiter}}{e}} $}
      \RightLabel{$\labscomp$}
      \UnaryInfC{$ \flabs{abcd}{\labs{e}{\funap{rec_{\aiter}}{e}}} $}
      \AxiomC{\mbox{}}
      \RightLabel{$\bvarax$}
      \UnaryInfC{$ \flabs{abc}{c} $}
      \RightLabel{$\Vacstreg$}
      \UnaryInfC{$ \flabs{abcd}{c}$}  
      \RightLabel{$\lappcomp$}
      \BinaryInfC{$ \flabs{abcd}{\lapp{\funap{rec_{\aiter}}{d}}{c}} $}
      \RightLabel{$\labscomp$}
      \UnaryInfC{$ \flabs{abc}{\labs{d}{\lapp{\funap{rec_{\aiter}}{d}}{c}}} $}
      \AxiomC{\mbox{}}
      \RightLabel{$\bvarax$}
      \UnaryInfC{$ \flabs{ab}{b} $}
      \RightLabel{$\Vacstreg$}
      \UnaryInfC{$ \flabs{abc}{b}$}  
      \RightLabel{$\lappcomp$}
      \BinaryInfC{$ \flabs{abc}{\lapp{\funap{rec_{\aiter}}{c}}{b}} $}
      \RightLabel{$\labscomp$}
      \UnaryInfC{$ \flabs{ab}{ \labs{c}{\lapp{\funap{rec_{\aiter}}{c}}{b}} } $}
      \AxiomC{\mbox{}}
      \RightLabel{$\bvarax$}
      \UnaryInfC{$ \flabs{a}{a} $}
      \RightLabel{$\Vacstreg$}
      \UnaryInfC{$ \flabs{ab}{a}$}
      \RightLabel{$\lappcomp$}
      \BinaryInfC{$ \flabs{ab}{\lapp{\funap{rec_{\aiter}}{b}}{a}} $}
      \RightLabel{$\labscomp$}
      \UnaryInfC{$ \flabs{a}{\labs{b}{\lapp{\funap{rec_{\aiter}}{b}}{a}}}$}
      \RightLabel{$\labscomp$}
      \UnaryInfC{$ \femptylabs{\labs{a}{\funap{rec_{\aiter}}{a}}}$}
      \DisplayProof
          }
  \end{equation*}
  But then, since this prooftree does not contain repetitions, also use of the rule \FIX\
  in order to discharge assumptions is impossible. Consequently, the term is
  not derivable in $\stReg$.

  For the $\RegARS$ and $\stRegARS$ rewriting sequences corresponding to the
  leftmost paths through the two proofs above, see
  Example~\ref{ex:entangled-infinite-path}. The corresponding transition graphs are
  displayed in Figure~\ref{fig:entangled-cons-ltg}. \todo{LTG}
\end{example}

%

%

\begin{proposition}\label{prop:derivationpaths:2:rewritesequences:Reg:stReg}
  Let $\Deriv$ be a derivation in $\Reg$ (in $\stReg$) with conclusion $\flabs{\vec{\avar}}{\aiter}$.
  
  Then every path in $\Deriv$ from the conclusion upwards corresponds to a
  $\sregred$-rewrite sequence (a $\stregred$-rewrite sequence) from $\flabs{\vec{\avar}}{\ater}$:
  while passes over instances of \FIX\ correspond to empty steps, 
  passes over instances of the rule ($\lappcomp$) to the left and to the right correspond to       
  $\slappdecompired{0}$- and $\slappdecompired{1}$-steps, respectively;
  passes over instances of the rules ($\labscomp$) and ($\Vacreg$) (the rule ($\Vacstreg$))
  correspond to $\slabsdecompred$-, and $\scompressregred$\nb-steps ($\scompressstregred$\nb-steps). 
  %
  The same holds for (finite or infinite) cyclic paths in $\Deriv$ that return, possibly repeatedly, 
  from a marked assumption at the top down to the conclusion of the instance of
  \FIX\ at which the respective assumption is discharged.
\end{proposition}

\begin{proof}
  The proposition is an easy consequence of the following facts:
  passes from a term in the conclusion of an instance $\ainst$ 
  of one of the rules ($\labscomp$), ($\Vacreg$), ($\Vacstreg$) to the term the premise of $\ainst$ 
  correspond to $\slabsdecompred$-, $\scompressregred$-, and $\scompressstregred$\nb-steps, respectively;
  passes from a term in the conclusion of an instance of ($\lappcomp$) to the left and the right premise 
  correspond to $\slappdecompired{0}$- and $\slappdecompired{1}$\nb-steps, respectively.
\end{proof}

Observe that, for the derivation $\Deriv$ in Example~\ref{example:stReg},
the $\sstregred$\nb-rewrite sequences that correspond to paths in $\Deriv$ as described in Proposition~\ref{prop:derivationpaths:2:rewritesequences:Reg:stReg}
are actually rewrite sequences with respect to the eager \extscopedelimiting\ strategy~$\eagscdelstratstreg$ for $\stRegARS$.
This illustrates the general situation, formulated by the lemma below:
paths in a derivation~$\Deriv$ in $\stReg$ (or in $\Reg$) from the conclusion upwards correspond
to rewrite sequences according to some, usually history-aware, \extscopedelimiting\ (\scopedelimiting) strategy~$\astrat$
for $\stReg$ (for $\Reg$),
which can be extracted from $\Deriv$.

\begin{lemma}[from $\Reg$/$\stReg$-derivations to scope/scope$^+$-delim.\ strategies
                                                                                    ]%
  \label{lem:derivations:Reg:stReg:2:scdelstrategies}
  Let $\aiter\in\Ter{\inflambdacal}$, 
  and let $\Deriv$ be a closed derivation in $\Reg$ (in $\stReg$) 
  with conclusion $\femptylabs{\aiter}$.
  Then there exists an, in general history-aware, \scopedelimiting\ strategy~$\astrati{\Deriv}$ for $\RegARS$
                                     (\extscopedelimiting\ strategy~$\astrati{\Deriv}$ for $\stRegARS$)
  with the following properties:
  \begin{enumerate}[(i)]
    \item{}\label{lem:derivations:Reg:stReg:2:scdelstrategies:item:i} 
      Every (possibly cyclic) path in $\Deriv$ from the conclusion upwards corresponds
      to a rewrite sequence with respect to $\astrati{\Deriv}$ starting on $\femptylabs{\aiter}$
      in the sense of Proposition~\ref{prop:derivationpaths:2:rewritesequences:Reg:stReg}
      where passes over instances of the rules ($\lappcomp$) to the left and to the right correspond
      to $\slappdecompistratred{0}{\astrati{\Deriv}}$- and $\slappdecompistratred{1}{\astrati{\Deriv}}$-steps, respectively,
      and passes over instances of ($\labscomp$) and of ($\Vacreg$) (of ($\Vacstreg$))
      correspond to $\slabsdecompstratred{\astrati\Deriv}$-, and $\scompressregstratred{\astrati\Deriv}$\nb-steps ($\scompressstregstratred{\astrati\Deriv}$\nb-steps).
    \item{}\label{lem:derivations:Reg:stReg:2:scdelstrategies:item:ii}
      Every 
            rewrite sequence that starts on $\femptylabs{\aiter}$ and proceeds according to $\astrati\Deriv$  
      corresponds to a (possibly cyclic) path in $\Deriv$ starting at the conclusion
      in upwards direction: thereby a $\slappdecompistratred{0}{\astrati\Deriv}$- and $\slappdecompistratred{1}{\astrati\Deriv}$\nb-step
      corresponds to a pass over (possibly successive \FIX-instances, or from a marked assumption to the instance of \FIX\ that binds it, followed by) 
      an instance of ($\lappcomp$) in direction left and right, respectively;
      a $\slabsdecompstratred{\astrati\Deriv}$- or $\scompressregstratred{\astrati\Deriv}$\nb-step ($\scompressstregstratred{\astrati\Deriv}$\nb-step)
      corresponds to a pass over (possibly \FIX-instances and assumption bindings to \FIX-instances) 
      an instance of ($\labscomp$) or ($\Vacreg$) (of ($\Vacstreg$)), respectively. 
      %
    \item{}\label{lem:derivations:Reg:stReg:2:scdelstrategies:item:iii}
      $\gSTstrat{\astrati\Deriv}{\aiter} 
         =
       \descsetexp{ \flabs{\vec{\bvar}}{\biter} }
                  {\text{the term $\flabs{\vec{\bvar}}{\biter}$ occurs in $\Deriv$}}$.
  \end{enumerate}  
\end{lemma}

\begin{proof}\label{prf:lem:derivations:Reg:stReg:2:scdelstrategies}
The proof defines a history-aware strategy $\astrati\Deriv$ for
$\stRegARS$ as a modification of an arbitrary (history-free)
strategy for $\stRegARS$ lifted to a labelled version of
$\stRegARS$. 
Thereby the modification is performed according to a given derivation $\Deriv$,
and the construction will guarantee that 
(\ref{lem:derivations:Reg:stReg:2:scdelstrategies:item:i}),
(\ref{lem:derivations:Reg:stReg:2:scdelstrategies:item:ii}),
and
(\ref{lem:derivations:Reg:stReg:2:scdelstrategies:item:iii})
hold.
\begin{figure}[t!]
\begin{center}  
  \framebox{
\begin{minipage}{350pt}
\begin{center}
  \mbox{}
  \\[1.5ex]
  \mbox{ 
    \AxiomC{}
    \RightLabel{\bvarax}
    \UnaryInfC{$\labflabs{\alabel}{\vec{\avar}\bvar}{\bvar}$}
    \DisplayProof
        }
  \hspace*{1.5ex}    
  \mbox{
    \AxiomC{$ \labflabs{\alabel}{\vec{x}\bvar}{\aiteri{0}} $}
    \RightLabel{$\labscomp$}
    \UnaryInfC{$ \labflabs{\alabel}{\vec{\avar}}{\labs{\bvar}{\aiteri{0}}} $}
    \DisplayProof
        }
  \hspace*{1.5ex}
  \mbox{
    \AxiomC{$ \labflabs{\alabel\, 0}{\vec{\avar}}{\aiteri{0}}$}
    \AxiomC{$ \labflabs{\alabel\, 1}{\vec{\avar}}{\aiteri{1}}$}
    \RightLabel{$\lappcomp$}
    \BinaryInfC{$ \labflabs{\alabel}{\vec{\avar}}{\lapp{\aiteri{0}}{\aiteri{1}}} $}
    \DisplayProof 
    }    
   \\[3.5ex]
  \mbox{
    \AxiomC{$ \labflabs{\alabel}{\avari{1}\ldots\avari{n-1}}{\aiter} $}
    \RightLabel{\Vacstreg\ $\;$
                \parbox[c]{65pt}{\small (if $\avari{n}$ does not\\[-0.5ex]\mbox{}$\;$ occur in ${\aiter}$)}
      }
    \UnaryInfC{$ \labflabs{\alabel}{\avari{1}\ldots\avari{n}}{\aiter} $}
    \DisplayProof
        }  
  \mbox{
    \AxiomC{$ [\labflabs{\alabel}{\vec{\avar}}{\aiter}]^u$}
    \noLine
    \UnaryInfC{$\Derivi{0}$}
    \noLine
    \UnaryInfC{$ \labflabs{\alabel}{\vec{\avar}}{\aiter} $}
    \RightLabel{$\sFIX,u$ $\,$
                \parbox{49pt}{\small (if $\depth{\Derivi{0}} \ge 1$)}
                }                   
    \UnaryInfC{$ \labflabs{\alabel}{\vec{\avar}}{\aiter} $}
    \DisplayProof
    }
  \\[1ex]
  \mbox{}  
\end{center}
\end{minipage}
   }
\end{center} 
  \caption{\label{fig:stReglab}%
           Proof system $\stReglab$ 
           for decorating $\stReg$-derivations with labels in $\setexp{0,1}^*$.
           }
\end{figure}   
  We establish the lemma only for the case of derivations in $\stReg$, since the case of derivations in $\Reg$
  can be treated analogously.
  So, let $\Deriv$ be a derivation in $\stReg$ with conclusion~$\femptylabs{\aiter}$.
  
  In a first step we decorate $\Deriv$ with position labels such that a derivation~$\Derivlab$
  with conclusion $\labfemptylabs{\rootpos}{\aiter}$ 
  in the variant proof system~$\stReglab$ in Figure~\ref{fig:stReglab} is obtained. 
  Note that the decoration process can be carried out in a bottom-up manner,
  where the label in the conclusion of a rule instance determines the label in the premise(s)
        if that is not an already labelled term,
  and where in the case of instances of the rule of \FIX\ also the labels in marked assumptions
  are determined. 
  
  In a second step we use the decorated version $\Derivlab$ of $\Deriv$ to define a history-aware
  strategy $\astrati{\Deriv}$ according to which the term $\femptylabs{\aiter}$ can
  be reduced as `prescribed' by $\Derivlab$. 
  Since the derivations can only determine the strategy $\astrati{\Deriv}$ on terms occurring in $\Deriv$,
  we also have to define $\astrati{\Deriv}$ on other terms.
  This will be done by choosing an arbitrary (but here: history-free) \extscopedelimiting~strategy 
  $\astrat$ for $\RegARS$, and basing the definition of $\astrati{\Deriv}$ on it.
  
  We start by defining a labelling of $\stRegARS$ as the \ARS\ for which
  $\astrati{\Deriv}$ will be defined as a history-free strategy,
  which together with an initial labelling~$\ainitlabelling$ 
  then yields a history-aware strategy for $\stRegARS$.  
  Assuming 
  $\stRegARS = \tuple{\Ter{\inflambdaprefixcal}, \,\steps, \,\ssrc, \,\stgt}$
  as the formal representation of $\stRegARS$, 
  we define the \ARS\ 
  $
    \stRegARSlab
      \defdby 
    \tuple{\Terlab{\inflambdaprefixcal},\, 
          \stepslab,\, 
          \ssrclab,\, 
          \stgtlab}
  $ 
  where
  \begin{align}
    &
    \Terlab{\inflambdaprefixcal} 
      \defdby
    \descsetexp{\labflabs{\alabel}{\vec{\bvar}}{\biter}}{\flabs{\vec{\bvar}}{\biter}\in\Ter{\inflambdaprefixcal},\, \alabel\in\setexp{0,1}^*}
    \displaybreak[0]
    \notag\\
    &  
    \stepslab 
      \defdby 
    \descsetexp{\triple{\labflabs{\alabel}{\vec{\bvar}}{\biter}}{\astep}{\labflabs{\alabel'}{\vec{\bvar}'}{\biter'}}}
               {\text{\eqref{eq1:prf:lem:derivations:Reg:stReg:2:scdelstrategies} holds}}   
    \notag\\
    & \hspace*{6ex}
    \left.    
      \parbox[c]{295pt}{there is an instance of ($\labscomp$), ($\lappcomp$), or ($\Vacstreg$) in $\stReglab$
                        with $\labflabs{\alabel}{\vec{\bvar}}{\biter}$ in the conclusion and 
                        the term $\labflabs{\alabel'}{\vec{\bvar}'}{\biter'}$ in the premise,
                        and with $\astep \funin \flabs{\vec{\bvar}}{\biter} \stregred \flabs{\vec{\bvar}'}{\biter'}$
                        (one of) the corresponding step(s) in $\stRegARS$}
    \right\}
    \label{eq1:prf:lem:derivations:Reg:stReg:2:scdelstrategies}
  \end{align}
  and where $\ssrclab,\, \stgtlab \funin \stepslab \to \Terlab{\inflambdaprefixcal}$
  are defined as projections on the first, and respectively, the third component of the triples that constitute steps in $\stepslab$. 
  Then the relation
  \begin{align*}
    \alabelling 
      \defdby &
    \descsetexp{\pair{\flabs{\bvar}{\biter}}{\labflabs{\alabel}{\bvar}{\biter}}}
               {\flabs{\bvar}{\biter}\in\Ter{\inflambdaprefixcal},\, \alabel\in\setexp{0,1}^*}
    \\
              & 
    {} \cup           
    \descsetexp{\pair{\astep}{\triple{\flabs{\bvar}{\biter}}{\astep}{\flabs{\bvar}{\biter}}}}
               {\triple{\flabs{\bvar}{\biter}}{\astep}{\flabs{\bvar}{\biter}}\in\stepslab}   
  \end{align*}
  is a labelling of $\stRegARS$ to $\stRegARSlab$.
  As initial labelling we choose the function $\ainitlabelling$ that is defined by
  $\ainitlabelling \funin \Ter{\inflambdaprefixcal} \to \Terlab{\inflambdaprefixcal},\,
    \flabs{\vec{\avar}}{\aiter} \mapsto \labflabs{\rootpos}{\vec{\avar}}{\aiter}$.
  and which adds the label `$\rootpos$'. 
  
Now we define the strategy
$ \astrati\Deriv = \tuple{\Terlab{\inflambdaprefixcal}, \,
                   \stepslabon{\Derivlab} \cup \stepslabnoton{\Derivlab}, \,
                   \ssrclab', \,
                   \stgtlab'} $  
with
\begin{align}
    &
    \stepslabon{\Derivlab}
      \defdby 
    \descsetexp{\triple{\labflabs{\alabel}{\vec{\bvar}}{\biter}}{\astep}{\labflabs{\alabel'}{\vec{\bvar}'}{\biter'}} \in \stepslab}
               {\text{\eqref{eq3:prf:lem:derivations:Reg:stReg:2:scdelstrategies} holds}}   
    \notag\\
    & \hspace*{6ex}
    \left.           
      \parbox[c]{295pt}{there is an instance of ($\labscomp$), ($\lappcomp$), or ($\Vacstreg$) in $\Derivlab$
                        with $\labflabs{\alabel}{\vec{\bvar}}{\biter}$ in the conclusion and 
                        the term $\labflabs{\alabel'}{\vec{\bvar}'}{\biter'}$ in the premise,
                        and with $\astep \funin \flabs{\vec{\bvar}}{\biter} \stregred \flabs{\vec{\bvar}'}{\biter'}$
                        (one of) the corresponding step(s) in $\stRegARS$}
    \right\}
    \label{eq3:prf:lem:derivations:Reg:stReg:2:scdelstrategies}
    \displaybreak[0]\\
    &
    \stepslabnoton{\Derivlab}
      \defdby 
    \descsetexp{\triple{\labflabs{\alabel}{\vec{\bvar}}{\biter}}{\astep}{\labflabs{\alabel'}{\vec{\bvar}'}{\biter'}} \in \stepslab}
               {\text{\eqref{eq4:prf:lem:derivations:Reg:stReg:2:scdelstrategies} holds}}   
    \notag\\
    & \hspace*{6ex}
    \left.           
      \parbox[c]{295pt}{$\labflabs{\alabel}{\vec{\bvar}}{\biter}$ does not occur in $\Derivlab$,
                        $\astep$ is a step according to $\astrat$}
    \right\}
    \label{eq4:prf:lem:derivations:Reg:stReg:2:scdelstrategies}
\end{align}  
where $\ssrclab'$, $\stgtlab'$ are the appropriate restrictions of $\ssrclab$ and $\stgtlab$.
  
  Note that, by its definition, $\astrati\Deriv$ is a sub\nb-\ARS\ of $\stRegARSlab$,
  Now for showing that $\astrati\Deriv$ is a (history-aware) strategy for $\stRegARSlab$,
  it has to be established that $\astrati\Deriv$ is a history-free strategy for the lifted version $\stRegARSlab$ of $\stRegARS$.
  For this it remains to show that every normal form of $\astrati\Deriv$ is also a normal form of $\stRegARSlab$.
  So, let $\labflabs{\alabel}{\vec{\bvar}}{\biter} \in \Ter{\inflambdaprefixcal}$ 
  be such that it is not a normal form of $\stRegARSlab$. 
  Then also $\flabs{\vec{\bvar}}{\biter}$ is not a normal form of $\stRegARS$.
  We will distinguish the cases that $\labflabs{\alabel}{\vec{\bvar}}{\biter}$
  occurs on $\Derivlab$ or not 
  for showing that there is a step in $\astrati\Deriv$ with this labelled term as a source.
  
  For the first case, we suppose that $\labflabs{\alabel}{\vec{\bvar}}{\biter}$ does not occur in $\Derivlab$.
  Then there is a step 
  $\astep \funin (\flabs{\vec{\bvar}}{\biter}) \stratred{\astrat} (\flabs{\vec{\bvar}'}{\biter'})$
  in the \scopedelimiting~strategy~$\astrat$ in $\stRegARS$, 
  which gives rise to the step 
  $\astep \funin (\labflabs{\alabel}{\vec{\bvar}}{\biter}) \red (\labflabs{\alabel'}{\vec{\bvar}'}{\biter'})$
  in $\stRegARSlab$ and in $\astrati\Deriv$. 
   
  For the second case, we suppose that $\labflabs{\alabel}{\vec{\bvar}}{\biter}$ occurs in $\Derivlab$,
  and we fix an occurrence $o$. 
  Since by assumption $\labflabs{\alabel}{\vec{\bvar}}{\biter}$ is not a normal form of $\stRegARSlab$,
  $o$ cannot be the occurrence of an axiom (\bvarax), and hence it is either an occurrence as the conclusion
  of an instance of one of the rules $(\labscomp)$, $(\lappcomp)$, $(\Vacstreg)$ in $\Derivlab$, or as a marked assumption in $\Derivlab$. 
  If $o$ is the conclusion of an instance $\iota$ of $(\labscomp)$, $(\lappcomp)$, or $(\Vacstreg)$,
  then $\iota$ defines a step on $\labflabs{\alabel}{\vec{\bvar}}{\biter}$ which also is a step in $\astrati\Deriv$.
  If $o$ is the conclusion of an instance of \FIX\ in $\Deriv$, then we consider an arbitrary path $\apath$ in $\Derivlab$
  from $o$ upwards towards a leaf of $\Derivlab$. Since, due to the side-condition of the rule \FIX, immediate subderivations
  of instances of \FIX\ consist of at least one rule application, $\apath$ cannot consist merely of applications of \FIX.
  Hence by following $\apath$ from $o$ upwards, after a number of successive instances of \FIX, each of which have
  $\labflabs{\alabel}{\vec{\bvar}}{\biter}$ as conclusion and premise, an instance of one of the rules $(\labscomp)$, $(\lappcomp)$, $(\Vacstreg)$ 
  follows, which witnesses a step with source $(\labscomp)$, $(\lappcomp)$, $(\Vacstreg)$ in in $\stRegARSlab$ and in $\astrati\Deriv$.
  Finally, if $o$ is an occurrence in a marked assumption at the top of the prooftree $\Derivlab$,
  then, since $\Derivlab$ is a closed derivation and due to the form of instances of the assumption-discharging rule \FIX, 
  there is also an occurrence $o'$ of $\labflabs{\alabel}{\vec{\bvar}}{\biter}$ as the conclusion of an instance of \FIX\ in $\Derivlab$.
  Now the argument above can be applied to the occurrence $o'$ to obtain a step of $\astrati\Deriv$ on $\labflabs{\alabel}{\vec{\bvar}}{\biter}$. 

By construction $\astrati\Deriv$ conforms to
(\ref{lem:derivations:Reg:stReg:2:scdelstrategies:item:i})
and
(\ref{lem:derivations:Reg:stReg:2:scdelstrategies:item:ii})
because of the inclusion of $\stepslabon{\Derivlab}$ and $\stepslabnoton{\Derivlab}$ respectively;
(\ref{lem:derivations:Reg:stReg:2:scdelstrategies:item:iii}) follows from
(\ref{lem:derivations:Reg:stReg:2:scdelstrategies:item:ii}).
\end{proof}

\begin{lemma}[from scope/scope$^+$-delim.\ strategies 
                                                   to $\Reg$/$\stReg$-derivations]%
  \label{lem:scdelstrategies:2:derivations:Reg:stReg}
  Let $\aiter\in\Ter{\inflambdacal}$, 
  and let $\astrat$ be a \scopedelimiting\ strategy for $\RegARS$ (a \extscopedelimiting\ strategy for $\stRegARS$)
  such that $\gSTstrat{\astrat}{\aiter}$ is finite. 
  Then there exists a closed derivation $\Deriv$ in $\Reg$ (in $\stRegzero$, and hence in $\stReg$) 
  with conclusion $\femptylabs{\aiter}$ 
  such that the following properties hold
  (note the minor differences with the items~(\ref{lem:derivations:Reg:stReg:2:scdelstrategies:item:i}), 
                                             (\ref{lem:derivations:Reg:stReg:2:scdelstrategies:item:ii}), 
                                         and (\ref{lem:derivations:Reg:stReg:2:scdelstrategies:item:iii})
  in Lemma~\ref{lem:derivations:Reg:stReg:2:scdelstrategies}):
  \begin{enumerate}[(i)]
    \item\label{lem:scdelstrategies:2:derivations:Reg:stReg:item:i} 
      Every (non-cyclic) path in $\Deriv$ from the conclusion upwards to a leaf of the prooftree~$\Deriv$ corresponds
      to a $\stratred{\astrat}$\nb-rewrite sequence 
                                                    starting on $\femptylabs{\aiter}$
      where passes over instances of the rules ($\lappcomp$) to the left and to the right correspond
      to $\slappdecompistratred{0}{\astrat}$- and $\slappdecompistratred{1}{\astrat}$-steps, respectively,
      and passes over instances of ($\labscomp$) and of ($\Vacreg$) (of ($\Vacstreg$))
      correspond to $\slabsdecompstratred{\astrat}$-, and $\scompressregstratred{\astrat}$\nb-steps ($\scompressstregstratred{\astrat}$\nb-steps);
      passes from the conclusion to the premise of instances of \FIX\ correspond to empty steps. 
    \item\label{lem:scdelstrategies:2:derivations:Reg:stReg:item:ii}
      Every sufficiently long $\stratred{\astrat}$\nb-rewrite sequence on $\flabs{\vec{\avar}}{\aiter}$ 
      has an initial segment that corresponds to a (non-cyclic) path in $\Deriv$ from the conclusion upwards
      to a leaf of the prooftree: 
      thereby a $\slappdecompistratred{0}{\astrat}$- or $\slappdecompistratred{1}{\astrat}$\nb-step
      corresponds to a pass over (possibly some \FIX-instances followed by) an instance of ($\lappcomp$) in direction left and right, respectively;
      a $\slabsdecompstratred{\astrat}$- or $\scompressregstratred{\astrat}$\nb-step ($\scompressstregstratred{\astrat}$\nb-step)
      corresponds to a pass over (possibly some \FIX-instances followed by) an instance of ($\labscomp$) or ($\Vacreg$) (of ($\Vacstreg$)), respectively. 
      %
    \item\label{lem:scdelstrategies:2:derivations:Reg:stReg:item:iii}
      $\gSTstrat{\astrat}{\aiter} 
         \subseteq 
       \descsetexp{ \flabs{\vec{\bvar}}{\biter} }
                  {\text{the term $\flabs{\vec{\bvar}}{\biter}$ occurs in $\Deriv$}}$.
  \end{enumerate}  
\end{lemma}

\begin{proof}{\label{prf:lem:scdelstrategies:2:derivations:Reg:stReg}}
  We will argue only for the part of the statement of the lemma concerning a \extscopedelimiting\ strategy for $\stRegARS$,
  since the case with a \scopedelimiting\ strategy for $\RegARS$ can be established analogously.
  
  Let $\aiter$ be an infinite \lambdaterm, and let $\astrat$ be a \extscopedelimiting\ strategy for $\RegARS$
  such that $\gSTstrat{\astrat}{\aiter}$ is finite.
  Now let $\Derivi{0}$ be the (trivial) derivation with conclusion $\femptylabs{\aiter}$,
  which, in case that this is not an axiom of $\stRegzero$ (and $\stReg$), is also an assumption, and then
  is of the form $(\femptylabs{\aiter})^u$, carrying an assumption marker $u$. 
  If $\Derivi{0}$ is an axiom, then it is easy to verify that the
  statements~(\ref{lem:scdelstrategies:2:derivations:Reg:stReg:item:i}), 
             (\ref{lem:scdelstrategies:2:derivations:Reg:stReg:item:ii}), 
             and (\ref{lem:scdelstrategies:2:derivations:Reg:stReg:item:iii}) hold.
  
  Otherwise we construct a sequence $\Derivi{1}$, $\Derivi{2}$, \ldots\
  of derivations  
  where each $\Derivi{n}$ satisfies the properties~(\ref{lem:scdelstrategies:2:derivations:Reg:stReg:item:i}), 
                                               (\ref{lem:scdelstrategies:2:derivations:Reg:stReg:item:ii}),
                                               and (\ref{lem:scdelstrategies:2:derivations:Reg:stReg:item:iii}),
  terms in marked assumptions are not also terms in axioms $\bvarax$,                                            
  and where $\Derivi{n+1}$ extends $\Derivi{n}$ by one additional rule instance
  above a marked assumption in $\Derivi{n}$:
  For the extension step on a derivation $\Derivi{n}$, a marked assumption 
  $(\flabs{\vec{\bvar}}{\biter})^{\amarker}$ in $\Derivi{n}$ is picked
  with the property that the term $\flabs{\vec{\bvar}}{\biter}$ does not appear in the thread
  down to the conclusion of $\Derivi{n}$.
  
  Suppose that the $\stratred{\astrat}$\nb-rewrite sequence
  from the conclusion of $\Derivi{n}$ up to the marked assumption is of the form:
  \begin{equation*}
    \arewseq \funin
    \femptylabs{\aiter} = \flabs{\vec{\avar}_0}{\aiteri{0}} 
      \stratred{\astrat}
    \flabs{\vec{\avar}_1}{\aiteri{1}}
      \stratred{\astrat}
    \ldots
      \stratred{\astrat}
    \flabs{\vec{\avar}_m}{\aiteri{m}} = \flabs{\vec{\bvar}}{\biter} 
  \end{equation*}
  Note that, since by assumption $\flabs{\vec{\bvar}}{\biter}$ is not a term in an axiom $\bvarax$ of $\stRegzero$, 
  it follows by Proposition~\ref{prop:rewprops:RegCRS:stRegCRS}, (\ref{prop:rewprops:RegCRS:stRegCRS:item:v}), 
  that it is not a $\sstregred$\nb-normal form. 
  Then depending on whether the possible next step(s) in an $\stratred{\astrat}$\nb-rewrite
  that extends $\arewseq$ 
  by one step
  is a $\slabsdecompstratred{\astrat}$-, $\scompressregstratred{\astrat}$\nb-step,
  or either a $\slappdecompistratred{0}{\astrat}$- or a $\slappdecompistratred{1}{\astrat}$\nb-steps,
  the derivation $\Derivi{n}$ is extended above the marked assumption $(\flabs{\vec{\bvar}}{\biter})^{\amarker}$
  by an application of $\labscomp$, $\Vacstreg$, or $\lappcomp$, respectively.   
  For example in the case that $\arewseq$ extends by one additional step to either of the two rewrite sequences:
  \begin{equation*}
    \arewseq_i \funin
    \femptylabs{\aiter} = \flabs{\vec{\avar}_0}{\aiteri{0}} 
      \stratred{\astrat}
    \ldots
      \stratred{\astrat}
    \flabs{\vec{\avar}_m}{\aiteri{m}} = \flabs{\vec{\avar}_m}{\lapp{\aiteri{m,0}}{\aiteri{m,1}}}
      \lappdecompistratred{i}{\astrat}
    \flabs{\vec{\avar}_m}{\aiteri{m,i}}
  \end{equation*}
  with $i\in\setexp{0,1}$, 
  the derivation $\Derivi{n}$ of the form: 
  \begin{gather*}
    \hspace*{-2.5ex}
    \begin{gathered}[c]
      \mbox{
        \AxiomC{$ \langle (\flabs{\vec{\avar}_m}{\lapp{\aiteri{m,0}}{\aiteri{m,1}}})^{\amarker} \rangle $}
        \noLine
        \UnaryInfC{$\Derivi{n}$}
        \noLine
        \UnaryInfC{$ \femptylabs{\aiter} $}
        \DisplayProof
            }
    \end{gathered}
    \text{is extended to $\Derivi{n+1}$:}\hspace*{-1.5ex}
    \begin{gathered}[c]
      \mbox{
        \AxiomC{$ (\flabs{\vec{\avar}_{m}}{\aiteri{m,0}})^{\amarkeri{0}} $}
        \AxiomC{$ (\flabs{\vec{\avar}_{m}}{\aiteri{m,1}})^{\amarkeri{1}} $}
        \insertBetweenHyps{\hspace*{1em}}
        \RightLabel{$\lappcomp$}
        \BinaryInfC{$ \langle \flabs{\vec{\avar}_m}{\lapp{\aiteri{m,0}}{\aiteri{m,1}}} \rangle $}
        \noLine
        \UnaryInfC{$\Derivi{n}$}
        \noLine
        \UnaryInfC{$ \femptylabs{\aiter} $}
        \DisplayProof
            }
    \end{gathered}
  \end{gather*}
  for two fresh assumption markers $\amarkeri{0}$ and $\amarkeri{1}$
  (the angle brackets $\langle\ldots\rangle$ are used here to indicate just a single formula occurrence at the top of the prooftree $\Derivi{n}$). 
  If either of $\flabs{\vec{\avar}_{m}}{\aiteri{m,0}}$ or $\flabs{\vec{\avar}_{m}}{\aiteri{m,1}}$ is an axiom,
  then the assumption marker is removed and the formula is marked as an axiom $\bvarax$, accordingly. 
  Note that, if 
  the statements~(\ref{lem:scdelstrategies:2:derivations:Reg:stReg:item:i}),
                 (\ref{lem:scdelstrategies:2:derivations:Reg:stReg:item:ii}), and
                 (\ref{lem:scdelstrategies:2:derivations:Reg:stReg:item:iii}) 
  are satisfied for $\Deriv = \Derivi{n}$,
  then this is also the case for $\Deriv = \Derivi{n+1}$.
  Furthermore, terms in marked assumptions are not terms in axioms of $\stRegzero$. 
  
  The extension process continues as long as $\Derivi{n}$ contains a marked assumption 
  $(\flabs{\vec{\bvar}}{\biter})^{\bmarker}$ without a `$\stRegzero$-admissible repetition' beneath it,
  by which we mean the occurrence $o$ of the term $\flabs{\vec{\bvar}}{\biter}$ on the
  thread down to the conclusion in $\Deriv$, but strictly beneath the marked assumption,
  such that furthermore all terms on the part of the thread down to $o$ have
  an abstraction prefix of length greater or equal to $\length{\vec{\bvar}}$. 
  (Note the connection to the side-condition on instances of the rule \FIX\ in $\stRegzero$,
   and, in particular, that marked assumptions with an $\stRegzero$-admissible repetition beneath it
   could be discharged by an appropriately introduced instance of \FIX\ in $\stRegzero$.)
%

  That the extension process terminates can be seen as follows:
  Suppose that, to the contrary, it continues indefinitely. 
  Then, since the derivation size increases strictly in every step,
  an infinite prooftree $\Deriv^{\infty}$ is obtained in the limit,
  which due to finite branchingness of the prooftree and K\H{o}nig's Lemma 
  possesses an infinite path $\apath$ starting at the conclusion.
  Now note that due to (\ref{lem:scdelstrategies:2:derivations:Reg:stReg:item:i}),
  $\apath$ corresponds to an infinite $\stratred{\astrat}$\nb-rewrite sequence.
  Due to Proposition~\ref{prop:grounded:cycle}, this infinite rewrite sequence
  must contain a grounded cycle. 
  However, the existence of such grounded cycle contradicts the
  termination condition of the extension process, because every grounded cycle
  provides an $\stRegzero$\nb-admissible repetition. 
                      
  Let $\Derivi{N}$, for some $N\in\nats$, be the derivation that is reached when no further extension step, as described, is possible. 
  By the construction
  the statements~(\ref{lem:scdelstrategies:2:derivations:Reg:stReg:item:i}),
                 (\ref{lem:scdelstrategies:2:derivations:Reg:stReg:item:ii}), and
                 (\ref{lem:scdelstrategies:2:derivations:Reg:stReg:item:iii}) 
  are satisfied for $\Deriv = \Derivi{N}$.
  Furthermore, $\Derivi{N}$ is a derivation in $\stReg$ and $\stRegzero$ 
  in which every leaf at the top is either an axiom $\bvarax$ 
  or an assumption $(\flabs{\vec{\bvar}}{\biter})^{\amarker}$ marked with a unique marker $\amarker$,
  and for every such marked assumption in $\Derivi{N}$,
  there is a $\stRegzero$\nb-admissible repetition strictly beneath it.
  This fact enables us to modify $\Derivi{N}$
  into a closed derivation in $\stRegzero$ by closing all open assumptions by newly introduced applications
  of the rule \FIX. More precisely, steps of the following kind are carried out repeatedly.
  A derivation with occurrences of a number of marked assumptions $(\flabs{\vec{\bvar}}{\biter})^{\amarkeri{i}}$
  highlighted together with a single occurrence of the term $\flabs{\vec{\bvar}}{\biter}$ in its interior 
  that indicates the $\stRegzero$\nb-admissible repetition for the displayed marked assumptions:
  \begin{gather*}
    \hspace*{-2.5ex}
    \begin{gathered}[c]
      \mbox{
        \AxiomC{$\langle (\flabs{\vec{\bvar}}{\biter})^{\amarkeri{1}} \rangle$}
        \insertBetweenHyps{\ldots}
        \AxiomC{$\langle (\flabs{\vec{\bvar}}{\biter})^{\amarkeri{k}} \rangle$}
        \noLine
        \BinaryInfC{$\Derivi{000}$}
        \noLine
        \UnaryInfC{$\langle \flabs{\vec{\bvar}}{\biter} \rangle$}
        \noLine
        \UnaryInfC{$\Derivi{00}$}
        \noLine
        \UnaryInfC{$\femptylabs{\aiter}$}
        \DisplayProof
            }
    \end{gathered}
    \text{is modified into:}\hspace*{-1.5ex}
    \begin{gathered}[c]
      \mbox{
        \AxiomC{$\langle (\flabs{\vec{\bvar}}{\biter})^{\cmarker} \rangle$}
        \insertBetweenHyps{\ldots}
        \AxiomC{$\langle (\flabs{\vec{\bvar}}{\biter})^{\cmarker} \rangle$}
        \noLine
        \BinaryInfC{$\Derivi{000}$}
        \noLine
        \UnaryInfC{$ \flabs{\vec{\bvar}}{\biter} $}
        \RightLabel{\sFIX, $\cmarker$}
        \UnaryInfC{$\langle \flabs{\vec{\bvar}}{\biter} \rangle$}
        \noLine
        \UnaryInfC{$\Derivi{00}$}
        \noLine
        \UnaryInfC{$\femptylabs{\aiter}$}
        \DisplayProof
            }
    \end{gathered}
  \end{gather*}
  where $\cmarker$ is a fresh assumption marker.
  In every such transformation step the number of open assumptions is strictly decreased,
  but
  the properties~(\ref{lem:scdelstrategies:2:derivations:Reg:stReg:item:i}),
                 (\ref{lem:scdelstrategies:2:derivations:Reg:stReg:item:ii}), and
                 (\ref{lem:scdelstrategies:2:derivations:Reg:stReg:item:iii}) 
  (for $\Deriv$ the resulting derivation of such a step) is preserved.               
  Hence after finitely many such transformation steps a derivation $\Deriv$ in $\stRegzero$
  without open assumptions and with 
  the properties~(\ref{lem:scdelstrategies:2:derivations:Reg:stReg:item:i}),
                 (\ref{lem:scdelstrategies:2:derivations:Reg:stReg:item:ii}), and
                 (\ref{lem:scdelstrategies:2:derivations:Reg:stReg:item:iii}) 
  is reached, and obtained as the result of this construction.
\end{proof}

As a consequence of the two lemmas above, derivability in $\stReg$ and in $\stRegzero$ coincides.

\begin{proposition}\label{prop:stReg:stRegzero}
  For all infinite \lambdaterms~$\aiter$: \mbox{} 
  $\derivablein{\stReg}{\femptylabs{\aiter}}$
    if and only if 
  $\derivablein{\stRegzero}{\femptylabs{\aiter}}$.
\end{proposition}

\begin{proof}
  The direction ``$\Leftarrow$'' follows by the fact that every derivation in $\stRegzero$
  is also a derivation in $\stReg$.
  For the direction ``$\Rightarrow$'', let $\aiter$ be an infinite term such that $\derivablein{\stReg}{\femptylabs{\aiter}}$.   
  By Lemma~\ref{lem:derivations:Reg:stReg:2:scdelstrategies} there exists a \extscopedelimiting~strategy~$\astratplus$
  such that $\gSTstrat{\astratplus}{\aiter}$ is finite.  
  But then it follows by Lemma~\ref{lem:scdelstrategies:2:derivations:Reg:stReg}
  that there is also a closed derivation in $\stRegzero$ with conclusion $\femptylabs{\aiter}$,
  and hence that $\derivablein{\stRegzero}{\femptylabs{\aiter}}$.
\end{proof}

Now we have assembled all auxiliary statements that we use for proving a theorem that 
tightly links derivability in
the proof system $\Reg$ with regularity, and derivability in $\stReg$ and in $\stRegzero$
with strong regularity, of infinite \lambdaterms. 

\begin{theorem}\label{lem:Reg:stReg}
  The following statements hold for the proof systems 
  $\Reg$, $\stReg$, $\stRegzero$:
  \begin{enumerate}[(i)]
    \item\label{lem:Reg:stReg:item:i} 
      $\Reg$ is sound and complete for regular infinite \lambdaterms. That is, for all $\aiter\in\Ter{\inflambdacal}$:
      \begin{equation*}
        \derivablein{\Reg}{\femptylabs{\aiter}}
          \hspace*{5ex}\text{if and only if}\hspace*{5ex}
        \text{$\aiter$ is regular.}
      \end{equation*}
    \item\label{lem:Reg:stReg:item:ii} 
       $\stReg$ and $\stRegzero$ are sound and complete for strongly regular infinite \lambdaterms. That is,
       for all $\aiter\in\Ter{\inflambdacal}$ the following assertions are equivalent:
       \begin{enumerate}[(a)]
         \item 
           $\aiter$ is strongly regular.
         \item 
           $\derivablein{\stReg}{\femptylabs{\aiter}}$.
         \item 
          $\derivablein{\stRegzero}{\femptylabs{\aiter}}$.    
       \end{enumerate}
  \end{enumerate}
\end{theorem}

\begin{proof}
  Since the proof of statement~(\ref{lem:Reg:stReg:item:ii}) of the theorem can be carried out analogously
  (taking into account Proposition~\ref{prop:stReg:stRegzero}), 
  we 
  argue here only for statement~(\ref{lem:Reg:stReg:item:i}).
  
  For ``$\Rightarrow$'' in (\ref{lem:Reg:stReg:item:i}),
  let $\aiter$ be an infinite \lambdaterm\ that is regular.
  Then there exists a \scopedelimiting\ strategy~$\astrat$ on $\RegARS$ such that
  $\gSTstrat{\astrat}{\aiter}$ is finite. 
  By Lemma~\ref{lem:scdelstrategies:2:derivations:Reg:stReg} it follows
  that there exists a closed derivation~$\Deriv$ in $\Reg$ with conclusion~$\femptylabs{\aiter}$.
  This derivation witnesses $\derivablein{\Reg}{\femptylabs{\aiter}}$.  
  For ``$\Leftarrow$'' in (\ref{lem:Reg:stReg:item:i}), suppose that $\derivablein{\Reg}{\femptylabs{\aiter}}$.
  Then there exists a closed derivation $\Deriv$ in $\Reg$ with conclusion~$\femptylabs{\aiter}$.
  Now Lemma~\ref{lem:derivations:Reg:stReg:2:scdelstrategies} entails the existence of
  a \scopedelimiting\ strategy~$\astrat$ in~$\stRegARS$ such that, 
  in particular, $\gSTstrat{\astrat}{\aiter}$ is finite.
  This fact implies that $\aiter$ is regular. 
\end{proof}

\begin{figure}[t!]
\begin{center}\framebox{\begin{minipage}{350pt}\begin{center}
  \mbox{}
\\[0.5ex]
  \mbox{ 
    \AxiomC{}
    \RightLabel{\bvarax}
    \UnaryInfC{$ \flabsCRS{n+1}{\vecOneToN{\avar}\bvar}{\bvar}  = \flabsCRS{n+1}{\vecOneToN{\cvar}\dvar}{\dvar}  $}
    \DisplayProof
        } 
\\[2ex]
  \mbox{
    \AxiomC{$ \flabsCRS{n}{\vecOneToN{\avar}}{\apreter}  =  \flabsCRS{n}{\vecOneToN{\cvar}}{\bpreter} $}
    \RightLabel{$\Vacstreg\;\;$ \parbox[c]{125pt}{(if $\bvar$ does not occur free in $\apreter$,\\[-0.5ex]
                                                  and $\fvar$ does not occur free in $\bpreter$)}}
    \UnaryInfC{$ \flabsCRS{n+1}{\vecOneToN{\avar}\bvar}{\apreter} = \flabsCRS{n+1}{\vecOneToN{\cvar}\fvar}{\bpreter} $}
    \DisplayProof
        }     
\\[2ex]
  \mbox{
    \AxiomC{$ \flabsCRS{n+1}{\vecOneToN{\avar}\bvar}{\apreter} 
                   = 
              \flabsCRS{n+1}{\vecOneToN{\cvar}\dvar}{\bpreter} $}
    \RightLabel{$\labscomp$}
    \UnaryInfC{$ \flabsCRS{n}{\vecOneToN{\avar}}{\labsCRS{\bvar}{\apreter}} 
                   = 
                 \flabsCRS{n}{\vecOneToN{\cvar}}{\labsCRS{\dvar}{\bpreter}} $}
    \DisplayProof
        }
\\[2ex]
  \mbox{
    \AxiomC{$ \flabsCRS{n}{\vecOneToN{\avar}}{\apreteri0}  =  \flabsCRS{n}{\vecOneToN{\bvar}}{\bpreteri0} $\hspace{-2.5ex}}
    \AxiomC{$ \flabsCRS{n}{\vecOneToN{\avar}}{\apreteri1}  =  \flabsCRS{n}{\vecOneToN{\bvar}}{\bpreteri1} $}
    \RightLabel{$\lappcomp$}
    \BinaryInfC{$ \flabsCRS{n}{\vecOneToN{\avar}}{\lapp{\apreteri0}{\apreteri1}} 
                   = 
                  \flabsCRS{n}{\vecOneToN{\bvar}}{\lapp{\bpreteri0}{\bpreteri1}} $}
    \DisplayProof 
    }     
\\[0.5ex]
  \mbox{}  
\end{center}\end{minipage}}\end{center} 
\caption{\label{fig:proof-equiv-preter}\label{fig:AlphaPreTer}
  Proof system $\AlphaPreTer$ for equality of preterms in 
                                                          $\inflambdaprefixcal$ modulo $\alphaequiv$.
}
\end{figure}   

\begin{figure}[t!]
\begin{center}\framebox{\begin{minipage}{350pt}\begin{center}
  \mbox{}
  \\[1.5ex]
  \mbox{ 
    \AxiomC{}
    \RightLabel{\bvarax}
    \UnaryInfC{$ \flabs{\vec{\avar}\bvar}{\bvar}  = \flabs{\vec{\cvar}\dvar}{\dvar}  $}
    \DisplayProof
        } 
  \hspace*{1ex}     
  \mbox{
    \AxiomC{$ \flabs{\vec{\avar}}{\ater}  =  \flabs{\vec{\cvar}}{\bter} $}
    \RightLabel{$\Vacstreg\;\;$ \parbox[c]{110pt}{(if $\bvar$ does not occur in $\ater$,
                                                  and $\fvar$ does not occur in $\bter$)}}
    \UnaryInfC{$ \flabs{\vec{\avar}\bvar}{\ater} = \flabs{\vec{\cvar}\fvar}{\bter} $}
    \DisplayProof
        }     
  \\[2.5ex]    
  \mbox{
    \AxiomC{$ \flabs{\vec{\avar}\bvar}{\ater} 
                   = 
              \flabs{\vec{\cvar}\dvar}{\bter} $}
    \RightLabel{$\labscomp$}
    \UnaryInfC{$ \flabs{\vec{\avar}}{\labs{\bvar}{\ater}} 
                   = 
                 \flabs{\vec{\cvar}}{\labs{\dvar}{\bter}} $}
    \DisplayProof
        }
  \hspace*{3.5ex} 
  \mbox{
    \AxiomC{$ \flabs{\vec{\avar}}{\ateri0}  =  \flabs{\vec{\bvar}}{\bteri0} $}
    \AxiomC{$ \flabs{\vec{\avar}}{\ateri1}  =  \flabs{\vec{\bvar}}{\bteri1} $}
    \RightLabel{$\lappcomp$}
    \BinaryInfC{$ \flabs{\vec{\avar}}{\lapp{\ateri0}{\ateri1}} 
                   = 
                  \flabs{\vec{\bvar}}{\lapp{\bteri0}{\bteri1}} $}
    \DisplayProof 
    }     
  \\[1ex]
  \mbox{}  
\end{center}\end{minipage}}\end{center}
\caption{\label{fig:proof-equiv-ter}\label{fig:EqTer}
  Proof system $\EqTer$ for equality of terms in 
                                                 $\inflambdaprefixcal$ in informal notation.}
\end{figure}

For defining, shortly, of a proof system for equality of strongly regular infinite \lambdaterms,
we first give a specialised version of the proof system $\Alpha$ for $\alpha$\nb-equivalence of \iCRS-preterms
for preterms, and a corresponding system on terms, in \inflambdaprefixcal.  

\begin{definition}[proof systems $\AlphaPreTer$, $\EqTer$]\label{def:equiv-preter}\label{def:AlphaPreTer:EqTer}
  The proof system $\AlphaPreTer$ for $\alpha$-equi\-va\-lence of infinite preterms in $\inflambdaprefixcal$
  consists of the rules displayed in Figure~\ref{fig:AlphaPreTer}. 
  The proof system~$\EqTer$ for equality of infinite terms in $\inflambdaprefixcal$
  consists of the rules displayed in Figure~\ref{fig:EqTer}. 
  Provability in $\AlphaPreTer$ and in $\EqTer$ is defined, analogously to the proof system $\Alpha$ from Definition~\ref{def:Alpha}, 
  as the existence of a completed (possibly infinite) derivation, and will, 
  as was first done so for $\Alpha$ in Definition~\ref{def:AlphaInfPreterm},
  be indicated using the symbol $\sinfderivable$.
\end{definition}

\begin{proposition}\label{prop:AlphaPreTer:EqTer}
  The following 
                statements hold for the proof systems $\AlphaPreTer$ and $\EqTer$.
  \begin{enumerate}[(i)]
    \item{}\label{prop:AlphaPreTer:EqTer:item:i}
      $\AlphaPreTer$ is sound and complete for $\salphaequiv$ on $\inflambdaprefixcal$\nb-preterms.
      That is:
      \begin{equation*}
        \infderivablein{\AlphaPreTer}{\femptylabsCRS{\apreter} = \femptylabsCRS{\bpreter}}
          \hspace*{5ex}\text{if and only if}\hspace*{5ex}
          \aipreter \alphaeq \bipreter \punc{.}
      \end{equation*}
      holds for all closed preterms $\apreter$ and $\bpreter$ in $\inflambdaprefixcal$. 
    \item{}\label{prop:AlphaPreTer:EqTer:item:ii}
      $\EqTer$ is sound and complete for equality between $\inflambdaprefixcal$\nb-terms.
      That is: 
      \begin{equation*}
        \infderivablein{\EqTer}{\aiter = \biter}
          \hspace*{5ex}\text{if and only if}\hspace*{5ex}
        \aiter = \biter \punc{.}
      \end{equation*}
      holds for all terms $\aiter$ and $\biter$ in $\inflambdaprefixcal$.
  \end{enumerate}
\end{proposition}

\begin{proof}
  For statement~(\ref{prop:AlphaPreTer:EqTer:item:i}) it suffices to show
  that, for an equation $\femptylabsCRS{\apreter} = \femptylabsCRS{\bpreter}$ between preterms of $\inflambdaprefixcal$, 
  derivability in $\AlphaPreTer$
    coincides with 
  derivability of this equation in the general proof system $\Alpha$ for $\alpha$\nb-equivalence 
  between \iCRS\nb-preterms in Definition~\ref{def:Alpha}.
  Given a derivation $\infDeriv$ in $\AlphaPreTer$, a derivation $\infDeriv_1$ in $\Alpha$
  results by replacing each formula occurrence 
  $\flabsCRS{n}{\vecOneToN{\avar}}{\apreter}  =  \flabsCRS{n}{\vecOneToN{\bvar}}{\bpreter}$
  by the formula occurrence
  $\fabsCRS{\vecOneToN{\avar}}{\apreter}  =  \fabsCRS{\vecOneToN{\bvar}}{\bpreter}$
  and adding an instance of the rule for the function symbol $\sflabsCRS{0}$ at the bottom.
  Then instances of the axioms and rules ($\bvarax$), ($@$), ($\labscomp$), and ($\Vacstreg$) in $\infDeriv$
  correspond to instances of axioms and rules ($\bvarax$), ($\slappCRS$), ($\absCRS{\hspace*{1pt}}{\niks}$), and ($\Vacstreg$) in $\infDeriv_1$,
  respectively.
  This proof transformation also has an inverse.
  
  For ``$\Leftarrow$'' in statement~(\ref{prop:AlphaPreTer:EqTer:item:ii})
  it suffices to note: Every \extscopedelimiting~strategy on $\stRegARS$
  can be used to stepwise extend finite derivations in $\EqTer$ 
  with conclusion $\femptylabs{\aiter} = \femptylabs{\aiter}$
  by one additional rule application above a leaf containing a formula
  $\flabs{\bvar}{\biter} = \flabs{\bvar}{\biter}$ that is not an axiom $\bvarax$,
  which implies that $\flabs{\bvar}{\biter}$ is not a normal form of $\sstregred$. 
  If these extensions are carried out
  in a fair manner by extending all non-axiom leafs at depth $n$ in the prooftree
  before proceeding with leafs at depth $>n$, then in the limit a completed
  derivation in $\EqTer$ is obtained.
  
  For ``$\Rightarrow$'' in statement~(\ref{prop:AlphaPreTer:EqTer:item:ii}), 
  suppose that $\infDeriv$ is a completed derivation in $\EqTer$
  with conclusion $\femptylabs{\aiter} = \femptylabs{\biter}$. 
  Let $\femptylabsCRS{\aipreter}$ and $\femptylabsCRS{\bipreter}$ be preterm representatives of $\femptylabs{\aiter}$ and $\femptylabs{\biter}$, respectively. 
  Now a completed derivation $\infDeriv_{\text{pter}}$ in $\AlphaPreTer$ can be found by developing it step by step from the conclusion 
  $\femptylabsCRS{\aipreter} = \femptylabsCRS{\bipreter}$ upwards, parallel to $\infDeriv$, and following the rules of $\AlphaPreTer$, which are invertible
  (that is, the premise/s of a rule instance is/are uniquely determined by the conclusion).  
  Then $\infDeriv_{\text{pter}}$ is a preterm representative version of $\infDeriv$. 
  By using (\ref{prop:AlphaPreTer:EqTer:item:i}), it follows that $\aipreter \alphaequiv \bipreter$.
  Since $\aipreter$ and $\bipreter$ are preterm representatives of $\aiter$ and $\biter$, respectively,
  $\aiter = \biter$ follows. 
\end{proof}

\begin{definition}[the proof system $\stRegeq$]\label{def:stRegeq}
  The natural-deduction-style proof system $\stRegeq$ for equality of strongly regular, infinite $\lambda$\nb-terms
  has all the rules of the proof system $\EqTer$ from Definition~\ref{def:AlphaPreTer:EqTer} and Figure~\ref{fig:EqTer},
  and additionally, the rule \FIX\ in Figure~\ref{fig:stRegeq}.
  But contrary to the definition in $\EqTer$, provability of an equation $\flabs{\vec{\avar}}{\aiter} = \flabs{\vec{\avar}}{\biter}$
  in $\stRegeq$ is defined as the existence of a \underline{finite} closed derivation
  with conclusion $\flabs{\vec{\avar}}{\aiter} = \flabs{\vec{\avar}}{\biter}$.
\end{definition}
\begin{figure}[t!]
\begin{center}\framebox{\begin{minipage}{200pt}\begin{center}
\mbox{}
\\[1.5ex]
\mbox{
    \AxiomC{$ [\flabs{\vec{\avar}}{\aiter} = \flabs{\vec{\bvar}}{\biter}]^u$}
    \noLine
    \UnaryInfC{$\Derivi{0}$}
    \noLine
    \UnaryInfC{$ \flabs{\vec{\avar}}{\aiter} = \flabs{\vec{\bvar}}{\biter} $}
    \RightLabel{$\sFIX,u$ $\,$
                \parbox{49pt}{\small (if $\depth{\Derivi{0}} \ge 1$)}
                }                   
    \UnaryInfC{$ \flabs{\vec{\avar}}{\aiter} = \flabs{\vec{\bvar}}{\biter} $}
    \DisplayProof
    }
  \\[1ex]
  \mbox{}  
\end{center}\end{minipage}}\end{center} 
\caption{\label{fig:stRegeq}
  The rule \FIX, which is added to the rules of $\EqTer$ from Figure~\ref{fig:EqTer}
  in order to obtain the proof system $\stRegeq$ for equality of strongly regular infinite $\lambda$\nb-terms.  
  }
\end{figure}

\begin{theorem}\label{thm:stRegeq}
  $\stRegeq$ is sound and complete for equality between strongly regular, infinite \lambdaterms.
  That is,
  for all strongly regular, infinite \lambdaterms\ $\aiter$ and $\biter$ it holds:
\[
\derivablein{\stRegeq}{\aiter = \biter}
\hspace{1cm}\text{if and only if}\hspace{1cm}
   \aiter = \biter \punc{.}
\]
\end{theorem}

\begin{proof}[Sketch of the Proof.]\label{prf:lem:stRegeq}
   Let $\aiter$ and $\biter$ be strongly regular, infinite \lambdaterms.
   In view of Proposition~\ref{prop:AlphaPreTer:EqTer}, (\ref{prop:AlphaPreTer:EqTer:item:ii}), 
   it suffices to show: 
   \begin{equation}\label{eq:prf:lem:stRegeq}
     \derivablein{\stRegeq}{\aiter = \biter}
       \hspace{1cm}\text{if and only if}\hspace{1cm}
     \infderivablein{\EqTer}{\aiter = \biter} \punc{.}
   \end{equation}
   
   For showing ``$\Leftarrow$'' in \eqref{eq:prf:lem:stRegeq}, 
   let $\infDeriv$ be a derivation in $\EqTer$ with conclusion $\femptylabs{\aiter} = \femptylabs{\biter}$.
   Since paths in $\infDeriv$ correspond to $\stregred$\nb-rewrite sequences, and since the number of
   generated subterms of both $\aiter$ and $\biter$ are finite (as a consequence of their strong regularity),
   on every infinite thread equation repetitions occur. These repetitions can be used to cut all infinite
   threads by appropriate introductions of instances of \FIX\ in order to obtain a finite and closed derivation in $\stRegeq$
   with the same conclusion. 
   
   For showing ``$\Rightarrow$'' in \eqref{eq:prf:lem:stRegeq}, 
   let $\Deriv$ be a closed derivation in $\stRegeq$ with conclusion 
   $\femptylabs{\aiter} = \femptylabs{\biter}$. Now $\Deriv$ can be unfolded into an infinite derivation $\infDeriv$ in $\EqTer$
   by repeatedly removing a bottommost instance of \FIX\ and inserting its immediate subderivation above each of the marked assumptions
   the instance discharges. If this process is organised in a fair manner with respect to bottommost instances of \FIX,
   then in the limit an infinite completed prooftree with conclusion $\femptylabs{\aiter} = \femptylabs{\biter}$ in $\EqTer$ is obtained. 
   For productivity of this process it is decisive that the side-condition on every instance $\ainst$ of the rule \FIX\ guarantees that on threads from
   the conclusion of $\ainst$ to a marked assumption that is discharged by $\ainst$ at least one instance of a rule different from
   \FIX\ is passed. 
\end{proof}

\begin{figure}[t!]
\begin{center}  
  \framebox{
\begin{minipage}{365pt}
\begin{center}
  \mbox{}
  \\[0.75ex]
  \mbox{
    \AxiomC{$ \{\: [ \flabs{\vec{\avar}}{\aconstnamei{\arecvari{i}}} ]^{\amarkeri{i}} \:\}_{i=1,\ldots,n} $}
    \noLine
    \UnaryInfC{$\Derivi{j}$}
    \noLine
    \UnaryInfC{$ \{ \:\ldots\ldots \;\; \flabs{\vec{\avar}}{\subst{\allteri{j}}{\vec{\arecvar}}{\vec{\aconstname}_{\vec{\arecvar}}}} \;\;\ldots\ldots\: \}_{j=1,\ldots,n} $}
    \AxiomC{$ \{\: [ \flabs{\vec{\avar}}{\aconstnamei{\arecvari{i}}} ]^{\amarkeri{i}} \:\}_{i=1,\ldots,n} $}
    \noLine
    \UnaryInfC{$\Derivi{n+1}$}
    \noLine
    \UnaryInfC{$ \flabs{\vec{\avar}}{\subst{\allteri{n+1}}{\vec{\arecvar}}{\vec{\aconstname}_{\vec{\arecvar}}}} $}
    \RightLabel{\sFIXletrec, $\amarkeri{1},\ldots,\amarkeri{n}$}   
    \BinaryInfC{$ \flabs{\vec{\avar}}{\,\letrec{\arecvari{1} = \allteri{1} \ldots \arecvari{n} = \allteri{n}}{\allteri{n+1}}} $}
    \DisplayProof
        }  
  \\[2.5ex]
  \parbox{330pt}
    {\small\centering
      where
      $\aconstnamei{\arecvari{1}}, \ldots,\aconstnamei{\arecvari{n}}$ are distinct constants fresh for $\allteri{1},\ldots,\allteri{n+1}$, 
      and 
      substitutions $\subst{\allteri{l}}{\vec{\arecvar}}{\vec{\aconstname}_{\vec{\arecvar}}}$
      stands short for 
      $\allteri{l} [\arecvari{1} \defdby \aconstnamei{\arecvari{1}}, \ldots, \arecvari{n} \defdby \aconstnamei{\arecvari{n}}] $.
      \\[0.75ex]
      \emph{side-conditions}:
        $\length{\vec{y}} \ge \length{\vec{x}}$ holds for the prefix length of
        every $\flabs{\vec{y}}{\bter}$ on a thread\\ in $\Derivi{j}$ for 
                                                                          $1\le j\le n+1$
        from an open assumptions $( \flabs{\vec{\avar}}{\aconstnamei{\arecvari{i}}} )^{\amarkeri{i}}$ downwards;\\
        for bottommost instances: 
        the arising derivation is guarded on access path cycles 
      } 
  \\[0.75ex]    
  \mbox{}
\end{center}
\end{minipage}
            }
\end{center}
  \caption{\label{fig:Regletrec:stRegletrec}
           The rule \FIXletrec\ for the natural-deduction style proof systems \protect\Regletrec\ and \protect\stRegletrec\ 
           on $\protect\lambdaletrec$\protect\nb-terms.
           }
\end{figure}

For the purpose of the following definition and the respective proof system in
Figure~\ref{fig:Regletrec:stRegletrec} we extend the signature $\sigllcCRS$ of
$\Regletrec$ and $\stRegletrec$ by an infinite set of constants for which we use 
the symbol $\aconstname$ as syntactical variables which frequently carry index subscripts. 

\todo{Argue that this proof system is not able to define the eager \scopedelimiting\ strategy}

\begin{definition}[proof systems \Regletrec, \stRegletrec, and \ugRegletrec, \ugstRegletrec]%
   \label{def:Regletrec:stRegletrec}
   The proof systems \stRegletrec\ and \Regletrec\
   for $\protect\lambdaletrec$\nb-terms
   arise from the proof systems \stReg\ and \protect\Reg\ 
   (see Definition~\ref{def:Reg:stReg:stRegzero}, Figure~\ref{fig:stReg:stRegzero} and Figure~\ref{fig:Reg}), respectively,
   by replacing the terms in the axioms ($\bvarax$), and rules ($\labscomp$), ($\lappcomp$), ($\Vacstreg$) and ($\Vacreg$)
   through \lambdaletrecterm{s} with abstraction prefixes accordingly, and 
   by replacing the rule \FIX\ with the rule \FIXletrec\ in Figure~\ref{fig:Regletrec:stRegletrec}. 
   The side-condition concerning access path cycles on the derivation arising by an instance of \FIXletrec\ 
   pertains only to bottommost occurrences of this rule, and 
   is explained below.  
   By \FIXletrecmin\ we mean the variant of the rule \FIXletrec\ in which the side-condition 
   concerning guardedness of the arising derivation on access path cycles has been dropped. 
   By \ugRegletrec/\ugstRegletrec\ we denote the variants of \Regletrec/\stRegletrec, respectively,
   in which the rule \FIXletrec\ is replaced by the rule \FIXletrecmin.
   
   Let $\Deriv$ be a derivation in one of these proof systems.
   By an \emph{access path} of $\Deriv$ we mean a (possibly cyclic) path $\apath$ in $\Deriv$ such that:
   \begin{enumerate}[(a)]
     \item{}\label{access:path:item:i}
       $\apath$ starts at the conclusion and can proceed in upwards direction;
     \item{}\label{access:path:item:ii}
       at instances of ($\lappcomp$), $\apath$ can step from the conclusion to one of the premises;  
     \item{}\label{access:path:item:iii}
       at instances of \FIXletrec, $\apath$ can step from the conclusion to the rightmost premise
       (which corresponds to the $\textsf{in}$\nb-part of the $\sletrec$\nb-term that is parsed by this instance);
     \item{}\label{access:path:item:iv}   
       when arriving at a marked assumption $(\flabs{\vec{\avar}}{\aconstnamei{\arecvari{i}}})^{\amarkeri{i}}$
       that is discharged at an application of \FIXletrec\ of the form as displayed in Figure~\ref{fig:Regletrec:stRegletrec},
       $\apath$ can step over to the conclusion
       $\flabs{\vec{\avar}}{\subst{\allteri{i}}{\vec{\arecvar}}{\vec{\aconstname}_{\vec{\arecvar}}}}$
       of the subderivation $\Derivi{i}$ of that application of \FIXletrec, and proceed from there, again in upwards direction.
   \end{enumerate}
   For every formula occurrence $o$ in $\Deriv$, by a \emph{relative access path} from $o$ we mean a path with
   the properties (\ref{access:path:item:ii})--(\ref{access:path:item:iv}) that starts at $o$ and proceeds in upwards direction. 
   An access path (or relative access path) in $\Deriv$ is \emph{cyclic} if there is a formula occurrence in $\Deriv$ that is visited more than once.   
   
   Now we say that $\Deriv$ is \emph{guarded on access path cycles} 
   if every cyclic access path contains, on each of its cycles, at least one \emph{guard}, that is, an instance of a rule ($\labscomp$) or ($\lappcomp$). 
   We say that $\Deriv$ is \emph{guarded} if every relative access path contains a guard on each of its cycles.
\end{definition}

\begin{example}
  The \lambdaletrecterm~$\allteri{1} = \letrec{\arecvar = \arecvar,\, \brecvar = \labs{\avar}{\avar}}{\labs{\bvar}{\brecvar}}$
  admits the following closed derivation $\Derivi{1}$ in $\Regletrec$/$\stRegletrec$:
  \begin{center}
    \mbox{
      \AxiomC{$ (\femptylabs{\aconstnamei{\arecvar}})^{\amarkeri{1}} $}
      \AxiomC{\mbox{}}
      \RightLabel{\bvarax}
      \UnaryInfC{$ \flabs{\avar}{\avar} $}
      \RightLabel{$\labscomp$}
      \UnaryInfC{$ \femptylabs{\labs{\avar}{\avar}}$ }
      \AxiomC{$ (\femptylabs{\aconstnamei{\brecvar}})^{\amarkeri{2}} $}
      \RightLabel{\Vacreg/\Vacstreg}
      \UnaryInfC{$ \flabs{\bvar}{\aconstnamei{\brecvar}} $}
      \RightLabel{$\labscomp$}
      \UnaryInfC{$ \femptylabs{\labs{\bvar}{\aconstnamei{\brecvar}}} $} 
      \RightLabel{$\sFIXletrec, {\amarkeri{1}, \amarkeri{2}} $}
      \insertBetweenHyps{\hspace*{6ex}}
      \TrinaryInfC{$ \femptylabs{\letrec{\arecvar = \arecvar,\, \brecvar = \labs{\avar}{\avar}}{\labs{\bvar}{\brecvar}}} $}
      \DisplayProof 
          }
  \end{center}  
  This derivation can be built in a straightforward way, from the bottom upwards. 
  Note that $\Deriv$ is guarded on access path cycles, and hence that the instance of \FIXletrec\ at the bottom is a valid one,
  because: $\Deriv$ does not possess any cyclic access paths. In particular, there is no access path in $\Derivi{1}$
  that reaches the first premise of the instance of \FIXletrec: this premise is the starting point of an 
  unguarded relative access path, which entails that $\Derivi{1}$ itself is not guarded. 
  
  Now consider the 
  \lambdaletrecterm~$\allteri{2} = \letrec{\arecvar = \arecvar,\, \brecvar = \labs{\avar}{\avar}}{\labs{\bvar}{\lapp{\arecvar}{\brecvar}}}$.
  When trying to construct a derivation in $\Regletrec$/$\stRegletrec$ for this term in a bottom-up manner,
  one arrives~at the closed derivation $\Derivi{2}$ in \ugRegletrec/\ugstRegletrec: 
  \begin{center}
    \mbox{
      \AxiomC{$ (\femptylabs{\aconstnamei{\arecvar}})^{\amarkeri{1}} $}
      \AxiomC{\mbox{}}
      \RightLabel{\bvarax}
      \UnaryInfC{$ \flabs{\avar}{\avar} $}
      \RightLabel{$\labscomp$}
      \UnaryInfC{$ \femptylabs{\labs{\avar}{\avar}}$ }
      \AxiomC{$ (\femptylabs{\aconstnamei{\arecvar}})^{\amarkeri{1}} $}
      \RightLabel{\Vacreg/\Vacstreg}
      \UnaryInfC{$ \flabs{\bvar}{\aconstnamei{\arecvar}} $}
      \AxiomC{$ (\femptylabs{\aconstnamei{\brecvar}})^{\amarkeri{2}} $}
      \RightLabel{\Vacreg/\Vacstreg}
      \UnaryInfC{$ \flabs{\bvar}{\aconstnamei{\brecvar}} $}
      \RightLabel{$\lappcomp$}
      \BinaryInfC{$ \flabs{\bvar}{\lapp{\aconstnamei{\arecvar}}{\aconstnamei{\brecvar}}} $} 
      \RightLabel{$\labscomp$}
      \UnaryInfC{$ \femptylabs{\labs{\bvar}{\lapp{\aconstnamei{\arecvar}}{\aconstnamei{\brecvar}}}} $} 
      \RightLabel{{$\sFIXletrecmin, \amarkeri{1}, \amarkeri{2}$} (but not {\FIXletrec}!)}
      \TrinaryInfC{$ \femptylabs{\letrec{\arecvar = \arecvar,\, \brecvar = \labs{\avar}{\avar}}{\labs{\bvar}{\lapp{\arecvar}{\brecvar}}}}  $}
      \DisplayProof 
          }
  \end{center}  
  However, $\Derivi{2}$ is not a valid derivation in $\Regletrec$/$\stRegletrec$, as the inference step at the bottom is
  an instance of \FIXletrecmin, but not of \FIXletrec, 
  because the side-condition on the arising derivation to be guarded on access path cycles is not satisfied:
  now there is an access path that reaches the first premise of the derivation and that continues looping on 
  this an unguarded cycle. 
  Since the bottom-up search procedure for derivations is deterministic in this case, 
  it follows that $\femptylabs{\allteri{2}}$ is not derivable in $\Reg$ nor in  $\stReg$. 
\end{example}

Similar as the correspondence, stated by Proposition~\ref{prop:derivationpaths:2:rewritesequences:Reg:stReg},
between (possibly cyclic) paths in a derivation in $\Reg$ and $\stReg$ starting at the conclusion
and rewrite sequences with respect to $\regred$ and $\stregred$ on the infinite term in the conclusion,
there is also the following correspondence between access paths in a derivation in $\Regletrec$ and $\stRegletrec$,
and rewrite sequences with respect to $\regred$ and $\stregred$ on the \lambdaletrecterm\ in the conclusion.

\begin{proposition}\label{prop:Regletrec:stRegletrec}
  Let $\Deriv$ be a derivation in $\Regletrec$ or $\ugRegletrec$ (in $\stRegletrec$ or in $\ugstRegletrec$) with conclusion $\flabs{\vec{\avar}}{\allter}$.
  \begin{enumerate}[(i)]
    \item\label{prop:Regletrec:stRegletrec:item:i} 
      Then every access path in $\Deriv$ to an occurrence $\aocc$ of a term $\flabs{\vec{\bvar}}{\bllter}$
      corresponds to a $\sregred$\nb-rewrite sequence 
      $\flabs{\vec{\avar}}{\allter} \regmred \flabs{\vec{\bvar}}{\letrec{\abindgroup}{\tilde{\bllter}}}$
      (to a $\sstregred$\nb-rewrite sequence 
       $\flabs{\vec{\avar}}{\allter} \stregmred \flabs{\vec{\bvar}}{\letrec{\abindgroup}{\tilde{\bllter}}}$),
      where $\abindgroup$ 
      arises as the union of all outermost binding groups in conclusions of instances of \FIXletrec\ below $\aocc$,
      and $\flabs{\vec{\bvar}}{\bllter} = \flabs{\vec{\bvar}}{\subst{\tilde{\bllter}}{\vec{\arecvar}}{\vecsub{\aconstname}{\vec{\arecvar}}}}$ 
      where $\vec{\arecvar}$ is comprised of the recursion variables occurring in $\abindgroup$ and $\vecsub{\aconstname}{\vec{\arecvar}}$
      distinct constants for $\vec{\arecvar}$ as chosen by~$\Deriv$. 
      More precisely:
      \begin{enumerate}[(a)]
        \item a pass over an instance of \FIXletrec\ corresponds to an empty or $\sunfletrecred$\nb-step,
          dependent on whether the instance is the bottommost \FIXletrec\nb-instance or not;
        \item a pass over an instance of the rule ($\lappcomp$) to the left/to the right
          corresponds to a $\slappdecompired{0}$\nb-step/$\slappdecompired{1}$\nb-step,
          which, if the application is somewhere above an instance of \FIXletrec, has to be preceded 
          by a $\sunflappred$\nb-step;
        \item a pass over an instance of the rule ($\labscomp$) corresponds to a $\slabsdecompred$\nb-step
          which, if the application is above an instance of \FIXletrec, has to be preceded 
          by a $\sunflabsred$\nb-step; 
        \item a pass over an instance of the rule $\Vacreg$ ($\Vacstreg$) 
          corresponds to a $\scompressregred$\nb-step ($\scompressstregred$\nb-step),
          possibly preceded by an application of $\sunfreducered$.
        \item a step from a marked assumption 
                                              to a premise of a \FIXletrec\nb-instances,
          a step as described in item~(\ref{access:path:item:iii}) of the definition of access paths,
          corresponds to an $\sunfrecred$\nb-step followed by a $\sunfreducered$\nb-step. 
      \end{enumerate}
    \item\label{prop:Regletrec:stRegletrec:item:ii}  
      $\allter$ is reduced if every formula occurrence in $\Deriv$ can be reached by an access path.
  \end{enumerate}
\end{proposition}

\begin{example}\label{ex:stRegletrec}
  The rewrite cycle in Example~\ref{ex:non-unfoldable} that witnesses that 
  the \lambdaletrecterm\ considered there, $\letrec{\arecvar = \letrec{\brecvar = \arecvar}{\brecvar}}{\arecvar}$,
  is not unfoldable can also be recognised, using the statement of Proposition~\ref{prop:Regletrec:stRegletrec},
  from the following derivation in $\ugstRegletrec$:
  \begin{equation*}
    \mbox{
      \AxiomC{$ (\aconstnamei{\arecvar})^{\amarker}$}
      \AxiomC{$ (\aconstnamei{\brecvar})^{\bmarker} $}
      \insertBetweenHyps{\hspace*{5ex}}
      \RightLabel{\sFIXletrec, $\bmarker$}
      \BinaryInfC{$ \femptylabs{\letrec{\brecvar = \arecvar}{\brecvar}} $}
      \AxiomC{$ (\aconstnamei{\arecvar})^{\amarker} $}
      \insertBetweenHyps{\hspace*{5ex}}
      \RightLabel{\sFIXletrecmin, $\amarker$}
      \BinaryInfC{$ \femptylabs{\letrec{\arecvar = \letrec{\brecvar = \arecvar}{\brecvar}}{\arecvar}}  $}
      \DisplayProof
          }
  \end{equation*}
  Note that the instance of \FIXletrecmin\ at the bottom is not an instance of \FIXletrec,
  since it is not guarded (has an unguarded cyclic access path that reaches and cycles on the left premise
  of the instance of $\FIXletrec$.
\end{example}

\begin{lemma}\label{lem:Regletrec:stRegletrec}
  Let $\Deriv$ be a closed derivation in $\ugRegletrec$ (in $\ugstRegletrec$) with conclusion $\femptylabs{\allter}$.
  Then there exists a \scopedelimiting\ (\extscopedelimiting) strategy~$\astrati{\Deriv}$ for $\RegletrecARS$ (for $\stRegletrecARS$) 
  with the following properties:
  \begin{enumerate}[(i)]
    \item{}\label{lem:Regletrec:stRegletrec:item:i}
      Every access path in $\Deriv$ corresponds
      to a rewrite sequence with respect to $\astrati{\Deriv}$ starting on $\femptylabs{\allter}$
      in the sense of Proposition~\ref{prop:Regletrec:stRegletrec}. 
    \item{}\label{lem:Regletrec:stRegletrec:item:ii}
      Every 
            rewrite sequence that starts on $\femptylabs{\allter}$ and proceeds according to $\astrati{\Deriv}$  
      corresponds to an access path in $\Deriv$ with correspondences as described 
      in Proposition~\ref{prop:Regletrec:stRegletrec}, (\ref{prop:Regletrec:stRegletrec:item:i}).
    \item{}\label{lem:Regletrec:stRegletrec:item:iii}
      $\gSTstrat{\astrati{\Deriv}}{\allter} 
         =
       \descsetexpBig{ \flabs{\vec{\bvar}}{\letrec{\abindgroup}{\tilde{\bllter}}} }
                  { \parbox{220pt}{
                      $\flabs{\vec{\bvar}}{\letrec{\abindgroup}{\tilde{\bllter}}}$
                      arises from an occ.\ of $\flabs{\vec{\bvar}}{\bllter}$ on\\[-0.5ex] an access path of $\Deriv$
                      as described in Prop.~\ref{prop:Regletrec:stRegletrec}, (\ref{prop:Regletrec:stRegletrec:item:i})
                                   } }
      $.              
      As a consequence of that $\Deriv$ is finite, $\allter$ is $\astrati{\Deriv}$\nb-regular.
    \item{}\label{lem:Regletrec:stRegletrec:item:iv}
      $\text{$\allter$ is $\astrati{\Deriv}$\nb-productive}
         \;\Leftrightarrow\;
       \text{$\Deriv$ is guarded (i.e.\ $\Deriv$ derivation in $\Regletrec$ ($\stRegletrec$))}        
      $.
  \end{enumerate}  
\end{lemma}

\begin{proof}
  Given a closed derivation $\Deriv$ with conclusion $\femptylabs{\allter}$ (for example) in $\ugstRegletrec$, 
  a \extscopedelimiting~strategy $\astrati{\Deriv}$ for $\stRegletrecARS$
  such that (\ref{lem:Regletrec:stRegletrec:item:i})--(\ref{lem:Regletrec:stRegletrec:item:iv}) hold 
  can be extracted from $\Deriv$ 
  similar as in the proof of 
  a \extscopedelimiting~strategy $\astrati{\Deriv}$ in $\stRegARS$
  was extracted from a closed derivation in $\stReg$.
  That the extracted strategy $\astrati{\Deriv}$ is productive/not productive for $\allter$
  if $\Deriv$ is guarded/not guarded can be seen by the fact that
  $\astrati{\Deriv}$\nb-rewrite sequences correspond to access paths of $\Deriv$
  in the sense as stated by Proposition~\ref{prop:Regletrec:stRegletrec}.
\end{proof}

Now we will prove that 
derivability in $\ugRegletrec$/$\ugstRegletrec$ is guaranteed for all \lambdaletrec\nb-terms,
and that derivability in $\Regletrec$/$\stRegletrec$ is a property of a \lambdaletrec\nb-term
that is decidable by an easy parsing process. 

\begin{proposition}\label{prop:derivability:Regletrec:stRegletrec}
  The following statements hold:
  \begin{enumerate}[(i)]
    \item{}\label{prop:derivability:Regletrec:stRegletrec:item:i}
      For every  \lambdaletrecterm~$\allter$, $\femptylabs{\allter}$ is derivable both in $\ugRegletrec$ and in $\ugstRegletrec$.
    \item{}\label{prop:derivability:Regletrec:stRegletrec:item:ii} 
      For every \lambdaletrecterm~$\allter$, derivability of $\femptylabs{\allter}$ in $\stRegletrec$ 
  is decidable in at most quadratic time in the size of $\allter$.   
  \end{enumerate}
\end{proposition}

\begin{proof}
  For (\ref{prop:derivability:Regletrec:stRegletrec:item:i}) note that 
  for every \lambdaletrecterm~$\allter$, a closed derivation $\Derivi{\allter}$ 
  with conclusion $\femptylabs{\allter}$ in $\ugstRegletrec$ 
  can be produced by a bottom-up construction following the term structure of~$\allter$.
  Hereby use of the rules ($\Vacstreg$) can be restricted to instances immediately below marked assumptions 
  such that, viewed from a (non-cyclic) path $\apath$ from the conclusion upwards to a marked assumption,
  these ($\Vacstreg$)\nb-instances are only introduced to shorten the frozen abstraction prefixes by all $\lambda$\nb-abstractions
  that have become frozen on $\apath$ (in order to conform to the side-condition on \FIXletrecmin\nb-instances 
  to have the same frozen abstraction prefix lengths in the discharged marked assumptions as in the conclusion and in the premises). 
  
  Now for (\ref{prop:derivability:Regletrec:stRegletrec:item:ii}) in order to decide derivability of $\femptylabs{\allter}$
  in $\stRegletrec$, it suffices to decide whether
  the derivation $\Derivi{\allter}$ in $\ugstRegletrec$ obtained as described above, 
  or its bottommost instance of \FIXletrecmin\ if there is any, is guarded on all of its access path cycles.
  (Note that in the construction of $\Derivi{\allter}$ only the freedom in placing instances of ($\Vacstreg$)
   has been used in a certain, namely lazy, way. The specific placement of instances of these rules does 
   not interfere with the existence or non-existence of guards, that is instances of $\labscomp$ or $\lappcomp$
   on cycles of access paths.)
  For this it remains to check whether every cycle on an access path in $\Derivi{\allter}$ has a guard. 
  This can be done by exploring the prooftree of $\Derivi{\allter}$ according to all possible access paths 
  (until for the first time a cycle is concluded) and checking for the existence of guards on cycles.
  
\end{proof}

We now can prove soundness and completeness of the proof system $\stRegletrec$ 
for the property of \lambdaletrecterms\ to unfold to infinite \lambdaterms.

\begin{theorem}\label{thm:Regletrec:stRegletrec}
  $\stRegletrec$ is sound and complete for the property of \lambdaletrec\nb-terms to unfold to an infinite \lambdaterm.
  That is,
  for every term $\allter\in\Ter{\lambdaletreccal}$ 
                                                           the following statements are equivalent:
  \begin{enumerate}[(i)]
    \item\label{thm:Regletrec:stRegletrec:item:i}
      $\allter$ expresses an infinite \lambdaterm.
    \item\label{thm:Regletrec:stRegletrec:item:iii}    
      $\derivablein{\stRegletrec}{\femptylabs{\allter}}$.
  \end{enumerate}
\end{theorem}

\begin{proof}
  For the proof of both directions of the equivalence, let $\allter\in\Ter{\lambdaletreccal}$. 
  
  For showing the implication 
  $\text{(\ref{thm:Regletrec:stRegletrec:item:i})}
     \Rightarrow
   \text{(\ref{thm:Regletrec:stRegletrec:item:iii})}$,
  we argue indirectly,
  and therefore assume that $\femptylabs{\allter}$ is not derivable in $\stRegletrec$.
  Then, while $\femptylabs{\allter}$ is not derivable in $\stRegletrec$,
  there is, by Proposition~\ref{prop:derivability:Regletrec:stRegletrec}, (\ref{prop:derivability:Regletrec:stRegletrec:item:i}),
  a derivation $\Deriv$ in $\ugstRegletrec$ that is not guarded.
  It follows by Lemma~\ref{lem:Regletrec:stRegletrec}, and in particular due to its item (\ref{lem:Regletrec:stRegletrec:item:iv}),
  that there is a \extscopedelimiting\ strategy $\astrati{\Deriv}$ for $\stRegARS$ 
  such that $\allter$ is not $\astrati{\Deriv}$\nb-productive. 
  Then it follows by Lemma~\ref{lem:unfolding:versus:scdelstrats}, using
  $(\text{\ref{lem:unfolding:versus:scdelstrats:item:i}}) \Rightarrow (\text{\ref{lem:unfolding:versus:scdelstrats:item:iv}})$ there,
  that $\allter$ does not unfold to an infinite \lambdaterm. 
  
  For showing the implication
      $\text{(\ref{thm:Regletrec:stRegletrec:item:iii})}
         \Rightarrow
       \text{(\ref{thm:Regletrec:stRegletrec:item:i})}$,
  let $\Deriv$ be a closed derivation in $\stRegletrec$ with conclusion $\femptylabs{\allter}$.        
  It follows by Lemma~\ref{lem:Regletrec:stRegletrec}
  that there is a \extscopedelimiting~strategy $\astrat$ for $\RegARS$ such that
  $\allter$ is $\astrat$\nb-productive.  
  Then Lemma~\ref{lem:unfolding:versus:scdelstrats} implies that $\allter$ unfolds to an infinite \lambdaterm.
\end{proof}

\begin{remark}[soundness and completeness for $\Regletrec$]
  Also the proof system $\Regletrec$ can be shown to be sound and complete for the property
  of \lambdaletrecterms\ to unfold to infinite \lambdaterms.
  To establish this in analogy with the route of proof we pursued here,
  a \CRS~$\ParseCRS$ similar to $\stParseCRS$ (see Definition~\ref{def:stParseCRS})
  could be defined by replacing the rule $(\srulebp{\sparse}{\scompressstreg})$
  by a rule $(\srulebp{\sparse}{\scompressreg})$ that can compress more abstraction prefixes,
  similar as the rule $(\srulebp{\sreg}{\scompressreg})$ of $\RegCRS$ can compress
  more abstraction prefixes than the rule $(\srulebp{\sstreg}{\scompressstreg})$ of $\stRegCRS$.
  Then furthermore also a lemma analogous to Lemma~\ref{lem:unfolding:versus:scdelstrats} 
  can be formulated, proved, and used in a similar way. 
\end{remark}

We now arrive at a theorem that states one direction of our main characterisation result
(Theorem~\ref{thm:ll-expressible:streg} in Section~\ref{sec:express}) 
that will link $\lambdaletrec$\nb-expressibility to strong regularity of infinite \lambdaterms.

\begin{theorem}\label{thm:ll-expressible:2:streg}
  Every \lambdaletrec\nb-expressible, infinite \lambdaterm\ is strongly regular.  
\end{theorem}

\begin{proof}
  Let $\aiter$ be an infinite \lambdaterm\ that is expressible by a \lambdaletrec\nb-term $\allter$,
  that is, $\allter \unfoldomegared \aiter$ holds.
  Then by Theorem~\ref{thm:Regletrec:stRegletrec} there exists a closed derivation $\Deriv$ in $\stRegletrec$
  with conclusion $\femptylabs{\allter}$. 
  Now Lemma~\ref{lem:Regletrec:stRegletrec} guarantees a \extscopedelimiting~strategy $\astrati{\Deriv}$ for $\stRegletrecARS$
  such that $\allter$ is $\astrati{\Deriv}$\nb-regular. 
  Then Lemma~\ref{lem:proj:scdelstrat:letrec:lambda} gives an \extscopedelimiting~strategy $\Check{\astrat}_{\Deriv}$ for $\stRegletrecARS$
  such that $\aiter = \Unf{\allter}$ is $\Check{\astrat}_{\Deriv}$\nb-regular.
  It follows that $\aiter$ is strongly regular. 
\end{proof}

\section{Binding--Capturing Chains}
  \label{sec:chains}
In this section we develop a characterisation for strong regularity of
an infinite \lambdaterm\ by means of a property of the `\bindcaptchains'
occurring in the term.
This concept is related to the notions of scope and
\extscope\ as explained informally at the start of Section~\ref{sec:regular}. 
Binding--capturing chains occur whenever scopes overlap, and they are contained within \extscopes. 
First we give definitions for the concepts involved: binding,
capturing, and \bindcaptchains. Then we show that strong regularity of regular infinite \lambdaterms\
is equivalent to the absence of infinite \bindcaptchains.

We will define binding and capturing as relations on the positions of a \lambdaterm.
Binding relates an abstraction with the occurrences of the variable it binds.
If $\apos$ is the position of an abstraction ($\labs{x}{\dots}$) that
abstracts over $x$ and $\bpos$ is the position of an occurrence of $x$ that is bound by the abstraction,
then we will write $\apos \binds \bpos$ and say that $\apos$ `binds' $\bpos$.
Capturing relates an abstraction with variable occurrences that are free in it.
If $\apos$ is the position of an abstraction, and $\bpos > \apos$ is the position of a variable
that is free in the entire subterm at position $\apos$, 
then we will write $\apos \captures \bpos$ and say that $\apos$ `captures' $\bpos$.  
See Figure~\ref{fig:entangled_bc} for an illustration of these concepts.

\begin{figure}
\fig{entangled-scopes}
\hspace{2cm}
\fig{entangled-extscopes}
\hspace{2cm}
\fig{entangled-chains}
\caption{
\label{fig:entangled_bc}
  The term graph from Example~\ref{ex:entangled}
  with its overlapping \scope{s} (left), its nested \extscope{s} (middle), and with indicated binding~$\sbinds$ and capturing~$\siscapturedby$ links (right).}
\end{figure}

When we speak of positions in \lambda-terms (and thus iCRS-terms) we act on the
assumption that positions on iCRS-terms are an established concept as for
example in \cite{kete:simo:2011}. Note, however, that we deviate slightly from
the scheme there in addressing the arguments of an $\slappCRS$ by $0$ and $1$ instead of $1$ and $2$.

Binding--capturing chains have been used in \cite{endr:grab:klop:oost:2011}
to study $\alpha$\nb-avoiding rewrite sequences in a rewrite calculus for $\mu$\nb-unfolding.
They originate from the notion of `gripping' due to \cite{mell:96}, and
from techniques developed in \cite{oost:97} concerning the notion of `holding' of redexes
(which is shown there as being `parting' for \CRS{s}, that is, never relating two residuals of the same redex).

We now define `binding' and `capturing' formally as binary relations on the set of
positions of infinite \lambdaterms.

\begin{definition}[binding, capturing]\label{def:bind:iscapturedby}
  For every $\aiter\in\Ter{\inflambdacal}$ we define the binary relations $\sbinds$ and $\siscapturedby$ 
  on the set $\Positions{\aiter}$ of positions of $\aiter$:
  \begin{enumerate}[(i)]
    \item\label{def:bind:iscapturedby:item:bind}
      We say that a binder at position $\apos$ \emph{binds} a variable occurrence at position $\bpos$,
      symbolically $\apos \binds \bpos$,
      if $\apos$ is a binder position, and $\bpos$ a variable position in $\aiter$,
      and the binder at position $\apos$ binds the variable occurrence at position $\bpos$.
    \item\label{def:bind:iscapturedby:item:iscapturedby}
      We say that a variable occurrence at position $\bpos$ \emph{is captured by} a binder at position $\apos$, symbolically $\bpos \iscapturedby \apos$
      (and that a binder at position $\apos$ \emph{captures} a variable occurrence at position $\bpos$, symbolically $\apos \captures \bpos$),
      if $\bpos$ is a variable position and $\apos < \bpos$ a binder position in $\aiter$,
      and there is no binder position $\bposi{0}$ in $\aiter$ with $\apos \le \bposi{0}$ and $\bposi{0} \binds \bpos$.
  \end{enumerate}
\end{definition}

\begin{definition}[\bindcaptchain]\label{def:bind:capt:chain}
  Let $\aiter$ be an infinite \lambdaterm. 
  A finite or infinite sequence $\enumsequence{\aposi{0},\aposi{1},\aposi{2},\ldots}$
  in $\Positions{\aiter}$ is called a \emph{\bindcaptchain\ in $\aiter$} 
  if $\aposi{0}$ is the position of an abstraction in $\aiter$,
  and the positions in the sequence are alternatingly linked via binding and capturing, starting with a binding:
  $\aposi{0} \binds \aposi{1} \iscapturedby \aposi{2} \binds \ldots$.
\end{definition}

Binding--capturing chains are closely related to the notion of \scope\ and \extscope.
In order to establish this, we first give precise definitions of the notions of \scope\ and \extscope\
in terms of an `in-scope' rewrite relation on the positions of a \lambdaterm:
While the \scope\ of a binder position $\apos$ is the set of positions between $\apos$ and variable positions bound at $\apos$
(the positions directly reachable by a single `in-scope' step),
the \extscope\ of $\apos$ is the set of positions reachable by a finite number of successive `in-scope' steps. 

\begin{definition}[\scope\ and \extscope]\label{def:scope:extscope}
  Let $\aiter$ be an infinite \lambdaterm.
  On the set $\Positions{\aiter}$ of $\aiter$, the \emph{in-scope} rewrite relation $\sscopered$ (for $\aiter$) is defined by:
  \begin{align*}
    \apos \scopered \bpos   
      \;\;& \Longleftrightarrow\;\;
      \left\{\,
    \begin{aligned}
      & 
      \text{$\apos$ a binder position}
      \\[-0.5ex]
      & \logand
      (\exists \apos'\in\Positions{\aiter}) \;\: \apos \bindseq \apos'  \:\logand\: \apos \le q \le p' 
    \end{aligned}  
    \,\right\}
    & & 
    \text{(for all $\apos,\bpos\in\Positions{\aiter}$)}
  \end{align*}
  where $\sbindseq$ denotes the reflexive closure of the binding relation $\sbinds$ (for $\aiter$).
  For every position $\apos\in\Positions{\aiter}$, the \emph{\scope\ of $\apos$ in $\aiter$} and the \emph{\extscope of $\apos$ in $\aiter$}
  are defined as the following sets of positions in $\aiter$:
  \begin{align*}
    \scopeof{\aiter}{\apos} 
      & \defdby 
    \descsetexp{\bpos\in\Positions{\aiter}}{\apos \scopered \bpos} 
    &
    \extscopeof{\aiter}{\apos} 
      & \defdby 
    \descsetexp{\bpos\in\Positions{\aiter}}{\apos \scopemorestepred \bpos}
  \end{align*}
  (Note that the \scope{s} and \extscope{s} of non-binder positions are empty sets of positions.)
\end{definition}

Now the following proposition establishes that \bindcaptchains\ starting at a binder position $\apos$ 
span the space of positions of the \extscope\ of $\apos$.

\begin{proposition}
  Let $\aiter$ be an infinite \lambdaterm. Then for all positions $\apos, \bpos\in\Positions{\aiter}$ 
  the following statements hold:
  \begin{enumerate}[(i)]
    \item
      $ \apos \scopered \bpos \,\logand\, \text{$\bpos$ is a binder position}
          \;\Longleftrightarrow\;
        (\text{$\apos = \bpos$ binder position}) 
          \logor
            \apos \binrelcomp{\sbinds}{\siscapturedby} \bpos
        $.  
    \item  
      $ \apos \scopemorestepred \bpos 
          \;\Longleftrightarrow\;
        (\text{$\apos$ binder pos.})
          \logand
          (\exists \apos'\in\Positions{\aiter}) \,
            \apos \binrelcomp{(\binrelcomp{\sbinds}{\siscapturedby})^*}{\sbindseq} \apos'
              \logand
            \apos \le \bpos \le \apos'$.
    \item 
      $
      \scopeof{\aiter}{\apos} 
        \defdby 
      \descsetexp{\bpos\in\Positions{\aiter}}{(\text{$\apos$ binder pos.})
                                              \logand
                                              (\exists \apos'\in\Positions{\aiter}) \: \apos \bindseq \apos'  \:\logand\: \apos \le q \le p'} 
      $.
    \item
      $   
      \extscopeof{\aiter}{\apos} 
        \defdby 
      \descsetexpBig{\bpos\in\Positions{\aiter}}%
               {\parbox{200pt}{$
                 (\text{$\apos$ binder position}) \logand
                 \\[-0.25ex]
                 (\exists \apos'\in\Positions{\aiter}) \: 
                   \apos \binrelcomp{(\sbinrelcomp{\sbinds}{\siscapturedby)^*}}{\bindseq} \apos'  \:\logand\: \apos \le q \le p'
                               $}} 
      $.
  \end{enumerate}
\end{proposition}

Conversely, positions between a binder position $\aposi{0}$ and a position $\aposi{n}$ on a \bindcaptchain\ starting at $\aposi{n}$
are in the \extscope\ of $\aposi{0}$.

\begin{proposition}\label{prop:bindcaptchains}
  Let $\enumsequence{\aposi{0},\aposi{1},\aposi{2},\ldots}$ be
  a binding--capt.\ chain in an infinite \lambdaterm~$\aiter$.
  Then it holds that $\aposi{0} < \aposi{2} < \ldots$, and $\aposi{0} < \aposi{1}$, $\aposi{2} < \aposi{3}$, \ldots.
  Furthermore, for all $\bpos$ such that $\aposi{0} \le \bpos \le \aposi{n}$ for some $n\in\nats$ with $\aposi{n}$ a position on the chain
  it holds that $\bpos\in\extscopeof{\aiter}{\aposi{0}}$.
\end{proposition}

In order to study the relationship between rewrite sequences in $\stRegARS$ and
\bindcaptchains{} we first introduce a position-annotated variant of
$\stRegCRS$.

Here the idea is that, when a prefixed term $\flabs{\bvari{1}\ldots\bvari{n}}{\biter}$ is obtained 
as a generated subterm of an infinite \lambdaterm\ $\aiter$ by a $\sregred$ or $\sstregred$ rewrite sequence $\arewseq$ on $\femptylabs{\aiter}$,
then in the position-annotated rewrite system a prefixed term 
$\flabspos{\bvari{1},\ldots,\bvari{n}}{\aposi{1}\ldots\aposi{n}}{\bpos}{\biter}$
is obtained by an annotated version $\arewseq^{\spos}$ of the rewrite sequence $\arewseq$ 
such that: the positions $\aposi{1},\ldots,\aposi{n}$ are the positions in (the original \lambdaterm) $\aiter$
from which the bindings $\slabs\bvari{1}$, \ldots, $\slabs\bvari{n}$ in the abstraction prefix descend,
and $\bpos$ is the position in $\aiter$ of the body $\biter$ of the subterm generated by $\arewseq$. 

On $\Ter{\inflambdaprefixcal}$ we consider the following rewrite rules in informal notation:
  \begin{align*}
    (\rulepos{\slappdecompi{i}}): 
      & &
    \flabspos{\avari{1}\ldots\avari{n}}{\aposi{1},\ldots,\aposi{n}}{\bpos}{\lapp{\aiteri{0}}{\aiteri{1}}}  
      & {} \red 
    \flabspos{\avari{1}\ldots\avari{n}}{\aposi{1},\ldots,\aposi{n}}{\bpos i}{\aiteri{i}}
      & & \hspace*{-10ex} (i\in\{0,1\})
    \\
    (\rulepos{\slabsdecomp}):
      & &
    \flabspos{\avari{1}\ldots\avari{n}}{\aposi{1},\ldots,\aposi{n}}{\bpos}{\labs{\bvar}{\aiteri{0}}}  
      & {} \red 
    \flabspos{\avari{1}\ldots\avari{n}\bvar}{\aposi{1},\ldots,\aposi{n},\bpos}{\bpos 00}{\aiteri{0}}  
    \displaybreak[0]\\
    (\rulepos{\snlvarsucc}):
      & & 
    \flabspos{\avari{1}\ldots\avari{n+1}}{\aposi{1},\ldots,\aposi{n+1}}{\bpos}{\aiteri{0}}   
      & {} \red 
    \flabspos{\avari{1}\ldots\avari{n}}{\aposi{1},\ldots,\aposi{n}}{\bpos}{\aiteri{0}} 
    \\  
      & & & & &
      \hspace*{-35ex} 
                           (\text{if the binding $\slabs\avari{n+1}$ is vacuous})
    \displaybreak[0]\\  
    (\rulepos{\scompress}):
      & &
    \flabspos{\avari{1}\ldots\avari{n+1}}{\aposi{1},\ldots,\aposi{n+1}}{\bpos}{\aiteri{0}}   
      & {} \red 
    \flabspos{\avari{1}\ldots\avari{i-1}\avari{i+1}\ldots\avari{n+1}}{\aposi{1},\ldots\aposi{i-1},\aposi{i+1}\ldots\aposi{n+1}}{\bpos}{\aiteri{0}}   
    \\   
      & & & & &
      \hspace*{-35ex} 
                           (\text{if the binding $\slabs\avari{i}$ is vacuous})
  \end{align*}
  Note that the change of the term-body position in a $\lambda$\nb-decomposition step is
  motivated by the underlying \CRS\nb-notation for terms in $\inflambdaprefixcal$: when a 
  term $\labsCRS{\bvar}{\aiteri{0}}$ representing a $\lambda$\nb-abstraction starts at position $\bpos$,
  then its binding is declared at position $\bpos 0$, and its  body $\aiteri{0}$ starts at position $\bpos 00$. 
   
\begin{definition}[position-annotated variants $\RegposCRS$ and $\stRegposCRS$]\label{def:RegposCRS:stRegposCRS}
  The \CRS-sig\-na\-ture for $\lambdaprefixposcal$, the \lambdacalculus\ with position-annotated abstraction prefixes is given by
  $\siglpposcCRS = \siglcCRS 
                  \cup 
                \descsetexpnormalsize{ \sflabsposCRS{\aposi{1},\ldots,\aposi{n}}{\bpos} }{ \aposi{1},\ldots,\aposi{n},\bpos\in\setexp{0,1}^* }$         
  where all of the function symbols $\sflabsposCRS{\aposi{1},\ldots,\aposi{n}}{\bpos}$ are unary.  
  We consider the following \CRS-rules over $\siglpposcCRS\,$:
  \begin{align*}
    (\rulepos{\slappdecompi{i}}): 
      & &
    \flabsposCRS{\avari{1}\ldots\avari{n}}{\aposi{1},\ldots,\aposi{n}}{\bpos}{\lappCRS{\aiteri{0}}{\aiteri{1}}} 
      & {} \;\red\;   
    \flabsposCRS{\avari{1}\ldots\avari{n}}{\aposi{1},\ldots,\aposi{n}}{\bpos i}{\aiteri{i}}
    \\
    (\rulepos{\slabsdecomp}):
      & &
    \flabsposCRS{\avari{1}\ldots\avari{n}}{\aposi{1},\ldots,\aposi{n}}{\bpos}{\labsCRS{\bvar}{\aiteri{0}}} 
      & {} \;\red\; 
    \flabsposCRS{\avari{1}\ldots\avari{n}\bvar}{\aposi{1},\ldots,\aposi{n},\bpos}{\bpos 00}{\aiteri{0}}  
    \displaybreak[0]\\
    (\rulepos{\snlvarsucc}):
      & & 
    \flabsposCRS{\avari{1}\ldots\avari{n+1}}{\aposi{1},\ldots,\aposi{n+1}}{\bpos}{\aiteri{0}}  
      & {} \;\red\;   
    \flabsposCRS{\avari{1}\ldots\avari{n}}{\aposi{1},\ldots,\aposi{n}}{\bpos}{\aiteri{0}} 
    \displaybreak[0]\\  
    (\rulepos{\scompress}):
      & &
    \flabsposCRS{\avari{1}\ldots\avari{n+1}}{\aposi{1},\ldots,\aposi{n+1}}{\bpos}{\aiteri{0}}  
      & {} \;\red\; 
    \\[-0.5ex]
      & & & \hspace*{-16ex}  
    \flabsposCRS{\avari{1}\ldots\avari{i-1}\avari{i+1}\ldots\avari{n+1}}{\aposi{1},\ldots\aposi{i-1},\aposi{i+1}\ldots\aposi{n+1}}{\bpos}{\aiteri{0}} 
  \end{align*}
  By $\RegposzeroCRS$ we denote the \CRS{} with the rules
  $\rulepos{\slappdecompi{i}}$ and $\rulepos\slabsdecomp$.
  By $\RegposCRS$ ($\stRegposCRS$) we denote the CRS consisting of all the above rules except the rule
  $\rulepos\scompress$ ($\rulepos\snlvarsucc$).
  
  By $\RegposzeroARS$, $\RegposARS$ and $\stRegposARS$ we denote the infinite
  abstract rewriting systems (\ARS{s}) induced by the iCRSs derived from
  $\RegposzeroCRS$, $\RegposCRS$, $\stRegCRS$, restricted to position-annotated terms in
  $\Ter\inflambdaprefixcal$.
  
  By $\sdroppos$ we denote an operation that drops the position annotations in \CRS\nb-terms. 
\end{definition}

\begin{proposition}\label{prop:position:lifting:projecting}
  The following two statements hold:
  \begin{enumerate}[(i)]
    \item\label{prop:position:lifting:projecting:item:i}
      Every rewrite sequence: 
      \begin{equation}\label{eq1:prop:position:lifting:projecting}
        \arewseq \;\funin\;\; 
          \flabs{\vec{\avar}_0}{\aiteri{0}}
            \rred
          \flabs{\vec{\avar}_1}{\aiteri{1}}
            \rred
          \ldots
            \rred
          \flabs{\vec{\avar}_n}{\aiteri{n}}  
      \end{equation}
      (with $r\in\setexp{\sregzero,\, \sreg,\, \sstreg}$) in $\RegzeroARS$, $\RegARS$, or $\stRegARS$
      can be transformed (lifted) step by step,
      for given $\bposi{0}\in\positions$ and $\vec{\apos}_0\in\vecpositions$ 
      with $\length{\vec{\apos}_0} = \length{\vec{\avar}_0}$,
      by adding these and appropriate further position annotations 
      $\bposi{1}, \ldots, \bposi{n}\in\positions$ 
      and $\vec{\apos}_1, \ldots, \vec{\apos}_n\in\vec{\positions}$,
      to a rewrite sequence:
      \begin{equation}\label{eq2:prop:position:lifting:projecting}
        \arewseq^{\spos} \;\funin\;\; 
          \flabspos{\vec{\avar}_0}{\vec{\apos}_0}{\bpos_0}{\aiteri{0}}
            \rred
          \flabspos{\vec{\avar}_1}{\vec{\apos}_1}{\bpos_1}{\aiteri{1}}
            \rred
          \ldots
            \rred
          \flabspos{\vec{\avar}_n}{\vec{\apos}_n}{\bpos_n}{\aiteri{n}}  
      \end{equation}
      (with $r\in\setexp{\sregzero,\, \sreg,\, \sstreg}$) in $\RegposzeroARS$, $\RegposARS$, or $\stRegposARS$,
      accordingly, such that the result of dropping the position annotations in the prefix of
      $\Hat{\arewseq}$ is again~$\arewseq$. 
    \item\label{prop:position:lifting:projecting:item:ii}
      Conversely, every rewrite sequence $\brewseq$ in $\RegposzeroARS$, $\RegposARS$, or $\stRegposARS$
      of the form \eqref{eq2:prop:position:lifting:projecting} (with $r\in\setexp{\sregzero,\, \sreg,\, \sstreg}$)
      can be transformed step by step, by dropping the position annotations in the prefix, to a rewrite sequence~$\Check{\brewseq}$
      of the form
      \eqref{eq2:prop:position:lifting:projecting} (with $r\in\setexp{\sregzero,\, \sreg,\, \sstreg}$)
      in $\RegzeroARS$, $\RegARS$, or $\stRegARS$, respectively.    
  \end{enumerate}
  The transformations in (\ref{prop:position:lifting:projecting:item:i}) and (\ref{prop:position:lifting:projecting:item:ii})
  preserve eagerness/laziness of rewrite sequences.
\end{proposition}

As a direct consequence we obtain, 
for the eager and lazy \scopedelimiting\ (\extscopedelimiting) strategies, 
the following direct correspondence between generated subterms
in the \ARS~$\RegARS$ ($\stRegARS$) and in the position-annotated version $\RegposARS$ ($\stRegposARS$).

\begin{proposition}\label{prop:gST:gSTpos}
  For all infinite \lambdaterm{s}~$\aiter$ it holds:
  \begin{enumerate}[(i)]
    \item\label{prop:gST:gSTpos:item:reg}
      $\gSTregstrat{\astrat}{\aiter} =  \droppos{\gSTregposstrat{\astrat}{\ater}} $
      for the strategies $\astrat \in \setexp{\eagscdelstratreg,\,\lazyscdelstratreg}$ 
      on $\RegARS$, $\RegposARS$. 
    \item\label{prop:gST:gSTpos:item:streg} 
      $\gSTstregstrat{\astrat}{\aiter} =  \droppos{\gSTstregposstrat{\astrat}{\ater}} $
      for the strategies $\astrat \in \setexp{\eagscdelstratstreg,\,\lazyscdelstratstreg}$
      on $\RegARS$, $\RegposARS$. 
  \end{enumerate}
\end{proposition}

The proposition below characterises the binding relation $\sbinds$ and the capturing relation $\iscapturedby$ on the positions of an
infinite term $\aiter$ with the help of rewrite sequences with respect to $\sregzerored$
on $\femptylabspos{\niks}{\rootpos}{\aiter}$ in $\RegposzeroCRS$
down to `variable occurrences' $\flabspos{\vec{\avar}}{\vec{\apos}}{\bpos}{\avari{i}}$ in $\aiter$.

\begin{proposition}\label{prop:bind:iscapturedby}
  For all $\aiter\in\Ter{\inflambdacal}$ and positions $\apos,\bpos\in\Positions{\aiter}$ it holds:
  \begin{align*}
    \apos \binds \bpos 
       \;\;&\Longleftrightarrow\;\;
    \parbox[t]{284pt}{there is a rewrite sequence
                      $\femptylabspos{\tuple{\niks}}{\rootpos}{\aiter} 
                          \regzeromred 
                       \flabspos{\avari{1}\ldots\avari{n}}{\aposi{1},\ldots,\aposi{n}}{\bpos}{\avari{i}}$\\[0.25ex]
                      with $\avari{1}\ldots\avari{n}$ distinct, $i\in\setexp{1,\ldots,n}$, and such that $\apos = \aposi{i}$}
    \\[0.25ex]
    \bpos \iscapturedby \apos 
       \;\;&\Longleftrightarrow\;\;
    \parbox[t]{290pt}{there is a rewrite sequence
                      $\femptylabspos{\tuple{\niks}}{\rootpos}{\aiter}
                         \regzeromred 
                       \flabspos{\avari{1}\ldots\avari{n}}{\aposi{1},\ldots,\aposi{n}}{\bpos}{\avari{i}}$\\[0.25ex]
                      with $\avari{1}\ldots\avari{n}$ distinct, $i\in\setexp{1,\ldots,n}$, and
                      such that $\apos\in\setexp{\aposi{i+1},\ldots,\aposi{n}}$}                  
  \end{align*}
\end{proposition}

The following lemmas describe the close relationship between, on the one hand,
\bindcaptchains\ in an infinite \lambdaterm~$\aiter$, and on the other hand,
$\stregred$\nb-rewrite sequences on $\femptylabspos{\niks}{\rootpos}{\aiter}$
in $\stRegposCRS$ that are guided by the eager \extscopedelimiting\ strategy.

\begin{lemma}[binding--capturing chains]\label{lem:bind:capt:chains:stRegpos} 
  For all $\aiter\in\iTer{\lambdacal}$ 
                                       it holds:
  \begin{enumerate}[(i)]
    \item\label{lem:bind:capt:chains:stRegpos:item:i}
      If $\femptylabspos{\niks}{\rootpos}{\ater}
            \stratmred{\eagscdelstratstreg}
          \flabspos{\avari{0}\ldots\avari{n_1}}{\aposi{0},\ldots,\aposi{n_1}}{\bpos}{\!\bter}
            \stratmred{\eagscdelstratstreg}
          \flabspos{\avari{0}\ldots\avari{n_1}\ldots\avari{n_2}}{\aposi{0},\ldots,\aposi{n_1},\ldots,\aposi{n_2}}{\bpos'}{\!\cter}$,
      then there exist $\bposi{n_1 +1},\ldots,\bposi{n_2}\in\Positions{\ater}$ such that
      $\aposi{n_1} \binds \bposi{n_1 +1} \iscapturedby \aposi{n_1 +1} \binds \ldots \binds \bposi{n_2} \iscapturedby \aposi{n_2}$.
    \item\label{lem:bind:capt:chains:stRegpos:item:ii} 
      If $\aposi{0} \binds \bposi{1} \iscapturedby \aposi{1} \binds \ldots \binds \bposi{n} \iscapturedby \aposi{n}$
      is a \bindcaptchain\ in $\aiter$,
      then there exist  $\cposi{0},\ldots,\cposi{m},\dpos\in\Positions{\ater}$ with $m\ge n$ such that
      $\femptylabspos{\niks}{\rootpos}{\ater} 
         \stratmred{\eagscdelstratstreg} 
       \flabspos{\avari{0}\ldots\avari{m}}{\cposi{0},\ldots,\cposi{m}}{\dpos}{\biter}$
      and furthermore
      $\aposi{0},\ldots,\aposi{n} \in \setexp{\cposi{0},\ldots,\cposi{m}}$
      such that $\aposi{0} < \aposi{1} < \ldots < \aposi{n} = \cposi{m}$.
  \end{enumerate}
\end{lemma}

%

\begin{lemma}\label{lem:fin:bind:capt:chains}
  Let $\aiter$ be an infinite \lambdaterm\ such that
  $\femptylabs{\aiter} \stratmred{\eagscdelstratstreg} \flabs{\avari{0}\ldots\avari{n}}{\biter}$. 
  Then $\aiter$ contains a \bindcaptchain\ of length $n$.
\end{lemma}  

\begin{proof}
  By Proposition~\ref{prop:position:lifting:projecting}, (\ref{prop:position:lifting:projecting:item:i}), 
  the assumed rewrite sequence
  $\femptylabs{\aiter} \stratmred{\eagscdelstratstreg} \flabs{\avari{0}\ldots\avari{n}}{\biter}$
  in $\stRegARS$
  can be lifted
  to a rewrite sequence 
  $\femptylabspos{\niks}{\rootpos}{\aiter} \stratred{\eagscdelstratstreg} \flabspos{\avari{0}\ldots\avari{n}}{\aposi{0},\ldots,\aposi{n}}{\bpos}{\biter}$
  in $\stRegposARS$.
  Then by Lemma~\ref{lem:bind:capt:chains:stRegpos}, (\ref{lem:bind:capt:chains:stRegpos:item:i}), 
  there exists a \bindcaptchain\ of length~$n$. 
\end{proof}

The notion of scope and \extscope\ helps to understand the relationship between
\bindcaptchains{} and rewrite sequences in $\stRegCRS$. A \bindcaptchain\
corresponds to the overlap of scopes, or in other words the nesting of
\extscope{s}.
 An infinite \bindcaptchain\ thus corresponds to a infinitely
deep nesting of \extscope{s} and therefore to an unrestricted growth of the
prefix in certain rewriting sequences in $\stRegCRS$.

\begin{lemma}[infinite \bindcaptchains]\label{lem:inf:bind:capt:chains}
  Let $\aiter$ be an infinite \lambdaterm, and let $\arewseq$ be an infinite rewrite sequence
  in $\RegARS$ w.r.t.\ the eager \extscopedelimiting\ strategy~$\eagscdelstratstreg\,$:
  \begin{equation}\label{eq:lem:inf:bind:capt:chains}
    \arewseq \;\funin\;\; 
       \femptylabs{\aiter} = \flabs{\vec{\avar}_0}{\aiteri{0}}
         \stratred{\eagscdelstratstreg}
       \flabs{\vec{\avar}_1}{\aiteri{1}}
         \stratred{\eagscdelstratstreg}
       \ldots
         \stratred{\eagscdelstratstreg}
       \flabs{\vec{\avar}_i}{\aiteri{i}}  
         \stratred{\eagscdelstratstreg}
       \ldots
  \end{equation}
  Furthermore suppose that 
  for $\spl \funin \nats\to\nats$, $i\mapsto\pl{i}\defdby \length{\vec{\avar}_i}$,
  the prefix length function associated with $\arewseq$,
  there exists a lower bound $\slb \funin\nats\to\nats$ such
  that $\slb$ is non-decreasing, and 
  $\lim_{n\to\infty} \lb{n} = \infty$. 
  Then there exists an infinite \bindcaptchain\ in $\aiter$.
\end{lemma}  

\begin{proof}
  Let $\aiter$, $\arewseq$, $\spl$, $\slb$ as in the assumption of the lemma. 
  We first note that 
  by Proposition~\ref{prop:position:lifting:projecting}, (\ref{prop:position:lifting:projecting:item:i}),
  the rewrite sequence~$\arewseq$ can be lifted to one with position annotations: 
  \begin{equation}\label{eq1:prf:lem:inf:bind:capt:chains}
    \arewseq^{\spos} \;\funin\;\; 
       \femptylabspos{\niks}{\rootpos}{\aiter} = \flabspos{\vec{\avar}_0}{\niks}{\rootpos}{\aiteri{0}}
         \stratred{\eagscdelstratreg}
       \flabspos{\vec{\avar}_1}{\vec{\apos}_1}{\bpos_1}{\aiteri{1}}
         \stratred{\eagscdelstratreg}
       \ldots
         \stratred{\eagscdelstratreg}
       \flabspos{\vec{\avar}_i}{\vec{\apos}_i}{\bpos_i}{\aiteri{i}}  
         \stratred{\eagscdelstratreg}
       \ldots
  \end{equation}
  where, for all $i\in\nats$,
  $\bposi{i}$ are positions and $\vec{\apos}_i = \tuple{\aposi{1},\ldots,\aposi{m_i}}$ vectors of positions,
  with $m_i\in\nats$.
  
  Next we define the function:
  \begin{equation*}
    \sst \funin \nats\to\nats, \;\;\;
                l \mapsto \st{l} \defdby \min \descsetexp{i}{\lb{i}\ge l}   
  \end{equation*}
  which is well-defined, since $\lim_{n\to\infty}{\lb{n} = \infty}$.
  It describes a prefix stabilisation property:
  for every $l\in\nats$, it gives
  the first index $i = \st{l}$ with the property that 
  the prefix of $\flabs{\vec{\avar}_i}{\aiteri{i}}$ contains more than $l$ abstractions,
  and (since $\slb$ is non-decreasing, and a lower bound for $\spl$) 
  that from $i$ onward
  the $l$\nb-th abstraction never disappears again, for $j\ge i$,
  in terms $\flabs{\vec{\avar}_j}{\aiteri{j}}$ that follow in $\arewseq$
  as well as in $\arewseq^{\spos}$.
  Furthermore, $\sst$ is non-decreasing, as an easy consequence of its definition,
  and unbounded:
  if $\sst$ were bounded by $M\in\nats$, then 
  \mbox{$\forall l\in\nats \, \exists i\in\nats. \,i\le M \logand \lb{i}\ge l$} would follow,
  which cannot be the case since $\setexp{\lb{0},\ldots,\lb{M}}$ is a finite set.
  By non-decreasingness and unboundedness it also follows that $\lim_{n\to\infty}{\st{n} = \infty}$. 
   
  So when the rewrite sequence $\arewseq^{\spos}$ is split into segments indicated in: 
  \begin{equation*}
       \femptylabspos{\niks}{\rootpos}{\aiter} 
         \stratmred{\eagscdelstratreg}
         \ldots
         \stratmred{\eagscdelstratreg}
       \flabspos{\vec{\avar}_{\st{i}}}{\vec{\apos}_{\st{i}}}{\bpos_{\st{i}}}{\aiteri{\st{i}}}
         \stratmred{\eagscdelstratreg}
       \flabspos{\vec{\avar}_{\st{i+1}}}{\vec{\apos}_{\st{i+1}}}{\bpos_{\st{i+1}}}{\aiteri{\st{i+1}}}  
         \stratmred{\eagscdelstratreg}
       \ldots
  \end{equation*}
  then it follows that all terms of the sequence after 
  $\flabspos{\vec{\avar}_{\st{i}}}{\vec{\apos}_{\st{i}}}{\bpos_{\st{i}}}{\aiteri{\st{i}}}$
  have an abstraction prefix of length greater or equal to $i$, for all $i\in\nats$.
  
  Now note that in a step 
  $\flabspos{\vec{\avar}}{\tuple{\aposi{1},\ldots,\aposi{n}}}{\bpos}{\citer} 
     \red 
   \flabspos{\vec{\avar}'}{\tuple{\apos'_{1},\ldots,\apos'_{n'}}}{\bpos'}{\citer'}$
  in $\RegposARS$
  that does not shorten the abstraction prefix it holds that $n\le n'$ and
  $\apos'_{1} = \apos_{1}$, \ldots, $\apos'_{n} = \apos_{n}$,
  that is, positions in the vector in the subscript of the abstraction prefix are preserved. 
  As a consequence it follows for the rewrite sequence $\arewseq^{\spos}$
  that, for all $i\in\nats$ and $j > i$, the position vector $\vec{\apos}_{\st{j}}$ in the
  term $\flabspos{\vec{\avar}_{\st{j}}}{\vec{\apos}_{\st{j}}}{\bpos_{\st{j}}}{\aiteri{\st{j}}}$ 
  is of the following form:
  \begin{align*}
    \vec{\apos}_{\st{j}} 
      =
    \tuple{\aposi{1,\st{i}},\ldots,\aposi{i,\st{i}},\aposi{j,\st{j}},\ldots,\aposi{j,m_j}}   
  \end{align*}
  This implies furthermore that for all $i\in\nats$:
  \begin{align*}
    \vec{\apos}_{\st{i}} 
      =
    \tuple{\aposi{1,\st{1}},\aposi{2,\st{2}},\ldots,\aposi{i,\st{i}},\ldots,\aposi{i,m_i}}   
  \end{align*}
  Then Lemma~\ref{lem:bind:capt:chains:stRegpos}, (\ref{lem:bind:capt:chains:stRegpos:item:i}), 
  implies the existence of positions $\bposi{2},\bposi{3},\ldots$ such that:
  \begin{equation*}
    \aposi{1,\st{1}} 
      \binds \bposi{2} \iscapturedby 
    \aposi{2,\st{2}}
       \binds \bposi{3} \iscapturedby  
    \ldots
       \iscapturedby
    \aposi{i,\st{i}}
       \binds \bposi{i+1} \iscapturedby
    \aposi{i+1,\st{i+1}} 
       \binds  
    \ldots \;
      \punc{.}
  \end{equation*}
  and thereby, an infinite \bindcaptchain\ in $\aiter$.
\end{proof}
  
Now we formulate and prove the main theorem of this section,
which applies the concept of \bindcaptchain\
to pin down, among all infinite \lambdaterms\ that are regular, 
those that are strongly regular.

\begin{theorem}\label{thm:streg:fin:bind:capt:chains}
  A regular infinite \lambdaterm\ is strongly regular 
    if and only if
  it contains only finite \bindcaptchains.
\end{theorem}

By adding the statement of Proposition~\ref{prop:def:reg:streg}, (\ref{prop:def:reg:streg:item:i}),
we obtain the following accentuation. 

\begin{corollary}\label{cor:thm:streg:fin:bind:capt:chains}
A infinite \lambdaterm\ is strongly regular if and only if it is regular and contains only finite \bindcaptchains.
\end{corollary}

\begin{proof}[Proof of Theorem~\ref{thm:streg:fin:bind:capt:chains}.]
  Let $\aiter$ be an infinite \lambdaterm\ that is regular. 
 
  For showing the implication ``$\Rightarrow$'', we assume that $\aiter$ is also strongly regular.
  Then there exists a \extscopedelimiting\ strategy $\astrat$ such that
  $\gSTstregstrat{\astrat}{\aiter}$ is finite. 
  By Proposition~\ref{prop:eager:strat:in:def:reg:streg}, (\ref{prop:eager:strat:in:def:reg:streg:item:i}) 
  it follows that then also $\gSTstregstrat{\eagscdelstratstreg}{\aiter}$ is finite
  for the eager \extscopedelimiting\ strategy $\eagscdelstratstreg$ in $\stRegARS$.
  Now let $n$ be the longest abstraction prefix of a term in $\gSTstregstrat{\eagscdelstratstreg}{\aiter}$.
  Then it follows by Lemma~\ref{lem:fin:bind:capt:chains} 
  that the length of
  every \bindcaptchain\ in $\aiter$ is bounded by $n-1$. 
  Hence $\aiter$ only contains finite \bindcaptchains.
  
  In the rest of this proof, we establish the implication ``$\Leftarrow$'' in the statement of the theorem.
  For this we argue indirectly: assuming that $\aiter$ is not strongly regular, we show the existence of 
  an infinite \bindcaptchain\  in $\aiter$.
  
  So suppose that $\aiter$ is not strongly regular. 
  Then for all \extscopedelimiting\ strategies~$\astrat$ in $\stRegARS$ it holds
  that $\gSTstregstrat{\astrat}{\aiter}$ is infinite.
  This means that in particular $\gSTstregstrat{\eagscdelstratstreg}{\aiter}$ is infinite
  for the eager \scopedelimiting\ strategy $\eagscdelstratstreg$ on $\stRegARS$.
  It follows that the number of 
                          $\sstratmred{\eagscdelstratstreg}$\nb-reducts,  
  and hence the induced sub\nb-ARS $\InducedSubARSmred{\femptylabs{\aiter}}{\sstratmred{\eagscdelstratstreg}}$ of $\femptylabs{\aiter}$ in $\stRegARS$
  is infinite. 
  Since $\sstratmred{\eagscdelstratstreg}$ on $\stRegARS$ has branching degree~$\le 2$ 
  (branching actually only happens at sources of $\slappdecompired{i}$\nb-steps), 
  it follows by K\H{o}nig's Lemma that there exists an infinite rewrite sequence:
  \begin{equation*}
    \arewseq \;\funin\;\; 
       \femptylabs{\aiter} = \flabs{\vec{\avar}_0}{\aiteri{0}}
         \stratred{\eagscdelstratstreg}
       \flabs{\vec{\avar}_1}{\aiteri{1}}
         \stratred{\eagscdelstratstreg}
       \ldots
         \stratred{\eagscdelstratstreg}
       \flabs{\vec{\avar}_i}{\aiteri{i}}  
         \stratred{\eagscdelstratstreg}
       \ldots
  \end{equation*}
  in $\stRegARS$ that passes through distinct terms. 
  By Lemma~\ref{lem:projection:RegCRS:stRegCRS}, (\ref{lem:projection:RegCRS:stRegCRS:item:reg}),
  this rewrite sequence projects to a rewrite sequence:
  \begin{equation}\label{eq2:prf:thm:streg:fin:bind:capt:chains}
    \Checkreg{\arewseq} \;\funin\;\; 
       \femptylabs{\aiter} = \flabs{\vec{\avar}'_0}{\aiteri{0}}
         \stratmred{\eagscdelstratreg}
       \flabs{\vec{\avar}'_1}{\aiteri{1}}
         \stratmred{\eagscdelstratreg}
       \ldots
         \stratmred{\eagscdelstratreg}
       \flabs{\vec{\avar}'_i}{\aiteri{i}}  
         \stratmred{\eagscdelstratreg}
       \ldots
  \end{equation}
  in $\RegARS$
  in the sense that:
  \begin{equation}\label{eq3:prf:thm:streg:fin:bind:capt:chains}
     \flabs{\vec{\avar}_i}{\aiteri{i}} \compressregmred \flabs{\vec{\avar}'_i}{\aiteri{i}}
       \;\;\;\;\;\; \text{(for all $i\in\nats$)} \punc{;}
  \end{equation}    
  note that, in the terms, the projection merely shortens the length of the abstraction prefix.
  Since $\aiter$ is regular, $\gSTregstrat{\eagscdelstratreg}{\aiter}$ is finite
  by Proposition~\ref{prop:eager:strat:in:def:reg:streg}, (\ref{prop:eager:strat:in:def:reg:streg:item:i}),
  and hence it follows that only finitely many terms occur in $\Checkreg{\arewseq}$.
  
  Now we will use this contrast with $\arewseq$, and the fact that the terms of $\arewseq$
  project to terms in $\Checkreg{\arewseq}$ via $\scompressregmred$\nb-prefix compression rewrite sequences,
  to show that the prefix lengths in terms of $\arewseq$
  are unbounded, and stronger still, that these lengths actually tend to infinity. More precisely,
  we show the following: 
  \begin{equation}\label{eq4:prf:thm:streg:fin:bind:capt:chains}
    (\forall l\in\nats)(\exists i_0\in\nats)\, (\forall i\ge i_0) \, [\, \length{\vec{\avar}_i} \ge l \,] \punc{.} 
  \end{equation}
  Suppose that this statement does not hold.
  Then there exists $l_0\in\nats$ such that $\length{\vec{\avar}_i} < l_0$ for infinitely many $i\in\nats$.
  This means that there is an increasing sequence $i_0 < i_1 < i_2 < i_3 < \ldots$ in $\nats$
  such that: 
  \begin{gather}
    S \defdby \descsetexp{\flabs{\vec{\avar}_{i_j}}{\aiteri{i_j}}}{j\in\nats}
      \text{ is infinite} 
    \label{eq6:prf:thm:streg:fin:bind:capt:chains}  
    \\
    \text{for all $\flabs{\vec{\avar}_{i_j}}{\aiteri{i_j}} \in S\,$: \; $\length{\vec{\avar}_{i_j}} < l_0$ }  
    \label{eq6a:prf:thm:streg:fin:bind:capt:chains}
  \end{gather}
  ($S$ is infinite since the terms on $\arewseq$ are distinct).
  On the other hand we have:
  \begin{gather}\label{eq7:prf:thm:streg:fin:bind:capt:chains}
    T \defdby
    \descsetexp{\flabs{\vec{\avar}'_{i_j}}{\aiteri{i_j}}}{j\in\nats} 
      \; \subseteq \;
    \gSTregstrat{\eagscdelstratreg}{\aiter} \text{ is finite}   
  \end{gather} 
  because $\aiter$ is regular. 
  However, since every term in $S$ has a $\scompressregmred$\nb-reduct in $T$ due to \eqref{eq3:prf:thm:streg:fin:bind:capt:chains},
  as well as an abstraction prefix of a length bounded by $l_0$,
  it follows by Proposition~\ref{prop:compress:prefix:RegARS}, (\ref{prop:compress:prefix:RegARS:item:ii}),
  that $S$ also has to be finite, conflicting with \eqref{eq6:prf:thm:streg:fin:bind:capt:chains}.
  We have reached a contradiction, and thereby established \eqref{eq4:prf:thm:streg:fin:bind:capt:chains}. 
  
  Now we are able to define a lower bound on the lengths of the prefixes in $\arewseq$
  that fulfils the requirements of Lemma~\ref{lem:inf:bind:capt:chains}.
  We define the function:
  \begin{equation*}
    \slb \funin \nats \to \nats, 
      \;\;\;
         n \mapsto \lb{n} \defdby \min \descsetexp{\length{\vec{\avar}_{n'}}}{n'\ge n}
  \end{equation*}
  Its definition guarantees that $\slb$ is a lower bound on the prefix lengths in $\arewseq$,
  and that $\slb$ is non-decreasing. 
  Furthermore also $\lim_{n\to\infty} \lb{n} = \infty$ follows by non-decreasingness, in addition to unboundedness of $\slb$:
  for arbitrary $l\in\nats$, by \eqref{eq4:prf:thm:streg:fin:bind:capt:chains}
  there exists $n_0\in\nats$ such that $\length{\vec{\avar}_n}\ge l$ holds for all $n\in\nats$, $n\ge n_0$;
  this entails $\lb{n_0} \ge l$.  
  
  Now since $\ater$, $\arewseq$, together with $\slb$ as defined above, satisfy the assumptions of
  Lemma~\ref{lem:inf:bind:capt:chains},
  this lemma can be applied, yielding an infinite \bindcaptchain\ in $\aiter$. 
\end{proof}

\begin{example}\label{ex:bindcaptchain}
  The infinite \lambdaterm\ from Example~\ref{ex:entangled} 
  with a representation as a higher-order recursive program scheme in Example~\ref{ex:entangled-infinite-path},
  which was recognised there to be regular but not strongly regular, 
  possesses an infinite \bindcaptchain\ as indicated on the right in Figure~\ref{fig:entangled_bc}.
\end{example}

\section{Expressibility by terms of the $\protect\lambda$-calculus with $\protect\sletrec$}
  \label{sec:express}

In this section we finish the proof of our main characterisation result:
we prove that every strongly regular \lambdaterm\ is \lambdaletrec\nb-expressible.
For this purpose we introduce an annotated variant of one of the proof systems for
strongly regular infinite \lambdaterms.
We show that every closed derivation in $\stRegzero$ with conclusion $\femptylabs{\aiter}$,
which witnesses that $\aiter$ is strongly regular, can be annotated, 
by adding appropriate \lambdaletrec\nb-terms to each prefixed term in the derivation,
into a derivation in the annotated system with conclusion $\femptylabsann{\allter}{\aiter}$
such that the \lambdaletrec\nb-term annotation $\allter$ expresses the infinite \lambdaterm~$\aiter$. 
We show the correctness of this construction by transforming the derivation in the annotated proof system
into a derivation in the proof system $\stRegeq$ 
with conclusion $\femptylabs{\Unf{\allter}} = \femptylabs{\aiter}$, and then
drawing upon  the soundness of $\stRegeq$ with respect to equality of strongly regular infinite \lambdaterms.

\begin{figure}[t!]
\begin{center}  
  \framebox{
\begin{minipage}{360pt}
\begin{center}
  \mbox{}
  \\[1ex]
  \mbox{ 
    \AxiomC{}
    \RightLabel{\bvarax}
    \UnaryInfC{$\flabsann{\vec{\avar}\bvar}{\bvar}{\bvar}$}
    \DisplayProof
        } 
  \hspace*{3ex}       
  \mbox{
    \AxiomC{$ \flabsann{\vec{\avar}\bvar}{\allter}{\ater} $}
    \RightLabel{$\labscomp$}
    \UnaryInfC{$ \flabsann{\vec{\avar}}{\labs{\bvar}{\allter}}{\labs{\bvar}{\ater}} $}
    \DisplayProof
        }  
  \hspace*{3ex}
  \mbox{
    \AxiomC{$ \flabsann{\vec{\avar}}{\allteri{0}}{\ateri{0}}$}
    \AxiomC{$ \flabsann{\vec{\avar}}{\allteri{1}}{\ateri{1}}$}
    \RightLabel{$\lappcomp$}
    \BinaryInfC{$ \flabsann{\vec{\avar}}{\lapp{\allteri{0}}{\allteri{1}}}{\lapp{\ateri{0}}{\ateri{1}}} $}
    \DisplayProof 
    }     
  \\[2ex]
  \mbox{
    \AxiomC{$ \flabsann{\avari{1}\ldots\avari{n-1}}{\allter}{\aiter} $}
    \RightLabel{\annVacstreg\
                \small (if the binding $\slabs\avari{n}$ is vacuous)
                       }
    \UnaryInfC{$ \flabsann{\avari{1}\ldots\avari{n}}{\allter}{\aiter} $}
    \DisplayProof
        }
  \\[2ex]
  \mbox{
    \AxiomC{$ [\flabsann{\vec{\avar}}{\aconstnamei{\aannvar}}{\ater}]^{\aannvar} $}
    \noLine
    \UnaryInfC{$\Derivi{0}$}
    \noLine
    \UnaryInfC{$ \flabsann{\vec{\avar}}{\subst{\allter}{\aannvar}{\aconstnamei{\aannvar}}}{\ater} $}
    \RightLabel{$\sFIX,u$ $\;$
                \parbox{195pt}{\small 
                               (if $\depth{\Derivi{0}} \ge 1$, and $\length{\vec{y}} \ge \length{\vec{x}}$ for all
                                $\flabs{\vec{y}}{\bter}$ on threads 
                                from open assumptions $(\flabsann{\vec{\avar}}{\aannvar}{\ater})^u$ down)}
                }                   
    \UnaryInfC{$ \flabsann{\vec{\avar}}{(\letrec{\aannvar = \allter}{\aannvar})}{\ater} $}
    \DisplayProof
    }
  \\[1.5ex]
  \mbox{}  
\end{center}
\end{minipage}
   }
\end{center} 
  \vspace*{-1.25ex}  
  \caption{\label{fig:annstRegzero}%
           Annotated natural-deduction style 
           proof system $\annstRegzero$ for strongly regular infinite $\lambda$\nb-terms,
           a version of $\stRegzero$ with 
           $\lambdaletrec$-terms as annotations.}
\end{figure}

We start by introducing a variant of the proof system~$\stRegzero$ in which the
formulas are closed, prefixed, \lambdaletrecterm\nb-annotated, infinite \lambdaterms.

\begin{definition}[the proof system $\annstRegzero$]
  The formulas of the proof system $\annstRegzero$ are closed expressions of the form $\flabsann{\vec{\avar}}{\allter}{\aiter}$
  with $\vec{\avar}$ a variable prefix vector, $\labs{\vec{\avar}}{\allter}$ a $\lambdaletrec$\nb-term,
  and $\labs{\vec{\avar}}{\aiter}$ a \lambdaterm. 
  The axioms and rules of \annstRegzero\
    are annotated versions of the axioms and rules of the proof system $\stRegzero$
    from Definition~\ref{def:Reg:stReg:stRegzero} and Figure~\ref{fig:stReg:stRegzero},
  and are displayed in Figure~\ref{fig:annstRegzero}.
\end{definition}

\begin{remark}
  As for an example that illustrates why we have chosen to formulate 
  an annotated version only of the proof system $\stRegzero$, but not of $\stReg$,
  please see Example~\ref{example:annstRegzero}. 
\end{remark}

The following proposition is a statement that is entirely analogous to Proposition~\ref{prop:Reg:stReg}. 

\begin{proposition}\label{prop:annstRegzero:sidecondition}
  For all for all instances $\ainst$ of the rule \FIX\ in a derivation $\Deriv$ 
  (possibly with open assumptions) in $\annstRegzero$ it holds:
  every thread from $\ainst$ upwards to a marked assumption that is discharged at $\ainst$
  passes at least one instance of a rule ($\labscomp$) or ($\lappcomp$).
\end{proposition}

The lemma below states a straightforward connection between derivations in
$\stRegzero$ and derivations in its annotated version $\annstRegzero$.

\begin{lemma}[from $\stRegzero$- to $\annstRegzero$-derivations, and back]%
  \label{lem:stRegzero:2:annstRegzero}
  The following transformations are possible between derivations in \stRegzero\ and derivations in \annstRegzero:
  \begin{enumerate}[(i)]
    \item{}\label{lem:stRegzero:2:annstRegzero:item:i} 
      Every derivation $\Deriv$ in \stRegzero\ with conclusion $\flabs{\vec{\avar}}{\aiter}$
      can be transformed into 
      a derivation $\Derivann$ in \annstRegzero\ with conclusion $\flabsann{\vec{\avar}}{\allter}{\aiter}$
      such that there is a bijective correspondence between marked assumptions $(\flabs{\vec{\bvar}}{\aiter})^{\amarker}$ in $\Deriv$
      and marked assumptions $(\flabsann{\vec{\bvar}}{\amarker}{\aiter})^{\amarker}$ in $\Derivann$. 
      (As a consequence, $\Derivann$ is a closed derivation if $\Deriv$ is closed.)
      More precisely, $\Derivann$ can be obtained from $\Deriv$ 
      by replacing every term occurrence $\flabs{\vec{\bvar}}{\biter}$ by
      an occurrence of $\flabsann{\vec{\bvar}}{\bllter}{\biter}$ for a prefixed \lambdaletrecterm~$\flabs{\vec{\bvar}}{\bllter}$ 
      with the property that every prefix variable $\bvari{i}$ bound in $\bllter$ is also bound in $\biter$.
      Thereby occurrences of marked assumptions and axioms $\bvarax$ in $\Deriv$ 
      give rise to 
      occurrences of marked assumptions and axioms $\bvarax$ in $\Derivann$, respectively;
      instances of the $\stRegzero$\nb-rules $\labscomp$, $\lappcomp$, $\Vacstreg$, and \FIX\ in $\Deriv$
      give rise to
      instances of $\annstRegzero$\nb-rules $\labscomp$, $\lappcomp$, $\Vacstreg$, and \FIX\ in $\Derivann$, respectively.
    \item{}\label{lem:stRegzero:2:annstRegzero:item:ii}
      From every closed derivation $\Deriv$ in \annstRegzero\ with conclusion $\flabsann{\vec{\avar}}{\allter}{\aiter}$
      a closed derivation $\check{\Deriv}$ in \stRegzero\ with conclusion $\flabs{\vec{\avar}}{\aiter}$
      can be obtained by dropping the annotations with \lambdaletrecterms. 
  \end{enumerate}
\end{lemma}

\begin{proof}
  Statement~(\ref{lem:stRegzero:2:annstRegzero:item:i}) of the lemma can be established
  through a proof by induction on the depth~$\depth{\Deriv}$ of a derivation $\Deriv$ in $\stRegzero$ 
  with possibly open assumptions. 
  In the base case, axioms ($\bvarax$) of $\stRegzero$ are annotated to axioms ($\bvarax$) of $\annstRegzero$,
  and marked assumptions $(\flabs{\vec{\bvar}}{\biter})^{\amarker}$ in $\stRegzero$ 
  to marked assumptions $(\flabsann{\vec{\bvar}}{\aconstnamei{\amarker}}{\biter})^{\amarker}$.
  In the induction step it has to be shown that a derivation $\Deriv$ in $\stRegzero$
  with immediate subderivation $\Derivi{0}$ can be annotated to a derivation $\Derivann$ in $\annstRegzero$,
  using the induction hypothesis which guarantees that an annotated version $\Derivanni{0}$ of $\Derivi{0}$ has already been obtained.
  Then for obtaining $\Derivann$ from $\Derivanni{0}$ the fact is used that
  the rules in $\annstRegzero$ uniquely determine the annotation in the conclusion of an instance once the annotation(s) in the premise(s)
  (and in the case of \FIX\ additionally the annotation markers used in the assumptions that are discharged) are given. 
  In order to establish that instances of $\Vacstreg$ in $\Deriv$ give rise to corresponding instances of $\Vacstreg$ in $\Derivann$,
  the part of the induction hypothesis is used which guarantees that the \lambdaletrecterm\ annotation in the premise
  contains not more variable bindings than the infinite \lambdaterm\ it annotates.
  
  Statement~(\ref{lem:stRegzero:2:annstRegzero:item:ii}) of the lemma
  is a consequence of the fact that, by dropping the \lambdaletrecterm-annotations,
  every instance of a rule of $\annstRegzero$ give rise to an instance of the corresponding rule in $\stRegzero$.
  Formally the statement can again be established by induction on the depth of derivations in $\annstRegzero$.  
\end{proof}

\begin{example}\label{example:annstRegzero}
  The derivation $\Derivi{l}$ in $\stRegzero$ from Example~\ref{example:stReg} on the left can be annotated,
  as described by Lemma~\ref{lem:stRegzero:2:annstRegzero}, (\ref{lem:stRegzero:2:annstRegzero:item:ii}), 
  to obtain the following derivation $\Hat{\Deriv}_{l}$ in $\annstRegzero\,$:
  \begin{center}
    \AxiomC{$ (\femptylabsann{\aconstnamei{\amarker}}{\aiter})^{\amarker} $}
    \RightLabel{$\Vacstreg$}
    \UnaryInfC{$ \flabsann{\avar}{\aconstnamei{\amarker}}{\ater} $}
    \RightLabel{$\Vacstreg$}
    \UnaryInfC{$ \flabsann{\avar\bvar}{\aconstnamei{\amarker}}{\ater} $}
    \AxiomC{\mbox{}}
    \RightLabel{$\bvarax$}
    \UnaryInfC{$ \flabsann{\avar\bvar}{\bvar}{\bvar} $}
    \RightLabel{$\lappcomp$}
    \BinaryInfC{$ \flabsann{\avar\bvar}{\lapp{\aconstnamei{\amarker}}{\bvar}}{\lapp{\ater}{\bvar}} $}
    \AxiomC{\mbox{}}
    \RightLabel{$\bvarax$}
    \UnaryInfC{$ \flabsann{\avar}{\avar}{\avar} $}
    \RightLabel{$\Vacstreg$}
    \UnaryInfC{$ \flabsann{\avar\bvar}{\avar}{\avar} $}
    \RightLabel{$\lappcomp$}
    \BinaryInfC{$ \flabsann{\avar\bvar}{\lapp{\lapp{\aconstnamei{\amarker}}{\bvar}}{\avar}}{\lapp{\lapp{\ater}{\bvar}}{\avar}} $}
    \RightLabel{$\labscomp$}
    \UnaryInfC{$ \flabsann{\avar}{\labs{\bvar}{\lapp{\lapp{\aconstnamei{\amarker}}{\bvar}}{\avar}}}{\lapp{\lapp{\ater}{\bvar}}{\avar}} $}
    \RightLabel{$\labscomp$}
    \UnaryInfC{$ \femptylabsann{\labs{\avar\bvar}{\lapp{\lapp{\aconstnamei{\amarker}}{\bvar}}{\avar}}}{\labs{\avar\bvar}{\lapp{\lapp{\ater}{\bvar}}{\avar}}} $}
    \RightLabel{$\sFIX, u$}
    \UnaryInfC{$ \femptylabsann{(\letrec{\amarker = \labs{\avar\bvar}{\lapp{\lapp{\amarker}{\bvar}}{\avar}}}{\arecvar})}{\ater} $}
    \DisplayProof
  \end{center}
  Note that the term in the conclusion, which has been extracted by the annotation procedure, 
  is actually the same as the \lambdaletrec\nb-term
  $\letrec{\arecvar = \labs{\avar\bvar}{\lapp{\lapp{\arecvar}{\bvar}}{\avar}}}{\arecvar}$
  which was used in Example~\ref{example:stReg} to define $\aiter$ as its infinite unfolding.
  
  Furthermore note that,
  in a variant of $\annstRegzero$ in which the `$\stRegzero$\nb-addition' 
  (concerning abstraction prefix lengths) to the side-condition of \FIX\ is dropped,
  the derivation $\Derivi{r}$ in Example~\ref{example:stReg} on the right
  could be annotated to obtain the following prooftree $\Hat{\Deriv}_r$:
\begin{equation*}
  \mbox{
    \AxiomC{$ (\flabsann{\avar}{\aconstnamei{\amarker}}{\labs{\bvar}{\lapp{\lapp{\ater}{\bvar}}{\avar}}})^{\amarker} $}
    \RightLabel{$\labscomp$}
    \UnaryInfC{$ \femptylabsann{\labs{\avar}{\aconstnamei{\amarker}}}{\ater} $}
    \RightLabel{$\Vacstreg$}
    \UnaryInfC{$ \flabsann{\avar}{\labs{\avar}{\aconstnamei{\amarker}}}{\ater} $}
    \RightLabel{$\Vacstreg$}
    \UnaryInfC{$ \flabsann{\avar\bvar}{\labs{\avar}{\aconstnamei{\amarker}}}{\ater} $}
    \AxiomC{\mbox{}}
    \RightLabel{$\bvarax$}
    \UnaryInfC{$ \flabsann{\avar\bvar}{\bvar}{\bvar} $}
    \RightLabel{$\lappcomp$}
    \BinaryInfC{$ \flabsann{\avar\bvar}{\lapp{(\labs{\avar}{\aconstnamei{\amarker}})}{\bvar}}{\lapp{\ater}{\bvar}} $}
    \AxiomC{\mbox{}}
    \RightLabel{$\bvarax$}
    \UnaryInfC{$ \flabsann{\avar}{\avar}{\avar} $}
    \RightLabel{$\Vacstreg$}
    \UnaryInfC{$ \flabsann{\avar\bvar}{\avar}{\avar} $}
    \RightLabel{$\lappcomp$}
    \BinaryInfC{$ \flabsann{\avar\bvar}{\lapp{(\lapp{\labs{\avar}{\aconstnamei{\amarker}})}{\bvar}}{\avar}}{\lapp{\lapp{\ater}{\bvar}}{\avar}} $}
    \RightLabel{$\labscomp$}
    \UnaryInfC{$ \flabsann{\avar}{\labs{\bvar}{\lapp{\lapp{(\labs{\avar}{\aconstnamei{\amarker}})}{\bvar}}{\avar}}}{\labs{\bvar}{\lapp{\lapp{\ater}{\bvar}}{\avar}}} $}
    \RightLabel{$\ainst$, $u$ \parbox{90pt}{(no instance of \FIX\\[-0.75ex] \hspace*{\fill}in $\annstRegzero\,$ !)}}
    \UnaryInfC{$ \flabsann{\avar}{\letrec{\amarker = \labs{\bvar}{\lapp{\lapp{\amarker}{\bvar}}{\avar}}}{\amarker}}{\labs{\bvar}{\lapp{\lapp{\ater}{\bvar}}{\avar}}} $}
    \RightLabel{$\labscomp$}
    \UnaryInfC{$ \femptylabsann{\labs{\avar}{\letrec{\amarker = \labs{\bvar}{\lapp{\lapp{\amarker}{\bvar}}{\avar}}}{\amarker}}}{\ater} $}
    \DisplayProof
        }
\end{equation*}
  Observe that, equally as was the case for $\Derivi{r}$, 
  also in $\Hat{\Deriv}_r$ there occurs,
  on the thread between the marked assumption at the top and the rule instance $\ainst$ at which this assumption is discharged,
  a formula, namely $\femptylabsann{\amarker}{\ater}$, 
  that has a shorter abstraction prefix than the formula in the premise and conclusion of $\ainst$ as well as in the assumption.
  Thus $\ainst$ is not an instance of the rule \FIX\ in $\annstRegzero$.

  Furthermore note that the \lambdaletrecterm\ extracted by $\Hat{\Deriv}_r$ does not
  unfold to $\aiter$, and hence does not express $\aiter$.
  This example shows that the side-condition on instances of \FIX\ in $\annstRegzero$ cannot be
  weakened to the form used for the rule \FIX\ in $\stReg$ 
  when the aim is to extract a \lambdaletrecterm\ that unfolds to the infinite \lambdaterm\ in the conclusion. 
\end{example}

The central property of the proof system $\annstRegzero$ still remains to be shown:
that the \lambdaletrecterms\ in the conclusion of a derivation in this system
does actually unfold to the infinite \lambdaterm\ in the conclusion.
This will be established below in Lemma~\ref{lem:annstRegzero:2:stRegeq} and Theorem~\ref{thm:annstRegzero}.
But as an intermediary proof system that will allow us to use results about the
proof system $\stRegletrec$ from Section~\ref{sec:proofs}, we also introduce an annotated
version of the rule $\sletrec$ in $\stRegletrec$, and an according annotated proof system.

\begin{figure}[t!]
\begin{center}  
  \framebox{
\begin{minipage}{365pt}
\begin{center}
  \mbox{}
  \\[0.5ex]
  \mbox{
    \AxiomC{$ \{\: [ \flabsann{\vec{\avar}}{\aconstnamei{\arecvari{i}}}{\aiteri{i}} ] \:\}_{i=1,\ldots,n} $}
    \noLine
    \UnaryInfC{$\Derivi{j}$}
    \noLine
    \UnaryInfC{$ \{\:\ldots\ldots \;\; \flabsann{\vec{\avar}}{\subst{\allteri{j}}{\vec{\arecvar}}{\vec{\aconstname}_{\vec{\arecvar}}}}{\aiteri{j}} \;\;\ldots\ldots\:\}_{j=1,\ldots,n} $}
    \AxiomC{$ \{\: [ \flabsann{\vec{\avar}}{\aconstnamei{\arecvari{i}}}{}\aiteri{i} ] \:\}_{i=1,\ldots,n} $}
    \noLine
    \UnaryInfC{$\Derivi{n+1}$}
    \noLine
    \UnaryInfC{$ \flabsann{\vec{\avar}}{\subst{\allteri{n+1}}{\vec{\arecvar}}{\vec{\aconstname}_{\vec{\arecvar}}}}{\aiteri{n+1}} $}
    \RightLabel{\sFIXletrec}
    \BinaryInfC{$ \flabsann{\vec{\avar}}
                           {\,(\letrec{\arecvari{1} = \allteri{1} \ldots \arecvari{n} = \allteri{n}}{\allteri{n+1}})}
                           {\aiteri{n+1}} 
                 $}          
    \DisplayProof
         }   
  \\[2.5ex]
  \parbox{340pt}{\small\centering
      where
      $\aconstnamei{\arecvari{1}}, \ldots,\aconstnamei{\arecvari{n}}$ are distinct constants fresh for $\allteri{1},\ldots,\allteri{n+1}$, 
      and 
      substitutions $\subst{\allteri{l}}{\vec{\arecvar}}{\vec{\aconstname}_{\vec{\arecvar}}}$
      stands short for 
      $\allteri{l} [\arecvari{1} \defdby \aconstnamei{\arecvari{1}}, \ldots, \arecvari{n} \defdby \aconstnamei{\arecvari{n}}]$.
      \\[0.75ex]
      \emph{side-conditions}: 
        $\length{\vec{y}} \ge \length{\vec{x}}$ holds for the prefix length of
        every $\flabs{\vec{y}}{\bter}$ on a thread\\ in $\Derivi{j}$ for 
                                                                          $1\le j\le n+1$
        from an open assumptions $( \flabs{\vec{\avar}}{\aconstnamei{\arecvari{i}}} )^{\amarkeri{i}}$ downwards;\\
        for bottommost instances: 
        the arising derivation is guarded on access path cycles.
      } 
  \mbox{}
\end{center}
\end{minipage}
            }
\end{center} 
  \vspace*{-1.25ex}
  \caption{\label{fig:annstRegletrec}
           The proof system \protect\annstRegletrec\ for \lambdaletrecterms\
           arises from the proof system \protect\annstRegzero\ (see
           Figure~\ref{fig:annstRegzero}) by replacing the rule \FIX\ with the
           rule \FIXletrec.}
\end{figure}

\begin{definition}[the proof system $\annstRegletrec$]
  The proof system $\annstRegletrec$ arises from $\stRegzero$ by replacing
  the rule \FIX\ by the rule \FIXletrec\ in Figure~\ref{fig:annstRegletrec},
  an annotated version of the rule \FIXletrec\ from
  Definition~\ref{def:Regletrec:stRegletrec} and
  Figure~\ref{fig:Regletrec:stRegletrec}. The side-condition on bottommost
  instances of \FIXletrec\ to be guarded on access path cycles is analogous as
  explained in Definition~\ref{def:Regletrec:stRegletrec}.
\end{definition}

\begin{proposition}[from $\annstRegletrec$- to $\stRegletrec$-derivations]\label{prop:annstRegletrec:2:stRegletrec}    
  Let $\Deriv$ be a closed derivation in \annstRegletrec\ with conclusion $\femptylabsann{\allter}{\aiter}$.
  Then a closed derivation $\check{\Deriv}$ in \stRegletrec\ with conclusion $\femptylabs{\allter}$
  can be obtained by removing the infinite \lambdaterms\ in $\Deriv$ while keeping the \lambdaletrecterm-annotations.
\end{proposition}

\begin{proposition}[from $\annstRegzero$- to $\annstRegletrec$-derivations]%
  \label{prop:annstRegzero:2:annstRegletrec}
  Every derivation $\Deriv$ in $\annstRegzero$ can be transformed into a derivation $\Deriv'$ in $\annstRegletrec$
  with the same conclusion and with the same open assumption classes.
\end{proposition}

\begin{proof}
  First note that the 
                      $\annstRegzero$ and $\annstRegletrec$ differ only by the specific version of 
  assumption-discharging rule in the system, \FIX\ in $\annstRegzero$ and \FIXletrec\ in $\annstRegletrec$. 
  For showing the proposition,
  let $\Deriv$ be a derivation in $\annstRegzero$. 
  
  We define a prooftree $\Deriv'$,
  (intended 
            to be a derivation in $\annstRegletrec$)
  by repeatedly replacing topmost occurrences of \FIX\ at the bottom of subderivations 
  of the form as depicted in Figure~\ref{fig:annstRegzero},
  by simulating subderivations of the form: 
  \begin{equation*}
    \begin{aligned}[c]
      \mbox{  
        \AxiomC{$ [\flabsann{\vec{\avar}}{\aconstnamei{\aannvar}}{\ater}]^{\aannvar} $}
        \noLine
        \UnaryInfC{$\Derivi{0}$}
        \noLine
        \UnaryInfC{$ \flabsann{\vec{\avar}}{\subst{\allter}{\aannvar}{\aconstnamei{\aannvar}}}{\ater} $}
        \AxiomC{$ (\flabsann{\vec{\avar}}{\aconstnamei{\aannvar}}{\ater})^{\aannvar} $}
        \RightLabel{$\sFIXletrec,u$}                   
        \BinaryInfC{$ \flabsann{\vec{\avar}}{(\letrec{\aannvar = \allter}{\aannvar})}{\ater} $}
        \DisplayProof
        }
    \end{aligned}
  \end{equation*}  
  until all occurrences of instances of \FIX\ have been replaced by instances of \FIXletrec. 
  The result is a prooftree with axioms and rules of $\addrule{\annstRegletrec}{\FIXletrecmin}$,
  with the same conclusion and the same classes of open assumptions as $\Deriv$, but
  in which rule instances carrying the label $\FIXletrec$ might actually be instances of $\FIXletrecmin$,
  unless actually proven (as will be done below) to be instances of $\FIXletrec$. 
  
  Now first note that, due to the form of the introduced instances of $\FIXletrec$, every formula occurrence
  in $\Deriv$ is reachable on an access path of $\Deriv'$. 
  Second, note that relative access paths $\apath'$ in $\Deriv'$ starting at the conclusion of an instance $\ainst'$ 
  of $\FIXletrec$ up to a marked assumption that is discharged at $\ainst'$
  descend from a thread $\apath$ in $\Deriv$ from the conclusion of an application $\ainst$ of $\FIX$ up to a marked
  assumption that is discharged at $\ainst$. Since by Proposition~\ref{prop:annstRegzero:sidecondition} the thread
  $\apath'$ passes at least one instance of a rule ($\labscomp$) or ($\lappcomp$), this is also the case for $\apath$.
  As a consequence, all cycles on relative access paths are guarded. Thus $\Deriv$ is guarded.
  Hence all occurrences of rule names $\FIXletrec$ in $\Deriv'$ rightly label occurrences of this rule,
  and $\Deriv'$ is a derivation in $\annstRegletrec$, which moreover is guarded.  
%
\end{proof}

\begin{example}\label{example:annstRegletrec}
  The closed derivation $\Hat{\Deriv}_{l}$ in Example~\ref{example:annstRegzero} 
  can be transformed into the following closed derivation in $\annstRegletrec\,$:
  \begin{center}
    \AxiomC{$ (\femptylabsann{\aconstnamei{\aannvar}}{\aiter})^{\amarker} $}
    \RightLabel{$\Vacstreg$}
    \UnaryInfC{$ \flabsann{\avar}{\aconstnamei{\aannvar}}{\ater} $}
    \AxiomC{\mbox{}}
    \RightLabel{$\bvarax$}
    \UnaryInfC{$ \flabsann{\avar}{\avar}{\avar} $}
    \RightLabel{$\lappcomp$}
    \BinaryInfC{$ \flabsann{\avar}{\lapp{\aconstnamei{\aannvar}}{\avar}}{\lapp{\ater}{\avar}} $}
    \RightLabel{$\Vacstreg$}
    \UnaryInfC{$ \flabsann{\avar\bvar}{\lapp{\aconstnamei{\aannvar}}{\avar}}{\lapp{\ater}{\avar}} $}
    \AxiomC{\mbox{}}
    \RightLabel{$\bvarax$}
    \UnaryInfC{$ \flabsann{\avar\bvar}{\bvar}{\bvar} $}
    \RightLabel{$\lappcomp$}
    \BinaryInfC{$ \flabsann{\avar\bvar}{\lapp{\lapp{\aconstnamei{\aannvar}}{\avar}}{\bvar}}{\lapp{\lapp{\ater}{\avar}}{\bvar}} $}
    \RightLabel{$\labscomp$}
    \UnaryInfC{$ \flabsann{\avar}{\labs{\bvar}{\lapp{\lapp{\aconstnamei{\aannvar}}{\avar}}{\bvar}}}{\lapp{\lapp{\ater}{\avar}}{\bvar}} $}
    \RightLabel{$\labscomp$}
    \UnaryInfC{$ \femptylabsann{\labs{\avar\bvar}{\lapp{\lapp{\aconstnamei{\aannvar}}{\avar}}{\bvar}}}{\labs{\avar\bvar}{\lapp{\lapp{\ater}{\avar}}{\bvar}}} $}
    \AxiomC{$ (\femptylabsann{\aconstnamei{\amarker}}{\aiter})^{\amarker} $}
    \RightLabel{\sFIXletrec, $u$}
    \BinaryInfC{$ \femptylabsann{(\letrec{\aannvar = \labs{\avar\bvar}{\lapp{\lapp{\aannvar}{\avar}}{\bvar}}}{\aannvar})}{\ater} $}
    \DisplayProof
  \end{center}
\end{example}

Now we concentrate on the remaining matter of proving that the \lambdaletrecterm\
obtained by the annotation process from a closed derivation in $\stRegzero$ to one in $\annstRegzero$
does indeed unfold to the infinite \lambdaterm\ it annotates. 
For this, we establish a proof-theoretic transformation from derivations in $\annstRegzero$
to derivations in $\stRegeq$.

\begin{lemma}[from $\annstRegzero$- to $\stRegeq$\nb-derivations]%
  \label{lem:annstRegzero:2:stRegeq}
  Let $\Deriv$ be a closed derivation in $\annstRegzero$ with conclusion $\femptylabsann{\allter}{\aiter}$.
  Then $\defd{\Unf{\allter}}$, and $\Deriv$ can be transformed into a closed derivation $\Deriv'$ in $\stRegeq$ 
  with conclusion $\femptylabs{\Unf{\allter}} = \femptylabs{\aiter}$ by:
  \begin{itemize}
    \item replacing each formula occurrence $\aocc$ of $\flabsann{\vec{\bvar}}{\bllter}{\biter}$ in $\Deriv$
      by an occurrence of the formula
      $ \Unf{\flabs{\vec{\bvar}}{\letrec{\abindgroup}{\tilde{\bllter}}}}
        =
       \flabs{\vec{\bvar}}{\biter}$ in $\Deriv'$,
      where 
      $\abindgroup$ 
      arises as the union of all outermost binding groups in conclusions of instances of \FIX\ at or below $\aocc$,
      and where
      $\flabs{\vec{\bvar}}{\bllter} = \flabs{\vec{\bvar}}{\subst{\tilde{\bllter}}{\vec{\arecvar}}{\vecsub{\aconstname}{\vec{\arecvar}}}}$ 
      and $\vec{\arecvar}$ is comprised of the recursion variables occurring in $\abindgroup$ and $\vecsub{\aconstname}{\vec{\arecvar}}$
      are distinct constants for $\vec{\arecvar}$ as chosen by~$\Deriv$;
      the unfoldings involved here are always defined.  
  \end{itemize}
\end{lemma}

\begin{proof}
  Let $\Deriv$ be a closed derivation in $\annstRegzero$ with conclusion $\femptylabsann{\allter}{\aiter}$.
  
  By Proposition~\ref{prop:annstRegzero:2:annstRegletrec}, $\Deriv$
  can be transformed into a closed derivation $\Deriv_1$ in $\annstRegletrec$
  with the same conclusion. 
  Due to Proposition~\ref{prop:annstRegletrec:2:stRegletrec},
  by dropping the infinite terms in $\Deriv_1$, a derivation $\Deriv_2$ in $\stRegletrec$
  with conclusion $\femptylabs{\allter}$ can be obtained.  
  Then it follows from Theorem~\ref{thm:Regletrec:stRegletrec} 
  that $\defd{\Unf{\allter}}$, that is, that $\allter$ unfolds to an infinite \lambdaterm. 
  
  We have to show that the transformation of $\Deriv$ into $\Deriv'$
  as described in the statement of the lemma
  is, on the one hand, possible (that is, the unfolding of each prefixed \lambdaletrec\nb-term is indeed defined),
  and on the other hand, that the prooftree $\Deriv'$ obtained by these replacements
  is indeed a valid derivation in $\stRegeq$.   
  
  We argue for the possibility of these replacements and for their correctness 
  locally, that is by carrying out the replacements from the bottom of $\Deriv$ upwards,
  thereby recognising for every replacement step that it is possible, and that it
  indeed produces a valid inference in $\stRegzero$. 
  
  As a typical example of the arguments necessary to establish this fact, we consider
  a derivation $\Deriv$ in $\annstRegzero$ with in it an instance of ($\labscomp$)
  that immediately succeeds an instance of \FIX: 
  \begin{equation*}
    \mbox{
      \AxiomC{$ [\flabsann{\vec{\avar}\bvar}{\aconstnamei{\amarker}}{\aiteri{0}}]^{\amarker} $} 
      \noLine
      \UnaryInfC{$ \Derivi{0} $}
      \noLine
      \UnaryInfC{$ \flabsann{\vec{\avar}\bvar}{\subst{\allteri{0}}{\amarker}{\aconstnamei{\amarker}}}{\aiteri{0}} $}
      \RightLabel{$\sFIX, \amarker$}
      \UnaryInfC{$ \flabsann{\vec{\avar}\bvar}{\letrec{\amarker = \allteri{0}}{\amarker}}{\aiteri{0}} $}
      \RightLabel{$\labscomp$}
      \UnaryInfC{$ \flabsann{\vec{\avar}}{(\labs{\bvar}{\letrec{\amarker = \allteri{0}}{\amarker}})}{\labs{\bvar}{\aiteri{0}}} $} 
      \noLine
      \UnaryInfC{$ \Derivi{00}' $}
      \noLine
      \UnaryInfC{$ \femptylabsann{\allter}{\aiter} $}
      \DisplayProof
          }
  \end{equation*}
  According to the statement of the lemma, $\Deriv$ is transformed into the following
  $\stRegeq$\nb-prooftree:
  \begin{equation*}
    \mbox{
      \AxiomC{$ \ldots 
                (\Unf{\flabs{\vec{\avar}\bvar}{\letrec{\abindgroupi{0},\,\amarker = \alltertildei{0},\,\abindgroup'}{\amarker}}} 
                   = 
                 \flabs{\vec{\avar}\bvar}{\aiteri{0}})^{\amarker} \ldots $} 
      \noLine
      \UnaryInfC{$ \Derivi{0}' $}
      \noLine
      \UnaryInfC{$ \Unf{\flabs{\vec{\avar}\bvar}{\letrec{\abindgroupi{0},\,\amarker = \alltertildei{0}}{\alltertildei{0}}}}
                     = 
                   \flabs{\vec{\avar}\bvar}{\aiteri{0}} $}
      \RightLabel{$\sFIX, \amarker$}
      \UnaryInfC{$ \Unf{\flabs{\vec{\avar}\bvar}{\letrec{\abindgroupi{0},\, \amarker = \alltertildei{0}}{\amarker}}} = \flabs{\vec{\avar}\bvar}{\aiteri{0}} $}
      \RightLabel{$\labscomp$}
      \UnaryInfC{$ \Unf{\flabs{\vec{\avar}}{\letrec{\abindgroupi{0}}{\labs{\bvar}{\letrec{\amarker = \alltertildei{0}}{\amarker}}}}}
                     = 
                   \flabs{\vec{\avar}}{\labs{\bvar}{\aiteri{0}}} $} 
      \noLine
      \UnaryInfC{$ \Derivi{00}' $}
      \noLine
      \UnaryInfC{$ \Unf{\femptylabs{\letrec{}\allter}} = \femptylabs{\aiter} $}
      \DisplayProof
          }
  \end{equation*} 
  where $\abindgroupi{0}$ 
  arises as the union of all outermost binding groups in conclusions of instances of \FIX\ 
  strictly below the visible instance of \FIX,
  and $\abindgroup'$ is the union of all outermost binding groups in conclusions of instances of \FIX\
  strictly above the visible instance of \FIX\ and below the indicated marked assumptions
  (this binding group differs for different marked assumptions of this assumption class),
  and where $\alltertildei{0}$ is the result of replacing in $\allteri{0}$ all occurrences of
  constants $\aconstnamei{\arecvar}$ by the recursion variable $\arecvar$ from which it originates.
  
  Now assuming that the unfolding in the conclusion of the visible
  instance of ($\lappcomp$) has been shown to exist, we want to recognise that this instance and the instance of \FIX\ above
  are valid instances in $\stRegzero$.
  For the instance of ($\labscomp$) we have to show:
  \begin{multline*}
    \defd{\Unf{\flabs{\vec{\avar}}{\letrec{\abindgroupi{0}}{\labs{\bvar}{\letrec{\amarker = \alltertildei{0}}{\amarker}}}}}}
    \\
    \Longrightarrow\;\;
      \text{for a term $\labs{\vec{\avar}\bvar}{\biteri{0}}$: } \;\,
      \begin{aligned}[t]
        &
        \defd{\Unf{\flabs{\vec{\avar}\bvar}{\letrec{\abindgroupi{0},\,\amarker = \alltertildei{0}}{\alltertildei{0}}}}}
          =
        \flabs{\vec{\avar}\bvar}{\biteri{0}}
        \\
          & \logand
        \Unf{\flabs{\vec{\avar}}{\letrec{\abindgroupi{0}}{\labs{\bvar}{\letrec{\amarker = \alltertildei{0}}{\amarker}}}}}  
          = 
        \flabs{\vec{\avar}}{\labs{\bvar}{\biteri{0}}}
      \end{aligned}  
  \end{multline*}
  This, however, is an easy consequence of the following $\unfoldred$\nb-rewrite steps:
  \begin{align*}
    \flabs{\vec{\avar}}{\letrec{\abindgroupi{0}}{\labs{\bvar}{\letrec{\amarker = \alltertildei{0}}{\amarker}}}}
      \;\unflabsred\; {} &
    \flabs{\vec{\avar}}{\labs{\bvar}{\letrec{\abindgroupi{0}}{\letrec{\amarker = \alltertildei{0}}{\amarker}}}}
    \\
      \;\unfletrecred\; {} & 
    \flabs{\vec{\avar}}{\labs{\bvar}{\letrec{\abindgroupi{0},\amarker = \alltertildei{0}}{\amarker}}}
  \end{align*}
  in view of the fact that, by Lemma~\ref{lem:unique:unfolding},
  unfolding is unique normalising (in at most $\omega$ steps).\
  And for the instance of \FIX\ we have to show:
  \begin{multline*}
    \defd{\underbrace{\Unf{\flabs{\vec{\avar}\bvar}{\letrec{\abindgroupi{0},\, \amarker = \alltertildei{0}}{\amarker}}}}_{= (\star)}}
    \Longrightarrow
      \begin{aligned}[t]
        &
        \defd{\Unf{\flabs{\vec{\avar}\bvar}{\letrec{\abindgroupi{0},\,\amarker = \alltertildei{0}}{\alltertildei{0}}}}}
           = (\star)
        \\
        &  \slogand \;\:
        \defd{\Unf{\flabs{\vec{\avar}\bvar}{\letrec{\abindgroupi{0},\,\amarker = \alltertildei{0},\,\abindgroup'}{\amarker}}}} 
           = (\star)
      \end{aligned}
  \end{multline*}  
  (actually the statement as in the second line has to be shown for every binding-group $\abindgroup'$
   that occurs for marked assumptions discharged at the instance of \FIX). 
  This implication is a consequence of the $\unfoldred$\nb-rewrite steps:
  \begin{gather*}
    \flabs{\vec{\avar}\bvar}{\letrec{\abindgroupi{0},\, \amarker = \alltertildei{0}}{\amarker}}
      \unfrecred
    \flabs{\vec{\avar}\bvar}{\letrec{\abindgroupi{0},\, \amarker = \alltertildei{0}}{\alltertildei{0}}}  
    \\
    \flabs{\vec{\avar}\bvar}{\letrec{\abindgroupi{0},\,\amarker = \alltertildei{0},\,\abindgroup'}{\amarker}}
      \unfreducered
    \flabs{\vec{\avar}\bvar}{\letrec{\abindgroupi{0},\, \amarker = \alltertildei{0}}{\amarker}}  
  \end{gather*}
  again in view of the statement of Lemma~\ref{lem:unique:unfolding}.
  
  The arguments used here are typical, and can be carried out similarly also for
  showing that axioms ($\bvarax$), and instances of rules ($\lappcomp$) and ($\Vacstreg$) in $\annstRegzero$\nb-derivations
  give rise to, under the transformation described in the statement of the lemma,
  valid instances of axioms ($\bvarax$), and instances of ($\lappcomp$) and ($\Vacstreg$), respectively, in $\stRegeq$\nb-derivations. 
\end{proof}

\begin{theorem}\label{thm:annstRegzero}
  If $\,\derivablein{\annstRegzero}{\femptylabsann{\allter}{\aiter}}$
    holds for a \lambdaletrec\nb-term $\allter$ and an infinite \lambdaterm~$\aiter$,
  then $\allter$ unfolds to, and hence expresses, $\aiter$. 
\end{theorem}

\begin{proof}
  Suppose that $\Deriv$ is a closed derivation in $\annstRegzero$ with conclusion $\femptylabsann{\allter}{\aiter}$.
  Lemma~\ref{lem:annstRegzero:2:stRegeq} entails  
  that $\allter$ unfolds to an infinite \lambdaterm, and moreover,
  that $\Deriv$ can be transformed into a closed derivation $\Deriv'$ in $\stRegeq$
  with conclusion $\femptylabsequat{\Unf{\allter}}{\aiter}$. 
  Then it follows by Theorem~\ref{thm:stRegeq} (applying soundness of $\stRegzero$ with respect
  to the property of \lambdaletrecterms\ to unfold to an infinite \lambdaterm\ that $\Unf{\allter} = \aiter$,
  and hence that $\allter \unfoldomegared \aiter$. 
  In this way we have found a \lambdaletrecterm~$\allter$ that expresses $\aiter$.
\end{proof}

We now arrive at our main characterisation result.

\begin{theorem}\label{thm:ll-expressible:streg}
  An infinite \lambdaterm\ is $\lambdaletrec$\nb-expressible 
    if and only if 
  it is strongly regular.
\end{theorem}

\begin{proof}
  Let $\aiter$ be an infinite \lambdaterm. 
  The direction ``$\Rightarrow$'' is the statement of Theorem~\ref{thm:ll-expressible:2:streg}. 
  %
  %
  For showing the direction ``$\Leftarrow$'' in the statement of the theorem,
  we assume that $\aiter$ is strongly regular. 
  
  Then by Lemma~\ref{lem:Reg:stReg}, (\ref{lem:Reg:stReg:item:ii}),
  there exists a closed derivation $\Deriv$ in $\stReg$ with conclusion $\femptylabs{\aiter}$.
  Due to Lemma~\ref{lem:stRegzero:2:annstRegzero}, (\ref{lem:stRegzero:2:annstRegzero:item:i}),
  $\Deriv$ can be transformed into a derivation $\Derivann$ in $\annstRegzero$ with conclusion $\femptylabsann{\allter}{\aiter}$,
  for some \lambdaletrecterm~$\allter$.
  Then it follows by Theorem~\ref{thm:annstRegzero} that the \lambdaletrecterm~$\allter$ expresses $\aiter$.
\end{proof}

As an immediate consequence of Theorem~\ref{thm:ll-expressible:streg} 
and of 
       Corollary~\ref{cor:thm:streg:fin:bind:capt:chains}
we obtain the following theorem, a summary of our main results. 

\begin{theorem}
  For all infinite \lambdaterms~$\aiter$ the following statements are equivalent:
  \begin{enumerate}[(i)]
    \item $\aiter$ is $\lambdaletreccal$\nb-expressible. 
    \item $\aiter$ is strongly regular.
    \item $\aiter$ is regular, and it only contains finite \bindcaptchains.
  \end{enumerate}
\end{theorem}

%
\section{\lambda-transition graphs}
\label{sec:ltgs}
In this section we introduce the concept of \lambda-transition graphs.
A \lambda-transition graph $\altg$ of a term $\ater$ can be understood as a
nameless graphical representation closely related to the term graph of $\ater$
in de-Bruijn notation. It is a graph that corresponds to the sub-ARS that is
induced by $\ater$ with respect to some \extscope-delimiting strategy for
$\stRegCRS$, but where no information can be extracted from the
objects. Consider, for example the sub-ARSs displayed in
Figure~\ref{fig:entangled-cons-ltg} and Figure~\ref{fig:distance} but ignore
the prefixes by which the nodes are annotated. To capture the notion of
`forgetting' the term associated with each object of the ARS we use the
formalism of labelled transition systems, in which only transitions are
observable (see Section~\ref{sec:prelims}).

We will show a coinduction principle for infinite \lambdaterms: two
\lambdaterms\ are equal if and only if they have bisimilar \lambdatg{s}.

\begin{definition}[transition systems induced by \CRS{s}]
Let $\aARS = \tuple{\objects,\steps,\ssrc,\stgt}$ 
be a sub\nb-\ARS, or a sub\nb-\ARS\ of a labelled version, of
an \ARS\ that is induced by a \CRS~$\aCRS$ with rules $\rules$ 
(see \cite[\mbox{11.2.24}]{terese:2003} for a definition of induced \ARSs).
In particular, every step in $\steps$ carries information according to from which rule of $\aCRS$ it stems from. 
By the \emph{\LTS\ induced by $\aARS$} we mean the \LTS~$\alts = \triple{\objects}{\rules}{\transitions}$
with transitions
\begin{equation*}
  {\transitions} 
    \defdby 
  \descsetexp{\triple{\aobj}{\arulename}{\aobj'}}
             {(\exists \astep\in\steps) \, \text{$\astep \funin \aobj \to \aobj'$ a step that stems from rule $\arulename$}}
\end{equation*}
in which the steps in $\aARS$ according to rule $\arulename$ are interpreted as transitions with label~$\arulename$.
And for a subset $\rules_0$ of $\rules$,
by \emph{the \LTS\ induced by $\aARS$ with silent $\rules_0$\nb-steps} 
we mean the \LTS~$\SilentLTS{\aARS}{\rules_0} = \triple{\objects}{\rules}{\transitions'}$
with
\begin{equation*}
  {\transitions'} 
    \defdby 
  \descsetexpBig{\triple{\aobj}{\arulename}{\aobj'}}
             {(\exists \astep\in\steps) \,
                \parbox[c]{208pt}
                  {$ \astep \funin \aobj \binrelcomp{\smredb{\rules_0}}{\sredb{\arulename}} \aobj'$
                   where $\smredb{\rules_0}$ are steps w.r.t.\ rules in $\rules_0$,
                   and $\sredb{\arulename}$ is a step w.r.t\ rule $\arulename\in\rules\setminus\rules_0$}}
\end{equation*}
in which the steps in $\aARS$ according to rules in $\rules_0$ are interpreted as silent transitions,
and the remaining rules as transitions according to their name.
\end{definition}

\begin{definition}[transition graph of an object]
If $\aobj$ is an object of the \ARS{} $\aARS$ that is induced by a \CRS~$\aCRS$
with rules $\rules$, and $\triple{\objects}{\rules}{\transitions} =
\lts{\InducedSubARS{\aobj}}$ the \LTS{} induced by $\InducedSubARS{\aobj}$, then we call
$\ltg\aARS\aobj \defdby \quadruple{\states}{\labels}{\aobj}{\transitions}$ the
transition graph of $\aobj$.

For an \LTS{} $\SilentLTS{\InducedSubARS{\aobj}}{\rules_0} =
\triple{\objects}{\rules}{\transitions}$ with silent $\rules_0$-steps, we call
$\SilentLTG{\aARS}{\rules_0}{\aobj} \defdby
\quadruple{\states}{\labels}{\aobj}{\transitions}$ the \emph{transition graph
of $\aobj$ with silent $\rules_0$-steps}.
\end{definition}

\begin{definition}[\lambda-transition graph]\label{def:ltg}
We call a labelled transition graph $\altg =
\quadruple{\states}{\labels}{\ainitialstate}{\transitions}$ a
\emph{\lambda-transition graph} if:
\begin{itemize}
\item it is connected
\item the labels are $\labels = \{\slabsdecomp,\snlvarsucc,\slappdecompi0,\slappdecompi1\}$
\item there are no infinite paths in $\altg$ consisting solely of $\snlvarsucc$-transitions
\item every state belongs to one of the following kinds: $\slabsdecomp$-states,
      $\snlvarsucc$-states, and $\slappdecomp$-states, where
      \begin{itemize}
      \item a $\slabsdecomp$-state $\astate$ is the source of precisely one $\slabsdecomp$-transition,
            and no other transitions:
            $\{\pair\alabel\bstate ~|~ \triple\astate\alabel\bstate \in
            {\transitions}\} = \{\pair\slabsdecomp\bstate\}$ for some $\bstate\in\states$.
      \item a $\snlvarsucc$-state $\astate$ is the source of precisely one $\snlvarsucc$-transition,
            and no other transitions:
            $\{\pair\alabel\bstate ~|~ \triple\astate\alabel\bstate \in
            {\transitions}\} = \{\pair\snlvarsucc\bstate\}$ for some $\bstate\in\states$.
      \item a $\slappdecomp$-state $\astate$ is the source of precisely one
            $\slappdecompi0$-transition and one $\slappdecompi1$-transition,
            but no other transitions:
            $\{\pair\alabel\bstate ~|~ \triple\astate\alabel\bstate \in
            {\transitions}\} = \{\pair{\slappdecompi0}\bstate,\pair{\slappdecompi1}\cstate\}$
            for some $\bstate,\cstate\in\states$.
      \end{itemize}
\end{itemize}
\end{definition}

\begin{proposition}\label{prop:s-is-a-lambda-ltg}
Let $\astratplus$ be a \extscope-delimiting strategy of $\stRegARS$.
For every term $\ater\in\Ter\inflambdaprefixcal$ the transition graph
$\ltg\astratplus\ater$ of $\ater$ is a \lambda-transition graph.
\end{proposition}

\begin{proof}
  In transition graphs $\ltg\astratplus\ater$, infinitely many successive $\snlvarsucc$-transitions
  are not possible because in the \ARS\ that induces $\ltg{\astratplus}{\ater}$,
  the rewrite relation $\scompressstreg$ is terminating, due to Proposition~\ref{prop:rewprops:RegCRS:stRegCRS}, (\ref{prop:rewprops:RegCRS:stRegCRS:item:iii}). 
\end{proof}

Along the lines of Proposition~\ref{prop:s-is-a-lambda-ltg} we can also view
transition graphs of \lambdaletrec-terms as \lambda-transition graphs, but only
when treating unfolding steps as silent transitions.
As hinted before in Remark~\ref{rem:nondet:unfolding},
here the restriction of \scopedelimiting\ (and \extscopedelimiting) strategies
to ones that prevent non-determinism in the application of unfolding rules is relevant.

\begin{proposition}
Let $\astratplus$ be a \extscope-delimiting strategy of $\stRegletrecARS$. For every
term $\allter\in\Ter\lambdaletrecprefixcal$, the transition graph
$\SilentLTG\astratplus\llunfCRS\allter$ of $\allter$ is a \lambda-transition graph.
\end{proposition}

\begin{definition}[\lambda-transition graph of a term]\label{def:ltg-of-a-term}
  \begin{enumerate}[(i)]
    \item 
      Let $\aiter\in\Ter\inflambdacal$. For a \extscope-delimiting strategy $\astratplus$ of $\stRegletrecARS$
      we call the transition graph $\ltg\astratplus\ater$ 
      \emph{the \lambda-transition graph of $\allter$ with respect to $\astratplus$}.
      And more generally, by a \emph{\lTG\ of} $\aiter$ we mean
      a transition graph that is bisimilar to the transition graph of $\aiter$ with respect to
      a \extscopedelimiting~strategy~$\astratplus$.
    \item 
      Let $\allter\in\Ter\lambdaletrecprefixcal$ be a (prefixed) $\astratplus$-productive \lambdaletrec-term.
      For a \extscope-delimiting strategy $\astratplus$ of $\stRegletrecARS$
      such that $\allter$ is $\astratplus$-productive,
      we call the transition graph $\SilentLTG\astratplus\llunfCRS\allter$
      \emph{the \lambda-transition graph of $\allter$ with respect to $\astratplus$}.
      And more generally, by a \emph{\lTG\ of} $\allter$ we mean
      a transition graph that is bisimilar to the transition graph of $\allter$ with respect to
      a \extscopedelimiting~strategy~$\astratplus$ with the property that $\allter$ is $\astratplus$\nb-productive. 
  \end{enumerate}
  For prefixed \lambda-terms in $\Ter\inflambdaprefixcal$ and in $\Ter\lambdaletreccal$ we use the terms
  `\lambda-transition graph' and `transition graph' synonymously.
  We also speak of \lambda-transition graphs of terms
  $\allter\in\Ter\lambdaletreccal$ or $\ater\in\Ter\inflambdacal$ by which we
  refer to the \lambda-transition graphs of $\femptylabs\allter$ and
  $\femptylabs\ater$, respectively.
\end{definition}

\begin{theorem}[coinduction principle for \inflambdacal]\label{thm:coinduction-principle}
  For all infinite \lambdaterms\ $\aiter$ and $\biter$ the following statements are equivalent:
  \begin{enumerate}[(i)]
    \item{}\label{thm:coinduction-principle:item:i}
      $\aiter = \biter$.
    \item{}\label{thm:coinduction-principle:item:ii}
      $\derivablein{\EqTer}{\ater = \biter}$.
    \item{}\label{thm:coinduction-principle:item:iii} 
      $\aiter$ and $\biter$ have bisimilar \lambdatg{s}.
  \end{enumerate}
\end{theorem}

\begin{proof}
  In view of Proposition~\ref{prop:AlphaPreTer:EqTer}, (\ref{prop:AlphaPreTer:EqTer:item:ii}), the logical equivalence 
  between (\ref{prop:AlphaPreTer:EqTer:item:i}) and (\ref{prop:AlphaPreTer:EqTer:item:ii}), 
  and the fact that $\text{(\ref{thm:coinduction-principle:item:i})} \Rightarrow \text{(\ref{thm:coinduction-principle:item:iii})}$ clearly holds,
  it suffices to show that $\text{(\ref{thm:coinduction-principle:item:iii})} \Rightarrow \text{(\ref{thm:coinduction-principle:item:ii})}$ holds.
  
  For this, suppose that $\ltg{\astratplus_1}{\aiter}$ and $\ltg{\astratplus_2}{\biter}$
  are bisimilar for some \extscopedelimiting~stra\-te\-gies $\astratplus_1$ and $\astratplus_2$ for $\RegARS$.
  But now bisimilarity of these transition graphs guarantees that 
  a derivation $\Deriv$ in $\EqTer$ with conclusion $\femptylabs{\aiter} = \femptylabs{\biter}$ 
  can be constructed such that all threads in $\Deriv$ correspond to 
  $\stratred{\astratplus_1}$\nb-rewrite sequences on $\aiter$ 
  and to  
  $\stratred{\astratplus_2}$\nb-rewrite sequences on $\aiter$,
  respectively. If the construction process is organised in a depth-fair manner (for example, all non-axiom leafs at depth~$n$
  are extended by appropriate rule instances, before extensions are carried out at depth greater than $n$),
  then in the limit a completed derivation $\infDeriv$ with conclusion $\femptylabs{\aiter} = \femptylabs{\biter}$ is obtained. 
  This establishes $\derivablein{\EqTer}{\ater = \biter}$.
\end{proof}

\begin{conjecture}[coinduction principle for \lambdaletreccal]
  For all $\allteri{1},\allteri{2}\in\Ter{\lambdaletreccal}$ it holds that
  $\allteri{1} = \allteri{2}$ if and only if $\allteri{1}$ and $\allteri{2}$ have bisimilar \lambdatg{s}.
\end{conjecture}

\begin{remark}[only $\stRegARS$ defines nameless representations]\label{rem:nameless-repr}
  For the concept of `\lambdatg\ of a term' 
  that is defined, in analogy with Definition~\ref{def:ltg-of-a-term},
  as transition graph of a term with respect to \scopedelimiting\ strategies
  (instead of with respect to \extscopedelimiting\ strategies),
  a similar coinduction principle does not hold.
   
  To realise this, consider the sub-ARS $\eagscdelstratreg$ of $\RegARS$,
  induced by the eager \scopedelimiting\ strategy for the term $\labs{\avar\bvar}{\lapp{\lapp\avar\avar}\bvar}$,
  in Figure~\ref{fig:four-strategies}. The corresponding \LTS\ can be obtained
  as the \LTS\ corresponding to a sub-\ARS\ of $\RegARS$ induced by an appropriate \scopedelimiting\ strategy 
  for each of the following four \lambdaterms:
\[
\labs{\avar\bvar}{\lapp{\lapp\avar\avar}\avar}
\hspace{1cm}
\labs{\avar\bvar}{\lapp{\lapp\avar\avar}\bvar}
\hspace{1cm}
\labs{\avar\bvar}{\lapp{\lapp\bvar\bvar}\avar}
\hspace{1cm}
\labs{\avar\bvar}{\lapp{\lapp\bvar\bvar}\bvar}
\]
  For both of the terms in the middle the eager \scopedelimiting~strategy can be chosen,
  but not for the terms on the left and on the right.
\end{remark}

The understanding of \lambda-transition graphs as nameless representations of
an infinite \lambda-terms implies that from a such a graph the corresponding
\lambda-term can be extracted. We define a function for this purpose by means
of a CRS which implements the assembly of a \lambda-term from the infinite
unfolding of a \lambda-transition graph. The function is closely related to the
$\stParseCRS$ in the sense that $\stParseCRS$ does both destruct and
reconstruct the scrutinised term while $\sreadback$ only implements the reconstruction.

\begin{definition}[readback]
\begin{align*}
\sreadback \funin \iTer{\setexp{\nlvarfo, \snllabsfo, \snllappfo,
\snlvarsuccfo }} & {} \rightharpoonup \Ter{\inflambdacal}\\
\bnllter & {} \mapsto \readback{\bnllter} \defdby ~ \parbox[t]{130pt}{infinite normal form of $\readwritezero{\bnllter}$ 
                                                                      w.r.t.\ the following CRS:}
\end{align*}
\begin{align*}
  \readwriten{n}{\sametavari{1},\ldots,\sametavari{n}}{\nllabsfo{\bnllteri{0}}}
    & {} \;\red\;
  \labsCRS{\avar}{\readwriten{n+1}{\sametavari{1},\ldots,\sametavari{n},\avar}{\bnllteri{0}}}  
  \\
  \readwriten{n}{\vec{\sametavar}}{\nllappfo{\bnllteri{0}}{\bnllteri{1}}}
    & {} \;\red\;
  \lappCRS{\labsCRS{\avar}{\readwriten{n}{\vec{\sametavar}}{\bnllteri{0}}}}
          {\labsCRS{\avar}{\readwriten{n}{\vec{\sametavar}}{\bnllteri{1}}}}
  \\
  \readwriten{n+1}{\vec{\sametavar},\avar}{\nlvarsuccfo{\bnllteri{0}}}
    & {} \;\red\;
  \readwriten{n}{\vec{\sametavar}}{\bnllteri{0}}
  \\
  \readwriten{n}{\sametavari{1},\ldots,\sametavari{n}}{\nlvarfo}
    & {} \;\red\;
  \sametavari{n}
\end{align*}
The function is partial because $\sreadwriten n$ is unproductive for infinite
$\snlvarsucc$-chains. That restriction comes forth accordingly in the definition
of \lambda-transition graphs (Definition~\ref{def:ltg}). The function is thus
complete on the subset of $\iTer{\setexp{\nlvarfo, \snllabsfo, \snllappfo,
\snlvarsuccfo }}$ that is obtained from unfolding a \lambda-transition graph.
\end{definition}

\section{Conclusion and Outlook}\label{sec:conclusion}
In this work we have introduced a number of formalisms for relating
infinite \lambdaterms\ and finite terms in the \lambdacalculus\ with \sletrec\
to each other. 
In the following we recapitulate the most important concepts briefly.

We provide CRS signatures to define the set of infinite \lambda-terms and the
set of \lambdaletrec-terms, which we connect by the CRS $\llunfCRS$ for
unfolding \lambdaletrec-terms to their corresponding \lambda-term. To determine
which \lambdaletrec-terms have an infinite unfolding we identify
\emph{productive} \lambdaletrec-terms.

To characterise the set of \lambda-terms for which there exists a corresponding
\lambdaletrec-term (such that the former can be obtained from the latter via
unfolding) we establish a framework of formalisms for `observing'
\lambda-terms coinductively. Firstly we introduce prefixed \lambda-terms that
enrich \lambda-terms by an abstraction prefix. On the prefixed terms we define
the CRS $\stRegARS$ in which a rewrite sequence corresponds to a
deconstruction of a term along one of its paths. In that sense a prefixed term
$\flabs{\vec\avar}\ater$ can be understood as a `suspended decomposition' which
has not advanced into subterm $\ater$ yet. Such a decomposition describes a
path through the term by observations of the form $\slabsdecompred$,
$\lappdecompired0$, $\lappdecompired1$, $\scompressregred$, where the latter
delimits the eager scope-delimiting strategy \extscope{} of an abstraction.

Since there is some freedom as to where \extscope{}-delimiters can be placed, we
define \extscope{}-delimiting strategies to formalise specific possible choices
eliminating  that freedom and thereby making 
the observations deterministic except for the forking into the left or the right
subterm of an application. By means of \extscope-delimiting strategies we can
formulate two important concepts: strong regularity and \lambda-transition
graphs.

The intuitive understanding of strong regularity is the property of a
infinite \lambda-term $\ater$ that from $\ater$ every `sufficiently eager' \extscope-delimiting
strategy can only generate a finite number of terms. We then show that
\lambdaletrec-expressibility coincides with strong regularity.

Every \extscope-delimiting strategy defines a \lambda-transition graph of a
term which can be viewed as a nameless graphical representation very similar to
its term graph in de-Bruijn notation with the difference that $S$-nodes are
not restricted to occur only near leafs but can be shared by variables.
\todo{We
proof the coinduction principle for \lambda-transition graphs that states that
two infinite \lambda-terms are equal if and only if their \lambda-transition
graphs w.r.t. to a \extscope-delimiting strategy are bisimilar.}
The eager \extscope-delimiting strategy yields finite \lambda-transition graphs
for strongly regular \lambda-terms.

We adapt the concepts of the CRS for observing terms, \extscope-delimiting
strategies, and \lambda-transition graphs and apply them to \lambdaletreccal{}
proving similar results as for \inflambdacal.

We provide a proof system that is sound and complete for the notion of strong
regularity and which admits finite proofs for strongly regular \lambda-terms. We
define an annotated version of the proof system which not unlike an
attribute-grammar definition implements the extraction of a \lambdaletrec-term
$\allter$ from a proof for term $\ater$ in that system, such that $\allter$
unfolds to the $\ater$. We show that every \extscope-delimiting strategy
induces a proof and that from a proof a corresponding history-aware strategy
can be deduced, which suggests a similar correspondence between
\lambda-transition graphs and proofs.

The following results are within reach but not worked out yet:
\begin{itemize}
\item {\it coinduction principle for \lambdaletreccal}:
      For all $\allter,\bllter\in\Ter{\lambdaletreccal}$ it holds that $\allter
      = \bllter$ if and only if $\allter$ and $\bllter$ have bisimilar \lambda-term-graph{s}.
\item a proof system for unfolding equivalence of \lambdaletrec-terms
\item a thorough coinductive treatment of \lambda-transition graphs and finality results
\item finite representations of regular \lambda-terms as higher-order recursive program schemes
  (cf.\ Example~\ref{ex:entangled-infinite-path})
  and their extractions from formalised proofs of regularity
\item characterisation of \lambdaletrec-expressible preterms of infinite \lambdaterms\
  as those that can be generated, up to $\alpha$\nb-equivalence, by first-order recursive program schemes   
\item a terminating readback function to extract \lambdaletrec-terms from transition graphs
\end{itemize}
\vspace{1ex}


We feel that in gathering these results we have gained a new perspective on the
\lambdacalculus\ with \sletrec\ and we find that the concepts and formalisms
introduced here have potential to be practically relevant for the
implementation of functional programming languages. In \cite{grab:roch:2013a}
we study various higher-order and first-order term-graph representations of
cyclic \lambdaterms. Their definitions draw heavily on the decomposition
rewrite systems in this paper. That is, every term in \lambdaletreccal\ can be
translated into a finite first-order `\lambda-term-graph' by applying the
eager scope-delimiting strategy
$\eagscdelstratstreg$ to the expressed strongly regular, infinite \lambdaterm.
Thereby vertices with the labels $\slabs$, $\lappcomp$, $\snlvarsucc$ are
created according to the kind of $\sstregred$\nb-step observed (plus variable
occurrence vertices with label $\snlvar$). The degree of sharing exhibited by
\lambda-term-graphs\ can be analysed with functional bisimulation. We identify
a class of first-order representations with eager application of scope closure
that faithfully preserves and reflects the sharing order on higher-order term
graphs. For practical applications this can be exploited in order to obtain:
\begin{itemize}
\item efficient decision of unfolding equivalence of \lambdaletrec-terms by means of the
      (nearly linear) DFA-equivalence algorithm of Hopcroft and Karp \cite{hopc:karp:1971}. 
\item a partial order for the amount of subterm sharing in a \lambdaletrec-term leading to
      \begin{itemize}
      \item a notion of maximal sharing for \lambdaletrec{}
      \item an efficient mechanism to derive the maximally shared form of a \lambdaletrec-term
            which generalises common subexpression elimination
      \end{itemize}
\end{itemize}

Another aspect is that functional programming languages based on the
\lambdacalculus\ with \stxtletrec\ restrict the set of (in the unfolding semantics) expressible terms 
to the strongly regular infinite \lambdaterms. But members of the superclass of regular terms
are also finitely expressible via sets of equations or \CRS\nb-rules. Therefore the question
arises whether finite representations of regular terms afford new opportunities
in compiling functional programming languages.

\todo{Link zur Implementierung}


\def\sortunder#1{}
\bibliography{main}

\begin{thebibliography}{}

\bibitem[Ariola and Blom, 1997]{ario:blom:1997}
Ariola, Z.~M. and Blom, S. (1997).
\newblock {Cyclic Lambda Calculi}.
\newblock In Abadi, M. and Ito, T., editors, {\em Proceedings of TACS'97},
  volume 1281 of {\em LNCS}, pages 77--106. Springer.

\bibitem[Ariola and Klop, 1997]{ario:klop:1997}
Ariola, Z.~M. and Klop, J.~W. (1997).
\newblock {Lambda Calculus with Explicit Recursion}.
\newblock {\em Information and Computation}, 139(2):154--233.

\bibitem[Blom, 2001]{blom:2001}
Blom, S. (2001).
\newblock {\em Term Graph Rewriting -- Syntax and Semantics}.
\newblock PhD thesis, Vrije Universiteit Amsterdam.

\bibitem[Courcelle, 1983]{cour:1983}
Courcelle, B. (1983).
\newblock Fundamental properties of infinite trees.
\newblock {\em Theoretical Computer Science}, 25(2):95--169.

\bibitem[Endrullis et~al., 2011]{endr:grab:klop:oost:2011}
Endrullis, J., Grabmayer, C., Klop, J.~W., and van Oostrom, V. (2011).
\newblock {On Equal $\mu$-Terms}.
\newblock In Bethke, I., Ponse, A., and Rodenburg, P.~H., editors, {\em
  Festschrift in Honour of Jan Bergstra}, Special Issue of TCS, 412 (28), pages
  3175--3202. Elsevier.

\bibitem[Grabmayer and Rochel, 2013]{grab:roch:2013a}
Grabmayer, C. and Rochel, J. (2013).
\newblock {Term Graph Representations for Cyclic Lambda Terms}.
\newblock In {\em Proc.\ of TERMGRAPH 2013}.
\newblock prelimary: {\tt http://rochel.info/ltgs.pdf}.

\bibitem[Hendriks and van Oostrom, 2003]{hend:oost:2003}
Hendriks, D. and van Oostrom, V. (2003).
\newblock {\reflectbox{$\lambda$}}.
\newblock In Baader, F., editor, {\em Proceedings CADE-19}, volume 2741 of {\em
  Lecture Notes in Artificial Intelligence}, pages 136--150. Springer--Verlag.

\bibitem[Hopcroft and Karp, 1971]{hopc:karp:1971}
Hopcroft, J. and Karp, R. (1971).
\newblock {A Linear Algorithm for Testing Equivalence of Finite Automata}.
\newblock Technical report, Cornell University.

\bibitem[Ketema and Simonsen, 2009]{kete:simo:2009}
Ketema, J. and Simonsen, J.~G. (2009).
\newblock Infinitary combinatory reduction systems: confluence.
\newblock {\em Logical Methods in Computer Science}, 5(4:3):1--29.

\bibitem[Ketema and Simonsen, 2010]{kete:simo:2010}
Ketema, J. and Simonsen, J.~G. (2010).
\newblock Infinitary combinatory reduction systems: Normalising reduction
  strategies.
\newblock {\em Logical Methods in Computer Science}, 6(1:7):1--35.

\bibitem[Ketema and Simonsen, 2011]{kete:simo:2011}
Ketema, J. and Simonsen, J.~G. (2011).
\newblock Infinitary combinatory reduction systems.
\newblock {\em Information and Computation}, 209(6):893 -- 926.

\bibitem[K\H{o}nig, 2001]{koen:2001}
K\H{o}nig, D. (2001).
\newblock {\em {Theorie der Endlichen und Unendlichen Graphen}}.
\newblock AMS Chelsea Publishing.

\bibitem[Klop, 1980]{klop:1980}
Klop, J.~W. (1980).
\newblock {\em {Combinatory Reduction Systems}}.
\newblock PhD thesis, Universiteit Utrecht.

\bibitem[Klop et~al., 1993]{klop:oost:raam:1993}
Klop, J.~W., van Oostrom, V., and van Raamsdonk, F. (1993).
\newblock Combinatory reduction systems: introduction and survey.
\newblock {\em Theoretical Computer Science}, 121(1-2):279 -- 308.

\bibitem[Melli\`es, 1996]{mell:96}
Melli\`es, P.-A. (1996).
\newblock {\em {Description Abstraite des Syst\`emes de R\'e\'ecriture}}.
\newblock PhD thesis, l'Universit\'{e} Paris 7.

\bibitem[Oostrom, 1997]{oost:97}
Oostrom, V.~v. (1997).
\newblock {FD} \`a la {M}elli\`es.
\newblock Vrije Universiteit Amsterdam.

\bibitem[Terese, 2003]{terese:2003}
Terese (2003).
\newblock {\em {Term Rewriting Systems}}, volume~55 of {\em Cambridge Tracts in
  Theoretical Computer Science}.
\newblock Cambridge University Press.

\bibitem[van Oostrom, 1994]{oost:1994}
van Oostrom, V. (1994).
\newblock {\em {Confluence for Abstract and Higher-Order Rewriting}}.
\newblock PhD thesis, Vrije Universiteit Amsterdam.

\bibitem[van Oostrom et~al., 2004]{oost:looi:zwit:2004}
van Oostrom, V., van~de Looij, K.-J., and Zwitserlood, M. (2004).
\newblock Lambdascope.
\newblock Extended Abstract for the Workshop on Algebra and Logic on
  Programming Systems (ALPS), Kyoto, April 10th 2004.

\end{thebibliography}

\appendix
\section{Proof: Confluence of \sletrec-unfolding}\label{app:conf_proof}
\raggedbottom
\begin{proof}[Proof of Proposition~\ref{prop:unf-confluence} (on page~\pageref{prop:unf-confluence}).]
First of all, we cannot use Newman's Lemma to prove the theorem, because
$\llunfCRS$ is not terminating. To show confluence of $\llunfCRS$ we use
the method of `decreasing diagrams' \cite[Sec.\hspace*{2pt}2.3]{oost:1994} \cite[Sec.\hspace*{2pt}14.2]{terese:2003}. 
We use it
however not to prove confluence of the rewriting relation $\sredp\sunf$ induced
by $\llunfCRS$ directly, but of the abstract reduction system
$\aARS = (\Ter\lambdaletreccal, \{\sparredp{\wDepth\arulename{d}} ~|~ (d,\arulename)
\in \nats \times R\})$ with $R$ as the set of rules of $\llunfCRS$ where
$\sparredp{\wDepth\arulename{d}}$ denotes the parallel rewriting relation on
$\Ter\lambdaletreccal$ induced by rule $\arulename$ at $\sletrec$-depth $d$.
As a precedence order we consider the order induced by the $\sletrec$-depth:
\[ \sparredp{\wDepth\arulename{d}} \geq \sparredp{\wDepth\brulename{d'}}
~ \Longleftrightarrow ~ d \geq d'\]

The \sletrec-depth of a redex in \lambdaletrec-term denotes the number of
$\sletrec$-nodes passed on the path from the root of the term tree to the
corresponding position. We write $\sredp{\wDepth\arulename{d}}$ to denote the
relation induced by applying rule $\arulename$ contracting a redex at
\sletrec-depth $d$.

Let us denote the rewriting relation induced by $\aARS$ by $\sredp\aARS$:
\[ \sredp\aARS = \bigcup\,\{\sparredp{\wDepth\arulename{d}} ~|~ (d,\arulename)
\in \nats \times R\}\]

If $\sredp\aARS$ is confluent then $\sredp\sunf$ is confluent because it holds:
$\sredp\sunf \subseteq \sredp\aARS \subseteq \smredp\sunf$ or equivalently
$\smredp\aARS = \smredp\sunf$
(see also \cite[Lemma\hspace*{2pt}2.2.5]{oost:1994}).
  
%

We use parallel steps because the preceding attempt to prove confluence of
$\llunfCRS$-steps themselves by decreasing diagrams was unsuccessful. As a
precedent order we considered an ordering on the rules and lexicographic extensions
of such orderings with the \sletrec-depth of the contracted redex. We came to
the conclusion that no such order could ensure decreasingness of the elementary
diagrams of both the critical pairs as well as the strictly nested redexes.
This was due to redex duplication induced by the diverging steps, so that
joining the diagram required a multi-step that disrupted decreasingness. In
order to resolve this problem we considered parallel steps such that the
problematic multi-step would become a single parallel step. This led to more
intricate diagrams but turned out to be a viable solution.

We will show that two diverging parallel steps in $\llunfCRS$ can be joined in an
elementary diagram of the following form with $d \leq e$.

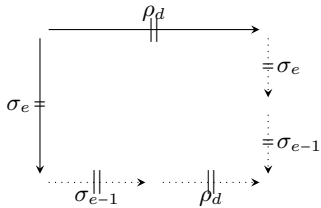
\begin{figure}[h]
\begin{tikzpicture}[>=stealth]
\matrix[row sep=0.8cm,column sep=1.3cm]{
\node(tl){};&&
\node(tr){};\\
&&\node(mr){};\\
\node(bl){};&
\node(bm){};&
\node(br){};\\
};
\draw[->       ](tl) to node{$||$} node[above]{$\wDepth\arulename{d}$}   (tr);
\draw[->       ](tl) to node{$=$ } node[left ]{$\wDepth\brulename{e}$}   (bl);
\draw[->,dotted](bl) to node{$||$} node[below]{$\wDepth\brulename{e-1}$} (bm);
\draw[->,dotted](bm) to node{$||$} node[below]{$\wDepth\arulename{d}$}   (br);
\draw[->,dotted](tr) to node{$= $} node[right]{$\wDepth\brulename{e}$}   (mr);
\draw[->,dotted](mr) to node{$= $} node[right]{$\wDepth\brulename{e-1}$} (br);
\end{tikzpicture}
\caption{\label{elem_dia}Elementary diagram}
\end{figure}

If we pick as the precedence order on the steps the order that is induced
by their \sletrec-depth, the diagram is decreasing. Note that in all the
diagrams we implicitly assume the reflexive closure for all arrows. The rest of
the proof is structured as follows. To justify the diagram we distinguish the
cases $d = e$ and $d < e$, for which we construct diagrams that are instances
of the diagram in Figure~\ref{elem_dia}.

\partitle{Case 1}
For $d = e$ we need to consider parallel diverging steps contracting redexes at
the same \sletrec-depth $d$. We construct the diagram below which is an
instance of the diagram above where the diverging parallel steps are in
sequentialised form.  
We write terms as fillings of a multihole context $\sacxt$ with all its holes
at $\sletrec$-depth $d$ such that the contracted $\wDepth\arulename{d}$- and
$\wDepth\brulename{d}$-redexes are filled into these holes. In this way we can make
explicit at which position a step takes place, i.e. at the root of the context
hole fillings. The topmost row and the leftmost column are respective
sequentialisations of the parallel diverging $\wDepth\arulename{d}$- and
$\wDepth\brulename{d}$-steps into single steps.

\begin{tikzpicture}[>=stealth]
\hspace{-2cm}
\matrix[row sep=1.2cm,column sep=0.5cm]{
\node(00){$\acxt{\allteri0^\TL, \dots, \allteri{n}^\TL}$};&
\node(10){$\acxt{\allteri0^\TR, \allteri1^\TL \dots, \allteri{n}^\TL}$};&
\node(20){$\acxt{\allteri0^\TR, \allteri1^\TR, \allteri2^\TL \dots, \allteri{n}^\TL}$};&
\node(i0){$\dots$};&
\node(n0){$\acxt{\allteri0^\TR, \dots, \allteri{n}^\TR}$};\\
\node(01){$\acxt{\allteri0^\BL, \allteri1^\TL, \dots, \allteri{n}^\TL}$};&
\node(11){$\acxt{\allteri0^\BR, \allteri1^\TL, \dots, \allteri{n}^\TL}$};&
\node(21){$\acxt{\allteri0^\BR, \allteri1^\TR, \allteri2^\TL, \dots, \allteri{n}^\TL}$};&
\node(i1){$\dots$};&
\node(n1){$\acxt{\allteri0^\BR, \allteri1^\TR, \dots, \allteri{n}^\TR}$};\\
\node(02){$\acxt{\allteri0^\BL, \allteri1^\BL, \allteri2^\TL, \dots, \allteri{n}^\TL}$};&
\node(12){$\acxt{\allteri0^\BR, \allteri1^\BL, \allteri2^\TL, \dots, \allteri{n}^\TL}$};&
\node(22){$\acxt{\allteri0^\BR, \allteri1^\BR, \allteri2^\TL, \dots, \allteri{n}^\TL}$};&
\node(i2){$\dots$};&
\node(n2){$\acxt{\allteri0^\BR, \allteri1^\BR, \allteri2^\TR, \dots, \allteri{n}^\TR}$};\\
\node(0i){$\vdots$};&
\node(1i){$\vdots$};&
\node(2i){$\vdots$};&
\node(ii){$\ddots$};&
\node(ni){$\vdots$};\\
\node(0n){$\acxt{\allteri0^\BL, \dots, \allteri{n}^\BL}$};&
\node(1n){$\acxt{\allteri0^\BR, \allteri1^\BL, \dots, \allteri{n}^\BL}$};&
\node(2n){$\acxt{\allteri0^\BR, \allteri1^\BR, \allteri2^\BL, \dots, \allteri{n}^\BL}$};&
\node(in){$\dots$};&
\node(nn){$\acxt{\allteri0^\BR, \dots, \allteri{n}^\BR}$};\\
};
\draw[->       ](00)  to            node[above]{$\wDepth\arulename{d}$}    (10);
\draw[->       ](10)  to            node[above]{$\wDepth\arulename{d}$}    (20);
\draw[->       ](20)  to            node[above]{$\wDepth\arulename{d}$}    (i0);
\draw[->       ](i0)  to            node[above]{$\wDepth\arulename{d}$}    (n0);
\draw[->,dotted](01)  to node{$||$} node[above]{$\wDepth\arulename{d}$}    (11);
\draw[->,dotted](11)  to            node[above]{$\wDepth\arulename{d}$}    (21);
\draw[->,dotted](21)  to            node[above]{$\wDepth\arulename{d}$}    (i1);
\draw[->,dotted](i1)  to            node[above]{$\wDepth\arulename{d}$}    (n1);
\draw[->,dotted](02)  to node{$||$} node[above]{$\wDepth\arulename{d}$}    (12);
\draw[->,dotted](12)  to node{$||$} node[above]{$\wDepth\arulename{d}$}    (22);
\draw[->,dotted](22)  to            node[above]{$\wDepth\arulename{d}$}    (i2);
\draw[->,dotted](i2)  to            node[above]{$\wDepth\arulename{d}$}    (n2);
\draw[->,dotted](0n)  to node{$||$} node[above]{$\wDepth\arulename{d}$}    (1n);
\draw[->,dotted](1n)  to node{$||$} node[above]{$\wDepth\arulename{d}$}    (2n);
\draw[->,dotted](2n)  to node{$||$} node[above]{$\wDepth\arulename{d}$}    (in);
\draw[->,dotted](in)  to node{$||$} node[above]{$\wDepth\arulename{d}$}    (nn);

\draw[->       ](00)  to            node[left ]{$\wDepth\brulename{d}$}    (01);
\draw[->,dotted](10)  to node{$=$}  node[left ]{$\wDepth\brulename{d}$}    (11);
\draw[->,dotted](20)  to node{$=$}  node[left ]{$\wDepth\brulename{d}$}    (21);
\draw[->,dotted](n0)  to node{$=$}  node[left ]{$\wDepth\brulename{d}$}    (n1);
\draw[->       ](01)  to            node[left ]{$\wDepth\brulename{d}$}    (02);
\draw[->,dotted](11)  to            node[left ]{$\wDepth\brulename{d}$}    (12);
\draw[->,dotted](21)  to node{$=$}  node[left ]{$\wDepth\brulename{d}$}    (22);
\draw[->,dotted](n1)  to node{$=$}  node[left ]{$\wDepth\brulename{d}$}    (n2);
\draw[->       ](02)  to            node[left ]{$\wDepth\brulename{d}$}    (0i);
\draw[->,dotted](12)  to            node[left ]{$\wDepth\brulename{d}$}    (1i);
\draw[->,dotted](22)  to            node[left ]{$\wDepth\brulename{d}$}    (2i);
\draw[->,dotted](n2)  to node{$=$}  node[left ]{$\wDepth\brulename{d}$}    (ni);
\draw[->       ](0i)  to            node[left ]{$\wDepth\brulename{d}$}    (0n);
\draw[->,dotted](1i)  to            node[left ]{$\wDepth\brulename{d}$}    (1n);
\draw[->,dotted](2i)  to            node[left ]{$\wDepth\brulename{d}$}    (2n);
\draw[->,dotted](ni)  to node{$=$}  node[left ]{$\wDepth\brulename{d}$}    (nn);
\end{tikzpicture}

Only the tiles on the diagonal require closer attention because for all other
tiles the vertical and horizontal steps take place in different holes of the
context, therefore they are disjoint and consequently commute.
In the tiles on the diagonal the diverging steps may be either due to a critical
pair or to identical steps. In the latter case the diagram is easily joined.
In case of a critical pair, since all steps take place at the same
$\sletrec$\nb-depth any such critical pair must arise from a root overlap.
Exhaustive scrutiny of all these critical pairs reveals that they
can be joined in a way that conforms to the tiles on the diagonal. Note that
the \sletrec-depths of the steps have to be increased by $d$ according to the
lifting into a context with its hole at \sletrec-depth $d$.

\newcommand\confDiaMatrix{\matrix[row sep=0.9cm,column sep=1.2cm]}

\begin{tikzpicture}[>=stealth]
\confDiaMatrix{
\node(tl){$\letrec{}{\labs\avar\allter}$};&
\node(tr){$\labs\avar{\letrec{}\allter}$};\\
\node(bl){$\labs\avar\allter$};&
\node(br){$\labs\avar\allter$};\\
};
\draw[->       ](tl) to node[above]{$\wDepth\sunflabs{0}$} (tr);
\draw[->       ](tl) to node[left ]{$\wDepth\sunfnil{0}$}  (bl);
\draw[->,dotted](tr) to node[right]{$\wDepth\sunfnil{0}$}  (br);
\draw[double   ](bl) to                                    (br);
\end{tikzpicture}
\begin{tikzpicture}[>=stealth]
\confDiaMatrix{
\node(tl){$\letrec\abindgroup{\labs\avar\allter}$};&
\node(tr){$\labs\avar{\letrec\abindgroup\allter}$};\\
\node(bl){$\letrec{\abindgroup'}{\labs\avar\allter}$};&
\node(br){$\labs\avar{\letrec{\abindgroup'}\allter}$};\\
};
\draw[->       ](tl) to node[above]{$\wDepth\sunflabs{0}$}   (tr);
\draw[->       ](tl) to node[left ]{$\wDepth\sunfreduce{0}$} (bl);
\draw[->,dotted](tr) to node[right]{$\wDepth\sunfreduce{0}$} (br);
\draw[->,dotted](bl) to node[below]{$\wDepth\sunflabs{0}$}   (br);
\end{tikzpicture}

\begin{tikzpicture}[>=stealth]
\confDiaMatrix{
\node(tl){$\letrec{}{\lapp\allter\bllter}$};&
\node(tr){$\lapp{(\letrec{}\allter)}{(\letrec{}\bllter)}$};\\
\node(bl){$\lapp\allter\bllter$};&
\node(br){$\lapp\allter\bllter$};\\
};
\draw[->       ](tl) to           node[above]{$\wDepth\sunflapp{0}$} (tr);
\draw[->       ](tl) to           node[left ]{$\wDepth\sunfnil{0}$}  (bl);
\draw[->,dotted](tr) to node{$=$} node[right]{$\wDepth\sunfnil{0}$}  (br);
\draw[double   ](bl) to                                              (br);
\end{tikzpicture}
\begin{tikzpicture}[>=stealth]
\confDiaMatrix{
\node(tl){$\letrec\abindgroup{\lapp\allter\bllter}$};&
\node(tr){$\lappbreak{\letrec\abindgroup\allter}{\letrec\abindgroup\bllter}$};\\
\node(bl){$\letrec{\abindgroup'}{\lapp\allter\bllter}$};&
\node(br){$\lappbreak{\letrec{\abindgroup'}\allter}{\letrec{\abindgroup'}\bllter}$};\\
};
\draw[->       ](tl) to           node[above]{$\wDepth\sunflapp{0}$}   (tr);
\draw[->       ](tl) to           node[left ]{$\wDepth\sunfreduce{0}$} (bl);
\draw[->,dotted](tr) to node{$=$} node[right]{$\wDepth\sunfreduce{0}$} (br);
\draw[->,dotted](bl) to           node[below]{$\wDepth\sunflapp{0}$}   (br);
\end{tikzpicture}

\begin{tikzpicture}[>=stealth]
\confDiaMatrix{
\node(tl){$\letrec\abindgroup{\arecvari{i}}$};&
\node(tr){$\letrec\abindgroup{\allteri{i}}$};\\
\node(bl){$\letrec{\abindgroup'}{\arecvari{i}}$};&
\node(br){$\letrec{\abindgroup'}{\allteri{i}}$};\\
};
\draw[->       ](tl) to            node[above]{$\wDepth\sunfrec{0}$}    (tr);
\draw[->       ](tl) to            node[left ]{$\wDepth\sunfreduce{0}$} (bl);
\draw[->,dotted](tr) to            node[right]{$\wDepth\sunfreduce{0}$} (br);
\draw[->,dotted](bl) to            node[below]{$\wDepth\sunfrec{0}$}    (br);
\end{tikzpicture}
\begin{tikzpicture}[>=stealth]
\confDiaMatrix{
\node(tl){$\letrec{}{\letrec\abindgroup\allter}$};&
\node(tr){$\letrec\abindgroup\allter$};\\
\node(bl){$\letrec\abindgroup\allter$};&
\node(br){$\letrec\abindgroup\allter$};\\
};
\draw[->    ](tl) to node[above]{$\wDepth\sunfletrec{0}$} (tr);
\draw[->    ](tl) to node[left ]{$\wDepth\sunfnil{0}$}    (bl);
\draw[double](tr) to                                      (br);
\draw[double](bl) to                                      (br);
\end{tikzpicture}

\begin{tikzpicture}[>=stealth]
\confDiaMatrix{
\node(tl){$\letrec\abindgroup{\letrec\bbindgroup\allter}$};&
\node(tr){$\letrec{\abindgroup~\bbindgroup}\allter$};\\
\node(bl){$\letrec{\abindgroup'}{\letrec\bbindgroup\allter}$};&
\node(br){$\letrec{\abindgroup'~\bbindgroup}\allter$};\\
};
\draw[->       ](tl) to node[above]{$\wDepth\sunfletrec{0}$}   (tr);
\draw[->       ](tl) to node[left ]{$\wDepth\sunfreduce{0}$} (bl);
\draw[->,dotted](tr) to node[right]{$\wDepth\sunfreduce{0}$}   (br);
\draw[->,dotted](bl) to node[below]{$\wDepth\sunfletrec{0}$}   (br);
\end{tikzpicture}
\begin{tikzpicture}[>=stealth]
\confDiaMatrix{
\node(tl){$\letrec\abindgroup\allter$};&
\node(tr){$\letrec{\abindgroup^\TR}\allter$};\\
\node(bl){$\letrec{\abindgroup^\BL}{\allter}$};&
\node(br){$\letrec{\abindgroup^\TR}\allter$};\\
};
\draw[->       ](tl) to node[above]{$\wDepth\sunfreduce{0}$}   (tr);
\draw[->       ](tl) to node[left ]{$\wDepth\sunfreduce{0}$} (bl);
\draw[->,dotted](tr) to node[right]{$\wDepth\sunfreduce{0}$} (br);
\draw[->,dotted](bl) to node[below]{$\wDepth\sunfreduce{0}$}   (br);
\end{tikzpicture}

\partitle{Case 2}
For $d < e$ we use the same approach as for $d = e$, the diagram is however
more involved. Again, we use a context $\sacxt$ with 
context holes at \sletrec-depth $d$. But since $e > d$, 
more than one $\wDepth\brulename{e}$-contraction may take place in one such hole.
Therefore a per-hole partitioning of the vertical steps requires a sequence of
parallel steps.

The diagram below fits the scheme of the elementary diagram (Figure~\ref{elem_dia})
when interleaving the $\wDepth\brulename{e}$-steps with the
$\wDepth\brulename{e-1}$-steps in the rightmost column such that steps at depth $e$
preceed those at depth $e-1$. Similarly for the bottommost row where the
$\wDepth\arulename{e-1}$-steps have to preceed the $\wDepth\brulename{d}$-steps.
These reorderings are possible since the segments
represent contractions within different holes of $\sacxt$. As in the previous
diagram the tiles which do not lie on the diagonal are unproblematic, which
leaves us to complete the proof by constructing the tiles on the diagonal.

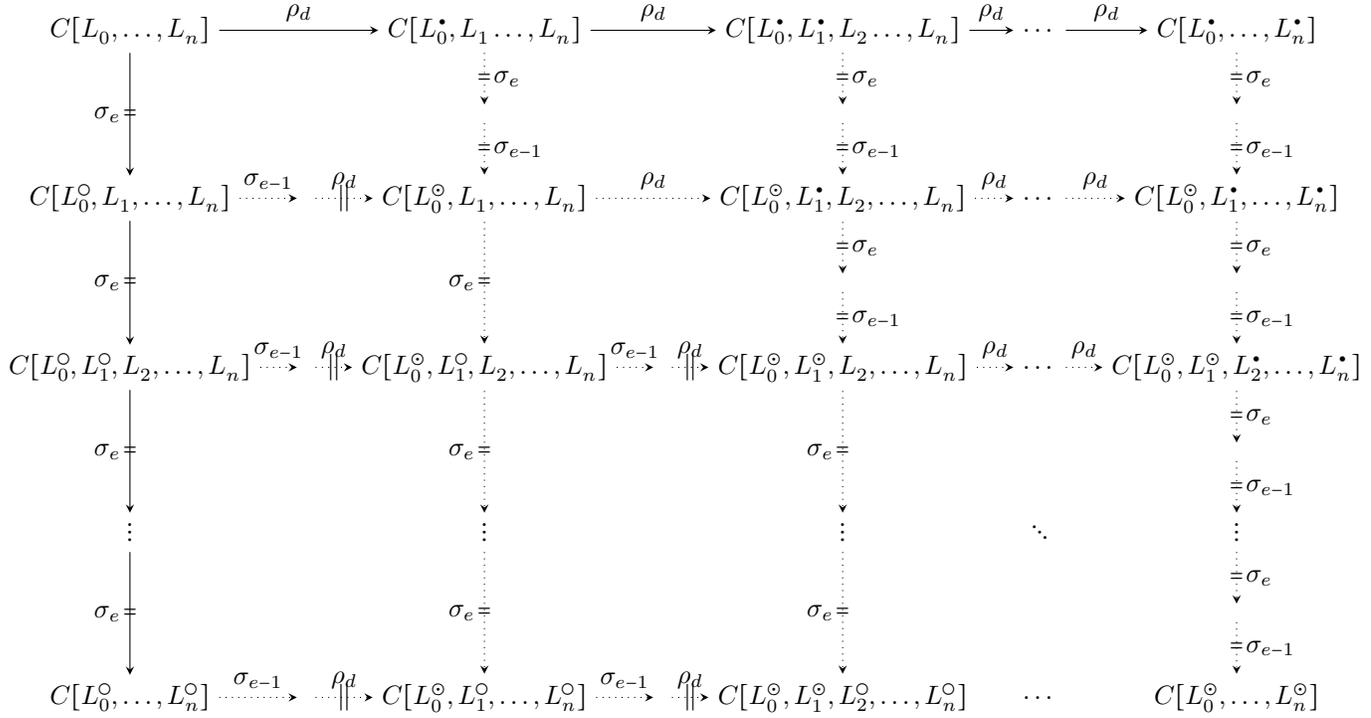
\begin{tikzpicture}[>=stealth]
\hspace{-3cm}
\matrix[row sep=0.7cm,column sep=0.5cm]{
\node(00){$\acxt{\allteri0^\TL, \dots, \allteri{n}^\TL}$};&
&
\node(10){$\acxt{\allteri0^\TR, \allteri1^\TL \dots, \allteri{n}^\TL}$};&
&
\node(20){$\acxt{\allteri0^\TR, \allteri1^\TR, \allteri2^\TL \dots, \allteri{n}^\TL}$};&
\node(i0){$\dots$};&
\node(n0){$\acxt{\allteri0^\TR, \dots, \allteri{n}^\TR}$};\\
\node(00h){};&&
\node(10h){};&&
\node(20h){};&
\node(i0h){};&
\node(n0h){};\\
\node(01){$\acxt{\allteri0^\BL, \allteri1^\TL, \dots, \allteri{n}^\TL}$};&
\node(0h1){};&
\node(11){$\acxt{\allteri0^\BR, \allteri1^\TL, \dots, \allteri{n}^\TL}$};&
&
\node(21){$\acxt{\allteri0^\BR, \allteri1^\TR, \allteri2^\TL, \dots, \allteri{n}^\TL}$};&
\node(i1){$\dots$};&
\node(n1){$\acxt{\allteri0^\BR, \allteri1^\TR, \dots, \allteri{n}^\TR}$};\\
\node(01h){};&&
\node(11h){};&&
\node(21h){};&
\node(i1h){};&
\node(n1h){};\\
\node(02){$\acxt{\allteri0^\BL, \allteri1^\BL, \allteri2^\TL, \dots, \allteri{n}^\TL}$};&
\node(0h2){};&
\node(12){$\acxt{\allteri0^\BR, \allteri1^\BL, \allteri2^\TL, \dots, \allteri{n}^\TL}$};&
\node(1h2){};&
\node(22){$\acxt{\allteri0^\BR, \allteri1^\BR, \allteri2^\TL, \dots, \allteri{n}^\TL}$};&
\node(i2){$\dots$};&
\node(n2){$\acxt{\allteri0^\BR, \allteri1^\BR, \allteri2^\TR, \dots, \allteri{n}^\TR}$};\\
\node(02h){};&&
\node(12h){};&&
\node(22h){};&
\node(i2h){};&
\node(n2h){};\\
\node(0i){$\vdots$};&&
\node(1i){$\vdots$};&&
\node(2i){$\vdots$};&
\node(ii){$\ddots$};&
\node(ni){$\vdots$};\\
\node(0ih){};&&
\node(1ih){};&&
\node(2ih){};&
\node(iih){};&
\node(nih){};\\
\node(0n){$\acxt{\allteri0^\BL, \dots, \allteri{n}^\BL}$};&
\node(0hn){};&
\node(1n){$\acxt{\allteri0^\BR, \allteri1^\BL, \dots, \allteri{n}^\BL}$};&
\node(1hn){};&
\node(2n){$\acxt{\allteri0^\BR, \allteri1^\BR, \allteri2^\BL, \dots, \allteri{n}^\BL}$};&
\node(in){$\dots$};&
\node(nn){$\acxt{\allteri0^\BR, \dots, \allteri{n}^\BR}$};\\
};
\draw[->       ](00)  to            node[above]{$\wDepth\arulename{d}$}    (10);
\draw[->       ](10)  to            node[above]{$\wDepth\arulename{d}$}    (20);
\draw[->       ](20)  to            node[above]{$\wDepth\arulename{d}$}    (i0);
\draw[->       ](i0)  to            node[above]{$\wDepth\arulename{d}$}    (n0);
\draw[->,dotted](01)  to            node[above]{$\wDepth\brulename{e-1}$}    (0h1);
\draw[->,dotted](0h1) to node{$||$} node[above]{$\wDepth\arulename{d}$}    (11);
\draw[->,dotted](11)  to            node[above]{$\wDepth\arulename{d}$}    (21);
\draw[->,dotted](21)  to            node[above]{$\wDepth\arulename{d}$}    (i1);
\draw[->,dotted](i1)  to            node[above]{$\wDepth\arulename{d}$}    (n1);
\draw[->,dotted](02)  to            node[above]{$\wDepth\brulename{e-1}$}    (0h2);
\draw[->,dotted](0h2) to node{$||$} node[above]{$\wDepth\arulename{d}$}    (12);
\draw[->,dotted](12)  to            node[above]{$\wDepth\brulename{e-1}$}    (1h2);
\draw[->,dotted](1h2) to node{$||$} node[above]{$\wDepth\arulename{d}$}    (22);
\draw[->,dotted](22)  to            node[above]{$\wDepth\arulename{d}$}    (i2);
\draw[->,dotted](i2)  to            node[above]{$\wDepth\arulename{d}$}    (n2);
\draw[->,dotted](0n)  to            node[above]{$\wDepth\brulename{e-1}$}    (0hn);
\draw[->,dotted](0hn) to node{$||$} node[above]{$\wDepth\arulename{d}$}    (1n);
\draw[->,dotted](1n)  to            node[above]{$\wDepth\brulename{e-1}$}    (1hn);
\draw[->,dotted](1hn) to node{$||$} node[above]{$\wDepth\arulename{d}$}    (2n);

\draw[->       ](00)  to node{$=$}  node[left ]{$\wDepth\brulename{e}$}    (01);
\draw[->,dotted](10)  to node{$=$}  node[right]{$\wDepth\brulename{e}$}    (10h);
\draw[->,dotted](10h) to node{$=$}  node[right]{$\wDepth\brulename{e-1}$}  (11);
\draw[->,dotted](20)  to node{$=$}  node[right]{$\wDepth\brulename{e}$}    (20h);
\draw[->,dotted](20h) to node{$=$}  node[right]{$\wDepth\brulename{e-1}$}  (21);
\draw[->,dotted](n0)  to node{$=$}  node[right]{$\wDepth\brulename{e}$}    (n0h);
\draw[->,dotted](n0h) to node{$=$}  node[right]{$\wDepth\brulename{e-1}$}  (n1);

\draw[->       ](01)  to node{$=$}  node[left ]{$\wDepth\brulename{e}$}    (02);
\draw[->,dotted](11)  to node{$=$}  node[left ]{$\wDepth\brulename{e}$}    (12);
\draw[->,dotted](21)  to node{$=$}  node[right]{$\wDepth\brulename{e}$}    (21h);
\draw[->,dotted](21h) to node{$=$}  node[right]{$\wDepth\brulename{e-1}$}  (22);
\draw[->,dotted](n1)  to node{$=$}  node[right]{$\wDepth\brulename{e}$}    (n1h);
\draw[->,dotted](n1h) to node{$=$}  node[right]{$\wDepth\brulename{e-1}$}  (n2);

\draw[->       ](02)  to node{$=$}  node[left ]{$\wDepth\brulename{e}$}    (0i);
\draw[->,dotted](12)  to node{$=$}  node[left ]{$\wDepth\brulename{e}$}    (1i);
\draw[->,dotted](22)  to node{$=$}  node[left ]{$\wDepth\brulename{e}$}    (2i);
\draw[->,dotted](n2)  to node{$=$}  node[right]{$\wDepth\brulename{e}$}    (n2h);
\draw[->,dotted](n2h) to node{$=$}  node[right]{$\wDepth\brulename{e-1}$}  (ni);

\draw[->       ](0i)  to node{$=$}  node[left ]{$\wDepth\brulename{e}$}    (0n);
\draw[->,dotted](1i)  to node{$=$}  node[left ]{$\wDepth\brulename{e}$}    (1n);
\draw[->,dotted](2i)  to node{$=$}  node[left ]{$\wDepth\brulename{e}$}    (2n);
\draw[->,dotted](ni)  to node{$=$}  node[right]{$\wDepth\brulename{e}$}    (nih);
\draw[->,dotted](nih) to node{$=$}  node[right]{$\wDepth\brulename{e-1}$}  (nn);
\end{tikzpicture}

\vspace{2ex}
Every hole on the diagonal is filled with at most one $\wDepth\arulename{d}$-redex
(at the root of the context hole fillings) but because of $d < e$ with possibly
many $\wDepth\brulename{e}$-redexes (properly inside of the fillings). There may or
may not be an overlap between the $\wDepth\arulename{d}$-step and a
$\wDepth\brulename{e}$-step, but there can be at most one, which is due to the
rules of $\llunfCRS$.

Therefore $\wDepth\brulename{e}$ contracts either an overlap and a number of nested
redexes, or only nested redexes without an overlap. These constellations are
depicted on the figure below. There is one $\wDepth\arulename{d}$-redex and three
$\wDepth\brulename{e}$-redexes. On the left, one of the $\wDepth\brulename{e}$-redexes
overlaps with the $\wDepth\arulename{d}$-redex while on the right all
$\wDepth\brulename{e}$-redexes are strictly nested inside the
$\wDepth\arulename{d}$-redex.

\vspace{3ex}
\centered{\fig{overlap}}
\vspace{3ex}

For the critical pairs due to a non-root overlap, and for all situations with
nested redexes, we construct diagrams of the following shape, respectively:\\
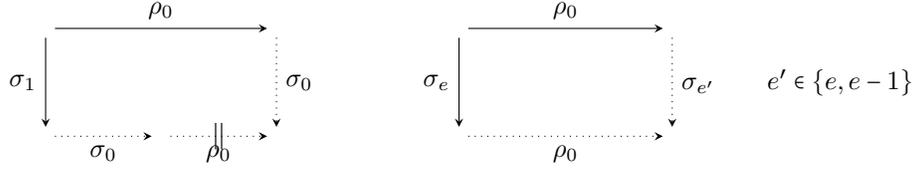
\begin{tikzpicture}[>=stealth]
\matrix[row sep=1.2cm,column sep=1.3cm]{
\node(tl){};&&
\node(tr){};\\
\node(bl){};&
\node(bm){};&
\node(br){};\\
\\
};
\draw[->       ](tl) to            node[above]{$\wDepth\arulename{0}$}   (tr);
\draw[->       ](tl) to            node[left ]{$\wDepth\brulename{1}$}   (bl);
\draw[->,dotted](bl) to            node[below]{$\wDepth\brulename{0}$} (bm);
\draw[->,dotted](bm) to node{$||$} node[below]{$\wDepth\arulename{0}$}   (br);
\draw[->,dotted](tr) to            node[right]{$\wDepth\brulename{0}$} (br);
\end{tikzpicture}
\hspace{1cm}
\begin{tikzpicture}[>=stealth]
\matrix[row sep=1.2cm,column sep=2.6cm]{
\node(tl){};&
\node(tr){};\\
\node(bl){};&
\node(br){};\\
\\
};
\draw[->       ](tl) to            node[above]{$\wDepth\arulename{0}$}   (tr);
\draw[->       ](tl) to            node[left ]{$\wDepth\brulename{e}$}   (bl);
\draw[->,dotted](bl) to            node[below]{$\wDepth\arulename{0}$}   (br);
\draw[->,dotted](tr) to            node[right]{$\wDepth\brulename{e'} ~~~~~~ e' \in \{e,e-1\}$} (br);
\end{tikzpicture}

\vspace{-3ex}
When lifted into a context of \sletrec-depth $d$ both of the diagrams comply to
the shape necessary for the diagonal tiles, but we need to be able to handle
situations as on as on the left of the above figure, where both nested
redexes as well as the overlapping redex are contracted. Firstly, since all
$\brulename$-redexes occur at the same \sletrec-depth, it must hold that $d = 0$
and $e = 1$, which is due to the rules of $\llunfCRS$. Secondly, none of the
involved redex contractions affect any of the nested redexes except for
duplicating or erasing them, which means that the residuals of the
$\brulename$\nb-steps after these steps are part of a a parallel
$\wDepth\brulename{e'}$-step (mind that we assume the reflexive closure of all
steps). Or in a diagram:\\
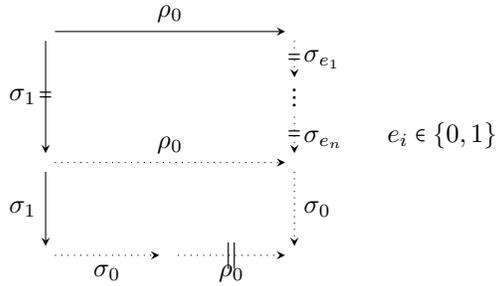
\begin{tikzpicture}[>=stealth]
\matrix[row sep=0.5cm,column sep=1.4cm]{
\node(tl){};&&
\node(tr){};\\
&&
\node(ir){$\vdots$};\\
\node(ml){};&&
\node(mr){};\\
\\
\node(bl){};&
\node(bm){};&
\node(br){};\\
};
\draw[->       ](tl) to            node[above]{$\wDepth\arulename{0}$}   (tr);
\draw[->       ](tl) to node{$= $} node[left ]{$\wDepth\brulename{1}$}   (ml);
\draw[->       ](ml) to            node[left ]{$\wDepth\brulename{1}$}   (bl);
\draw[->,dotted](ml) to            node[above]{$\wDepth\arulename{0}$}   (mr);
\draw[->,dotted](mr) to            node[right]{$\wDepth\brulename{0}$}   (br);
\draw[->,dotted](bl) to            node[below]{$\wDepth\brulename{0}$}   (bm);
\draw[->,dotted](bm) to node{$||$} node[below]{$\wDepth\arulename{0}$}   (br);
\draw[->,dotted](tr) to node{$= $} node[right]{$\wDepth\brulename{e_1}$}   (ir);
\draw[->,dotted](ir) to node{$= $} node[right]{$\wDepth\brulename{e_n} ~~~~~ e_i \in \{0,1\}$} (mr);
\end{tikzpicture}\\
The diagram is composed from the previous two diagrams. A parallel version of
the right one constitutes the top part, while the bottom part is an exact
replica of the left one. The top part settles the portion arising from the
nested redexes, the bottom part settles the portion arising from the overlapping
redex.

At last in order to fit that diagram into the scheme of the diagonal tiles the
steps on the right have to be reordered such that $\wDepth\brulename{e_i}$-steps
with $e_i = 1$ preceed $\wDepth\brulename{e_i}$-steps with $e_i = 0$. The
reordering is viable because every $\wDepth\brulename{e_i}$-step takes place in its
own residual of the $\wDepth\brulename1$-step from the left.

We conclude the proof by a comprehensive analysis all critical pairs that arise
from non-root overlaps in $\llunfCRS$ as well as the diagrams for joining
nested redexes.

\partitle{Diagrams for joining critical pairs}\\
\begin{tikzpicture}[>=stealth]
\confDiaMatrix{
\node(tl){$\letrec\abindgroup{\letrec\bbindgroup{\labs\avar{\allter}}}$};&
\node(tr){$\letrec{\abindgroup~\bbindgroup}{\labs\avar{\allter}}$};\\
\node(ml){$\letrec\abindgroup{\labs\avar{\letrec\bbindgroup\allter}}$};\\
\node(bl){$\labs\avar{\letrec\abindgroup{\letrec\bbindgroup\allter}}$};&
\node(br){$\labs\avar{\letrec{\abindgroup~\bbindgroup}\allter}$};\\
};
\draw[->       ](tl) to node[above]{$\wDepth\sunfletrec{0}$} (tr);
\draw[->       ](tl) to node[left ]{$\wDepth\sunflabs{1}$} (ml);
\draw[->,dotted](tr) to node[right]{$\wDepth\sunflabs{0}$}   (br);
\draw[->,dotted](ml) to node[left ]{$\wDepth\sunflabs{0}$}   (bl);
\draw[->,dotted](bl) to node[below]{$\wDepth\sunfletrec{0}$} (br);
\end{tikzpicture}

\begin{tikzpicture}[>=stealth]
\confDiaMatrix{
\node(tl){$\letrec\abindgroup{\letrec\bbindgroup{\lapp\allter\bllter}}$};&
\node(tr){$\letrec{\abindgroup~\bbindgroup}{\lapp\allter\bllter}$};\\
\node(ml){$\letrec\abindgroup{\lapp{(\letrec\bbindgroup\allter)}{(\letrec\bbindgroup\bllter)}}$};\\
\node(bl){$\lappbreak{\letrec\abindgroup{\letrec\bbindgroup\allter}}{\letrec\abindgroup{\letrec\bbindgroup\bllter}}$};&
\node(br){$\lappbreak{\letrec{\abindgroup~\bbindgroup}\allter}{\letrec{\abindgroup~\bbindgroup}\bllter}$};\\
};
\draw[->       ](tl) to            node[above]{$\wDepth\sunfletrec{0}$} (tr);
\draw[->       ](tl) to            node[left ]{$\wDepth\sunflapp{1}$} (ml);
\draw[->,dotted](tr) to            node[right]{$\wDepth\sunflapp{0}$}   (br);
\draw[->,dotted](ml) to            node[left ]{$\wDepth\sunflapp{0}$}   (bl);
\draw[->,dotted](bl) to node{$||$} node[below]{$\wDepth\sunfletrec{0}$} (br);
\end{tikzpicture}

\begin{tikzpicture}[>=stealth]
\confDiaMatrix{
\node(tl){$\letrec\abindgroup{\letrec\bbindgroup{\arecvari{i}}}$};&
\node(tr){$\letrec{\abindgroup~\bbindgroup}{\arecvari{i}}$};\\
\node(bl){$\letrec\abindgroup{\letrec\bbindgroup{\allteri{i}}}$};&
\node(br){$\letrec{\abindgroup~\bbindgroup}{\allteri{i}}$};\\
};
\draw[->       ](tl) to node[above]{$\wDepth\sunfletrec{0}$} (tr);
\draw[->       ](tl) to node[left ]{$\wDepth\sunfrec{1}$}    (bl);
\draw[->,dotted](tr) to node[right]{$\wDepth\sunfrec{0}$}    (br);
\draw[->,dotted](bl) to node[below]{$\wDepth\sunfletrec{0}$} (br);
\end{tikzpicture}

\begin{tikzpicture}[>=stealth]
\confDiaMatrix{
\node(tl){$\letrec\abindgroup{\letrec\bbindgroup{\letrec\cbindgroup\allter}}$};&
\node(tr){$\letrec\abindgroup{\letrec{\bbindgroup~\cbindgroup}\allter}$};\\
\node(bl){$\letrec{\abindgroup~\bbindgroup}{\letrec\cbindgroup\allter}$};&
\node(br){$\letrec{\abindgroup~\bbindgroup~\cbindgroup}\allter$};\\
};
\draw[->       ](tl) to node[above]{$\wDepth\sunfletrec{0}$}   (tr);
\draw[->       ](tl) to node[left ]{$\wDepth\sunfletrec{1}$}   (bl);
\draw[->,dotted](tr) to node[right]{$\wDepth\sunfletrec{0}$}   (br);
\draw[->,dotted](bl) to node[below]{$\wDepth\sunfletrec{0}$}   (br);
\end{tikzpicture}

\begin{tikzpicture}[>=stealth]
\confDiaMatrix{
\node(tl){$\letrec\abindgroup{\letrec{}\allter}$};&
\node(tr){$\letrec\abindgroup\allter$};\\
\node(bl){$\letrec\abindgroup\allter$};&
\node(br){$\letrec\abindgroup\allter$};\\
};
\draw[->       ](tl) to node[above]{$\wDepth\sunfletrec{0}$} (tr);
\draw[->       ](tl) to node[left ]{$\wDepth\sunfnil{1}$}    (bl);
\draw[double   ](tr) to                                      (br);
\draw[double   ](bl) to                                      (br);
\end{tikzpicture}
\begin{tikzpicture}[>=stealth]
\confDiaMatrix{
\node(tl){$\letrec\abindgroup{\letrec\bbindgroup\allter}$};&
\node(tr){$\letrec{\abindgroup~\bbindgroup}\allter$};\\
\node(bl){$\letrec\abindgroup{\letrec{\bbindgroup'}\allter}$};&
\node(br){$\letrec{\abindgroup~\bbindgroup'}\allter$};\\
};
\draw[->       ](tl) to node[above]{$\wDepth\sunfletrec{0}$}   (tr);
\draw[->       ](tl) to node[left ]{$\wDepth\sunfreduce{1}$} (bl);
\draw[->,dotted](tr) to node[right]{$\wDepth\sunfreduce{0}$}   (br);
\draw[->,dotted](bl) to node[below]{$\wDepth\sunfletrec{0}$}   (br);
\end{tikzpicture}

\vspace{3ex}
\partitle{Diagrams for joining nested redexes}\\
\begin{tikzpicture}[>=stealth]
\confDiaMatrix{
\node(tl){$\letrec\abindgroup{\labs\avar\allter}$};&
\node(tr){$\labs\avar{\letrec\abindgroup\allter}$};\\
\node(bl){$\letrec{\abindgroup'}{\labs\avar{\allter'}}$};&
\node(br){$\labs\avar{\letrec{\abindgroup'}{\allter'}}$};\\
};
\draw[->       ](tl) to            node[above]{$\wDepth\sunflabs{0}$} (tr);
\draw[->       ](tl) to            node[left ]{$\wDepth\brulename{e}$}    (bl);
\draw[->,dotted](tr) to            node[right]{$\wDepth\brulename{e}$}    (br);
\draw[->,dotted](bl) to            node[below]{$\wDepth\sunflabs{0}$} (br);
\end{tikzpicture}
\begin{tikzpicture}[>=stealth]
\confDiaMatrix{
\node(tl){$\letrec\abindgroup{\lapp{\allteri0}{\allteri1}}$};&
\node(tr){$\lapp{(\letrec\abindgroup{\allteri0})}{(\letrec\abindgroup{\allteri1})}$};\\
\node(bl){$\letrec{\abindgroup'}{\lapp{\allteri0'}{\allteri1'}}$};&
\node(br){$\lapp{(\letrec{\abindgroup'}{\allteri0'})}{(\letrec{\abindgroup'}{\allteri1'})}$};\\
};
\draw[->       ](tl) to            node[above]{$\wDepth\sunflapp{0}$} (tr);
\draw[->       ](tl) to            node[left ]{$\wDepth\brulename{e}$}    (bl);
\draw[->,dotted](tr) to node{$= $} node[right]{$\wDepth\brulename{e}$}    (br);
\draw[->,dotted](bl) to            node[below]{$\wDepth\sunflapp{0}$} (br);
\end{tikzpicture}

\begin{tikzpicture}[>=stealth]
\confDiaMatrix{
\node(tl){$\letrec\abindgroup{\arecvar}$};&
\node(tr){$\letrec\abindgroup{\allter}$};\\&
\node(mr){$\letrec{\abindgroup'}{\allter}$};\\
\node(bl){$\letrec{\abindgroup'}{\arecvar}$};&
\node(br){$\letrec{\abindgroup'}{\allter'}$};\\
};
\draw[->       ](tl) to                          node[above]{$\wDepth\sunfrec{0}$} (tr);
\draw[->       ](tl) to                          node[left ]{$\wDepth\brulename{e}$} (bl);
\draw[->,dotted](tr) to                          node[right]{$\wDepth\brulename{e}$} (mr);
\draw[->,dotted](mr) to node[at end, right]{$=$} node[right]{$\wDepth\brulename{e}$} (br);
\draw[->,dotted](bl) to                          node[below]{$\wDepth\sunfrec{0}$} (br);
\end{tikzpicture}
\begin{tikzpicture}[>=stealth]
\confDiaMatrix{
\node(tl){$\letrec\abindgroup{\letrec\bbindgroup\allter}$};&
\node(tr){$\letrec{\abindgroup~\bbindgroup}\allter$};\\
\node(bl){$\letrec{\abindgroup'}{\letrec{\bbindgroup'}{\allter'}}$};&
\node(br){$\letrec{\abindgroup'~\bbindgroup'}{\allter'}$};\\
};
\draw[->       ](tl) to            node[above]{$\wDepth\sunfletrec{0}$} (tr);
\draw[->       ](tl) to            node[left ]{$\wDepth\brulename{e}$}      (bl);
\draw[->,dotted](tr) to            node[right]{$\wDepth\brulename{e' ~~~~~~ e' \in \{e-1,e\}}$}      (br);
\draw[->,dotted](bl) to            node[below]{$\wDepth\sunfletrec{0}$} (br);
\end{tikzpicture}

\begin{tikzpicture}[>=stealth]
\confDiaMatrix{
\node(tl){$\letrec{}\allter$};&
\node(tr){$\allter$};\\
\node(bl){$\letrec{}{\allter'}$};&
\node(br){$\allter'$};\\
};
\draw[->       ](tl) to            node[above]{$\wDepth\sunfnil{0}$} (tr);
\draw[->       ](tl) to            node[left ]{$\wDepth\brulename{e}$}   (bl);
\draw[->,dotted](tr) to            node[right]{$\wDepth\brulename{e-1}$}   (br);
\draw[->,dotted](bl) to            node[below]{$\wDepth\sunfnil{0}$} (br);
\end{tikzpicture}
\begin{tikzpicture}[>=stealth]
\confDiaMatrix{
\node(tl){$\letrec\abindgroup\allter$};&
\node(tr){$\letrec{\abindgroup^\TR}\allter$};\\
\node(bl){$\letrec{\abindgroup^\BL}{\allter'}$};&
\node(br){$\letrec{\abindgroup^\BR}{\allter'}$};\\
};
\draw[->       ](tl) to            node[above]{$\wDepth\sunfreduce{0}$} (tr);
\draw[->       ](tl) to            node[left ]{$\wDepth\brulename{e}$}      (bl);
\draw[->,dotted](tr) to node[at end,right]{$=$} node[right]{$\wDepth\brulename{e}$}      (br);
\draw[->,dotted](bl) to            node[below]{$\wDepth\sunfreduce{0}$} (br);
\end{tikzpicture}
\end{proof}


\end{document}